\documentclass[prx,twocolumn,nofootinbib, superscriptaddress]{revtex4-2}

\usepackage[utf8]{inputenc}
\usepackage{amsmath}
\usepackage{amsthm}
\usepackage{enumitem}
\usepackage{amsfonts}
\usepackage{amssymb}
\usepackage{mathrsfs}
\usepackage{mathtools}
\usepackage{url}
\usepackage{graphicx}
\usepackage{comment}
\usepackage[caption=false]{subfig}
\usepackage{ragged2e}
\usepackage[colorlinks]{hyperref}
\hypersetup{citecolor=blue,
urlcolor=blue}
\newtheorem{theorem}{Theorem}

\newtheorem{corollary}[theorem]{Corollary}

\newtheorem{definition}[theorem]{Definition}
\newtheorem{example}[theorem]{Example}

\newtheorem{lemma}[theorem]{Lemma}

\newtheorem{proposition}[theorem]{Proposition}
\newtheorem{remark}[theorem]{Remark}

\newcommand{\ket}[1]{\lvert #1 \rangle}
\newcommand{\bra}[1]{\langle #1 \rvert}
\usepackage{soul}
\usepackage{color}
\usepackage{tikz}
\usepackage{balance}

\definecolor{nblue}{rgb}{0.3,0.3,1.0}
\definecolor{ngreen}{rgb}{0.2,0.7,0.2}
\definecolor{nred}{rgb}{0.9,0.1,0}
\definecolor{red2}{rgb}{0.6,0.2,0.2}
\definecolor{npurple}{rgb}{0.8,0.2,0.8}
\definecolor{golden}{rgb}{0.8,0.6,0.1}
\definecolor{nsilver}{rgb}{0.3,0.4,0.5}
\definecolor{nbrown}{rgb}{0.8,0.4,0.15}
\definecolor{nrose}{rgb}{0.7,0,0.35}
\definecolor{nviol}{rgb}{0.5,0,1.0}
\definecolor{nazur}{rgb}{0,0.35,0.7}
\definecolor{nchart}{rgb}{0.2,0.4,0}
\definecolor{nbrick}{rgb}{0.55,0.25,0.15}
\def\spc{7pt}

\newcommand{\blk}{\color{black}}
\newcommand{\red}{\color{nred}}

\newcommand{\expectation}[1]{\langle #1 \rangle}

\newtheorem{innercustomthm}{Theorem}
\newenvironment{customthm}[1]  {\renewcommand\theinnercustomthm{#1}\innercustomthm}
  {\endinnercustomthm}
\def\urlprefix{}

\begin{document}
\date{\today}

\title{Properties and Applications of Partially Deterministic Polytopes}
\author{Marwan Haddara}
\affiliation{Quantum and Advanced Technologies Research Institute, Griffith University, Yugambeh Country, Gold Coast, Queensland 4222, Australia}

\affiliation{Centre for Quantum Computation and Communication Technology (Australian Research Council), Quantum and Advanced Technologies Research Institute, Griffith University, Yuggera Country, Brisbane, Queensland 4111, Australia}

\author{Howard M. Wiseman}
\affiliation{Centre for Quantum Computation and Communication Technology (Australian Research Council), Quantum and Advanced Technologies Research Institute, Griffith University, Yuggera Country, Brisbane, Queensland 4111, Australia}

\author{Eric G. Cavalcanti}
\affiliation{Quantum and Advanced Technologies Research Institute, Griffith University, Yugambeh Country, Gold Coast, Queensland 4222, Australia}

\begin{abstract}
    The assumption of a deterministic local hidden variable model constrains the experimentally accessible statistics in a Bell experiment to be contained in a certain convex polytope (the Bell-local polytope).  \blk But what if the outputs for only a subset of the measurements at each site are predetermined by the model? In this work, we embark on a thorough mathematical exploration of this concept of `partial determinism', allowing for arbitrary numbers of parties, inputs and outputs per site. 
    We find that the assumption of (local) partial determinism implies that the behaviours are contained in convex polytopes which always contain the Bell polytope and are contained in the no-signalling polytope, recovering each of them as special cases, but also leading to distinct structures. Nontrivial equivalence classes of partially deterministic models arise, which we classify completely. In particular, we find that the Bell polytope for any scenario can be expressed in multiple different ways in terms of local partially deterministic models on a given scenario. This allows us to both generalise and strengthen Fine's theorem, recovering the original formulation as a special case, but finding new constraints otherwise. While a `partially determinisic model' may be thought of as a weakening of the local hidden variable assumption, we also discuss scenarios with different physical motivations, which do not require the causal structure of the Bell scenario, and where classes of partially deterministic polytopes arise as constraints.  Our example applications include device-independent quantum state inseparability witnesses, classes of broadcast-local polytopes, and Local Friendliness scenarios in quantum foundations.  We also point out instances in previous literature where classes of related objects have been studied. In the case of correlations compatible with the Local Friendliness assumptions in particular, we find a one-to-one correspondence between partially deterministic polytopes and sequential extended Wigner's friend scenarios, showing that every partially deterministic polytope has physical relevance. We discuss how the framework of partially deterministic polytopes captures a broad class of non-classicality notions, and identify an even broader notion of `composable sets', of which partially deterministic polytopes are special cases. 
\end{abstract}

\maketitle
\tableofcontents
\section{Introduction}

The celebrated work of Bell \cite{Bell1964} laid the groundwork for the theoretical approach now referred to as the \emph{device-independent} \cite{Brunner2014, Scarani2012} or \emph{black box} \cite{Acin2017BlackBoxUnspeakablesBook} framework, where the main objects of interest are \emph{behaviours} which represent the collection of data obtainable in a given experiment. In particular, Bell showed \cite{Bell1964} that it is possible to experimentally address the seemingly metaphysical question of the existence of local hidden-variable models for quantum behaviours. He did this by showing that the assumption of local determinism implies certain inequalities in an experiment involving space-like separated measurements, and that quantum theory predicts the violation of Bell's inequalities. This profound conclusion has since been supported by experiments on various kinds of quantum systems \cite{Hensen2015, Rosenfeld2017, Stortz2023}, demonstrating that nature violates local determinism, and indeed, Bell's later and weaker notion of \textit{local causality}~\cite{Bell76,Wiseman2017,Cavalcanti2021}.

Considered as a resource, the property of Bell inequality violation is related to advantages in various information processing tasks including randomness expansion \cite{Ma2016}, quantum key distribution \cite{Pirandola2020, Zapatero2023} and reduction of communication complexity \cite{Buhrman2010}. Furthermore, in many cases the devices can be guaranteed to perform the task correctly even if the manufacturer is untrusted (the \emph{device-independent} approach), and sometimes even without the need to assume the validity of quantum mechanics itself (the \emph{theory-independent} approach). Significant research effort has been devoted to understanding both the theory and the applicability of device-independent information processing  \cite{Pironio_2016, Pirandola2020, Brunner2014, Scarani2012}. 

On the foundational side, the seminal works of Tsirelson \cite{Tsirelson1980, Tsirelson1993} and Popescu and Rohrlich \cite{Popescu1994} drew attention to the fact that the set of quantum behaviours is in fact strictly contained in the set of no-faster-than-light-signalling (no-signalling for short) behaviours. This opened the question of what singles out quantum theory from possible alternative theories with non-classical but no-signalling correlations. Various principles have been proposed \cite{Brunner2014, Popescu2014, Scarani2012} in an ongoing investigation to better understand the physical significance of the quantum-superquantum boundary. Simultaneously, new insight and  implications of quantum correlations has been gained both by considering alternative ideas about causality, or specific models of the physics at play in Bell-scenarios \cite{Svetlichny1987,  Bancal2012, Cavalcanti2012, Bancal2013, Barnea2013, Wiseman2017, Liang2020} and by extending the investigation to more exotic situations  \cite{Branciard2010, Fritz_2012, Bong2020, Utreras-Alarcon2024, Ying2024relatingwigners, Haddara2025} where other forms of classicality assumptions are well motivated. 

Recent works in quantum foundations, exploring the Local Friendliness no-go theorem \cite{Bong2020, Cavalcanti2021Foundphys, Cavalcanti2021, Wiseman2023thoughtfullocal, Utreras-Alarcon2024, Ying2024relatingwigners, Haddara2025} in particular, have accentuated a new class of behaviour sets of interest. These behaviour sets share some of the defining mathematical features of both the local deterministic set and the no-signalling set, but in general equals neither of them. Classes of those objects had also been investigated under the name `partially deterministic polytopes' in the PhD thesis of Woodhead (\cite{Woodhead2014}, Appendix), which is the terminology we adopt in this work.  Intuitively,  a behaviour is partially deterministic (in this terminology) if it is consistent with a  model where the outputs of some given subset of inputs are \emph{locally} determined. A more precise terminology for this property would be `local partial determinism',  which we employ for now, however as we clarify later, there are justifiable reasons why the emphasis on the descriptive term `local' may be dropped in this instance.    

Previous works have only considered the implications of local partial determinism in specific cases, narrowing the attention to either bipartite scenarios \cite{Bong2020, Utreras-Alarcon2024, Woodhead2014} or multipartite scenarios with at most one deterministic input per site \cite{Haddara2025}. In this work, we explore the concept of partial determinism in full generality, allowing for arbitrary numbers of parties, inputs and outputs for each party, and arbitrary subsets of deterministic inputs per site. 

Among our results we show that, similarly to the sets of local deterministic and no-signalling behaviours \cite{Tsirelson1993, Brunner2014, Scarani2012}, local partially deterministic sets are convex polytopes. In addition to describing the form of their vertices,
we investigate the relationships of the different partially deterministic polytopes that may be defined over a given scenario, and contrast them with the other well known behaviour sets, including the sets of Bell-local, quantum, and no-signalling behaviours. We find that that partially deterministic models of different kinds form equivalence classes of polytopes, which we classify completely. The equivalence class corresponding to the no-signalling polytope always has a single element, while the equivalence class containing the local deterministic polytope always has more than one element.  We also show that, excluding the extreme cases where a local  partially deterministic polytope equals the local deterministic or no-signalling polytopes, the set of quantum behaviours is neither fully contained in, nor fully contains, any partially deterministic polytope. 

As a useful technical tool to study the properties of local partially deterministic polytopes, we develop a multipartite generalization of the `behaviour product' introduced by Woodhead \cite{Woodhead2014}, by which behaviours in a given scenario can be constructed from behaviours defined over strictly smaller scenarios. This allows us to define a general notion of `composable sets', of which local partially deterministic polytopes are a special case. While we mostly focus on the properties of local partially deterministic polytopes, we also provide some example cases where the more general composable sets are of interest. 

Inspired by one of the equivalences given in Fine's theorem \cite{Fine1982, Fine1982b}, namely that the set of local deterministic behaviours is equivalent to the set of local factorizable behaviours, we also consider the more general notion of local partial factorizability. We show that a partial analogue of Fine's theorem can be formulated, establishing the equivalence, among other statements,  of the existence of a partially deterministic local model with the existence of a partially factorizable local model.  As a consequence of the nontrivial equivalence class of partially deterministic polytopes corresponding to the local deterministic set, we also  extend Fine's theorem by adding the existence of partially extended distributions with suitable marginals to his list of equivalent statements. 

In this work, we mostly treat  partial determinism as a mathematical property, and describe broadly the behaviour sets compatible with such properties. This approach is intentional, as it allows us to emphasize the convex-geometric structure of partially deterministic polytopes without committing to any particular physical interpretation or motivation. Indeed, we will show that in addition to the previously known cases in the literature \cite{Woodhead2014, Bong2020, Wiseman2023thoughtfullocal, Utreras-Alarcon2024, Haddara2025}, there are various other situations where the mathematical form of a partially deterministic model is relevant. We believe that, owing to the generality of our presentation, classes of partially deterministic polytopes could be of significance in further contexts of device-independent enquiry. 

This work is structured as follows. In Section \ref{Section:preliminaries} we introduce the background theory and mathematical tools.  In order to make the main body of text as easy to follow as possible, we build a self-contained presentation starting from the most basic notions.  In Section \ref{Subsection:Setsofbehaviours} we introduce various common sets of behaviours, with emphasis on their mathematical definitions. Section \ref{section:GeometricalPerespectiveOnBehaviours} reviews the geometrical perspective on behaviour sets and introduces the basic notions of the theory of convex polytopes. In section \ref{Section:SubscenariosandRestrictions} we formally introduce the notion of a subscenario and the restriction map, which we make use of in proving various inclusion relations. 

Section \ref{Section:PartialDeterministicMathsSection} contains the main results. Section \ref{section:PartialDeterminismSubsection} introduces the sets of partially predictable and partially deterministic behaviours, and contains most of the mathematical results concerning partially deterministic (local) polytopes. It also introduces the notion of composable sets. In Section \ref{Section:PartialseparabilityandFactorizable} we define the more general notions of partial uncorrelatedness and partial factorizability, and show that the existence of a  partially factorizable local model is equivalent to the existence of a partially deterministic local model. Here we also prove both a generalized version of Fine's theorem, and a strengthened version of  its usual formulation. 

In Section \ref{Section:Applications} we exhibit some situations where partially deterministic polytopes arise. In Section \ref{Section:device-independentinseparabilitywitnesses} we show that classes of partially deterministic polytopes provide device-independent quantum state inseparability witnesses. We also show how the intersections, unions, and convex hulls of their unions are all related to other notions that have been previously investigated in the literature. In Section \ref{Section:BroadcastLocality} we show that partially deterministic polytopes are precisely the broadcast-local set in a family of broadcasting scenarios \cite{Bowles2021, Boghiu2023} and also discuss how the more general examples of composable sets attain relevance for other classes of broadcast scenarios. In Section \ref{Section:GeneralSequentialWignerSection} we generalise the sequential Wigner's friend scenario introduced in Ref.~\cite{Utreras-Alarcon2024} and show that the sets of behaviours compatible with the assumptions that underlie the Local Friendliness no-go theorem \cite{Bong2020} are in one-to-one correspondence with partially deterministic polytopes. We finish with concluding remarks in Section \ref{Section:conclusions}, and we conclude with a Finnish remark here: Lopussa kiitos seisoo.

\section{\label{Section:preliminaries}Preliminaries}
In this work, we take the background structure to be that of a \emph{correlation scenario}, also known as a Bell scenario \cite{Brunner2014}. A correlation scenario $S$ can be specified by a set $I=\{1,...,|I|\}$ of space-like separated parties or agents, a set $M_i$ of inputs for each agent $i\in I$, and a set $O_{x_i}$ of outputs for each input $x_i\in M_i$. Pictorially, in these scenarios each party $i$ is in possession of a black box, which on each round of experiment takes an input $x_i \in M_i$ and produces an output $a_{x_i} \in O_{x_i}$. An illustration of the setup is given in Fig.~\ref{fig:correlationscenarioFig}.
\begin{figure}
    \centering
    \includegraphics[width=\columnwidth]{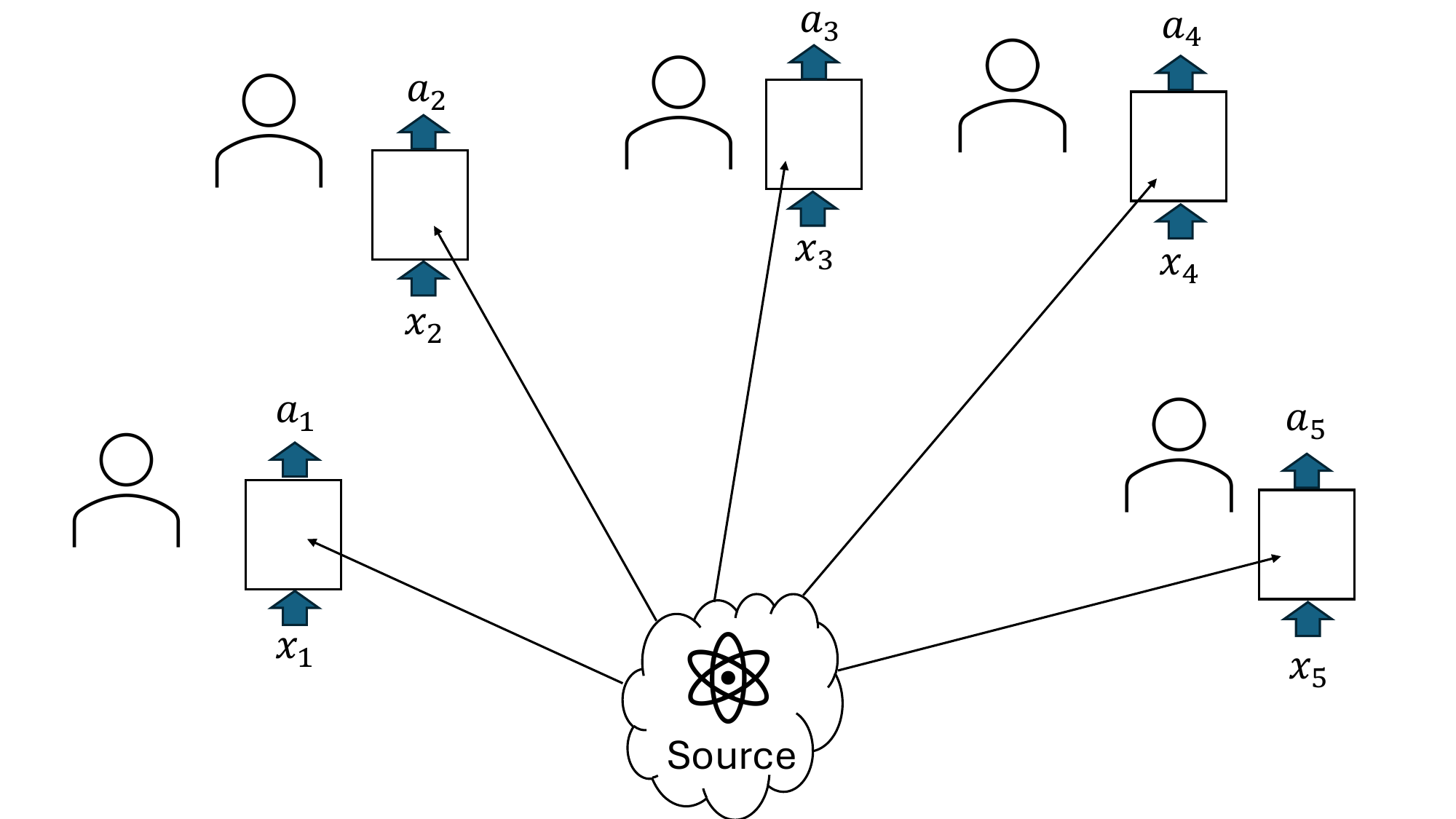}
    \caption{\label{fig:correlationscenarioFig} An example of a correlation scenario with five parties. Each party is in possession of a black box, which takes an input $x_i$ and produces an output $a_{x_i}$. It is assumed that the parties perform their interventions at space-like separation, but whatever is inside their boxes may have originated from a common source. 
    }
\end{figure} 

Following multiple rounds of the experiment, the statistics of the parties are gathered to estimate the collection of empirical probability distributions denoted by 
\begin{align}
   \wp \coloneqq \{\wp(a_{x_1}, a_{x_2}, \ldots , a_{x_{|I|}} | x_1, x_2, \ldots , x_{|I|} )\}  \label{behaviourEquation}
\end{align}
referred to as the \emph{behaviour} \cite{Brunner2014}.
We will often  refer directly to the individual probability distributions 
\begin{align}
  \wp(\vec{a}|\vec{x}) \coloneqq \wp(a_{x_1}, a_{x_2}, \ldots , a_{x_{|I|}} | x_1, x_2, \ldots , x_{|I|} ) \label{Equation:ProbabilitiesinBehaviour}
\end{align}
when discussing the properties of behaviours. Here we have adopted notation analogous to \cite{Haddara2025}, where $\Vec{x} \in M_1 \times M_2 \times \ldots M_{|I|} := \vec{M}$ represents a string of the measurements chosen by the agents, and $\Vec{a} \in O_{\Vec{x}} \coloneqq O_{x_1} \times O_{x_2} \times \ldots \times O_{x_{|I|}}$ is a string representing the outcomes of those measurements.  We highlight that the indices labelling the measurements have been dropped from the notation for the string $\Vec{a}$ in Eq.~\eqref{Equation:ProbabilitiesinBehaviour} above, as those can be implicitly taken to be specified by the value of $\Vec{x}$ in the conditional probability. Accordingly, we may also denote by $a_i$ the outcome of the $i$'th agent when the corresponding intervention $x_i$ is already clearly specified. Throughout this work, in order to exclude trivial cases, we will set $|M_i|, |O_{x_i}| \geq 2$ for all $i \in I, x_i \in M_i$, unless explicitly stated otherwise. The sets $I$, $ M_i$ and $O_{x_i}$ may otherwise be of arbitrary but finite size. 

Given a subset of parties $V \subset I$ we denote 
\begin{align}
    \wp(\Vec{a}_V | \Vec{x}) := \sum_{{ a_i:\, } i\in (I\setminus V)} \wp(\Vec{a} | \Vec{x})  
\end{align}
for the marginal distribution with  $\vec{a}_V $ representing a string consisting of outcomes in the sets $O_{x_{v}}$ for each $v\in V$. More generally, given a string $\Vec{s}$ we denote by $\Vec{s}_V$ the ordered substring of $\Vec{s}$ containing the elements $\{s_i\}_{i\in V}$. For a string of length one, such as  $\vec{a}_{\{i\}}$ we will simply write $a_i$.

Note that we employ the convention where the subset relation ``$A \subset B$'' is understood as the implication $p \in A \Rightarrow p \in B$, which does not rule out the possibility that $A=B$. For a strict subset the notation ``$\subsetneq$'' is used.

It is useful to use $M$ to represent the collection $\{M_i\}_{i\in I}$  of all the $M_i$, 
and, similarly, $O$ to represent $\{ O_{\Vec{x}} \}_{\Vec{x} \in \vec{M}}$. Then a correlation scenario $S$ can be identified by the triple $S = (I, M, O)$. In some works, e.g.~in \cite{Haddara2025}, the triple $S$ has also been referred to as the \emph{public scenario}; such terminology becomes particularly relevant, for example, in the type of extended Wigner's friend scenarios considered in Refs.~\cite{Brukner2017,Brukner2018,Bong2020,Utreras-Alarcon2024,Wiseman2023thoughtfullocal, Haddara2025}, where the experiment contains additional structure while maintaining the behaviour defined over the sets of public variables as the relevant object of interest.  Similar distinction will become useful for us in Section \ref{Section:GeneralSequentialWignerSection} as well. For now, however, by a \emph{scenario}, we mean any triple $S = (I,M,O)$.  A behaviour is always defined relative to some $S$, a fact which we incorporate in the definition in the next subsection for ease of reference.

\subsection{\label{Subsection:Setsofbehaviours}Sets of behaviours}
\definition{\label{definition:ThesetOfAllBehaviours}(The set of all  behaviours)\\Since the study of behaviour is ethology, the set of all behaviours for a given scenario $S$ is denoted by $\mathbf{E}(S)$. Every behaviour $\wp \in \mathbf{E}(S)$  satisfies

\begin{align}
    \wp(\vec{a}|\vec{x}) \geq 0 \hspace{0.3cm} \forall \vec{x},\vec{a}\in O_{\Vec{x}}, \hspace{0.2cm} \mathrm{and} \hspace{0.2cm}  \sum_{\vec{a}
    }\wp(\vec{a}|\vec{x}) =1 \hspace{0.3cm} \forall \vec{x}
\end{align}
by virtue of consisting of normalised probability distributions. 
}\\
\normalfont

In Definition \ref{definition:ThesetOfAllBehaviours} it is assumed that every input must produce one and only one output in every run of the experiment.

\definition{(Predictable behaviour)\label{predictablebehaviour}\\
A behaviour $\wp \in \mathbf{E}(S)$ is \emph{predictable} if and only if 
\begin{align}
    \wp(\vec{a}|\vec{x}) \in \{0,1 \} \hspace{0.3cm} \forall \vec{a}, \vec{x}.
\end{align}
The set of predictable behaviours is denoted by $\mathbf{P}(S)$.
\normalfont

\definition{(Uncorrelated behaviour)\label{DefinitionSeparablebehaviour}\\
A behaviour  $\wp \in \mathbf{E}(S)$ is termed \emph{uncorrelated}, or \emph{product } if and only if 
\begin{align}
  \label{Equation:separablebehaviour}  \wp(\vec{a}|\vec{x}) = \prod_{i\in I} \wp(a_i|\vec{x}) \hspace{0.3cm} \forall \vec{a}, \vec{x}.
\end{align}
The set of uncorrelated behaviours is denoted by $\mathbf{U}(S)$
}\\
\normalfont

Clearly $\mathbf{P}(S) \subsetneq \mathbf{U}(S)$. This can be seen, for example, by noting that by definition if $\wp\in \mathbf{P}(S)$ then from the definition of conditional probability

\begin{align}
    \wp(\vec{a}|\vec{x}) = \wp(\vec{a}_{I \setminus \{i \}}|\vec{x}, a_i )\cdot \wp(a_i|\vec{x}).\label{Equation:ConditionalProbForPredictableisSep}
\end{align}
The conditional expression $\wp(\vec{a}_{I \setminus \{i \}}|\vec{x}, a_i )$ is defined only for the events $a_i$ with nonzero probability, which owing to predictability of $\wp$ is unique. Hence the dependence on $a_i$ in  Eq.~\eqref{Equation:ConditionalProbForPredictableisSep} can be dropped. The  same kind of argument can be employed on the remaining marginal $\wp(\vec{a}_{I \setminus \{i \}}|\vec{x})$ and by iteration eventually leadsto  
\begin{align}
    \wp(\vec{a}|\vec{x}) = \prod_{i\in I} \wp(a_i|\vec{x}).
\end{align}
Thus, $\wp \in \mathbf{P}(S) \Rightarrow \wp \in \mathbf{U}(S)$ and hence $\mathbf{P}(S) \subset \mathbf{U}(S)$. To see that the inclusion is strict, take for example $\wp(a_i|\vec{x}) = 
1/{|O_{x_i}|}$ for all $a_i, \vec{x}$ in the defining Eq.~\eqref{Equation:separablebehaviour}, which provides a counterexample to the claim that the sets would be equal.

In Bell-type scenarios 
it is usually assumed that 
communication between the distant laboratories  is impossible, 
thus restricting to no-signalling behaviours: 
 
\definition{(No-signalling behaviours\label{No-signallingbehaviour})\\
  The subset $\mathbf{NS}(S) \subset \mathbf{E}(S)$ of no-signalling behaviours   consists of those $\wp$ for which  $\forall (i, x_i \neq x_i')$ 
  \begin{align}
     \wp(\Vec{a}_{I\setminus\{ i\}}| x_1\ldots x_i \ldots x_N) = \wp( \vec{a}_{I\setminus \{i\}}|x_1 \ldots x_i' \ldots x_N). \label{no-signalling}
 \end{align}}

\normalfont
For independence relations such as those in Eq.~\eqref{no-signalling}, we can equivalently write $\wp(\vec{a}_{I\setminus\{i\}}|\vec{x}) = \wp(\vec{a}_{I\setminus \{i\}}|\vec{x}_{I\setminus \{i\} })$,  which is to say that the distribution $\wp(\vec{a}_{I\setminus \{i\}}|\vec{x}_{I \setminus \{i\}})$ may represent any of the  marginal distributions $\wp(\vec{a}_{I\setminus\{i\}}|\vec{x})$ $\wp(\vec{a}_{I\setminus\{i\}}|\vec{x}')$  $\vec{x} \neq\vec{x}'$ which have $\vec{x}_{I\setminus \{i\}} = \vec{x}'_{I\setminus\{i\}}$, since they are equivalent.
 It can be seen, by repeated use of Equation \eqref{no-signalling}, that for all $\wp \in \mathbf{NS}(S)$,  $\wp(\Vec{a}_V| \Vec{x}) = \wp(\Vec{a}_V| \Vec{x}_V)$ for any $V \subset I$.

 The marginal independence  relations of Definition \ref{No-signallingbehaviour} guarantee, in particular,  that for every no-signalling behaviour the local statistics $\wp(a_i|\vec{x})$ at site $i \in I$ do not depend on the settings at other laboratories.  We introduce analogues of Definitions \ref{predictablebehaviour} and \ref{DefinitionSeparablebehaviour} which are consistent with this picture.  

\definition{(Predictable no-signalling behaviour)\\
The set $\mathbf{P}_{\mathbf{NS}}(S)$of predictable no-signalling behaviour is defined as the intersection 
\begin{align}
   \mathbf{P}_{\mathbf{NS}}(S) =  \mathbf{P}(S) \cap \mathbf{NS}(S).
\end{align}
}
\definition{(Uncorrelated no-signalling behaviour)\\
The set $\mathbf{U}_{\mathbf{NS}}(S)$ of uncorrelated no-signalling behaviour is defined as the intersection 
\begin{align}
   \mathbf{U}_{\mathbf{NS}}(S) = \mathbf{U}(S)\cap \mathbf{NS}(S).
\end{align}
}
\normalfont

A similar argument as in the case of $\mathbf{P}(S), \mathbf{U}(S)$ can be employed to establish that $\mathbf{P}_{\mathbf{NS}}(S) \subsetneq \mathbf{U}_{\mathbf{NS}}(S)$.

All of the definitions so far refer directly to properties of behaviours $\wp$, which represent the observational statistics in the experiment.  In the device-independent framework, important sets of behaviour are often characterized by indirect means, via constraints imposed on the \emph{models} which are expected to reproduce them. We introduce a few of such sets which are of particular importance.

\definition{\label{dEFINITION:lOCALDeterminism} A behaviour $\wp \in \mathbf{E}(S)$ is said to admit a local deterministic model, if and only if 
\begin{align}
    \wp(\vec{a}|\vec{x}) = \sum_{\lambda \in \Lambda} P(\lambda)\prod_{i\in I}D(a_i|x_i, \lambda),
\end{align}
for some finite set $\Lambda $, a probability distribution $P(\lambda)$ over $\Lambda$, and  a 
set of deterministic mappings $D(a_i|x_i, \lambda) \in \{0,1\}$ for all $a_i, x_i, \lambda$. The set of local-deterministic behaviours is denoted by $\mathbf{LD}(S)$.
}\\
\normalfont

That, in fact $\mathbf{LD}(S)\subset \mathbf{NS}(S)$ holds can be seen from the normalization of the individual distributions $D(a_i|x_i,\lambda)$.

The set $\mathbf{LD}(S)$ of Definition \ref{dEFINITION:lOCALDeterminism} consists of the behaviours which are compatible with a picture of the source encoding on each of the systems deterministic instructions for how to respond to all local measurements $x_i$  at each site $i\in I$. Thus, behaviour $\wp \in \mathbf{LD}(S)$ could in principle be produced by a purely classical (deterministic) process. Formally, if there  are some preparations $\xi_1$ and $\xi_2$ of the systems at the source  which, when probed in a scenario $S$ lead to behaviours $\wp_{\xi_1} = \{\wp_{\xi_1}(\vec{a}|\vec{x})\}$  and $\wp_{\xi_2} = \{\wp_{\xi_2}(\vec{a}|\vec{x})\}$, respectively, the mixture $\sum_{k\in \{1,2\}} P(\xi_k)\wp_{\xi_k}$  of those behaviours is another well-defined behaviour, which may be realised by having the source distribute systems prepared according to the different recipes $\xi_k$ randomly with probabilities $P(\xi_k)$ at each round of the experiment.  It is clear that the set $\mathbf{LD}(S)$ is exactly equivalent to all the set of behaviour which can be produced in the above manner with $\wp_{\xi_k}\in \mathbf{P}_{\mathbf{NS}}(S)$ for all $k$.  Mathematically this equivalence can be expressed as the set $\mathbf{LD}(S)$ being equal to the convex hull of $\mathbf{P}_{\mathbf{NS}}(S)$.

\definition{(\label{Definition:ConvexhULL}Convex hull of a set of behaviours)\\ Let $S$ be a scenario and $\mathbf{K}(S)\subset \mathbf{E}(S)$ some arbitrary subset of behaviours in that scenario. The convex hull of $\mathbf{K}(S)$ is defined as the set 
\begin{align}
&\mathrm{conv}[\mathbf{K}(S)] =  \left\{ \blk\sum_{j}\omega_j\wp_j  | \wp_j \in \mathbf{K}(S) \hspace{0.1cm}\forall j\in J , |J| < \infty,\right. \nonumber \\
& \hspace{3cm} \left. \omega_j \geq 0   \hspace{0,1cm} \forall j \in J , \sum_{j\in J} \omega_j = 1  \right\}. \blk
\end{align}
 The set $\mathbf{K}(S)$ is convex, iff it is closed under convex combinations, that is iff $\mathrm{conv}[\mathbf{K}(S)] = \mathbf{K}(S)$.
}\\
\normalfont

The convex hull is a standard notion in convex geometry, and  in Definition \ref{Definition:ConvexhULL} we have stated it in one of many equivalent ways it can be defined, while referring to the main objects of interest in this work. Thought of as an operator on sets, $\mathrm{conv}[\cdot]$ is idempotent, meaning that $\mathrm{conv}[\mathrm{conv}[\mathbf{K}(S)]] = \mathrm{conv}[\mathbf{K}(S)]$ for all sets $\mathbf{K}(S)$.

While the identification $\mathrm{conv}[\mathbf{P}_{\mathbf{NS}}(S)] = \mathbf{LD}(S)$ is mathematically clear,  it is worth drawing attention to a point of conceptual significance. Namely,  a distinction between the existence of a model with certain properties, and the claim that the behaviour is realised by such a process, should be made. For example, there are no-signalling behaviours $\wp \in \mathbf{NS}(S)$ that have convex decompositions $\wp=\sum_{j}\omega_j \wp_j$ where some of the $\wp_j \in \mathbf{E}(S)\setminus \mathbf{NS}(S)$, which is to say that the component $\wp_j$'s do not satisfy the no-signalling condition of Def.~\ref{No-signallingbehaviour}. If a physical theory postulates the existence of such ontic variables, it may therefore satisfy the operational no-signalling principle of Def.~\ref{No-signallingbehaviour}, defined at the level of behaviours, while violating an ontological notion of locality definable also at the level of such ontic variables -- see, for example, Ref.~\cite{Cavalcanti2012} for a discussion of this point.

For this reason, whenever defining various models for behaviours, we will use notation such as $P(\cdot)$ instead of $\wp(\cdot)$  to refer to probabilities specified by the model, which need not themselves represent operational statistics, and may also be defined over a probability space including further (``ontic'' or ``hidden'') variables postulated by the model. 

We  also recall  a set of behaviours compatible with a seemingly broader class of models. 

\definition{(LHV-modelable behaviour\label{LHVBehaviour})\\
  The subset $\mathbf{B}(S) \subset \mathbf{E}(S)$ of LHV-modelable, or Bell-local behaviour, consists of those $\wp$ which admit a factorizable local model, that is, which admit a decomposition of the form
  \begin{align}
      \wp(\Vec{a}|\Vec{x}) = \int_{\Lambda} p(\lambda)\cdot \prod_{i \in I} P(a_i | x_i, \lambda), \label{BellLocalDeterministiceq}
  \end{align}
where $\Lambda$ is a measurable set, $p(\lambda)\geq 0$ and $ \int_{\Lambda}p(\lambda)=1$. }\\
\normalfont

LHV models can be justified from various different sets of natural notions about locality and causality when applied to a Bell scenario, as discussed at length for example in Ref.~\cite{Wiseman2017}. We are mostly concerned with the mathematical features of behaviours which are compatible with such a decomposition. 

Clearly $\mathbf{B}(S) \supset \mathrm{conv}[\mathbf{U}_{\mathbf{NS}}(S)] \supset \mathrm{conv}[\mathbf{P}_{\mathbf{NS}}(S)] = \mathbf{LD}(S)$.  A well-known result established by Fine \cite{Fine1982, Fine1982b} implies, however, that in fact $\mathbf{B}(S) = \mathbf{LD}(S)$ and thus,  from a mathematical perspective, the question of the existence of either kind of model for a behaviour is equivalent.

\theorem{(A Fine equivalence theorem)\label{FinesTHRM}\\
The following statements are equivalent.

\begin{enumerate}[label=\roman*), ref=\ref{FinesTHRM}.\roman*]
    \item The behaviour $\wp$ is LHV-modelable in the sense of Definition \ref{LHVBehaviour}.
    \label{FinesLHVmodelable}

    \item There exists a joint distribution  \label{FinesJointDistribution} ${P^{S_{\infty}}}(\boldsymbol{\alpha})$, $\boldsymbol{\alpha} = (\alpha_{x_1 = x_1^1},\ldots ,\alpha_{x_1 = x^{|M_1|}_1}, \ldots ,\alpha_{x_{|I|} = x^{|M_{|I|}}_{|I|}}) \in O_{x_1 = x_1^{1}} \times \ldots \times O_{x_1 = x_1^{|M_1|}} \times \ldots \times O_{x_{|I|} = x^{|M_{|I|}|}_{|I|}} $
    over all the potential outcomes of all possible measurements for all agents, which recovers all the distributions $\wp(\vec{a}|\Vec{x})$ in the behaviour $\wp$ by appropriate marginalization. That is,
    \begin{equation}\label{JointdistributionMarginalizationEq}
     \wp(\vec{a}_{\vec{x}}|\vec{x}) = P^{S_{\infty}}(\boldsymbol{\alpha}_{\vec{x}} = \vec{a}_{\vec{x}}) 
    \end{equation}
    where 
    \begin{align}
          P^{S_{\infty}}(\boldsymbol{\alpha}_{\vec{x}}) := \sum_{\alpha_{x_i'}: x_i' \neq x_i} {P}^{S_{\infty}}( \boldsymbol{\alpha}).
    \end{align}

    \item \label{FinesDeterministicDistribution} The behaviour admits a local deterministic model in the sense of Definition \ref{dEFINITION:lOCALDeterminism}.

\end{enumerate}
}
\begin{proof}
   We omit the proof.  The theorem was stated in this exact form with proofs in Ref.~\cite{Haddara2025}. Other alternative presentations of the result, or some aspects of it,  can be found for example in Refs.~\cite{Fine1982, Fine1982b,Werner2001,Abramsky2011}.
\end{proof}
\normalfont

\remark{ \label{Remark:RemarkonFinesThrm} Fine's original work \cite{Fine1982} established the equivalences of Theorem \ref{FinesTHRM} in the special case of a bipartite scenario with two inputs 
$x_i \in \{x_i^1, x_i^2 \}$ and outputs $a_i \in \{a_i^1, a_i^2\}$ per site. That work \cite{Fine1982} also included the two additional statements in the list of equivalences of Theorem \ref{FinesTHRM} in the form stated below (in our notation): 
\begin{enumerate}
    \item The behaviour $\wp$ satisfies all the 4 versions obtainable by relabelings of inputs and outputs $(a_i^j,x_i^j) \leftrightarrow (a_{i}^{j'},x_i^{j'})$ of the Clauser-Horne (CH)-inequality \cite{Clauser1974}, which can be represented as
    \begin{align}
        -1 &\leq \wp(a_1^1, a_2^1|x_1^1,x_2^1) + \wp(a_1^1,a_2^2 |x_1^1, x_2^2) + \wp(a_1^2, a_2^2 |x_1^2,x_2^2) \nonumber \\
        &- \wp(a_1^2, a_2^2|x_1^2,x_2^2) - \wp(a_1^1|x_1^1) - \wp(a_2^2|x_2^2) \leq 0
    \end{align}
    \item There exists a pair of joint distributions $P_1(a_1^1, a_1^2, a_2^1|x_2^1)$ and $P_2(a_1^1, a_1^2, a_2^2|x_2^2)$ which agree on the marginal 
    \begin{align}
        P_1(a_1^1, a_1^2|x_2^1) =P_2(a_1^1,a_1^2|x_2^2) := P(a_1^1,a_1^2)
    \end{align}
    and which can be used to recover the distributions in the behaviour as marginals as in
    \begin{align}
        \wp(a^j_1, a^{j'}_{2}|x_1^j, x_2^{j'}) = P_{j'}(a_1^j, a_2^{j'}|x_2^{j'}).
    \end{align}
\end{enumerate}
}
\normalfont

In later work \cite{Fine1982b} Fine generalized some of the statements in Ref.~\cite{Fine1982} to the scope of Theorem \ref{FinesTHRM}, and included a result which generalized the second statement in Remark \ref{Remark:RemarkonFinesThrm} further to a certain multipartite form. The general applicability of these two bipartite statements in particular was challenged \cite{Garg1982,Fine1982Response} and finally shown not to be generally sufficient for the existence of an LHV model by a counterexample by Garg and Mermin \cite{Garg1982b}. A modern geometric view of Bell inequalities, which we review in Section \ref{section:GeometricalPerespectiveOnBehaviours}, gives a complete generalization of statement 1 that applies for all scenarios $S$. As a consequence of our main results,  we find, in Section \ref{Section:PartialseparabilityandFactorizable}, a multipartite, multi-input generalization of the second statement above which applies to all scenarios $S$, expanding the scope of Theorem \ref{FinesTHRM} even further and recovering both the form in above bipartite case, and Fine's later multipartite extension \cite{Fine1982b} of statement 2 as special cases.  

\normalfont

Finally, we define the set of behaviours compatible with a quantum mechanical description.

\definition{(Quantum behaviours\label{quantumbehaviour})\\
The set $\mathbf{Q}(S)$ of quantum behaviours in a Bell scenario, with  $|I|$ agents, is the set which can be modelled by a Hilbert space of the form $\mathcal{H}_{\Vec{A}} = \mathcal{H}_{A_1} \otimes \mathcal{H}_{A_2} \otimes \ldots \otimes \mathcal{H}_{A_{|I|}}$, a quantum state $\rho$, i.e. a positive unit-trace linear operator $\rho: \mathcal{H}_{\Vec{A}} \rightarrow \mathcal{H}_{\Vec{A}} $, and a set of positive operator valued measures (POVMs) with elements $M_{a_i|x_i}$ (i.e.~for which $M_{a_i|x_i} \geq 0 \; \forall i,x_i, a_i$ and $\sum_{a_i \in O_{x_i}} M_{a_i|x_i} = I_{\mathcal{H}_{A_i}} \forall i, x_i$), such that
\begin{align}
    \wp(\Vec{a}|\Vec{x}) = \mathrm{tr}[M_{a_1|x_1} \otimes M_{a_2|x_2} \otimes \ldots \otimes M_{a_{|I|} | x_{|I|}} \rho ]. \label{QuantumGeneralBornRule}
\end{align}
}
\normalfont

The set $\mathbf{Q}(S)$ defined in this manner is compatible with the scenario depicted in Fig.~\ref{fig:correlationscenarioFig} in the sense of  allowing only the implementation of local measurements on each system $i\in I$. It is immediately seen that $\mathbf{Q}(S)$ is a subset of $\mathbf{NS}(S)$, owing to the normalization of the local POVMs.

  Quantum behaviour are, in general, not predictable. It is easy to show, however, that for any predictable behaviour $\wp \in \mathbf{P}_{\mathbf{NS}}(S)$, there exists a quantum model that realises it. This immediately entails that $\mathbf{LD}(S)\subset \mathbf{Q}(S)$; a well known relation (see e.g.~\cite{Brunner2014}), which,  for completeness,  we prove in Appendix \ref{appendixProofthatLDisInQuantum}.

\theorem{\label{Theorem:LDisInQuantum}Let $S = (I,M,O)$ be a scenario. Then $\mathbf{LD}(S) \subset \mathbf{Q}(S)$. }
\begin{proof}
See Appendix \ref{appendixProofthatLDisInQuantum}. 
    \end{proof}

\normalfont

Theorem \ref{Theorem:LDisInQuantum} completes this part of the mathematical introduction by establishing the order $\mathbf{B}(S) \subset \mathbf{Q}(S) \subset \mathbf{NS}(S)$. In the next section, we will provide a geometric outlook to the sets of behaviour described here. We will discuss both Bell's theorem \cite{Bell1964}, which established the famous strictness relation $\mathbf{B}(S) \subsetneq \mathbf{Q}(S)$, and the works  of Tsirelson \cite{Tsirelson1980} along with Popescu and Rohrlich \cite{Popescu1994}, which made famous the strictness relation $\mathbf{Q}(S) \subsetneq \mathbf{NS}(S)$. Knowledge of these results will be useful for a general description of partial determinism in Section \ref{Section:PartialDeterministicMathsSection}.

\normalfont

\subsection{\label{section:GeometricalPerespectiveOnBehaviours}Geometrical perspectives on behaviours}

As seen in the previous section, some sets of behaviours admit a convex structure. It is possible to give a substantially sharper geometric picture however, which applies generally. This geometric perspective has proven very useful and is nowadays a standard tool in  both the study of device-independent, or black-box, information processing and foundations, see e.g. Refs.~\cite{Werner2001, Scarani2012, Brunner2014} for reviews. We will give a soft introduction to the topic here. The mathematical concepts we deal with are standard in the study of convex polytopes, and we present them as such.  For  more comprehensive mathematical treatments of the topic, we direct the reader to any of the textbooks available in the subject, e.g.~\cite{Brondsted1983, Grünbaum2003}. 

As a starting point, we note that any behaviour $\wp \in \mathbf{E}(S)$ may be thought of as a point in $\mathbb{R}^{d}$, with $d = \sum_{\vec{x}\in \vec{M}}|O_{\vec{x}}|$ the number of all the individual probabilities in a behaviour. By fixing coordinate systems labelled by the strings $(\vec{a},\vec{x})$ a behaviour $\wp$ may then be represented as a vector $\vec{\wp}\in \mathbb{R}^d$ (not to be confused with the string notation employed for $(\vec{a},\vec{x}$)). The labels in the sets $I, M, O$ may of course be arbitrary, but a choice has to be made as to how the behaviours are represented in $\mathbb{R}^d$.

A possible convention is what we term here a `natural representation' in which the elements in each set $O_{x_i}$ $M_i$, for all $i\in I, x_i$ are ordered by some choice of bijections $O_{x_i} \rightarrow \{1, 2, \ldots, |O_{x_i}|\}$, $M_i \rightarrow \{1,2,\ldots , |M_i| \}$ and so on. Thus associating each outcome $a_i \in O_{x_i}$ and input $x_i \in M_i$ for all $i\in I, x_i \in I$ with a natural number. The vector  $\vec{\wp}$ is then constructed by convention as the ordered set of distributions $\wp(\vec{a}|\vec{x})$ via $\vec{\wp}= [\wp( \vec{a}=\vec{1}|\vec{x}=\vec{1}), \wp(a_1 =2, \vec{a}_{I\setminus \{1\}}| \vec{x}=\vec{1}), \ldots, \wp(\vec{a}= |\vec{O}_{\vec{x}=\vec{1}}| |\vec{x}=\vec{1}), \ldots, \wp( \vec{a} = |\vec{O}_{\vec{x} = |\vec{M}|}| | \vec{x} = |\vec{M}|)]^{T}$, with $|\vec{O}_{\vec{x}}| = (|O_{x_1}|, \ldots , |O_{x_{|I|}} )$ and $|\vec{M|} = (|M_1|, \ldots |M_{|I|}|)$ etc., given the choice of bijection. Fortunately fixing such a convention will not be necessary for our purely analytical approach, but drawing attention to this fact, is instructive for the following discussion.

Since each individual probability satisfies $1 \geq \wp(\vec{a}|\vec{x})\geq 0$, the set of all behaviour $\mathbf{E}(S)$ is contained in the hypercube $\mathscr{C}^d = [0,1]^{d}$, which itself is clearly a convex,  closed and bounded subset of $\mathbb{R}^d$, with finite number of vertices. In other words, $\mathscr{C}^d $ is a \emph{convex polytope} \cite{Grünbaum2003}. 

\definition{(Convex polytope)\label{Definition:Convexpolytope}
A convex polytope is the convex hull of any finite set of points.
}\\
\normalfont

We emphasize for the purpose of upcoming discussions that Definition \ref{Definition:Convexpolytope} is general, in that it does not apply just to behaviour sets. There are also alternative equivalent ways to describe a convex polytope \cite{Grünbaum2003}. A particularly useful sharpening of Definition \ref{Definition:Convexpolytope} is that a convex polytope may be expressed as the convex hull of a specific finite set, namely the set of its extreme points, or vertices.

\definition{(Extreme point of a convex set)\label{Definition:Extremepointofaconvexset} Let $K \subset \mathbb{R}^d$ be a convex subset. The set $\mathrm{Ext}(K)$ of extreme points of $K$ are those $\vec{r}\in K$ for which if $\omega_j>0$, $\sum_{j\in j}\omega_j=1$ and $\vec{r}= \sum_j \omega_j \vec{r}_j$ for some $\vec{r}_j\in K$ for all $j\in J$, then $\vec{r}_j = \vec{r}$ for all $j\in J$. }\\
\normalfont

Intuitively, a point is extremal, if and only if it does not lie on a line between any two points in the convex set. The extreme points of a convex polytope provide its vertex representation, or $V$-representation, which is a minimal set of points which can be used to represent every point in the polytope by convex combinations. It is easy to see, that in the case of the hypercube $\mathscr{C}^d = [0,1]^d $ the extreme points are precisely the vectors $\vec{r}\in \{0,1\}^d$, which form the `corners' of the cube. This follows from the observation that for any $\vec{r}_1, \vec{r}_2\in \{0,1\}^d$, the vectors differ if and only if there is at least one component $r_{\vec{a}|\vec{x}}$ which equals to zero in one of the vectors and which is nonzero in the other. Therefore for that component, a convex combination gives a value in the interval $(0,1)$ and thus, it is not possible to express any third point $\vec{r}_3\in \{0,1\}^d$, $\vec{r}_1\neq\vec{r}_2 \neq \vec{r}_3$, in this manner. 

Yet another equivalent, and very useful way to define a convex polytope is as a bounded intersection of finite number of closed half-spaces of $\mathbb{R}^d$, meaning that the polytope is defined by the set of points $\vec{r}\in \mathbb{R}^d$ which satisfy a collection of linear inequalities of the form $\vec{b}\cdot\vec{r} \leq c$, with $\vec{b}$ a fixed vector of coefficients (with at least one non-zero) and $c$ a constant. Clearly every half-space is itself a convex set, and hence its intersection with any closed and bounded convex set gives in general another closed and bounded convex set. 

By definition, if an  inequality $\vec{b}\cdot \vec{r} \leq c$ is valid for a polytope $K$, then none of the points are in the open half-space $\vec{b}\cdot \vec{r} >c$. The hyperplane $H = \{ \vec{r}\in \mathbb{R}^d: \vec{b}\cdot \vec{r} =c \}$ defined by the closed half-space touches (or supports, since it does not cut the set) the polytope if the intersection $H\cap K$ is non-empty. In such a case either the intersection $H\cap K$ consists of a single point, or a set of points. In any case it is said to delimit a \emph{face} of the polytope $K$, which is easily seen to be another convex polytope, in general of smaller dimension.

Formally, for any subset $W \subset \mathbb{R}^d$, the (affine) dimension $\mathrm{dim} [W] $ of $W$ is defined as the dimension of the vector space that spans its affine hull $\mathrm{aff}[W]= \{\sum_{j\in J} \gamma_j w_j| w_j \in W \hspace{0.1cm} \forall j\in J \}$, with $J$ a finite set and $\gamma_j \in \mathbb{R}$ for all $j\in J$ and $\sum_j \gamma_j = 1 $. Note that in contrast with  convex combinations, the affine weights $\gamma_j$ are not required to be non-negative, while in contrast to general linear combinations the sum of the weights is required to be 1. A set of points $\{p_j\}_{j\in J}$ is affinely independent, if and only if none of the points $p_j$  is an affine combination of the others. It is seen that if a set of points is linearly independent, then it is also affinely independent.  On the other hand, the  affine independence of the set $\{p_j\}_{j\in J}$ is equivalent to the linear independence of a set with $|J|-1$ elements  of the form $\{p_j -p_{k} \}_{j\in J\setminus \{k\}}$. A $D$-dimensional object therefore contains $D+1$ affinely independent points. 

 Faces of different type for the polytope can be identified depending on their dimension, with a $0-$face corresponding to a single point, $1-$face corresponding to an edge, and so on.  A special class of half-spaces valid for a $D$-dimensional polytope are of particular importance: namely those, that delimit the \emph{facets} of the polytope, which are faces of dimension $D-1$. If the polytope $K$ is full-dimensional in $\mathbb{R}^d$, that is if $\mathrm{dim}[K] = d$, then the facet defining inequalities give the unique minimal half-space representation ($H$-representation) $K = \{ \vec{r}\in \mathbb{R}^d : \vec{r}\cdot \vec{b}_j \leq c_j \hspace{0.1cm} \forall j\in J \}$. In cases where $\mathrm{dim}[K]  < d$, the uniqueness of such a  representation is lost.

The hypercube $\mathscr{C}^d = [0,1]^d$ is full-dimensional in $d$, since in particular, the $d$ unit vectors $e_j, j\in J$ with $|J|= d$ and $e_j\cdot e_{j'} = \delta_{j, j'}$ for all $j, j' \in J $ are all contained in $\mathscr{C}^d$ and span the entirety of $\mathbb{R}^d$. Indeed, a set of $d+1$ affinely independent points contained in the cube can be formed from  the basis unit vectors $e_j$ and the origin $\vec{0}$. An $H$-representation is given by the $d^2$ inequalities $  \vec{r} \cdot e_j \geq 0 $ and  $\vec{r}\cdot e_j \leq 1$. These inequalities form the `sides' of the cube, and it is easily seen that they are, in fact,  facet defining. Namely, since for each inequality $\vec{r}\cdot e_j \geq 0$ there are,  $2^{d-1}$ vertices of the cube with that component of the vector zero, and in particular, those $2^{d-1}$ vectors contain the  $d-1$ vectors $e_{j'}, j' \in J\setminus \{j\}$ which are linearly independent and hence the claim follows as before. Exactly the same argument can be used in the case of all the inequalities $\vec{r}\cdot e_j\leq 1$ to get $d-1$ linearly independent vectors  and to that set can be added the vector $\vec{1}$ which has all components equal to one, which is affinely independent from the others and hence, every such inequality is indeed facet defining. 

The inclusion $\mathbf{E}(S) \subset \mathscr{C}^d$ was pointed out before. Now it is clear that not every vector $\vec{r}$ in the hypercube $\mathscr{C}^d$ represents a a valid behaviour. As an example, the origin $\vec{0}$ of $\mathbb{R}^d$ is one of the vertices of the cube, but does not correspond to any behaviour, which by virtue consisting of probability distributions have to have non-zero elements to be consistent with normalization. Since each behaviour $\wp$ has $\prod_{i\in I}|M_i|$ probability distributions labelled by the distinct $\vec{x}$, the set of points in  $\mathbf{E}(S)$ may be expressed as the intersection of $\mathscr{C}^d \cap \{\vec{r} \in \mathbb{R}^d| \sum_{\vec{a}}r_{\vec{a}|\vec{x}}=1 \hspace{0.1cm} \forall \vec{x}\} = \mathscr{C}^d \cap \{ \mathrm{Norm} \}$. The set $\{\mathrm{Norm}\}$  consists of those $\vec{r} \in \mathbb{R}^d$ which have certain numbers of their components, denoted individually by $r_{\vec{a}|\vec{x}}$, sum to one. 

The demand that $\vec{r}$ satisfies constraints of the form $\vec{r}\cdot \vec{b} = c$ may, of course, always be represented as the statement that $\vec{r}$ is contained in the intersection of two closed half spaces defined by $\vec{r}\cdot \vec{b} \leq c$ and $\vec{r}\cdot \vec{b}\geq c$. Thus, the intersection $\mathscr{C}\cap \{\mathrm{Norm}\}$, as the intersection of a convex polytope and a convex set (the intersection  of the two half-spaces), gives another convex polytope, namely $\mathbf{E}(S)$.

That $\mathbf{E}(S)$ is a polytope could, of course,  also have been established by simply noting that a convex combination of probability distributions is another valid probability distribution, and hence $\mathbf{E}(S)$ is convex. On the other hand, every probability distribution with a finite number of outcomes may be represented as a convex combination of precisely the predictable distributions, so that in fact any $\wp \in \mathbf{E}(S)$ would admit a convex decomposition  $\wp = \sum_{j\in J}\omega_j P_j$ with  $ P_j \in \mathbf{P}(S)$ for all $j\in J$. Thus $\mathbf{E}(S)$ is seen to be equivalent to the convex hull of predictable behaviour. Furthermore, each predictable behaviour is, by an analogous argument as in the case of corners of the cube, easily seen to be extremal. Thus, also $\mathrm{Ext}(\mathbf{E}(S)) = \mathbf{P}(S)$, so that the vertex representation of $\mathbf{E}(S)$ is easy to give. We believe providing this alternative, more generally grounded perspective to be insightful, however. 

The dimension $ \tilde{D}$ of $\mathbf{E}(S)$ is reduced by the marginalization constraints, since essentially only $|O_{\vec{x}}|-1$ probabilities of the outcomes $\vec{a}\in O_{\vec{x}}$ need to be specified in each distribution to determine the last one. The dimension of $\mathbf{E}(S)$ after the marginalization constraints are taken into account is therefore at most $\tilde{D} \leq \sum_{\vec{x}\in \vec{M}}|O_{\vec{x}}| - \prod_{i\in I}|M_i|$. To see that in fact $\tilde{D} = \sum_{\vec{x}\in \vec{M}}|O_{\vec{x}}| - \prod_{i\in I}|M_i|$, it is sufficient to consider predictable behaviours, of which there are $|\mathbf{P}(S)| = \prod_{\vec{x}\in \vec{M}}|O_{\vec{x}}|$. Sets of $\tilde{D}+1$ linearly independent points are easily seen to be included in the set. Consider for example  the predictable behaviour constructed of the $|O_{\vec{x}=\vec{1}}| + \sum_{\vec{x} \in \vec{M}, \vec{x}\neq 1}(|O_{\vec{x}}|-1) = \tilde{D}+1$ points which have the components $\wp(\vec{a}|\vec{x}=\vec{1})$ going through all permutations of the location of the nonzero output, while in every block with $\vec{x}\neq\vec{1}$ the nonzero component $\wp(\vec{a}|\vec{x})$ is fixed to some single output. Then, the next family of $|O_{\vec{x}}|-1$ vectors $\vec{x}\neq \vec{1}$ has the output-permutations of the component $\wp(\vec{a}|\vec{x})$ nonzero by going through all the $|O_{\vec{x}}|-1$ components not equal to the fixed value in the previous case and for every other case $\vec{x}'\neq \vec{1}, \vec{x}' \neq \vec{x}$ fixed to to the same single output as in the case of $\vec{x}=\vec{1}$, with $\wp(\vec{a}|\vec{x}=\vec{1})$ fixed to be nonzero for an arbitrary ouput. Clearly, these two families are linearly independent, since by construction, every such vector has a zero in a different location, and hence any linear linear combination of those points has to have their weights zero to represent the origin $\vec{0}$. This construction can be repeated for the remaining $\vec{x}$, thus leading to a set of $\tilde{D}+1$ linearly independent, and hence also affinely independent points, proving the claim. 

 By a similar counting argument we can also immediately establish, that the inequalities $\vec{r}\geq0$ remain facet defining for $\mathbf{E}(S)$. To do this, note first that for every inequality $r_{\vec{a}|\vec{x}}\geq 0$ the number $|\mathbf{P}(S)_{r_{\vec{a}|\vec{x}=0}}|$ of predictable points on that hyperplane  is

\begin{align}
  |\mathbf{P}(S)_{r_{\vec{a}|\vec{x}=0}}| = & (|O_{\vec{x}}| -1)(\prod\limits_{ \mathclap{\substack{\vec{x}' \in \vec{M}, \\ \vec{x}'\neq  \vec{x}}}} |O_{\vec{x}'}|), 
\end{align}  
we can then use a similar construction, by removing one of the $|O_{\vec{x}}
|$ first choices in the process which is set to zero to be compatible with $\vec{r}_{\vec{a}|\vec{x}}=0$. Now it is seen that it is possible  to find $\tilde{D}$  affinely independent points in $\mathbf{E}(S)$ which saturate the inequality $\vec{r}_{\vec{a}|\vec{x}}\geq 0$. 

In contrast, the inequalities $r_{\vec{a}|\vec{x}} \leq 1$ do  not define facets of $\mathbf{E}(S)$. Intuitively this is because that condition is guaranteed by the normalization conditions. More formally, fixing $r_{\vec{a}|\vec{\tilde{x}}}=1$ for some component of the vector $\vec{r}$ in the affine subspace $\{\mathrm{Norm}\}$ forces all the other components in the block labelled by $\vec{\tilde{x}}$ to equal zero. A simple counting argument can then be used to establish that the dimension $D'$ of vectors compatible with this constraint is strictly smaller than $\tilde{D}$, namely since $D' \leq \sum_{\vec{x}\in \vec{M}, \vec{x}\neq \vec{\tilde{x}}}|O_{\vec{x}}| - \prod_{i\in I}|M_i| < \sum_{\vec{x}\in \vec{M}}|O_{\vec{x}}| - \prod_{i\in I}|M_i| = \tilde{D}$. Therefore, the constraints $r_{\vec{a}|\vec{x}}\leq 1$ defining facets of the cube $\mathscr{C}^d$ could be omitted from the intersection. 

An illustration of $\mathbf{E}(S)$ in a trivial scenario with one party having one input with three outcomes is given in Fig.~\ref{fig:ProbabilitySimplex} In that case, $d =3$ and $\tilde{D}=2$ so a graphical representation is possible. Note that in the case of the figure, the set $\mathbf{E}(S)$ in fact forms a simplex, which has the special property of every convex decomposition in terms of extreme points being unique. In the cases of nontrivial scenarios $S$, the resulting polytopes will not in general be simplices due to the fact that there are at least two normalization conditions to be satisfied, which restricts the numbers of affinely independent points in the body. 

\begin{figure}
    \centering
    \includegraphics[width=\linewidth]{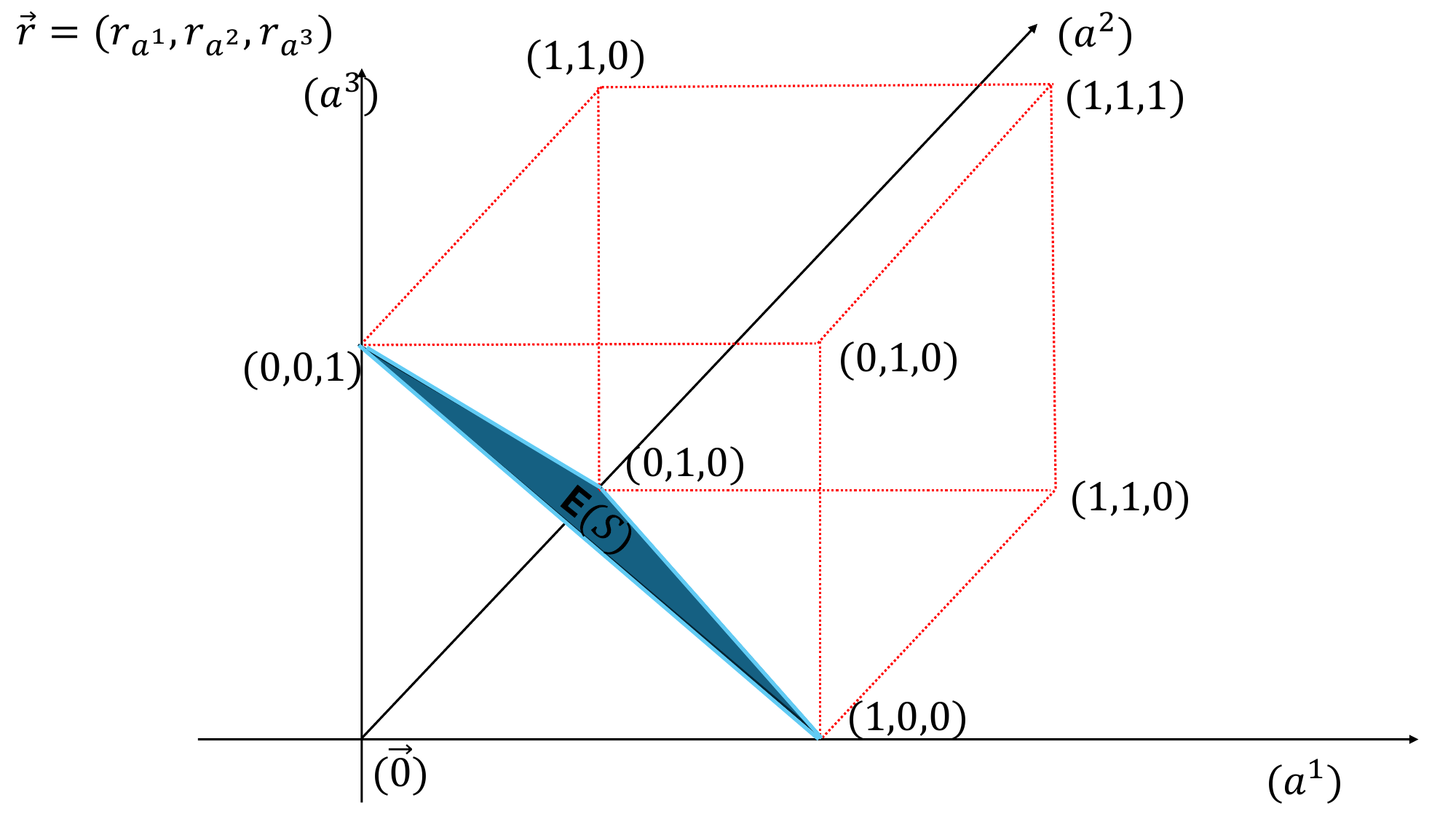}
    \caption{An illustration of the geometry of the set behaviour $\mathbf{E}(S)$, in a trivial scenario where the behaviour $\wp$ consists of a single probability distribution $\wp(a^k)$, $k\in \{1,2,3\}$. The vectors $\vec{r}$ representing the behaviour are necessarily contained in the cube $\mathscr{C}^3 = [0,1]^3$ (red dashed lines)  since each probability satisfies $ 1\geq \wp(a^k)\geq 0$. The behaviour are further constrained by the normalization constraints, which bounds the appropriate vectors to be in the set $\{\mathrm{Norm}\} = \{\vec{r}\in \mathbb{R}^3 : \sum_{k}r_{a^k}=1\}$ which defines a plane in $\mathbb{R}^3$. The set $\mathbf{E}(S)$ is the intersection $\mathscr{C}^3\cap \{\mathrm{Norm}\}$, a two-dimensional polytope (triangle) contained in the cube. The set $\mathbf{E}(S)$ is seen to be contained in the convex hull of the three predictable vertices $\wp(a^k)\in \{0,1\}$ which are a subset of the vertices of the cube $\mathscr{C}$. The facets of the set $\mathbf{E}(S)$ are the $1$-dimensional edges (light blue) corresponding to the positivity constraints $r_{a^k} \geq 0$. Since $\mathrm{dim}(\mathbf{E}(S)) = 2 <3$, the representation of the facet defining inequalities is not unique. For example constraining $\vec{r}$ to the hyperplane $r_{a^1}=0$ for the polytope $\mathbf{E}(S)$ is equivalent to $r_{a^2} +r_{a^3} = 1$, as may be seen by using the normalization constraints.  }
    \label{fig:ProbabilitySimplex}
\end{figure}

\normalfont
The sets $\mathbf{B}(S)$ and $\mathbf{NS}(S)$ are arguably more interesting. Both sets are contained in the intersection $\mathscr{C}^d\cap\{\mathrm{Norm}\}$, but also further constrained by the affine no-signalling constraints of Definition \ref{No-signallingbehaviour}.

The set $\mathbf{B}(S) =  \mathbf{LD}(S)$ is immediately identifiable as a convex polytope, since by Definition \ref{dEFINITION:lOCALDeterminism} it is the convex hull of the finite number of predictable no-signalling points $\mathbf{P}_{\mathbf{NS}}(S)$. Again, it is easily established that in fact $\mathbf{P}_{\mathbf{NS}}(S) = \mathrm{Ext}(\mathbf{LD}(S))$, since the only way a given component of the vector representing the behaviour may be zero in a convex sum is if that component is zero for all vectors in the decomposition. The vertex representation of $\mathbf{B}(S)$ is therefore easy to give.

The $H$-representations of $\mathbf{B}(S)$ are of great  importance. The inequalities that bound $\mathbf{B}(S)$ are precisely the Bell-inequalities \cite{Brunner2014}, which may be tested against the experimental data $\wp$, to probe whether the behaviour is compatible with a local deterministic model or not. While violation of a single inequality is sufficient to refute the existence of such a model, a complete set of inequalities gives both necessary and sufficient conditions for its existence.  

Constructions for complete sets of Bell inequalities are only known in some specific scenarios, such as the bipartite scenario mentioned in Remark \ref{Remark:RemarkonFinesThrm} where the CH-inequalities \cite{Clauser1974} are necessary and sufficient \cite{Fine1982}, or scenarios with specific structure (see e.g.~\cite{Brunner2014} for a review of some of those cases). Given knowledge of the vertices, it is possible to seek for an  $H$-representation computationally. The general task of computing a the facet defining inequalities of a polytope given its vertices is known as the \emph{facet enumeration problem} (see e.g.\cite{Avis1997}). In the case of Bell inequalities, the complexity of the facet enumeration problem grows quickly when the numbers parties, inputs and outputs are increased \cite{Pitowsky1989}. Therefore, though algorithms exist for solving for the facet defining inequalities, full solutions have been computed only in restricted cases.  Finding a complete set of inequalities bounding the polytope given knowledge of its vertices is the hard part in the case of $\mathbf{B}(S)$.

Note also  that in common terminology a Bell-inequality is any inequality valid for $\mathbf{B}(S)$ that may in principle witness $\wp \notin \mathbf{B}(S)$, and hence that terminology is not limited to just those that delimit the facets. The facet defining inequalites do provide a set that is 'optimal' in the sense of being minimal (but which is not unique, since $\mathbf{B}(S)$ is not full-dimensional).  

With this generally valid geometric description, Theorem \ref{FinesTHRM} can be sharpened by adding another equivalent statement, which refers to the $H$-representation of the Bell-polytope, as in the remark below. 

\remark{\label{Remark:FinestheoremSatisfiesBell} The statement `the behaviour satisfies any finite set of Bell-inequalities which form an $H$-representation of the polytope $\mathbf{B}(S)$' can be added to to the list of equivalent statements in Fine's theorem (Theorem \ref{FinesTHRM}). }
\normalfont
\\

 The set $\mathbf{NS}(S)$ is also a polytope, as it is defined by the intersection $\mathbf{E}(S)\cap \{NS\}$ =  $\mathscr{C}^d\cap \{\mathrm{Norm}\}\cap\{NS\}$ with the hyperplane $\{NS\}$ of vectors $\vec{r}$ which satisfy the affine no-signalling constraints 
\begin{align}
   \sum_{a_i} r_{\vec{a}|\vec{x}} = \sum_{a_i}r_{\vec{a}'|\vec{x}'}, 
\end{align}
enforced  for all $\vec{x}, \vec{x}'$ with $\vec{x}_{I\setminus\{i\}} = \vec{x}'_{I\setminus \{i\}}$ and $\vec{a}, \vec{a}'$ with $\vec{a}_{I\setminus \{i\}} =\vec{a}'_{I\setminus \{i\}}$. These are the only constraints on $\mathbf{NS}(S)$, and hence an $H$-representation of $\mathbf{NS}(S)$ is essentially known. The extreme points of $\mathbf{NS}(S)$ include all the predictable points $\mathbf{P}_{\mathbf{NS}}(S)$, a fact which can be established by analogous arguments as before, but also in general includes new extreme points which are not in  predictable and are non-classical, in the sense of being outside the set $\mathbf{B}(S)$. 

Complete descriptions of the extreme points of the no-signalling polytope have been obtained in some restricted scenarios (see e.g.~\cite{Barret2005, Jones2005}), however no compact generally valid classification of the vertices is known in general. Given that an $H$-representation of the no-signalling polytope is known, it is possible to try to find the vertices computationally.  The general problem of determining the vertices of a polytope given an $H$-representation is known as the \emph{vertex enumeration problem} \cite{Avis1997}. Unfortunately, as is the case with the facet enumeration problem for the Bell-local polytope, solving the vertex enumeration problem for a no-signalling polytope quickly becomes computationally prohibitive and complete computational solutions for the  have been obtained only in specific scenarios \cite{Pironio2011}. Finding  the vertices given the $H$-representation is the hard part in the case of $\mathbf{NS}(S)$, in contrast to $\mathbf{B}(S)$ where the vertices are known but the facet-defining inequalities are generally not.   

The affine hulls of $\mathbf{B}(S)$ and $\mathbf{NS}(S)$ have been completely characterized by Pironio \cite{Pironio2005}.

\theorem{\label{Theorem:Pironio'sTheorem}(Pironio's affine hull theorem\footnote{In Ref.~\cite{Pironio2005} the expression for the dimension $D$ is given as $D = \prod_{i\in I}(\sum_{x_i \in M_i}(|O_{x_i}|-1)+1) -1$ which is equivalent to the expression in Eq.~\eqref{Equation:PironiosDimensionEq.}. We choose to present the result in the form of Eq.~\eqref{Equation:PironiosDimensionEq.} as we find it easier to construct from the basic argument.})\\
$\mathrm{dim}(\mathbf{P}_{\mathbf{NS}}(S)) = \mathrm{dim}(\mathbf{NS}(S)) = D$, with 
\begin{align}
 \label{Equation:PironiosDimensionEq.}   D = \sum_{V\in 2^{I}\setminus \emptyset}\left( \prod_{i\in V}(\sum_{x_i \in M_i}( |O_{x_i}|-1))\right).
\end{align}
}
\begin{proof}[Proof sketch.]
    This is Theorem 1 in Ref.~\cite{Pironio2005}. We provide a sketch of the general idea and refer the reader to $\cite{Pironio2005}$ for further details. 

    First, note that by the defining feature of the no-signalling constraits $\wp(\vec{a}_V|\vec{x}) = \wp(\vec{a}_V|\vec{x}_V)$ for all $\vec{a}, \vec{x}$ and $V\subset I$. Consider the family of all such marginal distributions defined by $V \in 2^{I}\setminus \emptyset$, with $2^{I}$ the power set, i.e. the set of all possible subsets of $I$.  Each of these marginals $\wp(\vec{a}_{V}|\vec{x}_{V})$ further satisfies normalization constraints, thus one can remove some outcomes while still having the ability to construct the behaviour $\wp$ with the knowledge of these quantities. In fact, one can remove an output for each input $x_i$ in every distribution, leading in total to  
    \begin{align}
     D =    \sum_{V\in 2^{I}\setminus \emptyset}\left( \prod_{i\in V}(\sum_{x_i \in M_i}( |O_{x_i}|-1))\right)
    \end{align}
numbers. These are sufficient to construct the entire $\wp \in \mathbf{NS}(S)$ and hence $\mathrm{dim}(\mathbf{P}_{\mathbf{NS}}(S)) \leq \mathrm{dim}(\mathbf{NS}(S)) \leq D$. To see that the inequalities hold in the opposite order, it is sufficient to show that there are $D+1$ affinely independent predictable behaviour in $\mathbf{P}_{\mathbf{NS}}(S)$. In fact, a set of $D+1$ linearly independent predictable no-signalling points can be found, which thus span the entire affine hull of $\mathbf{NS}(S)$ and hence   it follows that in fact $\mathrm{dim}({\mathbf{P}}_{\mathbf{NS}}(S)) = \mathrm{dim}(\mathbf{NS}(S))$.
\end{proof}
\normalfont

As an immediate corollary of Pironios affine hull theorem, all the sets $\mathbf{LD}(S) \subset \mathbf{Q}(S) \subset \mathbf{NS}(S)$ live in the same affine subspace of $\mathbb{R}^d$ constrained by the no-signalling and normalization constraints, and spanned by the predictable no-signalling points $\mathbf{P}_{\mathbf{NS}}(S)$. It also follows that, these objects may be be represented in a space $\mathbb{R}^D$ with suitable parametrization, e.g. specifying the marginals $\wp(\vec{a}_J|\vec{x}_J)$ for all but one outcome for every input,  where the polytopes $\mathbf{NS}(S)$ and $\mathbf{LD}(S)$ are full-dimensional and that even the non-classical extreme points of $\mathbf{NS}(S)$ may be represented as affine combinations of points in $\mathbf{P}_{\mathbf{NS}}(S)$. 

Pironio also established \cite{Pironio2005}, that the positivity constraints remain facets of both $\mathbf{B}(S)$ and $\mathbf{NS}(S)$ as well, since they can, in analogy to the case of $\mathbf{E}(S)$, be saturated by $D+1$ linearly independent predictable no-signalling points. Note that since the inequalities $0 \leq r_{\vec{a}|\vec{x}} \leq 1$ are valid for the hypercube $\mathscr{C}^d$, they are valid for all of the behaviour sets as well. Inequalities of the form $r_{\vec{a}|\vec{x}} \leq 1$ were shown not to be facet defining for $\mathbf{E}(S)$, hence they cannot be facet defining for any of the smaller sets $\mathbf{NS}(S)$ or $\mathbf{B}(S)$ either, since they contain even less points. 

The geometry of the quantum set is significantly more sophisticated, and very rich even in the simplest cases. The quantum set is in general convex, but not a polytope see e.g.~Ref.~\cite{Brunner2014} for an introductory review and Refs~\cite{Goh2017,Le2023quantumcorrelations, Barizien-Bancal2025} for recent works that explore the geometry more thorougly in some specific scenarios.  We will not have much to say about it in this work, aside from some inclusion relations.

The sets $\mathbf{P}(S)$ and $\mathbf{P}_{\mathbf{NS}}(S)$ are finite sets of points, and obviously not convex. We can also establish that the sets $\mathbf{U}(S)$ and $\mathbf{U}_{\mathbf{NS}}(S)$ are not convex, despite not being finite. To illustrate this it is sufficient to consider an equal mixture $\frac{1}{2} \wp_1 + \frac{1}{2}\wp_2$, of distinct points with $\wp_1, \wp_2 \in \mathbf{P}_{\mathbf{NS}}(S)$. Let $\vec{a}^1$ and $\vec{a}^2$ represent the strings for which $\wp_k(\vec{a}|\vec{x}) =1$, respectively.  From the definition of conditional probabilities 
\begin{align}
    \wp(\vec{a}|\vec{x}) = \wp(\vec{a}_{I\setminus \{i\}}|\vec{x}_{I\setminus \{i\}}, a_i)\wp(a_i |x_i) 
\end{align}
but for the mixture, clearly 
\begin{align}
    \wp(\vec{a}_{I\setminus \{i\}}|\vec{x}_{I\setminus \{i\}}, a_i) = \begin{cases}
        \wp_1(\vec{a}_{I\setminus\{i \}}|\vec{x}_{I\setminus \{i\}})& \textrm{if } a_i =a_i^1\\
        \wp_2(\vec{a}_{I\setminus \{i\}}| \vec{x}_{I\setminus \{i\}}) & \textrm{if } a_i = a_i^2,
    \end{cases} 
\end{align}
which are in general not equal, and hence $\frac{1}{2}\wp_1 +\frac{1}{2}\wp_2 \notin \mathbf{U}_{\mathbf{NS}}(S)$. 

While not generally convex, the sets $\mathbf{U}_{\mathbf{NS}}(S)$ and $\mathbf{U}(S)$ do contain identifiable convex subsets of behaviour. For example, let $\wp_{1}, \wp_2 \in \mathbf{U}_{\mathbf{NS}}(S)$ be behaviour such that $\wp_1(a_1|x_1) \neq \wp_2(a_1|x_1)$ for some $x_1 \in M_1$ but otherwise $\wp_1(a_j|x_j) = \wp_2(a_j|x_j) = \wp(a_j|x_j)$. Then 
\begin{align}
   & \frac{1}{2}\prod_{i\in I}\wp_1(a_i|x_i) + \frac{1}{2}\prod_{i\in I}\wp_2(a_i|x_i) \\& \hspace{0.1cm} = \wp_3(a_1|x_1)\prod_{j \in I\setminus \{1\}}\wp(a_j|x_j),
\end{align}
where $\wp_3(a_i|x_i) = \frac{1}{2}\wp_1(a_i|x_i) + \frac{1}{2}\wp_2(a_i|x_i)$. This is clearly an element in $\mathbf{U}_{\mathbf{NS}}(S)$ as claimed. Intuitively, the sets $\mathbf{U}(S)$ and $\mathbf{U}_{\mathbf{NS}}(S)$ are closed under mixing of  outputs corresponding to some fixed single inputs, since these do not create correlations between parties. Geometrically, these sets may perhaps best be understood as  collections of hyperlines, which do not cover the entire convex hull.

A summary of the geometric description of the various sets of behaviour discussed so far is given in  Table~\ref{tab:SetsofBehaviourandSummaryofGeometry}. 

\begin{table*}[]
    \centering
    \subfloat[ \ \justifying{ Summary of the geometries of the sets of behaviour introduced in Section \ref{Subsection:Setsofbehaviours}. The sets $\mathbf{P}(S), \mathbf{P}_{\mathbf{NS}}(S), \mathbf{U}(S)$ and $\mathbf{U}_{\mathbf{NS}}(S)$ are not convex, but their convex hulls equal other convex sets in the list, namely the convex polytopes $\mathbf{E}(S)$ and $\mathbf{B}(S)$. The sets of behaviours which form convex polytopes are highlighited in red.}]{ 
    \begin{tabular}[t]{c|c|c||c}
       Set  & Geometry & Dimension & Convex hull \\ \hline
      $\mathbf{E}(S)$ & \red{Convex polytope} & $\tilde{D}$& equal\\ \hline
      $\mathbf{P}(S)$ & Set of points &$\tilde{D}$& $\mathbf{E}(S)$ \\ \hline
      $\mathbf{U}(S)$& Set of hyperlines & $\tilde{D}$ & $\mathbf{E}(S)$\\  \hline \hline
      $\mathbf{NS}(S)$ & \red{Convex polytope} & $D$ & equal\\ \hline
      $\mathbf{Q}(S)$ & Convex set & $D$ & equal\\ \hline
      $\mathbf{B}(S)$ & \red{Convex polytope} & $D$ & equal\\ \hline
      $\mathbf{P}_{\mathbf{NS}}(S)$ & Set of points & $D$ & $\mathbf{B}(S)$ \\ \hline
      $\mathbf{U}_{\mathbf{NS}}(S)$ & Set of hyperlines & $D$ & $\mathbf{B}(S)$
    \end{tabular}
    \label{Table:GeometryandDimensiontable}}
    \quad 
    \subfloat[ \justifying{ The vertex and half-space representations of the convex polytopes in Table \ref{Table:GeometryandDimensiontable}. The extreme points of $\mathbf{B}(S)$ are precisely the predictable no-signalling behaviours $\mathbf{P}_{\mathbf{NS}}(S)$, which are easy to give. The $H$-representations of $\mathbf{B}(S)$ consist of finite collections of Bell-inequalities which have to be necessarily satisfied, and the satisfcation of which is sufficient, for the behaviour $\wp$ be LHV-modelable. Precisely the opposite is the case for $\mathbf{NS}(S)$ for which an $H$-representation is easy to give in terms of the intersection of the polytope $\mathbf{E}(S)$ with the affine hyperplane $\{NS\}$ of vectors satisfying no-signalling type constraints.  $\mathrm{Ext}(\mathbf{NS}(S))$ always contains $\mathbf{P}_{\mathbf{NS}}(S)$, but in general $\mathbf{NS}(S)$ will also have vertices which are not contained in $\mathbf{B}(S)$. The set $\mathbf{E}(S)$ on the other hand can be represented as the intersection of the hypercube $\mathscr{C}^d=[0,1]^d$, which itself is an intersection of facet defining half spaces of the form $1\geq r_{\vec{a}|\vec{x}}\geq0$, and the affine hyperplane $\{\mathrm{Norm}\}$ of vectors satisfying normalization type constraints, or as the convex hull of predictable behaviour $\mathbf{P}(S)$.  }  ]{\begin{tabular}[t]{c|c|c}
        Polytope &  $V$-representation & $H$-representation    \\
         $\mathbf{E}(S)$ & $\mathrm{conv}[\mathbf{P}(S)]$ & $\mathscr{C}^d\cap \{\mathrm{Norm}\}$ \\ \hline
         $\mathbf{NS}(S)$ & generally unknown & $\mathscr{C}^d\cap \{\mathrm{Norm}\}\cap \{NS\}$ \\ \hline
         $\mathbf{B}(S)$& $\mathrm{conv}[\mathbf{P}_{\mathbf{NS}}(S)]$& generally unknown
\end{tabular}
\label{Table:PolytopesandRepresentationsiNGEOMETRY}}
   \caption{A description of the  main geometrical properties of the sets of behaviour introduced in Section \ref{Subsection:Setsofbehaviours}.  Table \ref{Table:GeometryandDimensiontable} outlines the basic geometrical features of the various sets, and provides a description of their convex hulls. The sets which form convex polytopes are highlighted in red, while a description of what is generally known about the vertex and half-space representations of those polytopes is given in  Table \ref{Table:PolytopesandRepresentationsiNGEOMETRY}.   }
    \label{tab:SetsofBehaviourandSummaryofGeometry}
\end{table*}

We shall finish this section with a case-study  of the simplest nontrivial Bell-scenario where the strict inclusions $\mathbf{B}(S) \subsetneq \mathbf{Q}(S) \subsetneq \mathbf{NS}(S)$ can be established. Namely, the bipartite scenario also described in Remark \ref{Remark:RemarkonFinesThrm}, with $I = \{1,2\}$, $M_i = \{x_i^1, x_i^2 \}$ and $O_{x_i} = \{a_i^1, a_i^2 \}$ for both $i\in I$.  Here complete geometric characterizations of the sets $\mathbf{B}(S)$ and $\mathbf{NS}(S)$ in particular, are known \cite{Tsirelson1993, Brunner2014, Barret2005}. A complete characterization of the quantum set even in this scenario has been missing until very recently \cite{Barizien-Bancal2025}, but for our purposes, only the very basic facts will suffice. 

In this scenario the behaviours can be represented as vectors in $\mathbb{R}^{d}$, with $d = \sum_{\vec{x}}|O_{\vec{x}}| = 4\times 4 =16$. The set of all valid behaviour $\mathbf{E}(S)$ in this scenario has dimension 
\begin{align}
    \tilde{D} = \sum_{\vec{x}\in \vec{M}}|O_{\vec
x}| - \prod_{i\in I}|M_i| = 4\times4 -4 = 12.
\end{align} The extreme points are the predictable behaviour $\mathbf{P}(S)$, of which there are $4^4 = 256$, since each of the four distributions $\wp(\vec{a}|\vec{x})$ have four individual probabilities in total. The sets $\mathbf{NS}(S)$  and $\mathbf{B}(S)$ have  their dimension $D$ further reduced to 
\begin{align}
    D &= \sum\limits_{V\in 2^I \setminus \emptyset}(\prod\limits_{i \in V}(\sum_{x_i\in M_i }(|O_{x_i}|-1))) \\
   &= 2(2-1) + 2(2-1) + 2(2-1)\times2(2-1) = 8.
\end{align} The extreme points of $\mathbf{B}(S)$ are precisely the predictable no-signalling behaviour, which are of the product form $\wp(a_1|x_1)\wp(a_2|x_2)$, since $\mathbf{P}_{\mathbf{NS}}(S) \subset \mathbf{U}_{\mathbf{NS}}(S)$. Thus $\mathbf{B}(S)$ has $2^2 \times 2^2 =16$ predictable extreme points. Note that these points in fact form a subset of the extreme points of $\mathbf{E}(S)$, and also the hypercube $\mathscr{C}^{16} = [0,1]^{16}\subset \mathbb{R}^{16}$. The nontrivial facets of $\mathbf{B}(S)$ are the symmetries of the CH-inequality \cite{Clauser1974} introduced in Remark \ref{Remark:RemarkonFinesThrm}. That they are indeed facets, can be seen, for example, by counting the number of affinely independent predictable behaviour saturating the bounds of the inequality. Thus, the polytope $\mathbf{B}(S)$ has a vertex representation as the convex hull of 16 vertices, or alternatively, a half-space representation as the intersection of 24 facet-defining inequalities, 16 of which correspond to the (trivial) positivity conditions and the 8 inequalities which are obtainable from the symmetries of the (upper and lower limits of the) CH-expression. 

The $H$-representation of $\mathbf{B}(S)$ in terms of the facet definining inequalities is not unique. Another famous equivalent representation of the nontrivial inequalities is given by the 8 relabeling symmetries of  Clauser-Horne-Shimony-Holt (CHSH)-inequalities  \cite{Clauser1969}, which can be expressed in terms of the parameters $  C_{\vec{x}} = \wp(a_1^1,a_2^1|\vec{x}) + \wp(a_1^2 ,a_2^2|\vec{x}) 
   - \wp(a_1^1 ,a^2_2|\vec{x}) -\wp(a_1^2, a_2^1 |\vec{x})$
as 

\begin{align}
    \sum_{\vec{x}\neq x_1^2x_2^2} C_{\vec{x}} - C_{x_1^2x_2^2} \leq 2.\label{Equation:CHSHinequality}
\end{align}

It is worth it to emphasize that the equivalence of the CH-expressions  used in Remark \ref{Remark:RemarkonFinesThrm} and the CHSH-expression of Eq.~\eqref{Equation:CHSHinequality}  as bounds for behaviours holds under the assumption that the behaviour are also constrained to the affine hyperplanes $\{\mathrm{Norm}\}$ and $\{NS\}$. If that property is relaxed for some reason, then the equivalence is lost. A nice discussion of this point, and the practical significance of  such possible 'relaxations' from the perspective of different post-processing strategies of experimental data can be found in Ref.~\cite{Czechlewski2018}. 

The no-signalling polytope in this scenario has 16 facets, corresponding to the (trivial) positivity constraints for each individual probability in the behaviour. The set of extreme points is nontrivial: in addition to the 16 predictable no-signalling behaviour which are shared with the polytope $\mathbf{B}(S)$, there are 8 non-classical extremal points \cite{Tsirelson1993, Brunner2014, Barret2005}. A representative of such an extremal no-signalling behaviour is given by the distributions
\begin{align}
 \label{equation:PRboxDefined} \wp(a_1^ja_2^{j'}|\vec{x}) = \begin{cases}\frac{1}{2} &  \mathrm{if }\hspace{0.1cm} \vec{x} \neq x_1^2x_2^2 \hspace{0,1cm}  \mathrm{and} \hspace{0.1cm} j=j' \hspace{0.1cm} \\
  \frac{1}{2} & \mathrm{if} \hspace{0.1cm} \vec{x}=x_1^2x_2^2 \hspace{0,1cm} \mathrm{and} \hspace{0.1cm} j\neq j'\\
  0 & \mathrm{otherwise.}
    \end{cases}
\end{align}
The other remaining extremal behaviour are obtained from Eq.~\eqref{equation:PRboxDefined} by relabeling the local inputs and outputs. 

It is easily seen that the behaviour in Eq.~\eqref{equation:PRboxDefined} obeys the no-signalling, as all the marginals $\wp(a_i|x_i)$ are unbiased. By inspection of Eq.~\eqref{Equation:CHSHinequality} and the definition of $C_{\vec{x}}$ that preceeded it, it can be seen that this behaviour also reaches the algebraic maximum 4 of the CHSH-expression. It has become customary to refer to behaviour of the type in Eq.~\eqref{equation:PRboxDefined} as Popescu-Rohrlich (PR)-correlations, owing to the work of those authors in Ref.~\cite{Popescu1994} which drew significant attention to the foundational role of such behaviour. The 8 variants of the PR-behaviour have nice relationship with the nontrivial facet-defining inequalities of the Bell-polytope $\mathbf{B}(S)$: each PR-box violates exactly one of the CHSH-inequalities.  The relationship between the nontrivial facets of the Bell-polytope and the nonclassical extreme points of the no-signalling polytope is generally more complicated however \cite{Barret2005}.

The celebrated work of Bell \cite{Bell1964} first established that there are indeed inequalities valid for the set $\mathbf{B}(S)$ which may be violated by some quantum behaviour, though the geometric understanding of these notions became later c.f.~\cite{Tsirelson1993, Brunner2014}. The CHSH-inequality of Eq.~\eqref{Equation:CHSHinequality}, for example, can be violated up to a maximum value of $2\sqrt{2}$ \cite{Tsirelson1980} by suitable measurements on the maximally entangled state $\rho^{\Psi_-} = \ket{\Psi_-}\bra{\Psi_-} \in \mathbb{S}(\mathcal{H}_{\vec{A}})$, $\mathcal{H}_{\vec{A}} \simeq \mathbb{C}^2\otimes \mathbb{C}^2$, with $\{\varphi_0, \varphi_1 \}$ an orthonormal basis for $\mathbb{C}^2$ and
\begin{align}
    \ket{\Psi_-} = \dfrac{1}{\sqrt{2}}(\varphi_0 \otimes \varphi_1 - \varphi_1 \otimes \varphi_0).
\end{align}

The seminal work of Tsirelson \cite{Tsirelson1980} established that the bound $2\sqrt{2}$ of the CHSH-expression reachable by local measurements on a two-qubit system is in fact a strict limit of quantum theory. Here the phrase 'strict limit' is meant in the sense in the sense that any pair of quantum systems, irrespective of dimensions or choices of observables, obey the inequality. Thus, the result  establishes the famous  strict inclusions $\mathbf{B}(S) \subsetneq \mathbf{Q}(S) \subsetneq \mathbf{NS}(S)$ in this scenario.

Intuitively,  that the strictness result holds in the simplest scenario implies that the strictness result holds in any other more general scenario as well. This is because the simplest scenario is a 'subscenario' of the others, and hence the CHSH-inequalities for the different pairs of parties and pairs of their inputs in $S$ in particular, give necessary (but not always sufficient as shown for example by the result  of Garg and Mermin \cite{Garg1982b}) conditions for a behaviour in a multipartite or multi-input scenario to be LHV-modelable. In the next section, we formalise this intuition and show how to rigorously extend the strictness result from the above special case to arbitrary scenarios $S$.

\normalfont
\subsection{\label{Section:SubscenariosandRestrictions}Subscenarios and restrictions of behaviour}

In this section we formalise the notion of a subscenario, and associate it with a corresponding restriction map, which operates on behaviours. We believe the basic ideas to be familiar to people working in the field. These concepts turn out, however, to be of significant utility   for our general treatment of partial determinism in Section \ref{Section:PartialDeterministicMathsSection}, which is why we believe giving a formal treatment of these ideas is worthwhile.  

\definition{(Restriction of a correlation scenario)\label{DefinitionRestrictionofaScenario}\\
Let $S$ be a scenario. Let  $M' = \{M'_i \}_i$ be a collection of subsets of the inputs $M'_i \subset M_i$ for each $i\in I$. Then, a restriction of $S$ to this collection, denoted $S_{|M'}= (I_{|M'}, M_{|M'}, O_{|M'})$, is defined by setting
\begin{align}
    I_{|M'} &= \{i \in I | M'_i \neq \emptyset \},\\
    M_{|M'} &= \{ M'_{i}\}_{i \in I_{|M'}},\\
    O_{|M'} &= \{ O_{\vec{x}_{|\vec{M'}}} \}.
\end{align}
Here the input string $\vec{x}_{|M'} \in  \vec{M'}$ = $\prod_{i\in I_{|M'}} M_i'$ is understood to be an element in the Cartesian product of the non-empty sets $M_i'$.
}\\
\normalfont

Note that contrary to the usual conventions of this work, the restricted scenarios $S_{|M'}$ may be trivial in the sense of, say, containing a single party, or a single input per party. Furthermore if $M' = M$, then by definition $S_{|M'} = S$ and the restriction may be considered trivial. At the other extreme, if $M_i' = \emptyset $ $ \forall i$, then the restriction is not well defined, since the resultant scenario would have empty sets of parties, inputs and outputs. In every other case, the mapping $S \mapsto S_{{|M'}}$ is nontrivial, and one may consider various definitions of sets of behaviour over $S_{|M'}$. Evidently, the number of possible restrictions $|S_{|M'}|$ equals the product of the cardinalities of the powersets of each $M_i$. Excluding the trivial and ill-defined cases, this equals $ |\{S_{|M'}\}| = \left(\prod_{i \in I } 2^{|M_i|} \right)-2$, for $M_i' \neq M_i, \emptyset$ for at least one $i\in I$.

For any two scenarios $S'$ and $S$ we may also say that $S'$ is a \emph{subscenario} of $S$ if there exists a collection $M'$ of subsets of inputs $M_i'\subset M_i$ for all $i\in I$ such that $S_{|M'} = S'$.

By means of Definition \ref{DefinitionRestrictionofaScenario}, it is possible to define a 'restriction mapping' $R_{|M'}: \mathbf{NS}(S)\longrightarrow \mathbf{NS}(S_{|M'})$, which will prove useful.

\definition{(Restriction of a behaviour)\label{Definition:RestrictionofaBehaviour}\\
Let $S_{|M'}$ be a restriction of $S$ in the sense of Definition \ref{DefinitionRestrictionofaScenario}. The restriction of a behaviour $\wp^S \in \mathbf{NS}(S)$ onto a behaviour $\wp^{S_{|M'}}\in \mathbf{NS}(S_{|M'})$ is defined via the mapping $R_{|M'}: \mathbf{NS}(S) \rightarrow \mathbf{NS}(S_{|M'})$ 

\begin{align}
   R_{|M'}(\wp^S) := \wp^{S_{|M'}} = \{\wp^{S}(\vec{a}_{|I_{M'}} |\vec{x}_{|M'})\}_{S_{|M'}}\hspace{0.2cm}, \label{Eq.restrictionmap2}
\end{align}
with which the restricted behaviour can be defined as
\begin{align}
    \wp^{S_{|M'}}(\vec{a}'|\vec{x}')&=\wp^S(\vec{a}_{|I_{M'}} = \vec{a}'|\vec{x}_{|M'} = \vec{x}')  
 \hspace{0.2cm}. 
\end{align}
Here the superscript $S_{|M'}$ has been added to emphasize that $\wp^{S_{|M'}}(\vec{a}'|\vec{x}')$ is a behaviour defined in the restricted scenario $S_{|M'}$.
}\\
\normalfont

The restriction map of Definition \ref{Definition:RestrictionofaBehaviour} produces a behaviour in a sub-scenario $S_{|M'}$ of $S$ by simply taking appropriate marginals of distributions in a behaviour over parties whose inputs are excluded  from $S$. This marginalization property is well defined owing to no-signalling, by appeal to which the distributions in Eq.~\eqref{Eq.restrictionmap2} can be recovered from any $\wp^S(\vec{a}| \vec{x})$ with $\vec{a}_{I_{M'}} = \vec{a}'$ and $\vec{x}_{I_{M'}} = \vec{x}'$ by marginalization. In cases with $I_{M'} = I$, where the restricted scenario has the same number of parties, the mapping simply removes distributions with certain inputs $x_i \notin M_i'$ from the behaviour. The action of this map is illustrated in Fig.~\ref{fig:RestrictionMap}.

\begin{figure}
    \centering
    \includegraphics[width=\linewidth]{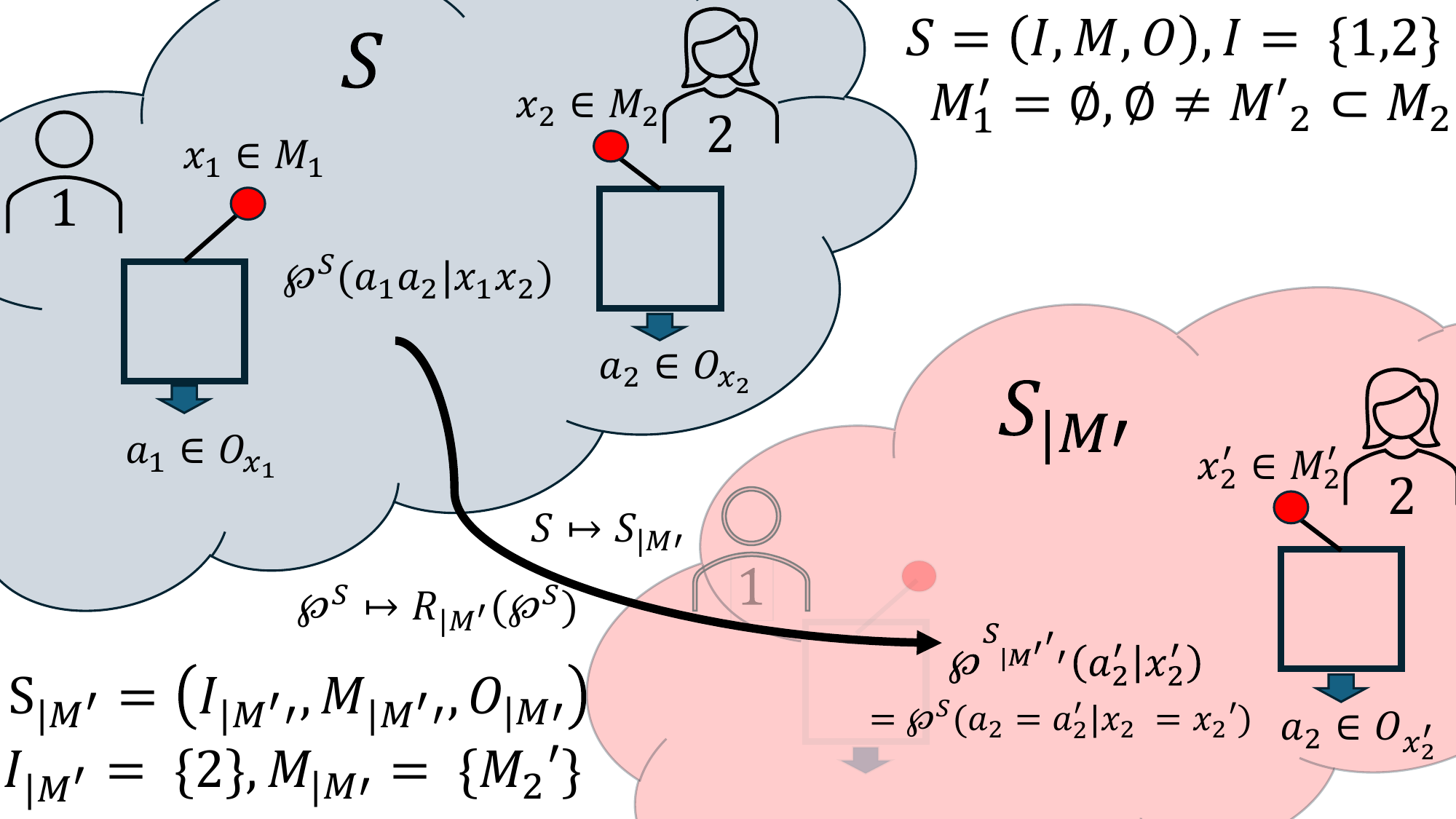}
    \caption{A bipartite scenario $S = (I,M,O), I= \{1,2\}$, is restricted to the scenario $S_{|M'}$ as defined by the collection $M'$ with $M_1'=\emptyset$ and $\emptyset\neq M_2'\subset M_2$. The scenario $S_{|M'} = (I_{|M'}, M_{|M'}, O_{|M'})$ contains only one agent $i =2$, with some subset $M_2'$ of the inputs of agent $i=2$ in scenario $S$. The restriction map $R_{|M'}$ represented in the figure maps the bipartite behaviour $\wp^{S}$ to a well-defined behaviour $\wp^{S_{|M'}}$ essentially by marginalizing over agent $i=1$, and omitting distributions for which $x_1 \notin M_1'$. }
    \label{fig:RestrictionMap}
\end{figure}

Note also that since  $\sum_{\vec{a}}\wp(\vec{a}|\vec{x}) = 1$ $\forall \vec{x}$, it is possible to formally allow for the case where $I_{M'} =\emptyset$ which on the level of scenarios corresponds to the situation where the restricted scenario $S_{|M'}$ itself is not well defined.  Conventionally we allow for this case, but we emphasize that  $R_{|M'= (\emptyset, \emptyset, \emptyset)}$ does not produce a behaviour, but instead, maps everything to a single point. We refer to either of the cases where $M_i' = M_i$ $\forall i$ or $M_i' =\emptyset$ $\forall i$ corresponding to the cases with $R_{|M'= M}(\wp) = \wp$ $\forall \wp \in \mathbf{NS}(S)$ and $R_{|M'= (\emptyset, \emptyset, \emptyset)}(\wp) = 1$ $\forall \wp \in \mathbf{NS}(S)$, respectively, as trivial restriction maps.

For any subset $\mathbf{K}(S)\subset \mathbf{NS}(S)$ we denote by $R_{|M'}(\mathbf{K}(S))$ the image of $\mathbf{K}(S)$ under the restriction $R_{|M'}$, that is the set $\{ \wp^{S_{|M'}} : \exists \wp^S \in \mathbf{K}(S) \hspace{0.1cm}\mathrm{ s.t. } \hspace{0.1cm} R_{|M'}(\wp^S) = \wp^{S_{|M'}} \}$.

We report an immediately evident feature of the restriction map as a lemma. 

\lemma{\label{Lemma:convexityofRestrictionmap}Let $\mathbf{K}(S) \subset \mathbf{NS}(S)$, $\mathrm{conv}(\mathbf{K}(S)) = \mathbf{K}(S)$ be a convex subset of no-signalling behaviours. Then for any collection of inputs $M'$
\begin{align}
    R_{|M'}(\mathbf{K}(S)) = \mathrm{conv}[R_{|M'}(\mathbf{K}(S))].
\end{align}
}
\begin{proof}
Immediate. We provide a proof for completeness.

Let $\wp =\omega\wp_1 + (1-\omega)\wp_2$, $\omega \in (0,1)$, and consider the image $R_{|M'}(\wp)  = \wp^{S_{|M'}}$ of $\wp$. By definition of the restriction map
\begin{align}
\wp^{S_{|M'}}(\vec{a}'|\vec{x}') &= \wp(\vec{a}_{I_{M'}} = \vec{a}'|\vec{x}_{|M'} = \vec{x}')\\
\begin{split}
&=\omega \wp_1(\vec{a}_{I_{M'}} =\vec{a}'|\vec{x}_{|M'} = \vec{x}') \\
 & \hspace{0.2cm}+ (1-\omega)\wp_2(\vec{a}_{I_{M'}}=\vec{a}' |\vec{x}_{|M'}=\vec{x}'),
\end{split}\\
&= \omega \wp_1^{S_{|M'}}(\vec{a}'|\vec{x}') + (1-\omega)\wp_2^{S_{|M'}}(\vec{a}'|\vec{x}')
\end{align}
which is to say that 
\begin{align}
    R_{|M'}(\wp) = \omega R_{|M'}(\wp_1)+ (1-\omega)R_{|M'}(\wp_2)
\end{align}
and hence the map $R_{|M'}$ preserves convexity. 
\end{proof}

\normalfont

The following proposition incorporates observations made with respect to Lemmas 16 and 17 from Ref.~\cite{Haddara2025} into a single, refined statement concerning restriction maps and the images of the subsets of $\mathbf{NS}(S)$ discussed in Section \ref{Subsection:Setsofbehaviours}. 

\proposition{\label{Proposition:restrictionmapidentities}Let $R_{|M'}:\mathbf{NS}(S) \rightarrow \mathbf{NS}(S_{|M'})$ be an arbitrary nontrivial restriction map. Then

\begin{enumerate}[label=(\roman*), ref=\ref{Proposition:restrictionmapidentities}.\roman*] 
    \item $R_{|M'}(\mathbf{P}_{\mathbf{NS}}(S)) = \mathbf{P}_{\mathbf{NS}}(S_{|M'})$ \label{prop:restrictionidentity1}
    \item $R_{|M'}(\mathbf{U}_{\mathbf{NS}}(S)) = \mathbf{U}_{\mathbf{NS}}(S|_{M'}).$\label{prop:restrictionidentity2}
    \item $R_{|M'}(\mathbf{B}(S)) = \mathbf{B}(S_{|M'}).$ \label{prop:restrictionidentity3}
    \item  $R_{|M'}(\mathbf{Q}(S)) =\mathbf{Q}(S_{|M'}).$\label{prop:restrictionidentity4}
    \item $ R_{|M'}(\mathbf{NS}(S)) = \mathbf{NS}(S_{|M'}).$\label{prop:restrictionidentity5}
\end{enumerate}}
\begin{proof}[Proof sketch]
    The proofs are straightforward and analogous in each case, and we will therefore only provide an outline of the general argument. 
    
    Let $\mathbf{K}(S)\in \{\mathbf{P}_{\mathbf{NS}}(S), \mathbf{U}_{\mathbf{NS}}(S), \mathbf{B}(S), \mathbf{Q}(S), \mathbf{NS}(S)\} $ represent any of the sets in \ref{prop:restrictionidentity1}-\ref{prop:restrictionidentity5}.  
    
    Clearly $R_{|M'}(\mathbf{K}(S)) \subset \mathbf{K}(S_{|M'})$ for any of the sets $\mathbf{K}(S)$.  For the cases \ref{prop:restrictionidentity1} and \ref{prop:restrictionidentity2}, this can be seen by noting that the marginals of predictable behaviour are themselves predictable, and the marginals of product behaviour are also of the product form. For \ref{prop:restrictionidentity3} and \ref{prop:restrictionidentity4}, note simply that the restriction map always either restricts inputs or omits parties, and in both cases the LHV/quantum model of a given behaviour always induces an LHV/quantum model for the image. In the case of \ref{prop:restrictionidentity5} the claim is trivial, since the map always produces a valid no-signalling behaviour contained in $\mathbf{NS}(S_{|M'})$. 

To see that also $R_{|M'}(K)(S) \supset \mathbf{K}(S_{|M'})$, it is sufficient to show that for every element $\wp^{S'} \in \mathbf{K}(S_{|M'})$ there is in fact an element $\wp^S$ such that $R_{|M'}(\wp^S) =  \wp^{S'}$. Consider then the `complement restriction'  $S_{|M'^{\perp}}$  of $S_{|M'}$ defined by  the collection of sets $M'^{\perp}_i$, which have $M'^{\perp}_i = M_i\setminus M_i'$ $\forall i \in I$. Let $F_{\vec{x}}= \{ i \in I| x_i \in M_i' \}$ denote a context-dependent set of parties in $S$ whose inputs $x_i \in M_i'$ in a given string $\vec{x}$ are inputs in the restricted scenario $S_{|M'}$, and likewise $F^{\perp}_{\vec{x}}= \{i \in I| x_i \in M_i'^{\perp} \} = I\setminus F_{\vec{x}}$ for the complement.

Let  $\wp^{S_{|M'}}\in \mathbf{K}(S_{|m'})$ be a behaviour in one of the sets of interest defined over the restricted scenario $S
_{|M'}$. We will show that this can always be lifted to a behaviour in $S$, which can be guaranteed to be an element in $\mathbf{K}(S)$. 

Since $\mathbf{P}_{\mathbf{NS}}(S) \subset \mathbf{K}(S)$ in particular, we can in every case use the same construction, namely by taking products with some distributions from the behaviour $\wp^{S_{|M'^{\perp}}} \in \mathbf{P}_{\mathbf{NS}}(S_{|M'^{\perp}})$ so that 
\begin{align}
\begin{split}
    \wp^S(\vec{a}|\vec{x}) &=   \wp^{S'}(\vec{a}'=
    \vec{a}_{F^{\perp}_{\vec{x}}} |\vec{x}'=\vec{x}_{F_{\vec{x}}}) \\
    & \hspace{0,2cm} \times \wp^{S_{|M'^{\perp}}}(\vec{a}'^{\perp}=\vec{a}_{F^{\perp}_{\vec{x}}}|\vec{x}'^{\perp}=\vec{x}_{F^{\perp}_{\vec{x}}}) \,.
    \label{Equation:restrictedbehaviourandProduct}
    \end{split}
\end{align}
The $\wp^S$ defined by the context-dependent expansion of Eq.~\eqref{Equation:restrictedbehaviourandProduct} is a well defined behaviour and clearly contained in  $\mathbf{K}(S)$ for every one of the sets $\mathbf{K}(S)$, owing to the fact that the predictable sub-behaviour are valid sub-behaviour of all of these sets.  Furthermore, by construction,  $R_{|M'}(\wp^S) = \wp^{S'}$ and hence   $R_{|M'}(\mathbf{K}(S)) \supset \mathbf{K}(S_{|M'})$ for each of the sets in \ref{prop:restrictionidentity1}-\ref{prop:restrictionidentity5}, as claimed. 
   \end{proof}
\normalfont

The construction of Eq.~\eqref{Equation:restrictedbehaviourandProduct}, which essentially maps behaviour from smaller scenarios to larger ones will turn out to be useful in its own right.  In the next section, we will define such operations formally, and employ them in our general investigation of `partial determinism', which shall be defined in due time.

Proposition \ref{Proposition:restrictionmapidentities} points to an important observation; the defining properties of the sets of behaviour are usually inherited by the images under the restriction map to arbitrary subscenarios. This may be used to formalise the intuition that for the purpose of proving strictness of inclusion relations for these sets via counterexamples, it is sufficient to narrow the attention to simple cases with small numbers of parties or inputs. 

Before showing a concrete example,  let us acknowledge that there is another natural  class of sub-scenarios, namely those scenarios with smaller sets of outputs for the given inputs, which we have not discussed. Such  subscenarios may be associated  with another notion of  'restrictions of a behaviour' which may be understood as the coarse-grainings of outputs of a given set of inputs.  Considering restrictions of  scenarios in the sense of Definition \ref{DefinitionRestrictionofaScenario} and the behaviour in terms of the restriction map $R_{|M'}$ of Definition \ref{Definition:RestrictionofaBehaviour} is of greater significance to our treatment of partial determinism in Section \ref{Section:PartialDeterministicMathsSection} than coarse-grainings of outputs, and hence we omit a detailed description of those 'output-restricted' subscenarios and any associated coarse-graining maps. The relevant observation which is sufficient for our purposes is stated in the remark below.

\remark{It is always possible to add more outcomes of probability zero to a distribution $\wp(\vec
a|\vec{x})$, thus for any behaviour $\wp$ there exists, in a trivial way, a behaviour $\wp'$ in a scenario with larger numbers of outcomes which can be coarse-grained to $\wp$, namely, the $\wp'$ which has zero probability for any additional outcomes, and otherwise the same probabilities as $\wp$. Therefore, while in general we allow for arbitrary (but finite) input sets $O_{x_i}$, for the purpose of specific counter examples we may narrow the outcome sets to any particular sizes, with the understanding that no generality of the conclusions is lost with such a choice.}\\
\normalfont

Finally, we may provide a formal argument to the well known fact that the strict inclusion relations $\mathbf{B}(S) \subsetneq \mathbf{Q}(S) \subsetneq \mathbf{NS}(S)$ shown for the bipartite case in Section \ref{section:GeometricalPerespectiveOnBehaviours} holds generally.

\proposition{\label{Proposition:BellisStrictSubsetofQisstrictofNS}Let $S$ be any nontrivial scenario. Then 
\begin{align}
    \mathbf{B}(S)\subsetneq \mathbf{Q}(S) \subsetneq \mathbf{NS}(S).
\end{align}
}\begin{proof}
Suppose that the inclusions were not strict for all nontrivial scenarios. Then, for some scenario $\widetilde{S}$ either $\mathbf{B}(\widetilde{S}) = \mathbf{Q}(\widetilde{S})$ or $\mathbf{Q}(\widetilde{S}) = \mathbf{NS}(\widetilde{S})$ in which case  also $R_{|M'}(\mathbf{B}(\widetilde{S})) = R_{|M'}(\mathbf{Q}(\widetilde{S}))$ or $R_{|M'}(\mathbf{Q}(\widetilde{S})) = R_{|M'}(\mathbf{NS}(\widetilde{S}))$ for all  restriction maps $R_{|M'}$.

    Since $\widetilde{S}$ is nontrivial, there exists in particular a restriction defined by a collection $M^*$ such that the subscenario $S_{|M^*}$ is bipartite $I_{|M^*} = \{j,j'\}$, with two inputs $x_i \in \{x_i^k, x_i^{k'} \}$  per site. In Section \ref{section:GeometricalPerespectiveOnBehaviours} it was shown that in such scenarios the strict inequalitities $\mathbf{B}(S_{|M^*}) \subsetneq \mathbf{Q}(S_{|M^*}) \subsetneq \mathbf{NS}(S_{|M^*})$ can be established for example by violations of the CHSH inequality, and thus, by Proposition \ref{Proposition:restrictionmapidentities} $R_{|M^*}(\mathbf{B}(S)) \subsetneq R_{|M^*}(\mathbf{Q}(S)) \subsetneq R_{|M^*}(\mathbf{NS}(S))$ contradicting the claim that the inclusions would not be strict. For if some of the sets would be equal, so would their images under the same restriction $R_{|M^*}$. 
\end{proof}
\normalfont

The proof of Proposition \ref{Proposition:BellisStrictSubsetofQisstrictofNS} hinges on the fact that if the images of two behaviour sets differ under a restriction map, the sets themselves have to differ. It is important to emphasize that the converse need not be true, despite the suggestive list of identities shown in Proposition \ref{Proposition:restrictionmapidentities}. This is is a simple consequence of the fact that some restricted scenarios $S_{|M'}$ do not have enough room for a distinction between seemingly different definitions of sets to emerge. We will demonstrate this fact by presenting a proof of another  well-known result, phrased in terms of a restricted scenario.

\proposition{\label{PropositionNSisBell} Let $S_{|M'} = (I_{|M'}, M_{|M'}, O_{|O'})$ be a restriction of some scenario $S$. Then $\mathbf{NS}(S_{|M'}) = \mathbf{B}(S_{|M'})$ if and only if $|M_i'| = 1$ for all but at most one  $i \in I_{|M'}$. }
\begin{proof}
    The `only if' is seen by the technique used in  Proposition \ref{Proposition:BellisStrictSubsetofQisstrictofNS}, by which $R_{|M'}(\mathbf{B}(S)) = \mathbf{B}(S_{|M'}) \subsetneq\mathbf{NS}(S_{|M'}) = R_{|M'}(\mathbf{NS}(S))$ holds whenever $S_{|M'}$ has at least two parties with two inputs each. 

    For the `if part', note first that since $R_{|M'}(\mathbf{B}(S)) = \mathbf{B}(S_{|M'}) \subset \mathbf{NS}(S_{|M'}) = R_{|M'}(\mathbf{NS}(S))$ still holds, it is sufficient to show that also $\mathbf{NS}(S_{|M'}) \subset \mathbf{B}(S_{|M'})$ holds under the given conditions. We will prove this, by using Fine's construction \cite{Fine1982b} to build a join distribution $P(\boldsymbol{\alpha})$ over all the outcomes for all potential inputs, which by Theorem \ref{FinesTHRM} is equivalent to LHV-modelability. 

    Let $\wp^{S_{|M'}} \in \mathbf{NS}(S_{|M'})$ be an arbitrary no-signalling behaviour in the restricted scenario $S_{|M'}$, where exactly one party $j\in I$ has more than one input. The behaviour $\wp^{S_{|M'}}$ therefore consists of $|M_i|$ distributions denoted by$\wp^{S_{|M'}}(\vec{a}'|\vec{x}_{I\setminus \{j\}}'x_j') = \wp^{S_{|M'}}_{x_i}(\vec{a}'|\vec{x}')$, since the inputs $\vec{x}_{I\setminus \{j\}}$ are fixed. Owing to no-signalling, the marginals $\wp_{x_i}^{S_{|M'}}({\vec{a}'_{I\setminus \{j \}}}|\vec{x}') = \wp^{S_{|M'}}(\vec{a}'_{I\setminus \{j\}}|\vec{x}'_{I\setminus \{j\}})$ are all independent of the input $x_j'$ and hence the product distribution  defined by
    \begin{align}
        P(\boldsymbol{\alpha}) = \dfrac{\prod_{x_i \in M_i}\wp_{x_i}^{S_{|M'}}(\vec{a}'|\vec{x}')}{\wp^{S_{|M'}}(\vec{a}'_{I\setminus \{j\}}|\vec{x}'_{I\setminus \{j\}})^{|M_i|-1}},
    \end{align}
    if $\wp^{S_{|M'}}(\vec{a}'_{I\setminus \{j\}}|\vec{x}'_{I\setminus \{j\}}) \neq 0$ and $P(\boldsymbol{\alpha}) = 0$ otherwise, is another well-defined probability distribution. This distribution can be used to  recover any distribution $\wp_{x_j}(\vec{a}'|\vec{x}')$ in the behaviour $\wp^{S_{|M'}}$ by marginalizing over all $\alpha_{x_j}= a'_{x_j}$ such that $x_j \neq x_j'$ and hence $P(\boldsymbol{\alpha})$ satisfies the criteria of Theorem \ref{FinesTHRM}, meaning that $\wp^{S_{|M'}}\in \mathbf{B}(S_{|M'})$ which was the claim. 
\end{proof}
\normalfont

Proposition \ref{PropositionNSisBell} shows that for two distinct sets of behaviours $\mathbf{K}(S) \neq K'(S)$ it may nonetheless be the case that for all nontrivial restriction maps $R_{|M'}(\mathbf{K}(S)) = R_{|M'}(\mathcal{K'}(S))$. Therefore, knowledge of (even all of the) restrictions of a behaviour is not in general sufficient to infer all the potentially relevant properties of the preimage. In section \ref{Section:PartialDeterministicMathsSection}, we introduce the notion of composable sets which, to some extent, have the property of being describable by means of the images of their restrictions. 

\normalfont

\section{Partially Deterministic Polytopes \label{Section:PartialDeterministicMathsSection}}

In this section we introduce a framework to describe, in a general manner, a certain broad class of behaviours. The main concept of interest will be that of partial determinism, which we shall build from the ground up. This leads to the family of mathematical objects  termed `partially deterministic polytopes' the properties of which we shall explore. The term `partially deterministic polytope' is due to Woodhead \cite{Woodhead2014}, who defined and explored these objects in the bipartite case. Similar objects, or special classes thereof, have emerged in different contexts with different motivations. In this work we choose this terminology for two reasons: first, the phrase `partially deterministic' is in our view an apt description of their mathematical features, and second, the terminology does not suggest any particular physical motivation underlying them. In this sense, `partial determinism' may be considered as a convenient unifying mathematical terminology to describe behaviour arising in various circumstances.  

We draw insight from previous works \cite{Woodhead2014, Bong2020, Haddara2025}  which investigated properties of some sub-classes of partially deterministic polytopes.  The first part of the  presentation of Section \ref{section:PartialDeterminismSubsection} is greatly inspired by the techniques of Refs.~\cite{Bong2020, Haddara2025}. The second part,  in particular, is greatly inspired by Woodhead's \cite{Woodhead2014} approach to the topic, and in some cases we provide liberal generalizations and rephrasings of his definitions. This leads us to identify the general notion of `composability' of sets, and to pinpoint partially deterministic polytopes as specific instances of composable sets. 

\subsection{Partial Determinism\label{section:PartialDeterminismSubsection}}

We build our definitions over the basic structure $S= (I,M,O)$. Similarly to the Definitions \ref{DefinitionRestrictionofaScenario} and \ref{Definition:RestrictionofaBehaviour} of restrictions of a scenario and behaviour, respectively, we will use notation such as $M'$ when referring to a collection of subsets $M_i' \subset M_i$, say.  When describing properties of behaviours, we will often also make use of context-dependent subsets such as $F_{\vec{x}}$ introduced for the first time in the proof-sketch of Proposition \ref{Proposition:restrictionmapidentities}.

\definition{(Partially predictable behaviour)\label{PartialPredictabilityDefinition}\newline
Let $S=(I,M,O)$ be a scenario, and $M'$ a collection of subsets $M'_i \subset M_i$ for all $i\in I$. Let $F_{\vec{x}} \subset I$ be a subset of the agents who choose an input $x_i \in M_i'$, so that $F_{\vec{x}} = \{i\in I| x_i \in M_i' \}$. A behaviour is termed partially predictable, or predictable with respect to the collection of inputs $M'$,  if and only if  the marginal on any subset $V_{\vec{x}} \subset {F_{\Vec{x}}}$ is always predictable, that is if and only if  
\begin{align}
   \wp(\Vec{a}_{{V_{\Vec{x}}}} | \Vec{x})  \in \{0, 1 \} \label{PartialPredictabilityEquation},
\end{align}
for all $\Vec{x}$ and all subsets ${V_{\Vec{x}}} \subset F_{\vec{x}}$.  The set  of behaviours partially predictable with respect to $M'$ is denoted by $\mathbf{PP}(S,M')$
}
\\

\normalfont
Note that since the marginals of a predictable distribution can themselves be identified as predictable, which follows from the normalization constraints,  an equivalent definition to that of Definition \ref{PartialPredictabilityDefinition} is to demand that Equation \eqref{PartialPredictabilityEquation} holds only for the marginals of maximal length; namely that Eq.~\eqref{PartialPredictabilityEquation} holds for $V_{\Vec{x}} = F_{\vec{x}}$. As an immediate corollary it also follows that if $M''$ is a collection of subsets with $M_i'' \subset M_i'$ for all $i \in I$  and $\wp \in \mathbf{PP}(S,M')$, then also $\wp \in \mathbf{PP}(S,M'')$. Clearly, that implication cannot in general be inverted. In fact, the following theorem is straightforwardly verified.

\theorem{Let $S = (I,M,O)$ be a scenario and $M'', M'$ arbitrary collections of subsets of $M$. \label{TheoremBasicInclusionTHRMALLPP}
\begin{enumerate}[label=\roman*), ref=\ref{TheoremBasicInclusionTHRMALLPP}.\roman*]
    \item If also $M_i'' \subset M_i'$ for all $i\in I$  then $\mathbf{PP}(S,M') \subset \mathbf{PP}(S,M'') $ \label{TheoremBasicInclusionTHRMALLPP1}
    \item $\mathbf{PP}(S,M') = \mathbf{PP}(S,M'') $ if and only if $M'' = M'$. \label{TheoremBasicInclusionTHRMAllPP2}
\end{enumerate}
}
\proof \

    \begin{enumerate}[label=\roman*), ref=\ref{BasicInclusionTHRMALLPPproof}.\roman*]
    \item We will show, that $\wp \in \mathbf{PP}(S,M') \Rightarrow \wp \in \mathbf{PP}(S,M'') $.

    Let ${F_{\Vec{x}}'} = \{ i \in I : x_i \in M_i' \}$ and ${F_{\Vec{x}}''} = \{ i \in I : x_i \in M_i''  \}$. By hypothesis always ${F_{\Vec{x}}''} \subset {F_{\Vec{x}}'}$. On the other hand from the defining Equation \eqref{PartialPredictabilityEquation}  if $\wp
 \in \mathbf{PP}(S,M') $ then 
    
    \begin{align}
   \wp(\Vec{a}_{{F_{\Vec{x}}}} | \Vec{x}) \in \{0, 1 \} ,
\end{align}
for all $\Vec{x}$, and hence also in particular
\begin{align}
  \sum_{i\in (F_{\Vec{x}}'\setminus F_{\Vec{x}}'')}  \wp(\Vec{a}_{{F_{\Vec{x}}'}} | \Vec{x}) = \wp(\Vec{a}_{F_{\Vec{x}}''}| \Vec{x}) \in \{0, 1 \}
\end{align}
for all $\Vec{x}$,  which means that also $\wp \in \mathbf{PP}(S,M'')$ as claimed.
    \item The `if' part is immediate, while the hypothesis $M'' \neq  M'$  can be easily seen from Definition \ref{PartialPredictabilityDefinition} to lead to a contradiction with the  the claim $\mathbf{PP}(S,M') = \mathbf{PP}(S,M'')$. We provide a proof of the latter claim for completeness. 
    
   Suppose then that $M'' \neq M'$, which means that there is at least one $\tilde{x}_i \in M_i$ such that either $\tilde{x}_i \in M_i'$ and $\tilde{x}_i \notin M_i''$ or  $\tilde{x}_i \notin M_i'$ and $\tilde{x}_i\in M_i''$. say, without loss of generality that the former is the case and let $\wp_{M'}\in \mathbf{PP}(S,M')$ be any partially predictable behaviour with respect to $M'$. Consider the 
   the marginals of $\wp_{M'} \in \mathbf{PP}(S,M')$ and $\wp_{M''} \in \mathbf{PP}(S,M'')$ onto  $F'_{\vec{x}}\cap F''_{\vec{x}} \cup \{i:\in I : x_i = \tilde{x_i} \}:= \widetilde{F}_{\vec{x}} $. Since $\tilde{x_i}\in M_i'$, it follows from  Definition \ref{PartialPredictabilityDefinition} that 
   \begin{align}
       \wp_{M'}(\vec{a}_{\widetilde{F}_{\vec{x}}}|\vec{x}) \in \{0,1\}
   \end{align}
   for all $\vec{x}$.    The same does not hold for all the behaviour $\wp_{M''}\in \mathbf{PP}(S,M'')$. Indeed, consider the behaviour $\wp_{|M''}$ defined by 
   \begin{align}
       \wp_{M''}(\vec{a}|\vec{x}) &= \prod_{i\in F''_{\vec{x}}}D(a_i|x_i)\wp(a_J|\vec{x})\\
       &= \prod_{i\in F''_{\vec{x}}}D(a_i|x_i) \times \dfrac{1}{|O_{\vec{x}_{I\setminus F''_{\vec{x}}}}|}
   \end{align}
   $\forall \vec{x}$ and for all $\vec{a}_J$. By assumption, there is at least one $\vec{x}'$ such that $\vec{x}'_{\{i\}} = \tilde{x_i}$ and so 
   \begin{align}
       \wp_{M''}(\vec{a}_{\{i\}}|\vec{x}') = \dfrac{1}{|O_{\tilde{x}_i}|} \neq \wp_{M'}(\vec{a}_{\{i\}}|\vec{x}') \in \{0,1\},
   \end{align}
which proves the claim
\qed
\end{enumerate}

\normalfont

Theorem \ref{TheoremBasicInclusionTHRMALLPP} highlights that partial predictability can be considered as a strictly hierarchical notion, with a fully predictable behaviour being partially predictable with respect to arbitrary collections of subsets of measurements $M$, while a fully nonpredictable behaviour is not partially predictable with respect to any collection of nonempty subsets.

Definition \ref{PartialPredictabilityDefinition} is fully general. We are particularly interested in the set of partially predictable behaviours that obey the no-signalling constraints.

\definition{(Partially predictable no-signalling behaviour)\label{DefinitionPartialDeterministicNosignallingbehaviours}\\
The set $\mathbf{PP}_{\mathbf{NS}}(S, M')$ of be partially predictable  no-signalling behaviour (wrt. $M'$) is defined as the intersection 
\begin{align}
     \mathbf{PP}_{\mathbf{NS}}(S, M') = \mathbf{PP}(S, M') \cap \mathbf{NS}(S).
\end{align}
}

\normalfont

Using the no-signalling constraints of Def.~\ref{No-signallingbehaviour}, and the defining property of partially predictable behaviour outlined in Equation \eqref{PartialPredictabilityEquation}, one can get a concise form for any no-signalling partially predictable behaviour. Namely, if $\wp \in \mathbf{PP}_{\mathbf{NS}}(S, M')$, then using properties of conditional probabilities it can be seen that 
\begin{align}
    \wp(\Vec{a}|\Vec{x}) &= \wp(\Vec{a}_{F_{\Vec{x}}} | \Vec{x})\cdot \wp(\Vec{a}_{I \setminus {F_{\Vec{x}}}} |\vec{a}_{F_{\vec{x}}},\Vec{x})\\
    &=  D(\Vec{a}_{F_{\Vec{x}}}|\Vec{x}_{F_{\Vec{x}}})\cdot \wp(\Vec{a}_{I \setminus {F_{\Vec{x}}}} |\vec{a}_{F_{\vec{x}}}, \Vec{x})  \\ &=D(\Vec{a}_{F_{\Vec{x}}}|\Vec{x}_{F_{\Vec{x}}})\cdot \wp(\Vec{a}_{I \setminus {F_{\Vec{x}}}} | \Vec{x}) 
    \\ 
    &= D(\Vec{a}_{F_{\Vec{x}}} | \Vec{x}_{F_{\Vec{x}}})\cdot \wp(\Vec{a}_{I \setminus {F_{\Vec{x}}}}| \Vec{x}_{I \setminus F_{\Vec{x}}}),\label{partialPredictableNosignallingequation}
\end{align}
where ${F_{\Vec{x}}} = \{ i \in I | x_i \in M_i' \} $ and  the notation $D(\Vec{a}_F|\Vec{x}_F)$ has been added to  emphasize that this distribution is by assumption predictable and thus completely determined by the measurements ${\Vec{x}}_{F_{\Vec{x}}}$. Note also that since the conditional probability $\wp(\Vec{a}_{I \setminus {F_{\Vec{x}}})} |\vec{a}_{F_{\vec{x}}}, \Vec{x})$ is defined only for events $\vec{a}_{F_{\vec{x}}}$ with nonzero probability, which owing to the predictability of $\vec{a}_{F_{\vec{x}}}$ is unique, the dependency on $\vec{a}_{F_{\vec{x}}}$  can be dropped, leading to  Eq.~\eqref{partialPredictableNosignallingequation}. The terms $\wp(\Vec{a}_{(I \setminus {F_{\Vec{x}}})}|\Vec{x}_{(I \setminus F_{\vec{x}})})$ in Eq.~\eqref{partialPredictableNosignallingequation}, on the other hand,  need not be predictable, but must in particular always satisfy no-signalling. 

Neither of the partially predictable sets $\mathbf{PP}(S,M')$ or $\mathbf{PP}_{\mathbf{NS}}(S,M')$ is convex, by appeal to similar arguments as in the case of the fully predictable sets $\mathbf{P}(S)$ and $\mathbf{P}_{\mathbf{NS}}(S)$, respectively.  They do, however, contain easily identifiable convex subsets of behaviours. To illustrate this, let $\wp^{1}, \wp^{2} \in \mathbf{PP}_{\mathbf{NS}}(S,M')$ be behaviours with $\wp^1(\vec{a}_{F_{\vec{x}}}|\vec{x}_{F_{\vec{x}}}) = \wp^2(\vec{a}_{F_{\vec{x}}}|\vec{x}_{F_{\vec{x}}}) = D(\vec{a}_{F_{\vec{x}}}|\vec{x}_{F_{\vec{x}}})$. Then, for $1 > \omega >0$ the mixture $\omega \wp^{1} + (1-\omega) \wp^{2}$ decomposes as

\begin{align}
\begin{split}
 \label{Equation:mixtureoftwoPartiallyPreditableNS}& \omega \wp^1 (\vec{a}|\vec{x}) + (1-\omega) \wp^{2}(\vec{a}|\vec{x}) =   D(\vec{a}_{F_{\vec{x}}}|\vec{x}_{F_{\vec{x}}}) \\ 
 & \times  \left[ \omega \wp^1(\vec{a}_{F_{\vec{x}}\setminus I}|\vec{x}_{F_{\vec{x}}\setminus I}) + (1-\omega) \wp^2(\vec{a}_{F_{\vec{x}}\setminus I}|\vec{x}_{F_{\vec{x}}\setminus I}) \right]
  \end{split}
\end{align}
which is likewise an element of $\mathbf{PP}_{\mathbf{NS}}(S,M')$. Evidently,  behaviours which share the same `predictable component' in this manner form such subsets. We will later use this fact in order to establish extremality results. 

While the set  $\mathbf{PP}(S,M')$ is not convex, its convex hull is not that interesting, since by Theorem \ref{TheoremBasicInclusionTHRMALLPP1} \emph{any} partially predictable set contains in particular, the set $\mathbf{P}(S)$ which corresponds to the partially predictable set $\mathbf{PP}(S,M')$ with  $M' =M$. This entails that $\mathrm{conv}[\mathbf{PP}(S,M')] = \mathbf{E}(S)$ for any collection $M'$. 

The situation is different for the convex hull $\mathrm{conv}(\mathbf{PP}_{\mathbf{NS}}(S, M'))$. We shall further characterize this set by generalizing the arguments used in  \cite{Bong2020} and \cite{Haddara2025}. 

First, let us observe that Eq.~\eqref{partialPredictableNosignallingequation} can be massaged into a useful form as a mixture of partially predictable behaviours with the same predictable component analogously to  Eq.~\eqref{Equation:mixtureoftwoPartiallyPreditableNS}.  Specifically,  since the behaviour obeys no-signalling, it must be the case that the marginals $\wp(\vec{a}_{I\setminus F_{\vec{x}}}| \Vec{x}_{I \setminus F_{\Vec{x}}})$ on the right side of Equation \eqref{partialPredictableNosignallingequation} can be recovered from \emph{any} distribution $\wp(\Vec{a}|\Vec{x}) $ in the behaviour where the substring $\Vec{x}_{I \setminus F_{\Vec{x}}}$ is present in $\Vec{x}$. Hence, in particular, those marginals ought to be recoverable from a distribution from any such suitable context $\vec{x}$ which is `maximal' in the sense that the set $I\setminus F_{\vec{x}} \subset I$ attains its maximum cardinality. To formalize this, let $ND = \{ \vec{x} : |I\setminus F_{\vec{x}}| = \max_{\vec{x}'}|I\setminus F_{\vec{x}'}|\} $ and let $\wp^{|ND}({\vec{a}|\vec{x}})$ represent any distribution of the behaviour with $\vec{x}\in ND$. Note that $ND = \emptyset$ if and only if $M_i' = M_i$ for all $i\in I$, which corresponds to the situation where Eq.~\eqref{partialPredictableNosignallingequation} has only distributions $D(\vec{a}|\vec{x})$ in the expansion.  In any other case, marginals of $\wp(\vec{a}|\vec{x})$ onto $I\setminus F_{\vec{x}}$ for any $\vec{x}$ satisfy

\begin{align}
    \wp(\vec{a}_{I\setminus F_{\vec{x}}}|\vec{x}_{I\setminus F_{\vec{x}}}) = \wp^{|ND}(\vec{a}_{I\setminus F_{\vec{x}}}|\vec{x}_{I\setminus F_{\vec{x}}}). \label{Equation:marginalfromNDequation}
\end{align}

Which is to say that information about the distributions in the subset $ND$ of contexts is sufficient to specify the marginals onto $I \setminus F_{\vec{x}}$ for any $\vec{x}$.  In fact, even the distributions $\wp^{|ND}(\vec{a}|\vec{x})$ may contain some redundant information for this purpose. Namely if $M_i' = M_i$ for some $i\in I$ then the outcomes $a_i$ for $i\in F_{\vec{x}}$ get marginalized over anyway. Let us then draw attention to the collection of marginal distributions of the following form
\begin{align}
    \sum_{i: M_i' = M_i }\wp^{|ND}(\vec{a}|\vec{x}) = \wp^{|ND} (\vec{a}_{I_{|M'^{\perp}}}|\vec{x}_{I_{|M'^{\perp}}}), \label{Equation:NDisBasicallyTheImage}
\end{align}
where $I_{|M'^{\perp}} = \{ i\in I : M_i' \neq M_i \}$. This set of marginals is complete in the sense that from information contained in them one can by construction  recover any $\wp(\vec{a}_{I\setminus F_{\vec{x}}}|\vec{x}_{I\setminus F_{\vec{x}}})$ in the sense of Eq.~\eqref{Equation:marginalfromNDequation}. Furthermore, these marginals form a minimal necessary collection as the information contained in them is also needed to specify all the $\wp(\vec{a}_{I\setminus F_{\vec{x}}}|\vec{x}_{I \setminus F_{\vec{x}}})$. To see this, one may consider the marginals of length $\max_{\vec{x}'}|I\setminus F_{\vec{x}}| = |I_{|M'^{\perp}}|$.

By examination of  Definition \ref{Definition:RestrictionofaBehaviour} of the restriction of a behaviour it can be verified that the set of distributions $\wp^{|ND} (\vec{a}_{I_{|M'^{\perp}}}|\vec{x}_{I_{|M'^{\perp}}})$ in Eq.~\eqref{Equation:NDisBasicallyTheImage} is in fact  identifiable precisely with the image of the behaviour $\wp^{S}\in \mathbf{PP}_{\mathbf{NS}}(S,M'^{\perp})$ under the restriction map $R_{|M'^{\perp}}$, with  $M'^{\perp}$ denoting the collection  of subsets $M_i^{\perp} = M_i \setminus M_i'$ for all $i\in I$. This collection maps the scenario $S = (I,M,O)$ to the substructure $S_{|M'^{\perp}} = (I_{|M'^{\perp}}, M_{|M'^{\perp}}, O_{M'^{\perp}}) $ in the sense of Definition \ref{DefinitionRestrictionofaScenario}.  Since the only constraint on these distributions in the subscenario $S_{|M'^{\perp}}$ is that they obey no-signalling, the sub-behaviour $R_{|M'^{\perp}}(\wp^{S})= \wp^{S_{|M'^{\perp}}}$ consisting of the distributions $ \wp^{S_{|M'^{\perp}}}(\vec{a}'' | \vec{x}'')$ can always be expanded in terms of the finite number of extreme points of the no-signalling polytope $\mathbf{NS}(S_{|M'^{\perp}})$. That is, one may write
\begin{align}
  \label{Equation:NDisexpandedExtPointsofSubscenario}  \wp^{|ND} (\vec{a}_{I_{|M'^{\perp}}}|\vec{x}_{I_{|M'^{\perp}}}) &= \wp^{S_{|M'^{\perp}}}(\vec{a}'' = a_{I_{|M'^{\perp}}} | \vec{x}'' = \vec{x}_{I_{|M'^{\perp}}})\\
    = \sum_{k}l_k P^{S_{|M'^{\perp}}}_{k}&(\vec{a}'' = \vec{a}_{I_{|M'^\perp }} | \vec{x}'' = \vec{x}_{I_{|M'^{\perp}}}),
\end{align}
where $l_k \in [0,1]$ gives the convex weight of the extremal no-signalling behaviours indexed by $k$, so that $\sum_k l_k = 1$.  Using  $\vec{a}_{I\setminus F_{\vec{x}}} = (\vec{a}_{I_{|M'^{\perp}}})_{F_{I \setminus \vec{x}}}$ and $\vec{x}_{I \setminus F_{\vec{x}}} = (\vec{x}_{I_{|M'^{\perp}}})_{I\setminus F_{\vec{x}}}$ which are valid for all $\vec{x}\in \vec{M}$,  we may  adopt the notational convention 
\begin{align}
\begin{split}
  P^{S_{|M'^{\perp}}}_k(\vec{a}_{I \setminus F_{\vec{x}}}|\vec{x}_{I\setminus F_{\vec{x}}})   \\=  \sum_{i\in F_{\vec{x}}} P^{S_{|M'^{\perp}}}_{k}&(\vec{a}'' = \vec{a}_{I_{|M'^\perp }} | \vec{x}'' = \vec{x}_{I_{|M'^{\perp}}}), \label{Equation:NotationalConvenience}
  \end{split}
\end{align}
which tracks the matching of the marginals $\vec{a}_{I
 \setminus F_{\vec{x}}}'' = (\vec{a}_{I_{|M'^{\perp}}})_{I
 \setminus F_{\vec{x}}} = \vec{a}_{I \setminus F_{\vec{x}}}$. This  should lead to no confusion since the substring $I\setminus F_{\vec{x}}$ consists of precisely the parties with inputs in $M'^{\perp}$ and well-defined since $P_k^{S_{|M'^{\perp}}}\in \mathbf{NS}(S_{|M'^{\perp}})$.
 
 Putting the identities in  Eqs.~\eqref{Equation:NDisexpandedExtPointsofSubscenario}-\eqref{Equation:NotationalConvenience} back to Eq.~\eqref{partialPredictableNosignallingequation}
it follows that \emph{any} partially predictable behaviour $\wp^{S}\in \mathbf{PP}(S,M')$ admits a decomposition of the form 

\begin{align}
    \wp^S(\vec{a}|\vec{x}) &=  D(\vec{a}_{F_{\vec{x}}}|\vec{x}_{F_{\vec{x}}}) \wp^{S_{|M'^{\perp}}}(\vec{a}_{I\setminus F_{\vec{x}}}|\vec{x}_{I\setminus F_{\vec{x}}})  \\
   &=  \sum_{k}l_k D(\vec{a}_{F_{\vec{x}}}|\vec{x}_{F_{\vec{x}}}) P^{S_{|M'^{\perp}}}_{k}(\vec{a}_{I\setminus F_{\vec{x}}}|\vec{x}_{I\setminus F_{\vec{x}}}). \label{Equation:PartiallypredictableExpandedeXTREMEOFns}
\end{align}

One can also run a similar argument for the maximal  contexts  $D = \{ \vec{x} : |F_{\vec{x}}| = \max_{\vec{x}'}|F_{\vec{x}'}| \}$. Now owing to no-signalling, the marginal distributions of the form $\wp(\Vec{a}_{F_{\Vec{x}}}|\Vec{x}_{F_{\Vec{x}}})$ in Eq.~\eqref{partialPredictableNosignallingequation} have to be recoverable from a distribution in any context  $\Vec{x} \in D$. Note also that $D = \emptyset$ if and only if $M_i' = \emptyset$ for all $i \in I$. We can then express any $D(\vec{a}_{F_{\vec{x}}}|\vec{x}_{F_{\vec{x}}})$
as 

\begin{align}
    D(\vec{a}_{F_{\vec{x}}}|\vec{x}_{F_{\vec{x}}}) = D^{|D}(\vec{a}_{F_{\vec{x}}}|\vec{x}_{F_{\vec{x}}}), \label{Equation:sAMEHoldsfORmaximald}
\end{align}
where the superscript $|D$ has been added on the right hand side to emphasize that one may take any distribution from the maximal contexts $\vec{x}\in D$ while no such constraint is placed on the left side. By similar steps, if $M_i' =\emptyset$ for some $i\in I$, the relevant information is seen to be contained in the collection of the marginals 

\begin{align}
    \sum_{i: M_i' = \emptyset} D^{|D}(\vec{a}|\vec{x}) = D^{|D}(\vec{a}_{I_{|M'}}|\vec{x}_{I_{|M'}}),
\end{align}
where $I_{|M'} = \{ i \in I : M_i' \neq \emptyset \}$. In any case, these correspond to the image of $\wp^{S}$ under the restriction $R_{|M'}$ defined over the substructure $S_{|M'} = (I_{|M'}, M_{M'}, O')$ so that one may write 
\begin{align}
    D^{|D}(\vec{a}_{I_{|M'}}|\vec{x}_{I_{|M'}}) = D^{S_{|M'}}(\vec{a}' = \vec{a}_{I_{|M'}}|\vec{x}'= \vec{x}_{I_{|M'}}).
\end{align} The sub-behaviour $\wp^{S_{|M'}}$ are clearly elements in the finite set $\mathbf{P}_{\mathbf{NS}}(S_{|M'})$ of predictable no-signalling behaviour defined over $S_{|M'}$. It will be useful to add the label $r$ to indicate each distinct sub-behaviour of this kind, so that for $r \neq r'$ it holds that $\exists \vec{x}' \in M_{|M'}$ with
\begin{align}
    D^{S_{|M'}}_{r}(\vec{a}'|\vec{x}') \neq D^{S_{|M'}}_{r'}(\vec{a}'|\vec{x}') \label{Equation:DistinctPredictableSubBehaviour}
    \end{align} for some $\vec{a}'\in O_{\vec{x}'}$. 

Again, we adopt the notational convention 
\begin{align}
  D^{S_{|M'}}(\vec{a}_{F_{\vec{x}}}|\vec{x}_{F_{\vec{x}}}) = \sum_{i\in I \setminus F_{\vec{x}}} D^{S_{|M'}}(\vec{a}' = \vec{a}_{I_{|M'}}|\vec{x}'= \vec{x}_{I_{|M'}}), \label{Equation:NotationalconventionforD} 
\end{align}
which simply tracks the matching of the marginals $\vec{a}_{F_{\vec{x}}}' = (\vec{a}_{I_{|M'}})_{F_{\vec{x}}} = \vec{a}_{F_{\vec{x}}}$ and $\vec{x}_{F_{\vec{x}}}' = (\vec{x}_{I_{|M'}})_{F_{\vec{x}}} = \vec{x}_{F_{\vec{x}}}$, which is well defined owing to no-signalling. 

    Combining the steps between Eqs.~\eqref{Equation:PartiallypredictableExpandedeXTREMEOFns}-\eqref{Equation:NotationalconventionforD} it is then found that any behaviour $\wp^{S}\in \mathbf{PP}_{\mathbf{NS}}(S,M')$ admits a decomposition in terms of the finite numbers of parameters as 
    \begin{align}
        \wp^{S}(\vec{a}|\vec{x}) = \sum_{k}l_{k} D_r^{S_{|M'}}(\vec{a}_{F_{\vec{x}}}|\vec{x}_{F_{\vec{x}}}) P^{S_{|M'^{\perp}}}_k(\vec{a}_{I\setminus F_{\vec{x}}}|\vec{x}_{I\setminus F_{\vec{x}}}),\label{Equation:expansionOfpARTIALLYpredictableNS}
    \end{align}
where $r$ labels a predictable sub-behaviour in $S_{|M'}$ and $k$ an extremal point of the no-signalling polytope $\mathbf{NS}(S_{|M'^{\perp}})$ in the sub-scenario $S_{|M'^{\perp}}$. Comparison with Eq.~\eqref{Equation:mixtureoftwoPartiallyPreditableNS} allows us to state that not only is the subset of behaviours with the same predictable component (determined by $r$) convex, but it is in fact a convex polytope, since any such partially predictable behaviour with fixed $r$ can be decomposed as a convex sum in terms of the same $|\mathrm{Ext}(\mathbf{NS}(S_{|M'^{\perp}}))|$ points $k$. 

We can immediately extract the relevant features of $\mathrm{conv}[\mathbf{PP}_{\mathbf{NS}}(S,M')]$ as well. From Eq.~\eqref{partialPredictableNosignallingequation} and \eqref{Equation:expansionOfpARTIALLYpredictableNS} it follows that  any mixture $\wp^{S} = \sum_{j}\omega_j \wp^{S}_j \in \mathrm{conv}(\mathbf{PP}_{\mathbf{NS}}(S, M'))$, with $\omega_j \geq 0$ $\forall j \in J$ and $\sum_{j} \omega_j = 1$,   decomposes  as

\begin{align}
    \wp^S(\Vec{a}|\Vec{x}) &= \sum_{j} \omega_j \wp^S_j(\vec{a}|\vec{x}) \\
    &=\sum_{j} \omega_j D_j(\Vec{a}_{F_{\Vec{x}}}|\Vec{x}_{F_{\Vec{x}}}) \cdot \wp_{j} (\Vec{a}_{I \setminus F_{\Vec{x}}}|\Vec{x}_{I\setminus F_{\Vec{x}}})\\
    \begin{split}
    &= \sum_{j,r,k} \omega_j l_k(\omega_j)s_r(\omega_j) D_r^{S_{|M'}}(\vec{a}_{F_{\vec{x}}}|\vec{x}_{F_{\vec{x}}}) \\
    & \hspace{0.5cm} \times P^{S_{|M'^{\perp}}}_k(\vec{a}_{I \setminus F_{\vec{x}}}|\vec{x}_{I\setminus F_{\vec{x}}}) \label{EqconvexPPNSEq}
     \end{split}
\end{align}
 where $l_k(\omega_j) \in [0,1],  s_{r}(\omega_j)\in \{0,1\}$ for all $k,r,j$ and $\sum_{j,k,r}\omega_j l_k(\omega_j)s_r(\omega_j) = 1$. The quantity $l_k(\omega_j)$ gives the  weight of the extremal no-signalling sub-behaviour over the scenario $S_{|M'^{\perp}}$ labelled by $k$  and similarly $s_{r}(\omega_j)$ gives the  weight of the predictable sub-behaviour over the scenario $S_{|M'}$ labelled by $r$ in the partially predictable distribution labelled by $j$, respectively. 

Inspection of Eq.~\eqref{EqconvexPPNSEq} allows us to verify the following theorem.

\theorem{The set $\mathrm{conv}[\mathbf{PP}_{\mathbf{NS}}(S, M')]$ is a convex polytope. \label{convPPNSisapolytopeTheorem}}
\begin{proof} 
The set $\mathrm{conv}[\mathbf{PP}_{\mathbf{NS}}(S, M')]$ is convex by definition. To prove that it is a polytope,  it is sufficient to show that the set can  in fact be defined as the convex hull of a finite number of points. This follows immediately from previous considerations. 

Namely, in the steps between Eqs.~\eqref{partialPredictableNosignallingequation}-\eqref{EqconvexPPNSEq} it was shown that any $\wp\in \mathrm{conv}[\mathbf{PP}_{\mathbf{NS}}(S,M')]$ can be given a convex decomposition in terms of the finite numbers of parameters $r,k$, which label the extreme points of the set of predictable no-signalling behaviours defined over the subscenario $S_{|M'}$ and the extreme points of the no-signalling polytope defined over the subscenario $S_{|M'^{\perp}}$:
\begin{align}
    \begin{split}
     \wp(\Vec{a}|\Vec{x})  &= \sum_{j,r,k} \omega_j l_k(\omega_j)s_r(\omega_j) D_r^{S_{|M'}}(\vec{a}_{F_{\vec{x}}}|\vec{x}) \\
    & \hspace{0.5cm} \times P^{S_{|M'^{\perp}}}_k(\vec{a}_{I \setminus F_{\vec{x}}}|\vec{x}_{I\setminus F_{\vec{x}}}) \\
    \end{split}
   \\[2ex]
   \begin{split}
    &= \sum_{r,k}t_{rk} D_r^{S_{|M'}}(\vec{a}_{F_{\vec{x}}}|\vec{x}_{F_{\vec{x}}})  P^{S_{|M'^{\perp}}}_k(\vec{a}_{I \setminus F_{\vec{x}}}|\vec{x}_{I \setminus F_{\vec{x}}})\label{Equation:ExpansionOfPDPintermsofpredictableandNSpoints},   
   \end{split}
\end{align}
where the quantities $t_{rk} \coloneq \sum_j \omega_j l_k(\omega_j)s_r(\omega_j)$ are by construction positive, and sum to one. 
\end{proof} 
\normalfont
The extreme points of this set are obtained for free. 

\theorem{$ $\label{TheoremextremepointsofPDP} The extreme points of $\mathrm{conv} [\mathbf{PP}_{\mathbf{NS}}(S,M')]$ are precisely the behaviours for which $t_{rk} \in \{0,1\}$ for all $r,k$ in the decomposition of Eq.~\eqref{Equation:ExpansionOfPDPintermsofpredictableandNSpoints}. That is,
\begin{align}
\begin{split}
   \wp &\in  \mathrm{Ext}(\mathrm{conv}[\mathbf{PP}_{\mathbf{NS}}(S,M'))]) \Longleftrightarrow \\
     \wp(\vec{a}|\vec{x}) &= D_r^{S_{|M'}}(\vec{a}_{F_{\vec{x}}}|\vec{x}_{F_{\vec{x}}})  P^{S_{|M'^{\perp}}}_k(\vec{a}_{I \setminus F_{\vec{x}}}|\vec{x}_{I \setminus F_{\vec{x}}}),
     \end{split}
\end{align}
for some $r,k$. 
}
\begin{proof}
    That the equivalence holds in the direction ``$\Rightarrow$'' was essentially established in the proof of Theorem \ref{convPPNSisapolytopeTheorem}. Hence, it is sufficient to show that for every pair $(r,k)$ such points are indeed extremal.

    Suppose then that this was not the case for some pair $(r=R,k=K)$. Then for $1>t_{r,k}>0, \sum_{r,k}t_{rk} = 1$, $(r,k) \neq (R,K)$  a decomposition 
    \begin{align}
      \wp_{RK}^{S}(\vec{a}|\vec{x}) = D^{S_{|M'}}_{R}(\vec{a}_{F_{\vec{x}}}|\vec{x}_{F_{\vec{x}}}) \times P^{S_{|M'^\perp}}_\mathbf{K}(\vec{a}_{I \setminus F_{\vec{x}}}|\vec{x}_{I\setminus F_{\vec{x}}}) \\
        = \sum_{(r,k) \neq (R,K)} t_{rk} D^{S_{|M'}}_{r}(\vec{a}_{F_{\vec{x}}}|\vec{x}_{F_{\vec{x}}}) \times P^{S_{|M'^\perp}}_k(\vec{a}_{I \setminus F_{\vec{x}}}|\vec{x}_{I\setminus F_{\vec{x}}}) \label{EquationConvexcombiofExtremepointscontradictionEQ}
    \end{align}
    exists for $\wp_{RK}^{S} \in \mathrm{Ext}(\mathrm{conv}[\mathbf{PP}_{\mathbf{NS}}(S,M'))])$. By no-signalling, therefore in particular
    \begin{align}
    \wp_{RK}^{S}(\vec{a}_{F_{\vec{x}}}|\vec{x}_{F_{\vec{x}}}) = D^{S_{|M'}}_R(\vec{a}_{F_{\vec{x}}}|\vec{x}_{F_{\vec{x}
        }})\\
        = \sum_{(r,k)} t_{rk} D^{S_{|M'}}_r(\vec{a}_{F_{\vec{x}}}|\vec{x}_{F_{\vec{x}}}),
    \end{align}
    which can only hold for all $F_{\vec{x}}$ if $t_{rk}= \delta_{r,R}t_{rk}= t_{Rk}$ for all $k$, since by construction $r$ labels the distinct deterministic extreme points of the restricted scenario $S_{|M'}$ which in turn uniquely determine the 'maximal' marginals for which $F_{\vec{x}} = I_{|M'}$.  This property fixes a significant portion of the decomposition of any behaviour in $\wp_{RK}^{S} \in \mathrm{Ext}(\mathbf{PD}(S,M'))$.

    Indeed, by a similar argument, one finds that the remaining degrees of freedom are in turn fixed by the properties of the marginals on the $I\setminus F_{\vec{x}}$ and the sub-behaviour $S_{|M'^{\perp}}$. By marginalizing over $F_{\vec{x}}$ in Eq.~\eqref{EquationConvexcombiofExtremepointscontradictionEQ} one gets
    \begin{align}
        P^{S_{{|M'^{\perp}}}}_\mathbf{K}(\vec{a}_{I\setminus F_{\vec{x}}}|\vec{x}_{I\setminus F_{\vec{x}}}) = \sum_{k}t_{Rk} P^{S_{|M'^{\perp}}}_{k}(\vec{a}_{I\setminus F_{\vec{x}}}|\vec{x}_{I\setminus F_{\vec{x}}})\label{Equation:marginalsofNSsubscenarioareeXTREMELAINpdpPROOF}
    \end{align}
    for all $F_{\vec{x}}$. In particular, since Eq.~\eqref{Equation:marginalsofNSsubscenarioareeXTREMELAINpdpPROOF} has to hold for the cases where $I\setminus F_{\vec{x}} = I_{|M'^{\perp}}$  it must be that in fact $t_{Rk}= \delta_{k,K}t_{Rk}$ since each $k$ labels by definition a distinct extreme point of $\mathbf{NS}(S_{|M'^{{\perp}}}$). Thus, each $(r,k)$ labels a unique extreme point of $\mathrm{conv}[\mathbf{PP}_{\mathbf{NS}}(S,M')]$ as claimed.    
 \end{proof}
\normalfont
Theorems \ref{convPPNSisapolytopeTheorem} and \ref{TheoremextremepointsofPDP} lead us to define the main object of interest of this work.

\definition{(Partially Deterministic behaviour\footnote{\label{footnote:PDnotLocalPD}Here we defined Partially Deterministic behaviour as those in the convex hull of the no-signalling partially predictable behaviour. If no-signalling is not enforced, that convex hull simply equals the set of all behaviour $\mathbf{E}(S)$. Since that case may be considered trivial, there is no loss of generality in defining Partial Determinism in this manner. })\\ \label{PartiallyDeterministicPolytopeDefinition}
The set $\mathbf{PD}(S,M')$ of behaviours partially deterministic with respect to the subset of inputs $M'$ is defined as $\mathbf{PD}(S,M') = \mathrm{conv}(\mathbf{PP}_{\mathbf{NS}}(S, M'))$. By Theorem $\ref{convPPNSisapolytopeTheorem}$ this set is a polytope, which is referred to as the Partially Deterministic Polytope (w.r.t. $M'$). Any $\wp \in \mathbf{PD}(S,M')$ admits a `partially deterministic model' meaning that the distributions $\wp(\vec{a}|\vec{x})$ decompose as 
\begin{align}
    \wp(\Vec{a}|\Vec{x}) = \sum_{j} \omega_j D_j(\vec{a}_{F_{\vec{x}}}|\vec{x}_{F_{\vec{x}}}) \wp_j(\vec{a}_{I \setminus F_{\vec{x}}}|\vec{x}_{I \setminus F_{\vec{x}}}), \label{Equation:partialdeterminismDefiningEQ.}
\end{align}
with $\omega_j \geq 0$ for all $j\in J$, $\sum_{j}\omega_j = 1$, $F_{\vec{x}} = \{i \in I | x_i \in M_i' \}$, $D_j(\vec{a}_{F_{\vec{x}}}|\vec{x}_{F_{\vec{x}}}) \in \{0,1\}$ for all $i\in I, \vec{a}_{F_{\vec{x}}}, \vec{x}$ and the distributions $\wp_j(\vec{a}_{I\setminus F_{\vec{x}}}| \vec{x}_{I\setminus F_{\vec{x}}})$ being compatible with no-signalling for all $j\in J$.

A partially deterministic model of the type in Eq.~\eqref{Equation:partialdeterminismDefiningEQ.} can be further expanded in terms of mixtures  of predictable behaviour $D_r^{S_{|M'}}$ over the substructure $S_{|M'}$ labelled by $r$ and the extreme points $P_k^{S_{|M'^{\perp}}}$ of the no-signalling polytope defined over $S_{|M'^{\perp}}$ labelled by $k$ so that any $\wp \in \mathbf{PD}(S,M')$ satisfies
\begin{align}
   \wp(\vec{a}|\vec{x}) = \sum_{r,k} t_{rk} D^{S_{|M'^{\perp}}}_{r}(\vec{a}_{F_{\Vec{x}}}| \Vec{x}_{F_{\Vec{x}}})P^{S_{|M'^{\perp}}}_{k}(\Vec{a}_{I \setminus F_{\Vec{x}}}|\Vec{x}_{I \setminus F_{\Vec{x}}}), \label{equation:partialdeterminismDefiningEq.EXTpoints}
\end{align}
with $t_{rk} \geq 0$ for all $r,k$ and $\sum_{rk}t_{rk} = 1$.
}\\
\normalfont
Let us again emphasize that a conceptual distinction between the existence of a model of a certain kind, and the operationally clear meaning of the convex combinations as mixing  of preparations which lead to partially predictable behaviours should be made. Indeed, mathematically Definition \ref{PartiallyDeterministicPolytopeDefinition} is equivalent to demanding that the behaviour $\wp$ has a model in terms of some local partially deterministic ontic variables $\lambda\in \Lambda$ which in the space-time relations of the Bell-scenario imply that
\begin{align}
    \wp(\vec{a}|\vec{x}) = \sum_{\lambda}P(\lambda)D(\vec{a}_{F_{\vec{x}}}|\vec{x}_{F_{\vec{x}}}, \lambda)P(\vec{a}_{I\setminus F_{\vec{x}}}|\vec{x}_{I\setminus F_{\vec{x}}}, \lambda).
\end{align}
The term `local' can be dropped in this case as well as stated in footnote \ref{footnote:PDnotLocalPD}, since nonlocal partially deterministic models are just arbitrary nonlocal models, which can be used to express any behaviour.

The two familiar polytopes $\mathbf{B}(S)$ and $\mathbf{NS}(S)$ are obtained as the two extreme cases of Definition \ref{PartiallyDeterministicPolytopeDefinition}. 

\theorem{\label{Theorem:BellandNSspecialcasesofPDP}Let $S= (I,M,O)$ be a scenario and $M'$  a collection of subsets of inputs with  $M_i'\subset M_i.$
\begin{enumerate}[label=(\roman*), ref=\ref{Theorem:BellandNSspecialcasesofPDP}.\roman*] 
    \item If $ M_i'= \emptyset$ for all $i\in I$ then $\mathbf{PD}(S,M') = \mathbf{NS}(S)$. 
    \item If $M_i' = M_i $ for all $i\in I$, then $\mathbf{PD}(S,M') = \mathbf{B}(S)$
\end{enumerate}
}
\begin{proof}
See, for example, Eq.~\eqref{equation:partialdeterminismDefiningEq.EXTpoints} in Definition \ref{PartiallyDeterministicPolytopeDefinition}. Cases $(i)$ and $(ii)$ correspond to the situation where $S_{|M'^{\perp}}$ = $S$ and $S_{|M'} = S$, respectively.  
\end{proof}
\normalfont

Theorem \ref{TheoremextremepointsofPDP} establishes the form of the extreme points of any partially deterministic polytope $\mathbf{PD}(S,M')$ in terms of predictable no-signalling behaviours in the restricted scenario $S_{|M'}$ and the extreme points of no-signalling behaviours in the scenario $S_{|M'^{\perp}}$. Hence, the vertex representation of $\mathbf{PD}(S,M')$ is essentially known. Unfortunately  this suggests that if a complete list of all the vertices of $\mathbf{PD}(S,M')$ is wanted, a vertex representation of the polytope $\mathbf{NS}(S_{|M'^{\perp}})$ needs to be solved first, which in general is computationally demanding \cite{Pironio2011}. Fortunately, we can say quite a bit about the structure of various partially deterministic polytopes $\mathbf{PD}(S,M')$ purely by analytical techniques.

Let us begin by pointing out some instantaneous consequences of the definitions considered so far.

\begin{theorem} \label{TheoremPDpolytopeBasicInclusionTHRM}
   Let $S = (I,M,O)$ be a scenario.  If $M_i'' \subset M_i' \subset M_i$, for all $i\in I$ then $\mathbf{PD}(S,M') \subset \mathbf{PD}(S,M'') $.
\end{theorem}
\begin{proof}
    Follows immediately from Theorem \ref{TheoremBasicInclusionTHRMALLPP} by which $M_i'' \subset M_i'$ $\forall i\in I$ $ \Rightarrow $ $\mathbf{PP}(S,M') \subset \mathbf{PP}(S, M'')$ and hence also 
         $\mathrm{conv}(\mathbf{PP}(S,M')\cap \mathbf{NS}(S)) \subset \mathrm{conv}(\mathbf{PP}(S,M'')\cap \mathbf{NS}(S))$ as claimed. 
\end{proof}
Theorem \ref{TheoremPDpolytopeBasicInclusionTHRM} can be used to establish the following result.

\normalfont

\theorem{Let $S= (I, M, O)$ be a scenario and  $M'$ an arbitrary collection of subsets $M_i' \subset M_i$ for all $i\in I$. Then $\mathrm{dim}(\mathbf{B}(S)) = \mathrm{dim}(\mathbf{PD}(S,M')) = \mathrm{dim}(\mathbf{NS}(S))$.  \label{AffineHullThrm}}
\begin{proof}
    In Theorem \ref{Theorem:BellandNSspecialcasesofPDP}  it was established that if $M^{\emptyset}$ is a collection of inputs with $M_i^\emptyset = \emptyset$ for all $i\in I$ then $\mathbf{PD}(S, M^\emptyset) = \mathbf{NS}(S)$ and that $\mathbf{PD}(S,M) = \mathbf{B}(S)$. Since  every collection  $M'$ of inputs of the parties satifies $\emptyset \subset M_i' \subset M_i$ for all $i\in I$, it follows  by virtue of Theorem \ref{TheoremPDpolytopeBasicInclusionTHRM} that $\mathbf{B}(S) \subset \mathbf{PD}(S,M') \subset \mathbf{NS}(S)$. This proves the claim when paired with Pironio's affine hull theorem (Theorem \ref{Theorem:Pironio'sTheorem}), which establishes that $\mathrm{dim}(\mathbf{B}(S)) = \mathrm{dim}(\mathbf{NS}(S)) $ since both those sets are constrained by the same the no-signalling and normalization requirements \cite{Pironio2005}.
\end{proof}
\normalfont

Notably the conditions for the subset relations in Theorem \ref{TheoremPDpolytopeBasicInclusionTHRM} to be strict are  more nuanced than one would perhaps immediately expect from Theorem \ref{TheoremBasicInclusionTHRMAllPP2}. Specifically, in previous works \cite{Woodhead2014, Bong2020, Haddara2025}  classes of nontrivial situations\footnote{The mathematical results in Refs.~\cite{Bong2020, Haddara2025},  were framed in the context of better understanding the Local Friendliness no-go result building on certain kinds of extensions of the Wigner's friend scenario \cite{Bong2020}. In Sec.~\ref{Section:GeneralSequentialWignerSection} we show how partially deterministic polytopes arise in the sequential Wigner's friend scenarios introduced in Ref.~\cite{Utreras-Alarcon2024} as physically motivated constraints.} where $M'' \neq M'$ but $\mathbf{PD}(S,M'')=\mathbf{PD}(S,M') = \mathbf{B}(S)$ have been found. We will later strengthen Theorem \ref{TheoremPDpolytopeBasicInclusionTHRM} by establishing precisely the cases when the inclusions are strict by finding all distinct equivalence classes of partially deterministic polytopes. 

Let us observe that Theorem \ref{AffineHullThrm} uses the fact that $\mathbf{B}(S) \subset \mathbf{PD}(S,M') \subset \mathbf{NS}(S)$. It is in fact possible to prove a strictly stronger relation concerning the extreme points of those sets. For this purpose, the following lemma will be useful.

\lemma{\label{LemmaExtremepointsimplication}Let $A,B$ and $C$ be  closed and bounded convex subsets of $\mathbb{R}^{d}$ with $A \subset B \subset C$. If $p \in \mathrm{Ext}(A)$ and $p \in \mathrm{Ext}(C)$ then also $p\in \mathrm{Ext}(B)$.}
\begin{proof}
   From the closedness of $A$ and $C$, it follows that $p\in A, C$. Therefore, by $A\subset B $ also $p\in B$. We will show, that $p$ is also an extreme point of $B$. 
   
   Suppose for the sake of contradiction, that $p\notin \mathrm{Ext}(B)$. If $p$ is not extremal, from the convexity of $B$ it follows that there are some points $b_1, b_2 \in B$ such that for $1>\omega >0$
   \begin{align}
      p =  \omega b_1 + (1-\omega)b_2.
   \end{align}
   From $B\subset C$ it follows that also $b_1, b_2 \in C$ but then, we have found a nontrivial convex decomposition of $p$ contradicting the assumption that $p\in \mathrm{Ext}(C)$. 
\end{proof}
    \normalfont
    
\theorem{Let $S= (I, M, O)$ be a scenario and  $M'$ an arbitrary collection of subsets $M_i' \subset M_i$ for all $i\in I$. Then $\mathrm{\mathrm{Ext}}(\mathbf{B}(S)) \subset \mathrm{\mathrm{Ext}}(\mathbf{PD}(S,M')) \subset \mathrm{\mathrm{Ext}}(\mathbf{NS}(S))$\label{TheoremExtremePointsInclusionBellPDPNS}}
\begin{proof}
  Note that by Defs.~\ref{predictablebehaviour}, \ref{No-signallingbehaviour} and \ref{LHVBehaviour} the set $\mathrm{Ext}(\mathbf{B}(S))$ consists precisely of the predictable no-signalling behaviour, which can be seen also to form a class of extreme points of the set of all behaviours $\mathbf{E}(S)$  so that $\mathrm{Ext}(\mathbf{B}(S)) \subset \mathrm{Ext}(\mathbf{E}(S))$. It then follows from Lemma \ref{LemmaExtremepointsimplication} that $\mathrm{Ext}(\mathbf{B}(S)) \subset \mathrm{Ext}(\mathbf{PD}(S,M'))$ and $\mathrm{Ext}(\mathbf{B}(S)) \subset \mathrm{Ext}(\mathbf{NS}(S))$, using the fact that all the sets in question are convex polytopes, and the inclusion $\mathbf{B}(S) \subset \mathbf{PD}(S,M') \subset \mathbf{NS}(S)$. It remains to show that $\mathrm{Ext}(\mathbf{PD}(S,M')) \subset \mathrm{Ext}(\mathbf{NS}(S))$.

Clearly the inclusion relation holds if $M_i' = \emptyset$ for all $i\in I$ when, by Theorem \ref{Theorem:BellandNSspecialcasesofPDP} $\mathbf{PD}(S,M') = \mathbf{NS}(S).$
We can demonstrate, by an argument analogous to that employed in the proof of Theorem \ref{TheoremextremepointsofPDP}, that points of this form are also extreme points of $\mathbf{NS}(S)$ in cases where $\mathbf{PD}(S,M')\neq \mathbf{NS}$. 

In Theorem \ref{TheoremextremepointsofPDP} it was shown that the extreme points  $\wp_{\mathrm{Ext}}\in \mathbf{PD}(S,M')$ are those which have the form 

\begin{align}
    \wp_{\mathrm{Ext}}(\vec{a}|\vec{x}) = D^{S_{|M'^{\perp}}}_{r}(\vec{a}_{F_{\Vec{x}}}| \Vec{x}_{F_{\Vec{x}}})P^{S_{|M'^{\perp}}}_{k}(\Vec{a}_{I \setminus F_{\Vec{x}}}|\Vec{x}_{I \setminus F_{\Vec{x}}}),\label{Equation:EXTpointofPDPinEXTofAlsonsPROOF}
\end{align}
with $r,k$ labelling distinct predictable behaviour in $S_{|M'}$ and distinct extreme points of the no-signalling polytope over $S_{|M'^{\perp}}$. 

Suppose for the sake of argument that points of the form of Eq.~\eqref{Equation:EXTpointofPDPinEXTofAlsonsPROOF} were not extreme points of $\mathbf{NS}(S)$ for all $r,k$. Then, there would exist some $P_j \in \mathbf{NS}(S)$, convex parameters $\omega_j$ such that $\omega_j \geq 0, \omega_j \neq 1$ for all $j\in J$, $\sum_j \omega_j = 1$ and a pair $r,k$ with  
\begin{align}
    \sum_{j}\omega_j P_j(\vec{a}|\vec{x})  = D^{S_{|M'^{\perp}}}_{r}(\vec{a}_{F_{\Vec{x}}}| \Vec{x}_{F_{\Vec{x}}})P^{S_{|M'^{\perp}}}_{k}(\Vec{a}_{I \setminus F_{\Vec{x}}}|\Vec{x}_{I \setminus F_{\Vec{x}}}).
\end{align}
This implies, by marginalizing over $I \setminus F_{\vec{x}}$ and using no-signalling, in particular that 
\begin{align}
    \sum_{j}\omega_j P_j(\vec{a}_{F_{\vec{x}
    }}|\vec{x}_{F_{\vec{x}}}) = D_r(\vec{a}_{F_{\vec{x}}}|\vec{x}_{F_{\vec{x}}})
\end{align}
for all $F_{\vec{x}}$. By the uniqueness of the predictable behaviour $D_r$ in scenario $S_{|M'}$ this can hold if and only if
\begin{align}
    P_j(\vec{a}_{F_{\vec{x}  }}|\vec{x}_{F_{\vec{x}}}) = D^{S_{|M'}}_r(\vec{a}_{F_{\vec{x}}}|\vec{x}_{F_{\vec{x}}}) \hspace{0,3cm} \forall j \label{EquationHastobePD}.
\end{align}

We can stop here,  since Eq.~\eqref{EquationHastobePD} essentially states that the $P_j \in \mathbf{NS}(S)$ above have to be partially predictable with respect to $M'$, and the form of the extreme points of $\mathbf{PD}(S,M')$ in terms of those behaviours is already known to match the form of Eq.~\eqref{Equation:EXTpointofPDPinEXTofAlsonsPROOF}. 
\end{proof}
\normalfont

As an immediate consequence a hierarchy for the extreme points of  partially deterministic polytopes can be obtained as a  strict sharpening  of Theorem \ref{TheoremPDpolytopeBasicInclusionTHRM}.

\theorem{\label{Theorem:extremepointsareIncludedinPDPsubset}Let $S=(I, M, O)$ be a scenario and $M', M'' $ collections of subsets of inputs such that $M_i'' \subset M_i' \subset M_i$ for all $i\in I$.  Then $\mathrm{\mathrm{Ext}}(\mathbf{PD}(S, M') \subset \mathrm{\mathrm{Ext}}(\mathbf{PD}(S, M''))$}
\begin{proof}
    By Theorem \ref{TheoremExtremePointsInclusionBellPDPNS}, $\wp \in \mathrm{Ext}(\mathbf{PD}(S,M') )\Rightarrow \wp \in \mathrm{Ext}(\mathbf{NS}(S))$. On the other hand, if $M'' \subset M'$ then by Theorem \ref{TheoremPDpolytopeBasicInclusionTHRM}, $\mathbf{PD}(S,M') \subset \mathbf{PD}(S,M'')$. Hence, the conditions of Lemma \ref{LemmaExtremepointsimplication} hold, and the claim $\wp \in \mathrm{Ext}(\mathbf{PD}(S,M')) \Rightarrow \wp \in \mathrm{Ext}(\mathbf{PD}(S,M''))$ follows.
\end{proof}
\normalfont

As a corollary of Theorems \ref{Theorem:BellandNSspecialcasesofPDP} and  \ref{Theorem:extremepointsareIncludedinPDPsubset}, if a full solution of the vertices of a no-signalling polytope $\mathbf{NS}(S)$ in particular are known, then the vertices of every single partially deterministic polytope $\mathbf{PD}(S,M')$ that may be defined over $S$ are contained in that solution. Since the vertices of $\mathbf{PD}(S,M')$ are precisely those behaviours which are partially predictable with respect to $M'$, excluding the vertices of $\mathbf{NS}(S)$ which are not predictable with respect to those subsets of inputs is an easy task.  

In deriving Theorems \ref{convPPNSisapolytopeTheorem}-\ref{Theorem:extremepointsareIncludedinPDPsubset}, the marginal distributions $\wp(\vec{a}|\vec{x})$ from the maximal contexts $D$ and $ND$, which we identified as equivalenty obtainable from the restricted behaviours $\wp^{S_{|M'}}$ and $\wp^{S_{M'^{\perp}}}$ defined over substructures $S_{|M'}$ and $S_{|M'^{\perp}}$, played a significant role in describing the features of partially deterministic behaviours $\wp \in \mathbf{PD}(S,M')$. In what follows, we describe this property in a sense from the opposite direction, by formally defining a notion of composability of scenarios and behaviours, by which sets of behaviours in a scenario $S$ are constructed from behaviours in strictly smaller scenarios. This leads to the notion of behaviour product, the properties of which turn out to be useful for describing the general features of partially deterministic polytopes in particular.

\definition{ (Bipartition of a scenario)\label{DefinitionBipartitionofaScenario}\\
Let $S=(I,M,O)$ be a scenario and $M'$ a collection of subsets $M'_i \subset M_i$ $\forall i\in I$. Let $M'^{\perp}$ denote the  complement of  $M_i$, that is, the collection of subsets  where $M_i'^{\perp} = M_i\setminus M_i'$ $\forall i$.  An (input-disjoint) bipartition of $S$ into two sub-scenarios $S'= (I', M', O')$ and $S'^{\perp} = (I'^{\perp}, M'^{\perp}, O'^{\perp})$ is defined  by the restrictions of $S$ defined by $M'$ and $M'^{\perp}$ so that we may identify
\begin{align}
  S' = S_{|M'} \hspace{0.5cm} \mathrm{and} \hspace{0.5cm} S'^{\perp} = S_{|M'^{\perp}}\label{EquationBipartitionofscenarioDefinesaScenarioforM'}.
\end{align}

When $S'$ and $S'^{\perp}$ form a bipartition of $S$, we may also equivalently say that $S$ is composed of $S'$ and $S'^{\perp}$, and write $S = S' \blacktriangle S'^{\perp}$ for the composition of $S$ via $I = I'\cup I'^{\perp}, M = \{M_i | M_i = M_i'\cup M_i'^{\perp} \}$ and $O = O' \cup O'^{\perp}$.
}\\
\normalfont

Note that similarly to the case of restricting $S$ to $S_{|M'}$ in Definition \ref{DefinitionRestrictionofaScenario}, some sub-scenarios $S'$ and $S'^{\perp}$ compatible with Def.~\ref{DefinitionBipartitionofaScenario} can be trivial in the sense of (say) containing a single party, or a single measurement for a number of parties. Accordingly,  in the extreme cases  the collection $M'$ compatible with Definition \ref{DefinitionBipartitionofaScenario} would consist completely of the sets $M_i = M_i \forall i$ or of empty sets $M_i' = \emptyset $ $\forall i$. Conventionally we allow for these situations with the understanding that the subscenarios generated in this way consist of the empty scenario (over which no behaviour are defined), and the original scenario itself. We refer to such extreme bipartitions as `redundant', since they do not map into novel pairs of substructures. A nonredundant bipartition therefore has that for at least one $i \in I$,  $ \emptyset \subsetneq M_i' \subsetneq M_i$.

The bipartition operation can be immediately generalized to the notion of $n$-fold partition of $S$ into $n$ scenarios $S_{k}$, by virtue of collections $M^{k}, k\in \{1, \ldots , n \}$  of subsets of inputs $M_i^{k}\subset M_i$ $\forall i$ with the properties 
\begin{align}
    M_i^{k_1} \cap M_i^{k_2} =\emptyset \hspace{0.1cm} \forall k_1,k_2 \in \{1, \ldots, n \}, k_1 \neq k_2
\end{align}
and \begin{align}
    \cup_{k} M_i^{k} = M_i \hspace{0.1cm} \forall i \in I.
\end{align}
Each $S^{k} = I^{k}, M^{k}, O^{k}$ would then be defined analogously to the recipe in Eq.~\eqref{EquationBipartitionofscenarioDefinesaScenarioforM'}.    

We will, for the most part, not need to deal with $n$-fold bipartitions rigorously in full generality. It is easy to see, however, that any $n$-partition of $S$ can be in principle realised by applying suitable bipartition operations of Definition \ref{DefinitionBipartitionofaScenario} iteratively $n$-times, starting from scenario $S$, then to one of the sub-scenarios and so on. In this sense, no generality is lost in focusing on the bipartite case. 

It is easy to see that the bipartite  'composition operator' $\blacktriangle$ is commutative and associative, since the composition is defined with respect to unions of sets forming the scenarios. When generalized to $n$-fold partitions, this indicates that the relevant structural properties of $S$  are encoded in the individual scenarios $S^{k}$ which form the $n$-partition,  and not in the order in which the scenario $S$ is constructed from them. Hence in such situations $S$ may be decomposed  as $S = S^1\blacktriangle S^{2} \blacktriangle \ldots \blacktriangle S^{n} = \blacktriangle_{k=1}^{k=n}S^{k}$, omitting the parentheses. 

With the definition of a bipartition in place, we are now ready to formally introduce the behaviour product. 

\definition{(The  behaviour product)\label{BehaviourProductDefinition}\\
Let $S= (I, M, O)$ be a scenario, and $M'$, $M'^{\perp}$ define  disjoint bipartition  $S$, into  the restricted scenarios $S' = (I', M', O')$, and $S'^{\perp}= (I'^{\perp}, M'^{\perp}, O'^{\perp}$). Let $\mathbf{K}(S') \subset \mathbf{NS}(S')$ and $\mathbf{K}(S'^{\perp}) \subset \mathbf{NS}(S'^{\perp})$ be some arbitrary sets of no-signalling behaviours defined in the restricted scenarios. 

 The behaviour product ${\odot}: \mathbf{K}(S') \times \mathbf{K}(S'^{\perp}) \rightarrow \mathbf{NS}(S)$ is defined as the binary operation 
\begin{align}
    \wp^{S'}{\odot}  \wp^{S'^\perp} = \wp^S,
\end{align}
 by enforcing that the behaviour $\wp^S\in \mathbf{E}(S)$ admits a context dependent expansion 
\begin{align}
\begin{split}
  \label{EquationBehaviourproductproduct}  \wp^S(\Vec{a}|\Vec{x}) &= \wp^{S'}(\vec{a}'_{F_{\vec{x}}} = \Vec{a}_{F_{\Vec{x}}}| \vec{x}'_{F_{\vec{x}}}= \Vec{x}_{F_{\Vec{x}}}) \\ & \hspace{0.3cm} \times  \wp^{S'^\perp}(\vec{a}'^{\perp} = \Vec{a}_{I \setminus F_{\Vec{x}}}| \vec{x}'^{\perp} = \Vec{x}_{I \setminus F_{\Vec{x}}})
  \end{split}
\end{align}
with   $ F_{\Vec{x}} = \{ i \in I | x_i \in M'_i   \}$. 
}
If the bipartition is redundant, i.e. $M_i' = M_i$ $\forall i \in I$ or $M_i' = \emptyset$ $\forall i \in I$, then $S' = S$ or $S'^{\perp} =S$, respectively, and  Eq.~\eqref{EquationBehaviourproductproduct} is understood to reduce to the identity map on $\mathbf{K}(S')$ or $\mathbf{K}(S'^{\perp})$, respectively. 
\normalfont
\\

We will employ the notation where the distributions in Eq.~\eqref{EquationBehaviourproductproduct} which explicititly take  care of the the `matching' $\wp^{S'}(\vec{a}'_{F_{\vec{x}}} = \vec{a}_{F_{\vec{x}}}|\vec{x}_{F_{\vec{x}}}' = \vec{x}_{F_{\vec{x}}})$  are replaced simply by $\wp^{S'}(\vec{a}_{F_{\vec{x}}}|\vec{x}_{F_{\vec{x}}})$. Hence, in place of Eq.~\eqref{EquationBehaviourproductproduct} for the product $\wp = \wp^{S'}\odot \wp^{S'^{\perp}}$ we may write 
\begin{align}
    \wp(\vec{a}|\vec{x})= \wp^{S'}(\vec{a}_{F_{\vec{x}}}|\vec{x}_{F_{\vec{x}}}) \wp^{S^{\perp}}(\vec{a}_{I\setminus F_{\vec{x}}}|\vec{x}_{I\setminus F_{\vec{x}}}).
\end{align}
This notation is well defined, even if the scenario $S'$ (resp. $S'^{\perp}$) does not support strings of length $|I|$, owing to the definition of $F_{\vec{x}}$ (resp. $I \setminus F_{\vec{x}})$  and the no-signalling constraints.  Occasionally we will revert back to the more precise notation of Eq.~\eqref{EquationBehaviourproductproduct} if needed. 

\begin{figure*}
    \centering
    \includegraphics[width=\textwidth]{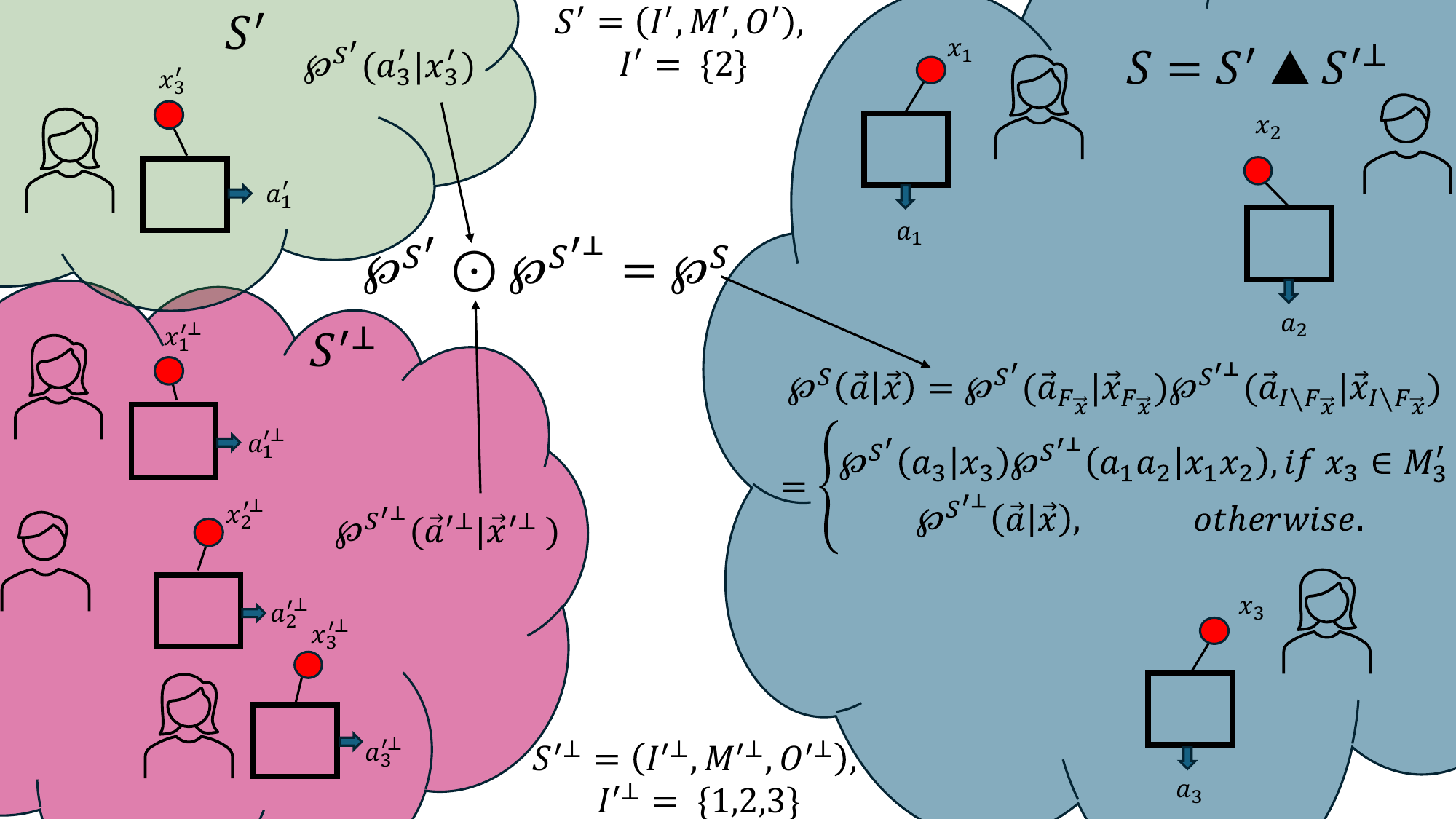}
    \caption{Illustration of the action of the behaviour product. The scenarios $S'=(I',M',O')$, which contains a single agent $i=3$, and $S'^{\perp}=(I'^{\perp}, M'^{\perp}, O'^{\perp})$, which contains three agents $i\in \{1,2,3\}$ form a disjoint bipartition of the tripartite scenario $S = (I,M,O)$. The behaviour product takes behaviour $\wp^{S'}, \wp^{S'^{\perp}}$ defined over the subscenarios $S', S'^{\perp}$ and produces a well-defined behaviour $\wp^S$ on scenario $S$. The expansion of the distributions $\wp^S(\vec{a}|\vec{x})$ is context-dependent, with each distribution given as the products of marginals  $\wp^{S'}(\vec{a}_{F_{\vec{x}}}|\vec{x}_{F_{\vec{x}}})$ and $\wp^{S'^{\perp}}(\vec{a}_{I\setminus F_{\vec{x}}}|\vec{x}_{I\setminus F_{\vec{x}}})$, with $F_{\vec{x}} = \{ i\in I : x_i \in M_i'\}$. In the situation depicted above, the general formula gives $\wp^S(\vec{a}|\vec{x}) = \wp^{S'}(a_3|x_3)\wp^{S'^{\perp}}(a_1a_2|x_1x_2)$ if $x_3 \in M_3'$ and $\wp^{S}(\vec{a}|\vec{x}) = \wp^{S'^{\perp}}(\vec{a}|\vec{x})$ otherwise, owing to the fact that the scenario $S'$ contains only the agent $i=3$, with input set $M_3'$. Note that expressions such as $\wp^{S'^{\perp}}(\vec{a}|\vec{x})$ should more precisely be understood as referring to the 'matching' $\wp^{S'^{\perp}}(\vec{a}'^{\perp}=\vec{a}_{I'}|\vec{x}'^{\perp}=\vec{x}_{I'})$. The precise notation is generally dropped for convenience, when there is no risk of confusion.  } 
    \label{fig:BehaviourProductIllustration}
\end{figure*}

Definition \ref{BehaviourProductDefinition} provides a unique way of constructing a behaviour in $S$ from two behaviours in two strictly smaller scenarios $S'$ and $S'^{\perp}$ which may be thought to be contained in $S$, as illustrated in Fig.~\ref{fig:BehaviourProductIllustration}. On the level of individual distributions $\wp(\vec{a}|\vec{x})$, it can be thought to generalise the notion of a product distribution $\wp(\vec{a}|\vec{x}) = \prod_{i\in I}\wp(a_i|x_i)$ to allow for correlations inside the partitions $F_{\vec{x}}$ and $I\setminus F_{\vec{x}}$, but not between them. 

The idea of combining behaviours from two subscenarios to a behaviour in a larger scenario is, of course, not new. For example, in a recent work of Ref.~\cite{Acin2015} many such products were studied in a general `hypergraph' approach to to contextuality.  

Instances of use of the type of operation in Definition \ref{BehaviourProductDefinition} in particular, or special cases thereof,  can be found in other works, see for example proof of Theorem 3 in Ref.~\cite{Mazzari2023} where properties of behaviours in experiments which combine both contextuality and Bell-nonlocality aspects into a single test where explored,   or Theorem 8 in Ref.~\cite{Haddara2025} where properties of Local Friendliness behaviours in extended Wigner's friend scenarios were studied. Our Definition \ref{BehaviourProductDefinition} is inspired by the more formal, structural approach in Woodhead's work \cite{Woodhead2014}, where a behaviour product for bipartite scenarios was defined\footnote{Similarly to Woodhead \cite{Woodhead2014} in the case of bipartite scenarios, we have defined the product over input disjoint subscenarios $S'$ and $S'^{\perp}$. Strictly speaking, this would not be necessary in order to define a map from smaller scenarios to larger ones. This convention, however, removes the choice as to how inputs (and even parties) of the smaller scenarios are mapped to inputs in the larger scenario, simplifying the presentation considerably. } and investigated. Indeed, the product of Definition \ref{BehaviourProductDefinition} can be thought of as a direct multipartite extension of the one defined by Woodhead \cite{Woodhead2014}. 

A closely related operation has also been investigated in the exclusivity graph approach to correlations in hybrid causal structures \cite{rodari2023characterizinghybridcausalstructures}, where the product of sub-behaviour was identified essentially as  a Hadamard-type product of distributions in the sub-scenarios. Note that the operation in our Definition \ref{BehaviourProductDefinition} is in general \emph{not} of that form, unless $M_i' \in \{M_i, \emptyset \}$ for all $i\in I$,  as it includes taking the products of different size marginals of the appropriate distributions. In fact, our operation subsumes that of Ref.~\cite{rodari2023characterizinghybridcausalstructures} by reducing to it as a special case.  To our knowledge the behaviour product has not been previously studied in this manner at this level of generality.

The behaviour product is commutative, since the ordinary product commutes. More precisely, an analogous mapping can be defined on the domain $\mathbf{K}(S'^{\perp}) \times \mathbf{K}(S')$, which uniquely maps to a behaviour in $S$. We equate this `symmetric' mapping with the product $\odot$ for simplicity, without affecting any of the conclusions. Furthermore, since the mapping is well defined over $\mathbf{NS}(S')\times \mathbf{NS}(S'^{\perp})$ in particular, we  use the same symbol $\odot$ for the cases in which the domain is restricted to  $\mathbf{K}(S') \times \mathbf{K}(S'^{\perp})$ with some subsets $\mathbf{K}(S') \subset \mathbf{NS}(S')$ and $\mathbf{K}(S'^{\perp}) \subset \mathbf{NS}(S'^{\perp})$ as well.

Let us also point out that, strictly speaking, enforcing no-signalling is not always necessary for the product of Definition \ref{BehaviourProductDefinition} to be well-defined. 

\example{(Behaviour product for signalling sets)\label{example:BehaviourproductCHSHoneWaysignalling}\\
Let $S= (I,M,O)$ be a bipartite two-input per-site scenario, that is a scenario with  $I = \{1,2 \}$  and  $|M_i|=2$ for $i\in I$. Let $S' = (I',M',O')$ be a minimal nontrivial subscenario defined by subsets $M_i'$ with $|M_1'| =1, |M_2'| = \emptyset$ so that $I' = \{1\}$. The complementary scenario $ S'^{\perp}$ has  $I'^{\perp} = \{1,2\}$ and $|M_1^{\perp}| = 1, |M_2^{\perp}|=2$. The behaviour product $\odot $ may be extended over $\mathbf{E}(S') \times \mathbf{E}(S'^{\perp})$  even though clearly $\mathbf{NS}(S'^{\perp}) \subsetneq\mathbf{E}(S'^{\perp})$.  Indeed note that the marginals $\sum_{a_1}\wp^{S'^{\perp}}(a_1a_2|x_1^{\perp}x_2) = \wp^{S'^{\perp}}(a_2|x_1^{\perp}x_2) \equiv \wp^{S'^{\perp}}(a_2|x_2)$ are uniquely defined for all $x_2$ since the site $i =1$ has only one fixed input $x_1^{\perp}$ in the subscenario $S'^{\perp}$.   By Eq.~\eqref{EquationBehaviourproductproduct} 
\begin{align}
    \wp^{S} = \wp^{S'}\odot \wp^{S'^{\perp}} = \begin{cases}
        \wp^{S'}(a_1|x_1) \wp^{S'^{\perp}}(b_2|x_2) &  \mathrm{if } \hspace{0.2cm} x_1 \in M'_1 \\
        \wp^{S'^{\perp}}(a_1a_2|x_1x_2)& \mathrm{otherwise} 
    \end{cases}
\end{align}
is well defined, despite the fact that the scenario $S'^{\perp}$ allows signalling from site $i=2$ to $i=1$.  
}\\
\normalfont

The case of Example \ref{example:BehaviourproductCHSHoneWaysignalling} where a product can be extended to arbitrary subsets of $\mathbf{E}(S') \times \mathbf{E}(S'^{\perp})$ is really an exception and not the norm. The following  example illustrates a case where the behaviour product of Def.~\ref{BehaviourProductDefinition}  \emph{is not} well defined over arbitrary subsets of behaviour.  

\example{(Failure of behaviour product for signalling sets)\label{example:behproductFailsforsignallingSETS}\\
Let $S = (I,M,O)$ be a bipartite scenario with $|M_i|=3$ for both $i$. Let $M', M'^{\perp}$ define a disjoint bipartition of $S$ into $S', S'^{\perp}$ such that $I' = I'^{\perp} = I$ and $|M_i'| = 1$ for both $i\in I$ so that $|M_i'^{\perp}| = 2$ for both $i\in I$.  The recipe of Eq.~\eqref{EquationBehaviourproductproduct} works fine for the cases where $x_1 \in M_1', x_2 \in M_2'$ or $x_1 \in M_1'^{\perp}, x_2\in M_2'^{\perp}$ by giving
\begin{align}
    \wp^{S}(\vec{a}|\vec{x}) = \begin{cases}
        \wp^{S'}(a_1a_2|x_1x_2) & \mathrm{if } \hspace{0.2cm} x_1 \in M_1', x_2\in M_2'\\
        \wp^{S'^{\perp}}(a_1a_2|x_1x_2) &\mathrm{if} \hspace{0.2cm}   x_1 \in M_1'^{\perp}, x_2 \in M_2'^{\perp}
    \end{cases} 
\end{align}
but fails in `mixed' contexts where $x_1 \in M_1', x_2\in M_2'^{\perp}$ or $x_1\in M_1'^{\perp}, x_2 \in M_2'$ which would demand 
\begin{align}
    \wp^{S}(\vec{a}|\vec{x}) = \
        \wp^{S'^{}}(a_1|x_1) \wp^{S'^{\perp}}(a_2|x_2) \hspace{0.1cm}  & \mathrm{if } \hspace{0.1cm} x_1 \in M_1', x_2 \in M_2'^{\perp} 
        \end{align}
        and
        \begin{align}
        \wp^{S}(\vec{a}|\vec{x})=\wp^{S'^{\perp}}(a_1|x_1)\wp^{S'^{\perp}}(a_2|x_2)  \hspace{0.1cm}  & \mathrm{if } \hspace{0.1cm} x_1 \in M_1'^{\perp}, x_2 \in M_2'.
\end{align}
The problem is that all of the marginal expressions referring to the subscenario $S'^{\perp}$ in the two equations above are ill-defined, as for example, $\sum_{a_2}\wp^{S'^{\perp}}(a_1a_2|x_1x_2) = \wp^{S'^{\perp}}(a_1|x_1x_2)$ which may explicitly depend on $x_2$, since now $\wp^{S^{\perp}} \in \mathbf{E}(S'^{\perp})$.  
}\\ 
\normalfont

 We will mostly be focusing on the case where no-signalling is assumed, and so investigation of the generally valid conditions under which the behaviour product of Def.~\ref{BehaviourProductDefinition} may be used is left for future work. For the purposes of later discussion however, we identify a class of useful cases which does allow for the product operation to be used. 

\remark{(Behaviour product for classes of signalling  sets)\label{remark:behaviourProductOFsignALLINGsubsets}\\
Let $S = (I,M,O)$ be a scenario and  $M^{I'}, M^{{I'}^{\perp}}$ a disjoint bipartition of $M$ specified by $M_i^{I'} = M_i$ if $i\in I'$ and $M_{i}^{I'}= \emptyset$ otherwise. Denote by $S^{I'}$ and $S^{I'^\perp}$ the restricted scenarios corresponding to this partition. Then, the behaviour product of Definition \ref{BehaviourProductDefinition} can be extended to a map $\widetilde{\odot}: \mathbf{K}(S^{I'}) \times \mathbf{K}(S^{I'^{\perp}}) \rightarrow \mathbf{E}(S)$  for any (not necessarily no-signalling) sets $\mathbf{K}(S^{I'})\subset \mathbf{E}(S^{I'})$ and $\mathbf{K}(S^{I'^{\perp}}) \subset \mathbf{E}(S^{I'^{\perp}})$  }\\
\normalfont

The bipartition $S^{I'}, S^{I'^{\perp}}$ in Remark \ref{remark:behaviourProductOFsignALLINGsubsets} guarantees that the distributions comprising the behaviour in arbitrary contexts $\vec{x}$ in the scenario $S$ are well defined. Such a condition on the partition is needed in this case, since if no-signalling is not enforced  the marginal distributions of the behaviour defined in $S'$ and $S'^{\perp}$ are allowed to depend on arbitrary remaining inputs in those scenarios leading to the kind of problem illustrated in Example \ref{example:behproductFailsforsignallingSETS}.  Intuitively, the image $\mathbf{K}(S^{I'}) \widetilde{\odot}\mathbf{K}(S^{I'^{\perp}}) \subset \mathbf{E}(S)$ in Remark \ref{remark:behaviourProductOFsignALLINGsubsets} consists of those behaviours which allow signalling only inside the partitions $I'$ and $I'^{\perp}$, but not across them.

The behaviour product of Definition \ref{BehaviourProductDefinition} formalises and generalises the example used in the proof of Proposition \ref{Proposition:restrictionmapidentities}, where it was used to demonstrate that behaviours defined in a subscenario $S_{|M'}$ of $S$ are always in the image of the corresponding restriction map $R_{|M'}$. Formally, this means that if $S'$ and $S'^{\perp}$ is an arbitrary bipartition of a scenario $S$ and  $\mathbf{K}(S') \subset \mathbf{NS}(S')$ and $\mathbf{K}(S'^{\perp}) \subset \mathbf{NS}(S'^{\perp})$ are some sets of no-signalling behaviours defined in those scenarios, then
\begin{align}
    R_{|M'}(\mathbf{K}(S') \odot \mathbf{K}(S'^{\perp})) = \mathbf{K}(S') \label{Equation:restrictiontoM'isKS'}
\end{align}
and 

\begin{align}
    R_{|M'^{\perp}}(\mathbf{K}(S') \odot \mathbf{K}(S'^{\perp})) = \mathbf{K}(S'^{\perp}).\label{Equation:restrictiontoM'isKS'perp}
\end{align}

Both of the above identities can in fact, be obtained as special cases of the following general result.

\proposition{\label{Proposition:ImageRestrictionofBehproductistheDomain} Let $S'$ and $S'^{\perp}$ be any bipartition of a scenario $S= I,M,O$ defined by collections $M'$ and $M'^{\perp}$. Let $\mathbf{K}(S') \subset \mathbf{NS}(S')$ and $\mathbf{K}(S'^{\perp}) \subset \mathbf{NS}(S'^{\perp})$ be arbitrary sets of no-signalling behaviour defined in those scenarios and $\wp = \wp^{S'} \odot \wp^{S'^{\perp}} \in \mathbf{K}(S') \odot \mathbf{K}(S'^{\perp})$ . If $V$ is an arbitrary collection  of subsets $V_i \subset M_i$ of inputs for all $i \in  I$, then 
\begin{align}
R_{|V}(\wp^{S'} \odot \wp^{S'^{\perp}}) = R_{|M'\cap V}(\wp^{S'}) \odot R_{|M'^{\perp}\cap V}(\wp^{S'^{\perp}}) 
\end{align}
Here, $M'\cap V$ and $M'^{\perp}\cap V$ refer to the collections of input sets  $M_i' \cap V_i $, $i \in I$ and $M_i'^{\perp} \cap V_i$, $ i\in I$ respectively.
}

\begin{proof}
    Follows immediately from Definitions \ref{Definition:RestrictionofaBehaviour} and \ref{BehaviourProductDefinition}  of the restriction map and behaviour product. 
    
    We provide the proof for completeness. 
    
    From  $\wp \in \mathbf{K}(S')\odot \mathbf{K}(S'^{\perp})$ it follows that the behaviour decomposes as 
    \begin{align}
        \wp(\vec{a}|\vec{x}) = \wp^{S'}(\vec{a}_{F_{\vec{x}}}|\vec{x}_{F_{\vec{x}}}) \cdot \wp^{S'^{\perp}}(\vec{a}_{I\setminus_{F_{\vec{x}}}}|\vec{x}_{I \setminus F_{\vec{x}}}), 
    \end{align}
    which, with $I_V = \{ i\in I : V_i \neq 
    \emptyset \} $ and $F^V_{\vec{x}'} = \{ i \in I_V : x_i' \in M_i' \}$  implies that the image $R_{|V}(\wp) = \wp^{S_{|V}}$ consists of distributions 
    \begin{align}
     \wp^{S_{|V}}(\vec{a}'|\vec{x}') &= \wp(\vec{a}_{I_V} = \vec{a}'|\vec{x}_{|V}=\vec{x}') \\
     \begin{split}
&= \wp^{S'}( \vec{a}'_{F^V_{\vec{x}'}}| \vec{x}'_{F^V_{\vec{x}'}}) \wp^{S'^{\perp}}(\vec{a}'_{I_V \setminus F^V_{\vec{x}'}}|\vec{x}'_{I_V \setminus F^V_{\vec{x}'}})
      \end{split}
\end{align}

 For the converse, note first that $ R_{|M'\cap V}(\wp^{S'}) \odot R_{M'^{\perp}\cap V}(\wp^{S'^{\perp}}) = \wp_w$ is a well defined behaviour in scenario $S_{|V}$, since $M'_i$ and $M_i'^{\perp}$ are collections of of disjoint subsets  $M_i$ and $(M'_i\cap V_i)\cup (M_i'^{\perp}\cap V_i) = V_i$ for all $i \in I$. A direct calculation gives

\begin{align}
\begin{split}
   \wp_w(\vec{a}'|\vec
   x') = & \wp^{S'}(\vec{a}'_{F^V_{\vec{x}}}|\vec{x}_{F^V_{\vec{x}}})\cdot \wp^{S'^{\perp}}(\vec{a}_{I_V \setminus F^V_{\vec{x}}}|\vec{x}_{I_V \setminus F^V_{\vec{x}}}),
    \end{split}
\end{align}
which establishes the claim. 

\end{proof}
\normalfont

Proposition \ref{Proposition:ImageRestrictionofBehproductistheDomain} implies, of course, that  
\begin{align}
\begin{split}
        R_{|V}&(\mathbf{K}(S') \odot \mathbf{K}(S'^{\perp})) \\& = R_{|M'\cap V}(\mathbf{K}(S')) \odot R_{|M'^{\perp}\cap V}(\mathbf{K}(S'^{\perp}))
    \end{split}
\end{align}
from which Equations \eqref{Equation:restrictiontoM'isKS'} and \eqref{Equation:restrictiontoM'isKS'perp} are recovered by setting $V= M'$ or $V = M'^{\perp}$, respectively.

 A further few remarks may be readily made, which we gather in the following propositions.

\proposition{\label{pROPOSITION:BehaviourProdForsignallingIdentities} Let $S = (I,M,O)$ be a scenario and $S^{I'}, S^{I'^{\perp}}$ a nonredundant bipartition of $S$ defined by the collections $M^{I'}$ $M^{I'^{\perp}}$ where $M_i^{I'} = M_i$ if $i\in I'$ and $M_i^{I'}=\emptyset$ otherwise, and similarly for $M^{I'^{\perp}}$. This class of bipartitions is consistent Remark \ref{remark:behaviourProductOFsignALLINGsubsets}, and allows use of behaviour product between signalling sets defined over the subscenarios $S^{I'}$ and $S^{I'^{\perp}}$. The following identities hold 

\begin{enumerate}[label=(\roman*), ref=\ref{PropositionBehaviourproductProperties}.\roman*] 
    \item $\mathbf{P}(S') \odot \mathbf{P}(S'^{\perp}) \subsetneq \mathbf{P}(S). $ \label{PropositionBehaviourproductProperties.signallingP}
  \item $\mathbf{U}(S') \odot \mathbf{U}(S'^{\perp}) \subsetneq \mathbf{U}(S).$ \label{PropositionBehaviourproductProperties.signallingS} 
  \item $\mathbf{E}(S^{I'})\odot \mathbf{E}(S^{I'^{\perp}}) \subsetneq \mathbf{E}(S)$
\end{enumerate}
}
\begin{proof}
    \textit{(i,ii,iii)} The proofs follows immediately from the definition of the behaviour product. First, let us observe that  the inclusions $\mathbf{P}(S^{I'})\odot \mathbf{P}(S^{I'^{\perp}})\subset \mathbf{P}(S)$, $\mathbf{U}(S^{I'}) \odot \mathbf{U}(S^{I'^{\perp}}) \subset \mathbf{U}(S)$ and $\mathbf{E}(S^{I'})\odot \mathbf{E}(S^{I'^{\perp}}) \subset \mathbf{E}(S)$ clearly hold. In the first two cases, owing to the fact that products of product distributions are product distributions, and products of predictable distributions are predictable, and in the last since $\mathbf{E}(S)$ is by definition the set of all valid behaviour over $S$. It is therefore sufficient to show that the inclusions do hold in the other direction. 
    
    We will sketch the general argument for the case (iii), the other cases can be proved analogously. 

    Let $\wp \in \mathbf{E}(S^{I'})\odot \mathbf{E}(S^{I'})$ so that by definition 
\begin{align}
    \wp(\vec{a}|\vec{x}) &= \wp^{S^{I'}}(\vec{a}_{F_{\vec{x}}}|\vec{x}_{F_{\vec{x}}})\wp^{S^{I'^{\perp}}}(\vec{a}_{I\setminus F_{\vec{x}}}|\vec{x}_{I\setminus F_{\vec{x}}})\\
& = \wp^{S^{I'}}(\vec{a}_{I'}|\vec{x}_{I'})\wp^{S^{I'^{\perp}}}(\vec{a}_{I\setminus I'}|\vec{x}_{I\setminus I'}).
\end{align}
Since the bipartition is nonredundant, this expansion is nontrivial.  By inspection, the marginals $\wp(\vec{a}_{I'}|\vec{x}), \wp(\vec{a}_{I\setminus I'}|\vec{x})$ of any $\wp$ in the product set are seen to satisfy the 'partial independence' relations $\wp(\vec{a}_{I'}|\vec{x}) = \wp(\vec{a}_{I'}|\vec{x}_{I'})$ and $ \wp(\vec{a}_{I\setminus I'}|\vec{x}) = \wp(\vec{a}_{I\setminus I'}|\vec{x}_{I\setminus I'})$, which clearly need not be the case for every $\wp^{S}\in \mathbf{E}(S)$ proving the claim. 
\end{proof}
\normalfont

\proposition{$ $ \label{PropositionBehaviourproductProperties} Let $S$ be an arbitrary nontrivial scenario and  $S'$ and $S'^{\perp}$  an arbitrary (nonredundant) bipartition of $S$. Then
\begin{enumerate}[label=(\roman*), ref=\ref{PropositionBehaviourproductProperties}.\roman*] 
    \item $ \mathbf{P}_{\mathbf{NS}}(S') \odot \mathbf{P}_{\mathbf{NS}}(S'^{\perp}) =\mathbf{P}_{\mathbf{NS}}(S).$ \label{PropositionBehaviourproductProperties.PNS} 
    \item  $\mathbf{U}_{\mathbf{NS}}(S') \odot \mathbf{U}_{\mathbf{NS}}(S'^{\perp}) =  \mathbf{U}_{\mathbf{NS}}(S).$ \label{PropositionBehaviourproductProperties.SNS} 
    \item $ \mathbf{B}(S') \odot\mathbf{B}(S'^{\perp}) \subsetneq \mathbf{B}(S).$\label{PropositionBehaviourproductProperties.BELLPRODUCT} 
    \item $\mathbf{NS}(S') \odot \mathbf{NS}(S'^{\perp}) \subsetneq \mathbf{NS}(S).$ \label{PropositionBehaviourproductProperties.NOSIGNALIGNPRODUCT} 
\end{enumerate}
}

\begin{proof}
   \textit{(i,ii)} Immediate. The sets $\mathbf{P}_{\mathbf{NS}}(S)$ and $\mathbf{U}_{\mathbf{NS}}(S)$ consist by definitions of distributions which can be represented as products of single party marginals, which can be grouped arbitrarily to the two subscenarios $S'$, $S'^{\perp}$. Conversely, the products of product distributions are product distributions, and products of predictable distributions are themselves predictable, and the whole behaviour can be specified by giving all the single party distributions, which can be fully specified with respect to any bipartition.

    \textit{(iii, iv)} A similar argument can be made in both cases. We will sketch  the argument for the case of the Bell-polytopes. 
    
    Clearly $\mathbf{B}(S') \odot \mathbf{B}(S'^{\perp}) \subset \mathbf{B}(S)$. It suffices to show that the converse inclusion does not hold. 
    
    By definition if $\wp \in \mathbf{B}(S')\odot \mathbf{B}(S'^{\perp})$, then the distributions $\wp(\vec{a}|\vec{x})$ decompose as
    \begin{align}
        \label{equation:anexpansionofbEHPRODUCT}\wp(\vec{a}|\vec{x}) = \wp^{S'}(\vec{a}_{F_{\vec{x}}}|\vec{x}_{F_{\vec{x}}})\wp^{S'^{\perp}}(\vec{a}_{I\setminus F_{\vec{x}}}|\vec{x}_{I\setminus F_{\vec{x}}}).
    \end{align}
    Furthermore, from the fact that the bipartition is nonredundant it follows, in particular, that at least for some $\vec{x}=\vec{x}'$  the distribution $\wp(\vec{a}|\vec{x})$ is uncorrelated with respect to the partitions $F_{\vec{x}'}, I\setminus F_{\vec{x}'}$ meaning that 
    \begin{align}
        \wp(\vec{a}_{F_{\vec{x}'}}|\vec{x}', \vec{a}_{I\setminus F_{\vec{x}'}}) = \wp(\vec{a}_{F_{\vec{x}}}|\vec{x}'_{F_{\vec{x}}})
    \end{align}
    and 
    \begin{align}
         \wp(\vec{a}_{I\setminus F_{\vec{x}'}}|\vec{x}', \vec{a}_{F_{\vec{x}'}}) = \wp(\vec{a}_{I\setminus F_{\vec{x}}}|\vec{x}'_{I \setminus F_{\vec{x}}}),
    \end{align}
    which follow immediately from the definition of conditional probabilities and the expansion of $\wp$  in Eq.~\eqref{equation:anexpansionofbEHPRODUCT}. Clearly, such constraints need not be obeyed by distributions $\wp^S(\vec{a}|\vec{x})$ of a general behaviour $\wp^S \in \mathbf{B}(S)$, proving the claim. 
\end{proof}

\normalfont
For the cases where no-signalling is not assumed, the identities in Proposition \ref{pROPOSITION:BehaviourProdForsignallingIdentities} can be seen simply noting that behaviours $\wp$ obtained as a behaviour product from behaviour in subscenarios do not in general depend on the whole context and hence in general $\mathbf{E}(S') \odot \mathbf{E}(S'^{\perp}) \subsetneq \mathbf{E}(S)$. The cases in Proposition \ref{PropositionBehaviourproductProperties} where no-signalling is assumed are more interesting, and Propositions \ref{PropositionBehaviourproductProperties.BELLPRODUCT} and \ref{PropositionBehaviourproductProperties.NOSIGNALIGNPRODUCT} formalize the intuition that a behaviour that can be obtained as a product does not in general allow correlations between the disjoint scenarios $S'$ and $S'^{\perp}$. From a mathematical perspective, this can be understood as a demonstration of the fact that the product does not in general preserve convexity. In fact, the following statement, reported here as a lemma, holds in full generality. 

\lemma{\label{LemmaConvexitynotPreservedInBehaviourproduct} Let $S'$ and $S'^{\perp}$ be an arbitrary bipartition of a scenario $S$. Let $\mathbf{K}(S') \subset \mathbf{NS}(S')$ and $\mathbf{K}(S'^{\perp}) \subset \mathbf{NS}(S'^{\perp})$ be arbitrary sets of no-signalling behaviour defined in those scenarios. Then 
\begin{align}
    \mathrm{conv}[\mathbf{K}(S')] \odot \mathrm{conv}[\mathbf{K}(S'^{\perp})] \subset \mathrm{conv}[\mathbf{K}(S') \odot \mathbf{K}(S'^{\perp})].
\end{align}
}
\begin{proof}
     Follows imediately from definitions. Namely if $\wp \in \mathrm{conv}(\mathbf{K}(S')) \odot \mathrm{conv}(\mathbf{K}(S'^{\perp})) $ then 
    \begin{align}
        \wp(\vec{a}|\vec{x}) = \sum_{j,k}\omega_j\omega_k \wp^{S'}_j(\vec{a}_{V_{\vec{x}}}|\vec{x}) \cdot \wp^{S'^{\perp}}_k(\vec{a}_{I\setminus V_{\vec{x}}}|\vec{x})
    \end{align}
    with 
    \begin{align}
        \wp^{S'}_j \in \mathbf{K}(S') \hspace{0.2cm} \mathrm{ and } \hspace{0.2cm} \wp^{S'^{\perp}}_k \in \mathbf{K}(S'^{\perp})
    \end{align}
    for all $j,k$. Since  $\sum_j \omega_j = 1$, $\sum_k \omega_k =1$ and $\omega_j , \omega_k \geq 0$ for all $j,k$ it is evident that the products $\omega_{jk}= \omega_j \omega_k$ themselves satisfy the criteria of being convex parameters and hence also $\wp \in \mathrm{conv}(\mathbf{K}(S') \odot \mathbf{K}(S'^{\perp})$, as claimed. 
    \end{proof}
\normalfont

Propositions \ref{PropositionBehaviourproductProperties.PNS}-\ref{PropositionBehaviourproductProperties.SNS}, point out to an interesting observation. Since  Proposition \ref{PropositionBehaviourproductProperties.PNS} in particular is equivalent to 
\begin{align}
    \mathrm{Ext}(\mathbf{B}(S')) \odot \mathrm{Ext}(\mathbf{B}(S'^{\perp})) = \mathrm{Ext}(\mathbf{B}(S)) \label{equationExtofBSisExtofBehproduct}
\end{align} 
which naturally implies that
\begin{align}
    \mathrm{conv}[\mathrm{Ext}(\mathbf{B}(S')) \odot \mathrm{Ext}(\mathbf{B}(S'^{\perp}))] = \mathbf{B}(S).
\end{align}
    Therefore, the behaviour product can be used to provide an alternative characterization of the set  $\mathbf{B}(S)$, provided that the convex hull is taken in the end.

We can, however, sharpen these considerations even further. Namely, following a bipartition of $S$ into $S'$ and $S'^{\perp}$, one can consider further bipartitions of those scenarios as well leading in full generality to a notion of an $n$-fold partition of $S = \blacktriangle_{i=1}^{n}S^{i}$.  Evidently, the behaviour product of sets defined between pairs of sets $\mathbf{K}(S^{i})$ defined over the the $S^{i}$ then inherits the associativity property of the ordinary product so that for example if $S^1, S^2, S^3$ form a 3-fold partition of $S$, say, then
\begin{align}
\begin{split}
    \left(\wp^{S^1}(\vec{a}^1|\vec{x}^1) \odot \wp^{S^2}(\vec{a}^2|\vec{x}^2)\right)\odot \wp^{S^3}(\vec{a}^{3}|\vec{x}^{3})\\
   = \wp^{S^1}(\vec{a}^1|\vec{x}^1) \odot \left(\wp^{S^2}(\vec{a}^2|\vec{x}^2)\odot \wp^{S^3}(\vec{a}^{3}|\vec{x}^{3})\right),
    \end{split}
\end{align}
which is an element in $\mathbf{NS}(S)$. 
  For our purposes, there is no need to consider the features of $n$-fold operations rigorously in full generality. Instead, we will just point out that since the scenario $S= (I,M,O)$ consists of finite sets, after a finite number $n=n^{max}$ of iterations of the nontrivial bipartition operation, one is left with a family of ''singleton scenarios`` $S^{k}= I^{k}, M^{k}, O^{k} $ with $ I^{k} = \{i_k \}, M_i^{k} =\{x_{i_{k}}\} \forall k$ and $O^{k}= \{O_{x_{i_k}} \}$. Such scenarios consist of a single agent $i_k \in I$ and a single measurement ${x_{i_k}} \in M_{i_k}$ and it is easy to convince oneself that $n^{max}= |\{S^{k}\}|= \prod_{i\in I}|M_i|$, since the bipartition operation is defined by a collection $M', M'^{\perp}$ of mutually disjoint subsets of inputs, and is nontrivial only when at least for one $i\in I$ the subset relation $M_i' \subset M_i$ is strict. 

One can then consider various sets of behaviour on the singleton scenarios: most of the definitions will trivialise, since behaviour on those scenarios will consists of just a single distribution. Nonetheless, since  Proposition $\ref{PropositionBehaviourproductProperties.PNS}$ in particular applies to the singletons as well, it is evident that the set $\mathbf{B}(S)$ can also be expressed as
\begin{align}
    \mathbf{B}(S) = \mathrm{conv}[\overset{n^{\max}}{\odot} \mathrm{Ext}(\mathbf{B}(S^{k}))],\label{Equation:B(S)isequaltoConvprodSingletonscenarios}
\end{align}
the convex hull of extreme points of Bell-behaviour in singleton scenarios. That this indeed is the case can be seen from the fact that the set $\mathbf{P}_{\mathbf{NS}}(S)$ always consists of (predictable) product behaviour. Equation \eqref{Equation:B(S)isequaltoConvprodSingletonscenarios} is a generalization of an observation made by Woodhead \cite{Woodhead2014}, who  pointed out the result of Eq.~\eqref{Equation:B(S)isequaltoConvprodSingletonscenarios} in bipartite scenarios. 

Appeal to the extreme cases with singleton scenarios could also be employed to provide another demonstration of the fact that in general $\mathbf{NS}(S') \odot \mathbf{NS}(S'^{\perp})  \subsetneq \mathbf{NS}(S)$ since if those sets were equal, then so would their $n^{max}$-fold product in the singleton scenarios as well. Since the sets can be equal only if their extreme points match, and the extreme points of each singleton consists of the predictable behaviour, a contradiction would follow. 

It is also possible to show that one does not explicitly need to deal with the extreme points of $\mathbf{B}(S)$ to establish an identity using the behaviour product.  To illustrate this, it  will be useful to spell out the following observation in the form of a lemma for ease of reference.

\lemma{\label{LemmaBheaviourproductsubsetlemma}
Let $S'$ and $S'^{\perp}$ be a bipartition of $S$ and suppose $\mathbf{K}^{\alpha}(S')$ and $\mathbf{K}^{\alpha^{\perp}}(S'^{\perp})$ are some subsets of behaviours for $\alpha \in \{1\ldots N' \}$ and $\alpha^{\perp} \in \{1, \ldots N'^{\perp} \}$ defined in the restricted scenarios respectively. If also $ \alpha_i \leq \alpha_j \Rightarrow \mathbf{K}^{\alpha_i}(S') \subset \mathbf{K}^{\alpha_j}(S') $ for all $\alpha_i, \alpha_j \leq N'$ and similarly for $S^{\perp}$, then 
\begin{align}
    \mathbf{K}^{\alpha'_i}(S') \odot \mathbf{K}^{\alpha^{\perp}_i}(S'^{\perp}) \subset \mathbf{K}^{\alpha'_j}(S') \odot \mathbf{K}^{\alpha^{\perp}_j}(S'^{\perp})
\end{align}
for all $\alpha'_i \leq \alpha'_j, \alpha^{\perp}_i \leq  \alpha^{\perp}_j$. In particular 
\begin{align}
    \mathbf{K}^{\alpha'_i}(S') \odot \mathbf{K}^{\alpha^{\perp}_i}(S'^{\perp}) \subset \mathbf{K}^{N'}(S') \odot \mathbf{K}^{N'^{\perp}}(S'^{\perp})
\end{align}
for all $\alpha_i$ and $\alpha^{\perp}_i$.
}
\begin{proof}
    Follows immediately from Definition \ref{BehaviourProductDefinition} of the behaviour product. 
\end{proof}
\normalfont
This can be used to establish the following theorem. 
\theorem{\label{TheoremProductBellPolytope}Let $S = (I, M, O)$ be a scenario and let $S', S'^{\perp}$ be a bipartition of $S$.  Then 
\begin{align}
    \mathrm{conv}[\mathbf{B}(S')) \odot \mathbf{B}(S'^{\perp})] = \mathbf{B}(S)
    \end{align}
\begin{proof}
    By Eq.~\eqref{equationExtofBSisExtofBehproduct}, Lemma  \ref{LemmaBheaviourproductsubsetlemma} and Proposition \ref{PropositionBehaviourproductProperties.BELLPRODUCT}
    \begin{align}
    \begin{split}
        \mathrm{Ext}(\mathbf{B}(S')) \odot \mathrm{Ext}(\mathbf{B}(S'^{\perp})) = \mathrm{Ext}(\mathbf{B}(S)) \\
        \subset   \mathbf{B}(S')\odot \mathbf{B}(S^{'\perp})\label{EquationSandwichingBellextremepoints1}
       \subset \mathbf{B}(S).
       \end{split}
    \end{align}
   Equation \eqref{EquationSandwichingBellextremepoints1} then implies, by taking the convex hull, that
   \begin{align}
       \mathbf{B}(S) \subset \mathrm{conv}(\mathbf{B}(S')\odot \mathbf{B}(S'^{\perp})) \subset \mathbf{B}(S),
   \end{align}
   which proves the claim.
\end{proof}
\normalfont

Proposition \ref{PropositionBehaviourproductProperties} and Theorem \ref{TheoremProductBellPolytope} show that there are some sets, such as $\mathbf{P}_{\mathbf{NS}}(S)$ and $\mathbf{B}(S)$, which can be constructed by the behaviour product $\odot$, or by taking the convex hull of  behaviour products of sets defined over some non-redundant sub-scenarios $S', S'^{\perp}$. It is easy to see that if $\mathbf{K}(S')\subset \mathbf{NS}(S)$ and $\mathbf{K}(S'^{\perp}) \subset \mathbf{NS}(S'^{\perp})$ are some finite sets of behaviour, then the set 
\begin{align}
  \mathbf{K}(S) =   \mathbf{K}(S') \odot \mathbf{K}(S'^{\perp})
\end{align} 
constructed from the behaviour product is also finite. Therefore by definition, the convex hull of the product of finite sets always produces a convex polytope. The generalization of Theorem \ref{TheoremProductBellPolytope} allows extension of the statement to include products of arbitrary polytopes themselves, since their extreme points are finite sets.

\theorem{\label{Theorem:generalConvexbehaviourproduct}Let $S = (I,M,O)$ be a scenario and $S', S'^{\perp}$ some arbitrary disjoint bipartition of $S$ defined by collections $M'$ $M'^{\perp}$. Let $\mathbf{K}(S') \subset \mathbf{NS}(S')$ and $\mathbf{K}(S'^{\perp})\subset \mathbf{NS}(S'^{\perp})$ be some arbitrary convex polytopes defined over the sub-scenarios $S', S'^{\perp}$. Then, the convex hull of the product 
\begin{align}
   \mathbf{K}(S) =  \mathrm{conv}[\mathbf{K}(S') \odot \mathbf{K}(S'^{\perp})]
\end{align}
is itself a convex polytope. Furthermore, its extreme points are precisely the products of the extreme points of $\mathbf{K}(S')$ and $\mathbf{K}(S'^{\perp}).$ That is
\begin{align}
    \mathrm{Ext}(\mathbf{K}(S)) = \mathrm{Ext}(\mathbf{K}(S') \odot \mathrm{Ext}(\mathbf{K}(S'^{\perp})). 
\end{align}
}
\begin{proof}
    Note that by definition
    \begin{align}
        \mathbf{K}(S) = \mathrm{conv}[\mathrm{conv}[\mathrm{Ext}(\mathbf{K}(S')))] \odot \mathrm{conv}[\mathrm{Ext}(\mathbf{K}(S'^{\perp}))]] \label{Equation:DefinitionofConvHullofProd}
    \end{align}
    which by application of Lemma \ref{LemmaConvexitynotPreservedInBehaviourproduct} implies that 
    \begin{align}
        \mathbf{K}(S) \subset \mathrm{conv}[\mathrm{Ext}(\mathbf{K}(S') \odot \mathrm{Ext}(\mathbf{K}(S'^{\perp}))].\label{Equation:KisConvExtK'andExtK'perp}
    \end{align}
     On the other hand, application of Lemma \ref{LemmaBheaviourproductsubsetlemma} on the right hand side of Eq.~\eqref{Equation:KisConvExtK'andExtK'perp} implies the converse, that is 
    \begin{align}
    \begin{split}
         \mathrm{conv}[\mathrm{Ext}(\mathbf{K}(S') \odot \mathrm{Ext}(\mathbf{K}(S'^{\perp}))] \\
         \subset \mathrm{conv}[\mathbf{K}(S') \odot \mathbf{K}(S'^{\perp})],
          \end{split}
    \end{align}
    which by comparison of Eqs.~\eqref{Equation:DefinitionofConvHullofProd} and \eqref{Equation:KisConvExtK'andExtK'perp} means that 
    \begin{align}
        \mathbf{K}(S) \subset  \mathrm{conv}[\mathrm{Ext}(\mathbf{K}(S') \odot \mathrm{Ext}(\mathbf{K}(S'^{\perp}))] \subset \mathbf{K}(S).
    \end{align}
    Therefore $\mathbf{K}(S)$ is identifiable as  the convex hull of the finite number of product behaviour, namely the extreme points of $\mathbf{K}(S')$ and $\mathbf{K}(S'^{\perp}).$
    To show that the products of the extreme points of $\mathbf{K}(S')$ and $\mathbf{K}(S'^{\perp})$ are  in fact the extreme points of $\mathbf{K}(S)$, we may use a similar technique as in the proof of  Theorem \ref{TheoremextremepointsofPDP} where the extreme points of the partially deterministic polytope $ \mathbf{PD}(S,M')$ where identified. 

    First note that Eq.~\eqref{Equation:KisConvExtK'andExtK'perp} already entails that $\mathrm{Ext}(\mathbf{K}(S))\subset \mathrm{Ext}(\mathbf{K}(S')) \odot \mathrm{Ext}(\mathbf{K}(S'^{\perp}))$ and hence it is sufficient to show that this relation can also be inverted. 

Suppose then that that  $\wp_{rk}  =\wp_{r}^{S'} \odot \wp^{S'^{\perp}}_{k}$ with $\wp^{S'}_r \in \mathrm{Ext}(\mathbf{K}(S'))$ and $\wp_{k}^{S'^{\perp}}\in \mathrm{Ext}(\mathbf{K}(S'^{\perp}))$ with $r,k$ labelling distinct extreme points in those polytopes. If it were the case that $\wp\notin \mathrm{Ext}(\mathbf{K}(S)$, then  $\wp_{rk} $ admits a convex decomposition $\wp = \sum_{j}\omega_j\wp_j$, $\omega_j \geq 0$ for all $j\in J$ and $\sum_j \omega_j =1$ with 
    \begin{align}
         \wp_{rk}(\vec{a}|\vec{x}) &= \sum_{j} \omega_j \wp_j^{S'}(\vec{a}_{F_{\vec{x}}}|\vec{x}_{F_{\vec{x}}}) \wp_j^{S'^{\perp}}(\vec{a}_{I\setminus F_{\vec{x}}}| \vec{x}_{I\setminus F_{\vec{x}}})\\
         & = \wp_r^{S'}(\vec{a}_{F_{\vec{x}}}|\vec{x}_{F_{\vec{x}}})\wp_k^{S'^{\perp}}(\vec{a}_{I\setminus F_{\vec{x}}}|\vec{x}_{I\setminus F_{x}}).
    \end{align}
    By considering the images of $\wp_{rk}$ under the restrictions $R_{|M'}$ and $R_{|M'^{\perp}}$ one  therefore gets
    \begin{align}
       \sum_j \omega_j \wp_j^{S'} = \wp_{r}^{S'} 
    \end{align}
    and 
    \begin{align}
        \sum_j \omega_j \wp_{j}^{S'^{\perp}} = \wp_k^{S'^{\perp}},
    \end{align}
    which would contradict the assumption of extremality of the behaviours labelled by $r,k$ in the sub-scenarios $S'$, $S'^{\perp}$, respectively. Therefore, it has to be that $\wp_{rk} \in \mathrm{Ext}(\mathbf{K}(S))$ also. 
\end{proof}
\normalfont

Theorem \ref{Theorem:generalConvexbehaviourproduct} suggests a powerful method to construct convex polytopes of behaviour from polytopes in sub-scenarios $S, S'^{\perp}$, generalizing the results of Proposition \ref{PropositionBehaviourproductProperties} and Theorem \ref{TheoremProductBellPolytope}. We can also show, however, that some convex polytopes of behaviour can not be constructed this way, and hence this property is a nontrivial feature of a class of sets.  This idea leads to the notion of 'composable sets', which we define in full generality. 

\definition{(Composable sets)\label{definition:composableset}\\
Let $S= (I,M,O)$ be a scenario and $\mathbf{K}(S)\subset \mathbf{NS}(S)$ an arbitrary subset of no-signalling behaviour in that scenario. Let $S'$ and $S'^{\perp}$ be some arbitrary non-redundant bipartition of $S$.  If 
    \begin{align}
        \mathbf{K}(S) = \mathbf{K}(S') \odot \mathbf{K}(S'^{\perp})
    \end{align}
    for some sets $\mathbf{K}(S') \subset \mathbf{NS}(S'),$ $\mathbf{K}(S'^{\perp}) \subset \mathbf{NS}(S'^{\perp})$ then the set $\mathbf{K}(S)$ is said to be strictly composed of the sets $\mathbf{K}(S') $ and $ \mathbf{K}(S')$ (with respect to the  product $\odot$). Similarly if,  
    \begin{align}
         \mathbf{K}(S) = \mathrm{conv}[\mathbf{K}(S') \odot \mathbf{K}(S'^{\perp})]
    \end{align}
then we say that $\mathbf{K}(S)$ is convex-composed (c-composed) of the sets $\mathbf{K}(S')$ and $\mathbf{K}(S')$. We refer to both types of sets as composable sets, unless further specification is warranted. 

If the set $\mathbf{K}(S)$ is not strictly composed or c-composed of any sets of sub-behaviours, we say that $\mathbf{K}(S)$ is not composable or, that the set $\mathbf{K}(S)$ is solid.  
}\\
\normalfont

Intuitively a solid set has, for every bipartition at least some behaviour which cannot be represented as a product or $c-$product of behaviour defined over the substructures in the product.  Indeed by definition, if a set $\mathbf{K}(S)^{\alpha}$ is solid, then for any composable set $\mathbf{K}(S)$ it holds that $\mathbf{K}^{\alpha}(S)\not\subset \mathbf{K}(S)$.

We have defined composability with respect to the bipartitions $S'$ and $S'^{\perp}$. Naturally one could consider $n$-fold partitions as well. Note however that if a set $\mathbf{K}(S)$ is strictly $n$-fold composable, that is if
\begin{align}
    \mathbf{K}(S) = \bigodot_{i = 1}^{n} \mathbf{K}(S^{i}),
\end{align}
then it is also composable in the sense of Definition \ref{definition:composableset}, for example with respect to any bipartition $S^{j}$ and $\blacktriangle_{i\neq j}^{n}S^{i}$. Similarly, $n$-fold c-composability entails the bipartite case. The converse clearly need not be true, and hence to prove solidity of $\mathbf{K}(S)$ with respect to any $n$-fold partition, it is sufficient to consider solidity in the sense of Definition \ref{definition:composableset}. For this reason, we leave more general investigations of the nyances related to $n$-fold composability for future work. 

The images of composable sets under arbitrary restriction maps behave in a very specific way, as was established for the case of strictly composable sets in Proposition \ref{Proposition:ImageRestrictionofBehproductistheDomain}. From Lemma \ref{Lemma:convexityofRestrictionmap} it follows that the image of a $c$-composable set under any restriction map is the convex hull of the image of the corresponding strictly composable set and thus the nontrivial properties under the restriction extend in a way to the $c$-composable sets as well. This feature turns out to be very useful for our later investigation. 

 From Proposition \ref{PropositionBehaviourproductProperties} we may extract that even if the images restriction maps agree with the hypothesis of strict composability, this does not alone entail that the set $\mathbf{K}(S)$ is strictly composable. Indeed, by Proposition \ref{Proposition:restrictionmapidentities} $R_{|M'}(\mathbf{B}(S)) = \mathbf{B}(S')$ for any collection $M'$, and thus one might expect that $\mathbf{B}(S)$ is strictly composed of $\mathbf{B}(S')$ and $\mathbf{B}(S'^{\perp})$ in accordance of the 'distributive' properties of the restriction map established in Proposition \ref{Proposition:ImageRestrictionofBehproductistheDomain}. In  Proposition \ref{PropositionBehaviourproductProperties} it was shown, however,  that the set $\mathbf{B}(S)$ is not strictly composed of $\mathbf{B}(S')$ and $\mathbf{B}(S')$, though Theorem \ref{TheoremProductBellPolytope} shows that  $\mathbf{B}(S)$ is $c$-composed of the sets $\mathbf{B}(S')$ and $\mathbf{B}(S')$. Evidently, strict composability is a stronger constraint on a set of behaviour than $c$-composability. 
 
The following result shows that there are also some sets which are solid, and thus that composability is a nontrivial property of a set of behaviour. 

\proposition{\label{Proposition:NS(s)isSolidSimplecase}Let $S = (I,M,O)$ be a bipartite scenario with $|M_i| = 2$ for all $i\in I$. The set $\mathbf{NS}(S)$ is not composable. }
\begin{proof}
    Suppose $\mathbf{NS}(S)$ were composable. Then there would exist some nonredundant partition of $S$ to $S'$ and $S'^{\perp}$ as defined by collections $M'$ and $M'^{\perp}$ along with  some sets $\mathbf{K}(S')\subset \mathbf{NS}(S')$ and $\mathbf{K}(S'^{\perp}) \subset \mathbf{NS}(S')$ such that 
    \begin{align}
        \mathbf{NS}(S) = \mathrm{conv}[\mathbf{K}(S') \odot \mathbf{K}(S'^{\perp})].
    \end{align}
    Then, by Propositions \ref{Proposition:restrictionmapidentities} and \ref{Proposition:ImageRestrictionofBehproductistheDomain} and Lemma \ref{Lemma:convexityofRestrictionmap} it follows in particular that
    \begin{align}
        R_{|M'}(\mathbf{NS}(S))= \mathbf{NS}(S') = \mathrm{conv}[\mathbf{K}(S')]
    \end{align}
    and 
    \begin{align}
        R_{|M'^{\perp}}(\mathbf{NS}(S)) = \mathbf{NS}(S'^{\perp}) = \mathrm{conv}[\mathbf{K}(S'^{\perp})]. 
    \end{align}
    Therefore, by use of Lemma \ref{LemmaBheaviourproductsubsetlemma}
    \begin{align}
        \mathbf{NS}(S) \subset \mathrm{conv}[\mathbf{NS}(S') \odot \mathbf{NS}(S'^{\perp})].\label{Equation:NSisSubsetofNS'andNs'perp}
    \end{align}
    On the other hand, any nonredundant bipartition of $S$ has in both subscenarios $S', S'^{\perp}$ at least one of the two agents with one input. Thus, by Proposition \ref{PropositionNSisBell} it has to be that $\mathbf{NS}(S') = \mathbf{B}(S')$ and $\mathbf{NS}(S'^{\perp}) = \mathbf{B}(S'^{\perp})$ and so by combining Theorem \ref{TheoremProductBellPolytope} with Eq.~\eqref{Equation:NSisSubsetofNS'andNs'perp} one gets
    \begin{align}
        \mathbf{NS}(S) \subset \mathbf{B}(S),
    \end{align}
    a contradiction. Therefore, $\mathbf{NS}(S)$ is not composable, as claimed. 
\end{proof}
\normalfont

The relevant feature used in the proof of Proposition \ref{Proposition:NS(s)isSolidSimplecase} was that possible nonredundant bipartitions of $S$  are essentially trivial scenarios, which strongly restricts the composable sets which can be constructed on $S$. This property turns out to be  one of the defining characteristic of composable sets beyond the simple case of Proposition \ref{Proposition:NS(s)isSolidSimplecase}, owing to the fact that the two-party two-input scenario is contained as a subscenario in any other non-trivial scenario $S$, as we will later see. For now, let us also point out that by a similar argument as in Proposition \ref{Proposition:NS(s)isSolidSimplecase} it can be established that the set of quantum behaviour $\mathbf{Q}(S)$ is solid in any bipartite scenario with two inputs per site. 

In this work, we are most interested in the properties of the partially deterministic polytopes $\mathbf{PD}(S,M')$. As it turns out, the behaviour product is a very useful tool to understand the structure of these objects. Let us begin, by providing an alternative characterization of the set $\mathbf{PD}(S,M')$ by means of the behaviour product as a specific instance of a composable set.

 \theorem{\label{TheoremExtremepointsPDPproductpreservationTHRM}Let $S = (I, M, O)$ be a scenario and let $M', M'^{\perp}$  define any disjoint bipartition of $S$. Then $\mathrm{\mathrm{Ext}}(\mathbf{B}(S')) \odot \mathrm{\mathrm{Ext}}(\mathbf{NS}(S'^{\perp})) = \mathrm{\mathrm{Ext}}(\mathbf{PD}(S,M')) $}
 \begin{proof}
     Immediate. See Definition \ref{BehaviourProductDefinition} of the behaviour product and Theorem \ref{TheoremextremepointsofPDP} where the form of the extreme points of $\mathbf{PD}(S,M')$ was derived.
 \end{proof}
 \normalfont

As is the case for $\mathbf{B}(S)$ in Theorem \ref{TheoremProductBellPolytope}, one need not define $\mathbf{PD}(S,M')$ in terms of c-composition of the extreme points of the respective sub-scenarios. 

 \theorem{\label{TheoremPDPProduct}Let $S = (I, M, O)$ be a scenario and let $M', M'^{\perp}$  define any disjoint bipartition of $S$ into $S'$ and $S'^{\perp}$. Then $\mathrm{conv}[\mathbf{B}(S') \odot \mathbf{NS}(S'^{\perp}) ] = \mathbf{PD}(S,M')$}
\begin{proof}
    By Theorem \ref{Theorem:generalConvexbehaviourproduct} $\mathrm{conv}[\mathbf{B}(S') \odot \mathbf{NS}(S'^{\perp}) ]$ as a c-composition of two polytopes $\mathbf{B}(S')$ and $\mathbf{NS}(S')$ is a convex polytope the extreme points of which are given by the products of the sub-polytopes as $\mathrm{\mathrm{Ext}}(\mathbf{B}(S')) \odot \mathrm{\mathrm{Ext}}(\mathbf{NS}(S'^{\perp}))$. Comparison with Theorem \ref{TheoremExtremepointsPDPproductpreservationTHRM}  regarding the extreme points of $\mathbf{PD}(S,M')$ proves the claim. 
\end{proof}
\normalfont

Theorems \ref{TheoremExtremepointsPDPproductpreservationTHRM} and \ref{TheoremPDPProduct} serve as our starting point to establishing general facts related to the polytopes $\mathbf{PD}(S,M')$. Let us begin by identifying precisely the sets of parameters which are necessary and sufficient for the partially deterministic set $\mathbf{PD}(S, M')$ to equal the Bell polytope $\mathbf{B}(S)$. 

\theorem{\label{TheoremMainPDequalBellTHRM}Let $S = (I, M, O)$ be a scenario and $M'$ a collection of subsets of of the $M_i$. Then $\mathbf{PD}(S,M') = \mathbf{B}(S)$ if and only $|M_i| - |M'_i| \in \{0, 1 \}$ for all but at most one $i \in I$. }
\begin{proof} 
For the 'if' part, note that by Theorem \ref{TheoremPDPProduct}
\begin{align}
    \mathbf{PD}(S,M') = \mathrm{conv}[\mathbf{B}(S') \odot \mathbf{NS}(S'^{\perp})]. 
\end{align}
On the other hand, by hypothesis, $|M_i| - |M'_i| = |M_i'^{\perp}|  \in \{0,1 \}$ for all but at most one $i \in I$, which means also that $|M_i'^{\perp}| \in \{0,1\}$ for all but at most one $i$. Hence, by Proposition \ref{PropositionNSisBell}, $\mathbf{NS}(S'^{\perp}) = \mathbf{B}(S'^{\perp})$ which combined with Theorem \ref{TheoremProductBellPolytope} implies that 
\begin{align}
    \mathbf{PD}(S,M') = \mathrm{conv}[\mathbf{B}(S') \odot \mathbf{B}(S'^{\perp})] = \mathbf{B}(S).
\end{align}
The 'only if' is easily seen by appeal to Proposition \ref{PropositionNSisBell}, by which $\mathbf{NS}(S'^{\perp}) = \mathbf{B}(S'^{\perp})$ if and only if said conditions hold,  and the properties of the restriction map $R_{|M^{'\perp}}$ which would otherwise lead to a contradiction. Namely, by Lemma \ref{Lemma:convexityofRestrictionmap}, Proposition \ref{Proposition:ImageRestrictionofBehproductistheDomain} and Proposition  \ref{Proposition:restrictionmapidentities} 
\begin{align}
  &R_{|M'{\perp}}(\mathrm{conv}[\mathbf{B}(S') \odot \mathbf{NS}(S'^{\perp})]
    = \mathrm{conv}[\mathbf{NS}(S'^{\perp})] \\ &= \mathbf{NS}(S'^{\perp}).
\end{align}
Thus, to have $\mathbf{NS}(S'^\perp) = \mathbf{B}(S'^{\perp})$ is necessary for $\mathbf{PD}(S,M') = \mathbf{B}(S)$ to hold, and hence also  that $ |M_i'^{\perp}| = |M_i|-|M_i'| \in \{0,1\}$ for all but at most one $i\in I$, as claimed. 
\end{proof}
\normalfont

\normalfont
 Intuitively, the reason why the equivalence class $\mathbf{PD}(S,M') = \mathbf{B}(S)$ emerges is due to the fact that the 'non-deterministic' sub-scenario $S'^{\perp}$ does not have enough space to allow for behaviour outside $\mathbf{B}(S'^{\perp})$ in the sense of Proposition \ref{PropositionNSisBell}.  As it turns out, the remaining equivalence classes of partially deterministic polytopes can be completely characterised by only slightly extending this intuition by generalizing the technique used in the proof of Proposition \ref{Proposition:NS(s)isSolidSimplecase}.   

First, let us demonstrate a simple but useful 'partial' analogue of Proposition \ref{PropositionNSisBell}, which we report as a proposition phrased in terms of the behaviour product.

\proposition{\label{PropositionNSisPDproposition} Let $S = (I, M, O)$ be a (possibly trivial) scenario and $V  = \{i\in I : |M_i| = 1 \}\subset I$  the subset of agents with a single input. Let $M_V' = \{ M_i' \subset M_i| M_i' = M_i \hspace{0.2cm} \textrm{if} \hspace{0.2cm} i\in V \hspace{0.2cm} \textrm{and} \hspace{0.2cm} M_i' = \emptyset \hspace{0.1cm} \textrm{otherwise.} \} $  
and $S_V', S_V'^{\perp}$ a bipartition of $S$ as defined by $M_V', M_V'^{\perp}$. Then
\begin{enumerate}[label=(\roman*), ref=\ref{PropositionNSisPDproposition}.\roman*] 
    \item  $\mathrm{\mathrm{Ext}}(\mathbf{NS}(S)) = \mathrm{\mathrm{Ext}}(\mathbf{B}(S_V')) \odot \mathrm{\mathrm{Ext}}(\mathbf{NS} (S_V'^{\perp})) \label{PropositionNSisPDproposition1}$
   
    \item $\mathbf{NS}(S) = \mathrm{conv}[\mathbf{B}(S_V') \odot \mathbf{NS}(S_V'^{\perp})]$.
    \label{PropositionNSisPDproposition2}
\end{enumerate}
}
\begin{proof}
Clearly the claims holds if $V = I$ or $V = \emptyset$, where either $\mathbf{NS}(S) = \mathbf{B}(S)$ by Proposition \ref{PropositionNSisBell}, or the product reduces to the identity map $\mathbf{NS}(S) \rightarrow \mathbf{NS}(S)$. It remains to show, that the claims hold if $\emptyset \neq V \neq I $.

Note that by Theorem \ref{TheoremExtremepointsPDPproductpreservationTHRM} the right hand side of Proposition \ref{PropositionNSisPDproposition1} can be identified as the set $\mathrm{\mathrm{Ext}}(\mathbf{PD}(S,M_V'))$ consisting of the extreme points of the partially deterministic polytope $\mathbf{PD}(S,M_V')$. By appeal to Theorem \ref{TheoremPDPProduct} then clearly $(i) \Rightarrow (ii)$, and hence it is sufficient to only prove $(i)$.  Furthermore, by Theorem \ref{TheoremExtremePointsInclusionBellPDPNS}, also $\mathrm{\mathrm{Ext}}(\mathbf{NS}(S)) \supset \mathrm{\mathrm{Ext}}( \mathbf{PD}(S,M_V'))$, which means that it is sufficient to show that under the given conditions also $\mathrm{\mathrm{Ext}}(\mathbf{NS}(S)) \subset \mathrm{\mathrm{Ext}}( \mathbf{PD}(S,M_V'))$. 
   
   Suppose then that $\wp\in \mathrm{\mathrm{Ext}}(\mathbf{NS}(S))$, which owing to no-signalling implies, in particular,  that $\wp(\vec{a}|\vec{x})$ can in any context\footnote{In fact, the $\Vec{x}_V$ could be omitted from the notation since they are always fixed hence make no difference for specifying the context. }, using the properties of conditional probabilities, be decomposed as
   \begin{align}
       \wp(\Vec{a}|\Vec{x}) &= \wp(\Vec{a}_V|\Vec{x}_{V})\cdot \wp (\Vec{a}_{(I_A \setminus V)}|\Vec{a}_V,\Vec{x}).
   \end{align}
   One can create a deterministic LHV model for the term $\wp(\vec{a}_V|\vec{x}_V )$ above, by introducing  $\Vec{\alpha} \in O_{\Vec{x}_V}$ and a deterministic measure $D(\Vec{a}_V| \alpha, \Vec{x}_V)$ and express
   \begin{align}
       \wp(\Vec{a}_V|\Vec{x}_{V}) = \sum_{\Vec{\alpha}}D(\Vec{a}_V|\Vec{x}_V, \Vec{\alpha}) \cdot \wp(\Vec{\alpha}|\Vec{x}_{V}),
   \end{align}
   with the understanding that 
   \begin{align}
       D(\Vec{a}_V|\Vec{x}_V, \Vec{\alpha}) = \begin{cases}
           1,& \textrm{if } \alpha_i = a_i \hspace{0.2cm} \forall i \in V \\
           0, & \textrm{otherwise.}
       \end{cases}
   \end{align}
   Therefore equivalently
   \begin{align}
        \wp(\Vec{a}|\Vec{x}) = \sum_{\Vec{\alpha}} \wp(\Vec{\alpha}|\Vec{x}_V) D(\Vec{a}_V|\Vec{\alpha}, \Vec{x}_V) \cdot \wp (\Vec{a}_{(I_A \setminus V)}|\alpha,\Vec{x}_{I\setminus V}),
   \end{align}
   which is to say that $\wp \in  \mathbf{PD}(S,M_V')$. It can be demonstrated that also $\wp \in \mathrm{\mathrm{Ext}}(\mathbf{PD}(S,M_V')$,  by appeal to a contradictory conclusion otherwise. Namely if  $\wp$ were not an extreme point of $\mathbf{PD}(S,M_V')$, then there there would exist some $\wp_j\in \mathbf{PD}(S,M_V')$ and parameters $ 1>\omega_j >0$, $\sum_{j} \omega_j = 1$ such that 
   \begin{align}
       \wp(\vec{a}|\vec{x}) = \sum_{j}\omega_j \wp_j(\vec{a}|\vec{x}).\label{EquationconvexcombinationofPDPisPropositionProof}
   \end{align} But since $\mathbf{PD}(S,M_V') \subset \mathbf{NS}(S)$, this would indicate that also $\wp_j\in \mathbf{NS}(S)$ for all $i$ and hence Eq.~\eqref{EquationconvexcombinationofPDPisPropositionProof} represents a nontrivial convex combination of points in $\mathbf{NS}(S)$, contradicting the claim that $\wp \in \mathrm{\mathrm{Ext}}(\mathbf{NS}(S))$ to begin with. Therefore it has to be that $ \wp \in \mathrm{\mathrm{Ext}}(\mathbf{PD}(S,M_V'))$ proving the claim.
\end{proof}
\normalfont

Proposition \ref{PropositionNSisPDproposition} is the natural statement that if the scenario $S$ has a subset of parties with only one input, then one can identify this set of parties and their inputs forming a sub-structure $S_V'$ which only allows deterministic correlations, and hence a no-signalling polytope $\mathbf{NS}(S)$ is identifiable as the partially deterministic polytope with that substructure $S_V'$ having determinism imposed on it.

Let us remark that the set $V=\{i \in I : |M_i| = 1 \}$ in the statement of Proposition \ref{PropositionNSisPDproposition} could in fact be replaced by any subset $V^* \subset V \subset I$ with analogous but seemingly stronger conclusions with respect to the bipartition of $S$ into $S_{V^*}', S_{V^*}'^{\perp}$.  More exactly, it is clear that if Proposition \ref{PropositionNSisPDproposition} holds for any $V^* \subset V$, then it holds for $V$ also. On the other hand, assuming that Proposition \ref{PropositionNSisPDproposition1} holds,  by defining a bipartition of  the scenario $S_{V}'$ into $S_{V^*}^*, S^*_{V^*}$ by the collection $M^*_{V^*} = \{M_i^* \subset M_i' | M_i^* = M_i'\hspace{0.2cm} \textrm{if} \hspace{0.2cm} i\in V^* \hspace{0.2cm} \textrm{and} \hspace{0.2cm} M_i^* = \emptyset \hspace{0.1cm} \textrm{otherwise.} \} $  it then follows that 
\begin{align}
    & \mathrm{Ext}(\mathbf{NS}(S)) =  \mathrm{\mathrm{Ext}}(\mathbf{B}(S_V')) \odot \mathrm{\mathrm{Ext}}(\mathbf{NS} (S_V'^{\perp})) \\
     &= \mathrm{\mathrm{Ext}}(\mathbf{B}(S_{V^*}^*) \odot \mathbf{B}(S^{*\perp}_{V^*})) \odot \mathrm{\mathrm{Ext}}(\mathbf{NS}(S_{V}'^{\perp})\\
     &= \mathrm{\mathrm{Ext}}(\mathbf{B}(S_{V^*}^{*})) \odot \left(\mathrm{\mathrm{Ext}}( \mathbf{B}(S_{V^*}^{*\perp})) \odot \mathrm{\mathrm{Ext}}(\mathbf{NS}(S'^{\perp}_{V})\right)\label{Equation:ArbitraryVequation3inNSequalPDP}\\
     &= \mathrm{\mathrm{Ext}}(\mathbf{B}(S'_{V^*})) \odot \mathrm{\mathrm{Ext}}(\mathbf{NS}(S'_{V^*})),
\end{align}
which is to say that equivalently the set $V$ can be replaced with any $V^* \subset V$ as claimed. Here the second line uses Proposition \ref{PropositionBehaviourproductProperties.PNS} (Eq.~\eqref{equationExtofBSisExtofBehproduct}) following a bipartition of $S_V'$ into $S^*_{V^*}*, S_{V^*}^{*\perp}$, the third line follows from the associativity of the behaviour product, and the last line from the observation that $S_{V^*}' = S_{V^*}^*$ and the application of Proposition \ref{PropositionNSisPDproposition1} to the scenario $S_{V^*}'^{\perp}$ identifiable on the right hand side of Eq.~\eqref{Equation:ArbitraryVequation3inNSequalPDP}.

We are now ready to complement the results of Theorem \ref{TheoremMainPDequalBellTHRM} by determining the equivalence classes of partially deterministic polytopes when neither equals the Bell-polytope $\mathbf{B}(S)$.

\theorem{\label{Theorem:nonBellPDP'sEqualIFF} Let $S = (I, M, O)$ be a scenario and $M'$, $M''$ distinct collections of subsets $M_i', M_i'' \subset M_i$ for all $i \in I$. Suppose that $\mathbf{PD}(S,M'), \mathbf{PD}(S,M'') \neq \mathbf{B}(S)$ and denote by $V = \{i \in I| M_i' \neq M_i'' \} \neq \emptyset$ the set of parties $i\in I$ for which the subsets $M_i'$ and $M_i''$ are distinct. Then $\mathbf{PD}(S,M') = \mathbf{PD}(S,M'') $ if and only if $|M_i|- |M'_i| \in \{0,1 \}$, $|M_i| - |M''_i| \in \{ 0,1 \}$ for all $i \in V \subset I $.}
\begin{proof}
    First we shall demonstrate that if $|M_i|-|M_i'| \in \{0,1\}$ and $|M_i|-|M_i''| \in \{0,1\}$ for all $i\in V = \{i\in I | M_i' \neq M_i'' \}$ then $\mathbf{PD}(S,M') = \mathbf{PD}(S,M'')$. Hence, this is a sufficient condition for two partially deterministic polytopes to equal each other. We then verify, that if $\mathbf{PD}(S,M'),\mathbf{PD}(S,M'') \neq \mathbf{B}(S)$, then this condition is also necessary for the equality to hold.

The proof proceeds by showing that 
\begin{align}
    \mathbf{PD}(S,M') = \mathbf{PD}(S,M'') = \mathbf{PD}(S,M^{\min}),
\end{align}
where $M^{\min}$ is a collection consisting of sets $M_i^{\min} = M_i$ if $i \in V$ and $M_i^{\min} = M_i' = M_i''$ otherwise. We will only prove one of the equalities, the other case is analogous. 

Let us then, in addition to $M^{\min}$, introduce the collection $M'^{\max}$ with $M_i^{\max }\subset M_i' $ $\forall i\in I$ where the subsets $M_i'^{\max} $ are defined such that  $|M_i| - |M_i'^{\max}| = 1$ for all $i \in V$ and $M_i'^{\max}= M_i'$ otherwise. By construction $M_i'^{\max} \subset M_i' \subset M_i$ $\forall i\in I$, and hence from Theorem \ref{TheoremPDpolytopeBasicInclusionTHRM} it immediately follows that 
\begin{align}
    \mathbf{PD}(S,M^{\min}) \subset \mathbf{PD}(S,M') \subset \mathbf{PD}(S,M'^{\max}).
\end{align}
It is therefore sufficient to show that $\mathbf{PD}(S,M^{\min}) = \mathbf{PD}(S,M'^{\max})$ or equivalently, that the sets of   extreme points of these two polytopes coincide.

By Theorem \ref{TheoremExtremepointsPDPproductpreservationTHRM}
    \begin{align}
    \begin{split}
       &\mathrm{\mathrm{Ext}} (\mathbf{PD}(S, M^{\min})) \\ = & \mathrm{\mathrm{Ext}}([\mathbf{B}(S^{\min}) )\odot \mathrm{\mathrm{Ext}}(\mathbf{NS}(S^{\min ^{\perp}}))
        \end{split}
    \end{align}
    and 
    \begin{align}
    \begin{split}
        & \mathrm{\mathrm{Ext}}(\mathbf{PD}(S,M'^{\max})) \\
        =& \mathrm{\mathrm{Ext}}(\mathbf{B}(S'^{\max})) \odot \mathrm{\mathrm{Ext}}(\mathbf{NS}(S'^{ ^{\max \perp}})).\label{Equation:maxEXTPDequationinTheorem}
        \end{split}
    \end{align}
    From the definition of $M'^{\max^{\perp}}$ it follows that for all $i\in V$ $|M_i'^{\max^{\perp}}| = 1$, and hence by Proposition \ref{PropositionNSisPDproposition1} 
    \begin{align}
\mathrm{\mathrm{Ext}}(\mathbf{NS}(S'^{\max^{\perp}})) = \mathrm{\mathrm{Ext}}(\mathbf{B}(S'^V)) \odot \mathrm{\mathrm{Ext}}(\mathbf{NS}(S'^{V^{\perp}})),\label{Equation:splittingofMaximuminproofofEqualityofPDP}
    \end{align}
     Where $S'^V$, $S'^{V^{\perp}}$ are a disjoint bipartition of $S'^{max^{\perp}}$ defined by the collection $M'^V= \{M_i^V \subset M_i'^{\max^{\perp}}| M_i'^V  =  M_i'^{\max^{\perp}} \hspace{0.2cm} \mathrm{if} \hspace{0.2cm} i \in V , \hspace{0.2cm} M_i'^{V}= \emptyset \hspace{0.2cm} \mathrm{otherwise} \}$. Putting the identity \eqref{Equation:splittingofMaximuminproofofEqualityofPDP} back to Eq.~\eqref{Equation:maxEXTPDequationinTheorem}, and observing that 
     
     \begin{align}
         \mathrm{\mathrm{Ext}}(\mathbf{NS}(S'^{V^{\perp}})) = \mathrm{\mathrm{Ext}}(\mathbf{NS}(S'^{\min})
     \end{align}
     and, in same vein, by  using Proposition \ref{PropositionBehaviourproductProperties.PNS}
\begin{align}
    \mathrm{\mathrm{Ext}}(\mathbf{B}(S^{\max}) \odot \mathrm{\mathrm{Ext}}(\mathbf{B}(S'^{V})) = \mathrm{\mathrm{Ext}}(\mathbf{B}(S^{\min}))
\end{align}
 the desired claim $\mathbf{PD}(S,M^{min}) = \mathbf{PD}(S,M'^{\max})$ is established. As stated before, a similar collection $M''^{\max}$ can be defined to prove the claim for the case of $\mathbf{PD}(S,M'')$, and since both are supersets of the same polytope $\mathbf{PD}(S,M^{\min})$ the equality $\mathbf{PD}(S,M') = \mathbf{PD}(S,M'')$ follows. It remains to show the 'only if' part of the claim. 

 First, let us observe that if $\mathbf{PD}(S,M'), \mathbf{PD}(S,M'') \neq \mathbf{B}(S)$ then from Theorem \ref{TheoremMainPDequalBellTHRM} it follows that for each polytope defined by $M', M''$ there exists, respectively at least two parties with at least two non-deterministic inputs. Formally, this means that there exists sets $K' = \{i \in I : |M_i| -|M_i'| \geq 2 \}$ and $ K'' = \{i\in I : |M_i|-|M_i''| \geq 2 \}$ with $|K'| \geq 2$ and $|K''| \geq 2$.
 
 Suppose then for the sake of demonstrating a contradiction that $\mathbf{PD}(S,M') = \mathbf{PD}(S,M'')$, but that there exists at least one $v \in V = \{i\in I : M_i' \neq M_i'' \}$ such that either  $|M_v|-|M_v'| \geq2$ or $|M_v|-|M_v''|\geq 2$. Then, it has to also be the case that $v\in K'$ or $V \in K''$. Suppose, without loss of generality that $v\in K'$. By definition of $V$ which entails $M_v' \neq M_v''$, there exists at least one input $\Tilde{x}_v \in M_v$ such that  simultaneously $\Tilde{x}_v \in M_v'$ and $\tilde{x}_v \notin M_v''$, with the latter being equivalent to $\tilde{x}\in M_v''^{\perp}$. From the properties of $K'$ and $K''$ it further follows that, in particular there exists at least one input $x_v'' \in M_v''^{\perp}$. Whether it is the case that $x_v'' \in M'_v$ or $x_v'' \in M_v'^{\perp}$, turns out to not play a in the argument. Now since $|K''| \geq 2$, there is also a party $k'' \in K'' \setminus \{v\}$ with a pair of inputs  $ x^{1}_{k''},x^{2}_{k''}\in M_i''^{\perp}$. We will soon show that again,  whether the $x^{1/2}$ are elements of $M_{k''}'$ or $M_{k''}'^{\perp}$ turns out not to make a difference for the argument.  
 
Owing to Theorem \ref{TheoremPDPProduct} we may write
 \begin{align}
 \mathbf{PD}(S,M') = \mathrm{conv}[\mathbf{B}(S') \odot \mathbf{NS}(S'^{\perp})]
 \end{align}
and 
\begin{align}
    \mathbf{PD}(S,M') = \mathrm{conv}[\mathbf{B}(S'') \odot \mathbf{NS}(S''^{\perp})].
\end{align}

 Consider the restriction to the bipartite scenario $S_{{|\{v,k''\}}}$ defined by the collection $M^{\{v,k''\}}$ with $M^{\{v,k''\}}_v = \{ \tilde{x}_v, x_v''\}$, $M^{\{v,k''\}}_{k''} = \{x^1_{k''}, x^2_{k''}$ \}, $M^{v,k''}_i = \emptyset$ otherwise. By construction, $\tilde{x}_v, x_v'' \in M_v''^{\perp}$ and $x^{1}_{k''}, x^{2}_{k''}\in M_{k''}''^{\perp}$ and therefore, from the fact that the restriction map preserves convexity as established in Proposition \ref{Lemma:convexityofRestrictionmap} and by the 'distributivity over behaviour product' property established in  Proposition \ref{Proposition:ImageRestrictionofBehproductistheDomain}  it is seen that
\begin{align}
\begin{split}
  &  R_{|\{v,k''\}}(\mathbf{PD}(S,M'')) \\
    &= \mathrm{conv}[R_{M^{\emptyset}}(\mathbf{B}(S'')) \odot R_{ M^{\{v,k''\}}}(\mathbf{NS}(S''^{\perp}))]
    \end{split}\\
    &= \mathbf{NS}(S_{\{v,k''\}}),
\end{align} 
where $M^{\emptyset} $ represents the collection of empty sets arising from the intersection $M''_i\cap M_i^{\{v,k''\}} = \emptyset$ for all $i\in I$. We will show that image of the same restriction map for the polytope $\mathbf{PD}(S,M')$ is different which gives the desired contradiction.

First, note that while by construction $\tilde{x_v}\in M'_i$, it may or may not be the case that the remaining three inputs in the collection $M^{v,k''}$ can be found in $M'$. From Theorem \ref{TheoremPDpolytopeBasicInclusionTHRM} it follows, however, that if $M'^{\max}$ is a collection where $M_v'^{\max} = M_v' \setminus \{x_v''\}$, $M^{\max}_{k''} = M_v'\setminus \{x^1_{k''}, x^2_{k''} \}$ and $M_i'^{max} = M_i'$ otherwise, then 
\begin{align}
   \mathbf{B}(S) \subset \mathbf{PD}(S,M') \subset \mathbf{PD}(S,M'^{\max})
\end{align}
and so, with the aid of Proposition \ref{Proposition:restrictionmapidentities}
\begin{align}
  \mathbf{B}(S_{\{v,k''\}}) \subset  R_{\{v,k''\}}(\mathbf{PD}(S,M') ) \label{Equation:ProvingDistinctnessEq1} \\
  \subset R_{\{v,k''\}}(\mathbf{PD}(S,M'^{\max})).
\end{align}
Now by construction  $M'^{\max}_{v} \cap \{\tilde{x}_v, x_v'' \} = \{\tilde{x}_v\}$ and $M'^{\max}_{k''}\cap \{x^{1}_{k''}, x^{2}_{k''}\} = \emptyset$ so that 
\begin{align}
\begin{split}
  & R_{|\{v,k''\}}( \mathbf{PD}(S,M'^{\max}) ) \\
  & = R_{|\{v,k''\}}(\mathrm{conv}[\mathbf{B}(S'^{\max}) \odot \mathbf{NS}(S'^{\max \perp})])
\end{split}\\
& = \mathrm{conv}(\mathbf{B}(S^{\{v,k''\}}_{|\{\tilde{x}_v \}} \odot \mathbf{NS}(S^{\{v,k''\} \perp}_{|\{\tilde{x}_v\}}),\label{Equation:provingDistinctnessofPDP's2}
\end{align}

with $S^{\{v,k''\}}_{|\{\tilde{x}_v \}}$ representing the sub-scenario consisting of a single agent $v$ with a single input $\tilde{x}_v$ and $S^{\{v,k''\} \perp}_{|\{\tilde{x}_v\}})$ the sub-scenario with two agents $v, k''$, with $v$ having a single input $x_v''$ and $k''$ having two inputs $x^1_{k''}$ and $x^2_{k''}$ which form a disjoint bipartition of $S_{\{v,k''\}}$.  From Proposition \ref{PropositionNSisBell} it follows that $\mathbf{NS}(S^{\{v,k''\} \perp}_{|\{\tilde{x}_v\}}) = \mathbf{B}(S^{\{v,k''\} \perp}_{|\{\tilde{x}_v\}})$ and hence by Theorem \ref{TheoremProductBellPolytope} and Eqs.~\eqref{Equation:ProvingDistinctnessEq1}-\eqref{Equation:provingDistinctnessofPDP's2}
\begin{align}
    \mathbf{B}(S_{\{v,k''\}}) \subset  R_{|\{v,k''\}}(\mathbf{PD}(S,M'))  \\
    \subset  R_{|\{v,k''\}}( \mathbf{PD}(S,M'^{\max}) )  \subset  \mathbf{B}(S_{\{v,k''\}}).
\end{align}
Therefore $R_{|\{v,k''\}}(\mathbf{PD}(S,M')) = \mathbf{B}(S_{\{v,k''\}}) \neq \mathbf{NS}(S_{\{v,k''\}}) = R_{|\{v,k''\}}(\mathbf{PD}(S,M''))$ which contradicts the assumption that  $\mathbf{PD}(S,M') = \mathbf{PD}(S,M'')$, for if the sets were equal, so would the images under the same restriction map. 

\end{proof}
\normalfont
 
It is worth emphasizing  that for the proof of the  'if part' of Theorem \ref{Theorem:nonBellPDP'sEqualIFF} above the hypothesis that neither of the partially deterministic polytopes $\mathbf{PD}(S,M')$, $\mathbf{PD}(S,M'')$ equals the Bell-polytope $\mathbf{B}(S)$ was not strictly needed and so the condition that $|M_i|-|M_i'| \in \{0,1\}$ and $|M_i|-|M_i''| \in \{0,1\}$ for all $i\in V = \{i\in I | M_i' \neq M_i'' \}$ is generally sufficient for the two polytopes to equal each other. The 'only if' part of the claim is sensitive to the assumption that the polytopes are not equal to the Bell polytope however, and since the equivalence class corresponding to the Bell polytopes was completely classified in Theorem \ref{TheoremMainPDequalBellTHRM}, we have ruled that situation in the proof. 

Theorem \ref{Theorem:nonBellPDP'sEqualIFF} imposes strong constraints for the collections $M'$ and $M''$ to define the same non-Bell partially deterministic polytope: the collections may differ only for parties $i\in I$ which have at most one non-deterministic input. Armed with this result, we could straight away sharpen Theorem \ref{TheoremPDpolytopeBasicInclusionTHRM} by determining under what conditions partially deterministic polytopes satisfy a strict inclusion and investigate other consequences of this conclusion. We believe it to be illuminating to digress for a moment however, and frame the result in the broader context of composable sets which in the specific case of partially deterministic polytopes leads to a simple statement based on which the distinct equivalence classes can be determined. 

Verily, an alternative, perhaps more intuitive way of looking at the result of Theorem \ref{Theorem:nonBellPDP'sEqualIFF} can be found by observing that the proof employed the fact that the image of the map  $R_{|\{v,k''\}}$ for one of the polytopes $\mathbf{PD}(S,M'')$ was distinct (in fact solid, by Proposition \ref{Proposition:NS(s)isSolidSimplecase}), in contrast to the the image of the other partially deterministic polytope $\mathbf{PD}(S,M')$ owing to their different composability structure. This is essentially a consequence of Proposition \ref{PropositionNSisPDproposition}, where the no-signalling polytope $\mathbf{NS}(S)$ was shown, under certain conditions, to be c-composed of Bell- and no-signalling sub-polytopes $\mathbf{B}(S')$ and $\mathbf{NS}(S'^{\perp})$. Using this, we can further sharpen Proposition \ref{PropositionNSisPDproposition} to a general statement about solidity of $\mathbf{NS}(S)$.

\theorem{\label{Theorem:strictcomposabilityOfNS} Let $S = (I,M,O)$ be an arbitrary scenario, and let $S'$ $S'^{\perp}$ be any nonredundant bipartition of $S$. Then $\mathbf{NS}(S)$ is composable , that is there exists some sets $\mathbf{K}(S') \subset \mathbf{NS}(S')$ and $\mathbf{K}(S'^{\perp}) \subset \mathbf{NS}(S'^{\perp})$ such that
\begin{align}
    \mathbf{NS}(S) = \mathrm{conv}[\mathbf{K}(S') \odot \mathbf{K}(S'^{\perp})],
\end{align}
if and only $|M_i|=1$ for at least one $i\in I$. 
}
\begin{proof}
  The 'if' part of the claim was established in Proposition \ref{PropositionNSisPDproposition} where it was shown that for $\emptyset \neq V = \{i \in I: |M_i| = 1\}$ the set $\mathbf{NS}(S)$ is c-composed of $\mathbf{B}(S'_V)$ and $\mathbf{NS}(S'^{\perp}_V)$, where $S_V'$ and $S_v'^{\perp}$ are a bipartition formed by having all parties in the set $V$ with their inputs in  $S'_V$ and the remaining parties and their inputs in $S_V'^{\perp}$. Thus, it suffices to show that this is a necessary condition for composability of $\mathbf{NS}(S).$ 

  Suppose then that $V =\emptyset$ and $\mathbf{NS}(S)$ were composable. Since the scenario $S$ is nontrivial, it contains in particular bipartite two-input sub-scenarios $S_{\{i, j \}}$, defined  by  collections of the form $M^{\{k,j\}}$ with $M^{\{k,j\}}_j = \{x_j, x_j' \}$, $M^{\{k,j\}}_k = \{x_k, x_k' \}$ and $M^{\{k,j\}}_i = \emptyset$ otherwise. Consider the image of $\mathbf{NS}(S)$ under any  restriction $R_{\{k,j\}}$ to such subscenarios. From the hypothesis of composability, the use of Propositions \ref{Proposition:restrictionmapidentities} and \ref{Proposition:ImageRestrictionofBehproductistheDomain} along with the fact that the restriction map preserves convexity as established in  Lemma \ref{Lemma:convexityofRestrictionmap} it therefore follows that 
  \begin{align}
      R_{\{k,j\}}(\mathbf{NS}(S)) &= \mathbf{NS}(S_{\{k,j\}})\label{Equation:NSisSolidEqOfRESTRICTION} \\
      =& \mathrm{conv}[\mathbf{K}(S_{|M'\cap M^{\{k,j\}}}) \odot \mathbf{K}(S_{|M'^{\perp}\cap M^{\{k,j\}}})],
  \end{align}
where as before, $M'\cap M^{\{k,j\}}$ refers to the collection of intersections $M_i'\cap M_i^{\{k,j\}}$ for all $i\in I$ and so on. Since by hypothesis $S'$ and $S'^{\perp}$ form a nonredundant bipartition of $S$, it must be, that some pairs $k,j$  and subsets  $\{x_k, x_k' \} \subset M_k$ and $\{x_j, x_j' \} \subset M_j$ can be found with the property  that $M'_k\cap M^{\{k,j\}}_k \neq \emptyset$ and $M_j'^{\perp} \cap M_j^{\{k,j\}} \neq \emptyset$. This means that for such collections the subscenarios $S_{|M'\cap M^{\{k,j\}}}$ and $S_{|M'^{\perp}\cap M^{\{k,j\}}}$ become trivial, leading by the use of Proposition \ref{PropositionNSisBell} to the sequence
\begin{align}
\begin{split}
 &\mathrm{conv}[\mathbf{K}(S_{|M'\cap M^{\{k,j\}}}) \odot \mathbf{K}(S_{|M'^{\perp}\cap M^{\{k,j\}}})] \\
& \subset  \mathrm{conv}[\mathbf{NS}(S_{|M'\cap M^{\{k,j\}}}) \odot \mathbf{NS}(S_{|M'^{\perp}\cap M^{\{k,j\}}})]\\
 \end{split}
  \\[2ex]
& =  \mathrm{conv}[\mathbf{B}(S_{|M'\cap M^{\{k,j\}}}) \odot \mathbf{B}(S_{|M'^{\perp}\cap M^{\{k,j\}}})]\\
& =  \mathbf{B}(S_{\{k,j\}})
\end{align}
  contradicting Eq.~\eqref{Equation:NSisSolidEqOfRESTRICTION} and completing the proof. 
\end{proof}
\normalfont

\corollary{\label{corollary:NSisSolidoraPDP}Let $S = (I,M,O)$ be any, possibly trivial, scenario and $V = \{i\in I : |M_i| = 1\}$. Then $\mathbf{NS}(S)$ is composable and 
\begin{align}
 \mathbf{NS}(S)  =  \mathrm{conv}[\mathbf{B}(S'_V) \odot \mathbf{NS}(S_V'^{\perp})] = \mathbf{PD}(S,M'^{V})
    \end{align}
if and only if  $V \neq \emptyset$. }
\begin{proof}
    This is a corollary of Proposition \ref{PropositionNSisPDproposition} and Theorem \ref{Theorem:strictcomposabilityOfNS}.
\end{proof}
\normalfont

By a similar argument, one can also see an analogous result to Theorem \ref{Theorem:strictcomposabilityOfNS} holds for the composability of the  quantum set $Q(S)$. 

\theorem{\label{Theorem:strictcomposabilityresultofQ(S)}Let $S = (I,M,O)$ be an arbitrary, possibly trivial scenario and $V = \{i\in I : |M_i|=1 \}$. Then the quantum set $\mathbf{Q}(S)$ is composable, and 
\begin{align}
    \mathbf{Q}(S) = \mathrm{conv}[\mathbf{B}(S'_V) \odot \mathbf{Q}(S'^{\perp}_{V})] 
\end{align}
if and only if $V\neq \emptyset$. Here $S'_V$ is defined by the collection $M'^V$ satisfying $M'^V_i= M_i$ if $i\in V$ and $M'^{V}_i = \emptyset$ otherwise, and conversely for $S_V'^{\perp}$.}
\begin{proof}
That $\mathbf{Q}(S)$ is composable only if $V \neq \emptyset$, can be shown exactly in the same manner as Theorem \ref{Theorem:strictcomposabilityOfNS} for the composability of the no signalling set $\mathbf{NS}(S)$. We will therefore only prove the 'if' part.

        Note first that if $V \neq \emptyset$ then, using Corollary \ref{corollary:NSisSolidoraPDP},
        \begin{align}
            \mathbf{Q}(S)\subset \mathbf{NS}(S) = \mathrm{conv}[\mathbf{B}(S'_V) \odot \mathbf{NS}(S'^{\perp})].\label{Equation:QissubsetofComposableNS}
        \end{align}
        If, however $\mathbf{Q}(S)$ were solid, then with any composable set $\mathbf{K}(S)$ one would expect $\mathbf{Q}(S)\not\subset \mathbf{K}(S)$, which would contradict Eq.~\eqref{Equation:QissubsetofComposableNS}. Hence, it has to be that $\mathbf{Q}(S)$ is itself composable. 

        Suppose then that $\mathbf{Q}(S)= \mathrm{conv}[\mathbf{K}(S')\odot \mathbf{K}(S'^{\perp})]$ for some bipartition of $S$ into $S'$ and $S'^{\perp}$.
        By Proposition \ref{Proposition:restrictionmapidentities} $R_{|M'}(\mathbf{Q}(S)) = \mathbf{Q}(S_{|M'})$ for any collection $M'$, and as so when combined with Eq.~\eqref{Equation:QissubsetofComposableNS}, in particular 
        \begin{align}
           & R_{|M'^{ V \perp}} (\mathbf{Q}(S)) = \mathbf{Q}(S_V'^{\perp})\\
            &= \mathrm{conv}[R_{M'\cap M'^{V \perp}}(\mathbf{K}(S)) \odot R_{M'\cap M'^{V \perp}}\mathbf{K}(S'^{\perp}))].
        \end{align}
        The image $\mathbf{Q}(S_V'^{\perp})$ is, by the 'only if' condition established before, solid and therefore, it has to be that either $M'^V \cap M' = M'^{V \perp}$ and $M'^V \cap M' = \emptyset$ so that  $\mathrm{conv}[\mathbf{K}(S'^{\perp})] = \mathbf{Q}(S'^{\perp}_V)$  or $M'^V \cap M' = \emptyset$ and $M'^V \cap M'^{\perp} = M'^{V \perp}$ with $\mathrm{conv}[\mathbf{K}(S')] = \mathbf{Q}(S'^{\perp}_V)$. Suppose without loss of generality, that the former. Then, the composition of a quantum set  has to have the form
        \begin{align}
            \mathbf{Q}(S) &= \mathrm{conv}[\mathbf{K}(S'_V) \odot \mathbf{K}(S'^{\perp}_{V})]. \label{Equation:NecessarycompositionformofqSET}
        \end{align} with  $\mathrm{conv}[\mathbf{K}(S')] = \mathbf{Q}(S'^{\perp}_V)$.

        On the other hand, the image $R_{|M'^V}(\mathbf{Q}(S))$ has to satisfy, by using Eq.~\eqref{Equation:NecessarycompositionformofqSET}
        \begin{align}
          &\mathbf{B}(S'_V) \subset  R_{|M'^V}(\mathbf{Q}(S)) \\
          &= \mathrm{conv}[ R_{|M'^V}\mathbf{K}(S'_V) \odot R_{|M'^V } K](S'^{\perp}_V)\\ 
          & \subset  R_{|M'^V}(\mathbf{NS}(S)) = \mathbf{B}(S'_V),
        \end{align}
        so that $\mathrm{conv}[\mathbf{K}(S'_V)] = \mathbf{Q}(S'_V) = \mathbf{B}(S'_V)$. By use of Lemmas  \ref{LemmaConvexitynotPreservedInBehaviourproduct} and \ref{LemmaBheaviourproductsubsetlemma} it can be then seen that 
        \begin{align}
            \mathbf{Q}(S) &= \mathrm{conv}[\mathbf{K}(S'_V) \odot \mathbf{K}(S'^{\perp}_{V})] \\ &\subset \mathrm{conv}[\mathbf{B}(S'_V) \odot \mathbf{Q}(S'^{\perp}_V)] \\
            & \subset \mathrm{conv}[\mathbf{K}(S'_V) \odot \mathbf{K}(S'^{\perp}_{V})] = \mathbf{Q}(S),
        \end{align}
        which proves the claim. 
\end{proof}
\normalfont

Corollary \ref{corollary:NSisSolidoraPDP} specifically,  along with the definition of a partially deterministic polytope $\mathbf{PD}(S,M')$ as the composable set $\mathrm{conv}[\mathbf{B}(S') \odot \mathbf{NS}(S'^{\perp})]$ provides another way to view Theorems \ref{TheoremMainPDequalBellTHRM} and \ref{Theorem:nonBellPDP'sEqualIFF} which together establish all the equivalence classes of partially deterministic polytopes. Namely, by Corollary \ref{corollary:NSisSolidoraPDP} exactly one of the two situations can happen: either the  'nondeterministic component' $\mathbf{NS}(S'^{\perp})$ is solid which happens if and only if $|M_i'^{\perp}| \geq 2$ for all $i\in I'^{\perp}$, in which case it cannot be further expanded as a composed set, or it is another partially deterministic polytope $\mathbf{PD}(S'^{\perp},M'')$ defined over the scenario $S'^{\perp}$. If the latter is the case, one can use the associativity of the behaviour product to absorb the 'deterministic component' $\mathbf{B}(S'')$ of that polytope to that of the original expression $\mathbf{B}(S')$, and hence be left with another equivalent expression of the same partially deterministic polytope defined with respect to a deterministic  collection $M^* := M'\cup M''$ with a non-deterministic component $\mathbf{NS}(S'^*)=\mathbf{NS}(S''^{\perp})$. Clearly this process can be repeated until the non-deterministic component is solid (or empty) in which case all the freedom  in the composition of $\mathbf{PD}(S,M')$ has been exhausted. To formalise this observation, let us introduce the notion of a maximal solid fragment, as a useful tool for classifying  partially deterministic polytopes. 

\definition{\label{Definition:solidfragmentofPDP}(Maximal solid fragment)\\
Let $S = (I,M,O)$ be a scenario, and $M'$ some arbitrary collection of subsets $M_i' \subset M_i$ for all $i\in I$ of inputs which defines the partially deterministic polytope 
\begin{align}
    \mathbf{PD}(S,M') = \mathrm{conv}[\mathbf{B}(S')\odot \mathbf{NS}(S'^{\perp})].
\end{align} 
Let $M'^{\perp \min}$ be the collection of subsets of inputs defined by $M_i'^{\min} = \emptyset$, if $|M_i'^{\perp}|\in \{0,1\}$ and $M_i'^{\min}= M_i'^{\perp}$ otherwise and $V'^{\min} = \{i\in I : M_i'^{\min} \neq \emptyset \}$. The object defined by the mapping
\begin{align}
\begin{split}
     & \mathrm{MSF}(\mathbf{PD}(S,M'))\\ &:= \begin{cases}
          \bot, \hspace{0,3cm} \mathrm{if} \hspace{0,1cm} |V'^{\min}|   \in \{0,1\} \\
          R_{|M'^{\min}}(\mathbf{PD}(S,M'))&  \mathrm{otherwise, }
     \end{cases}
       \end{split}
\end{align}
is termed the maximal solid fragment of the  polytope $\mathbf{PD}(S,M')$.
If $\mathrm{MSF}(\mathbf{PD}(S,M'))= \bot,$ we say that $\mathbf{PD}(S,M')$ has no solid fragments.  
}
\\
\normalfont

By construction,  a partially deterministic polytope $\mathbf{PD}(S,M')$ has no solid fragments if and only if it is equal to the Bell polytope $\mathbf{B}(S)$. Otherwise, the maximal solid fragment satisfies 
\begin{align}
    \mathrm{MSF}(\mathbf{PD}(S,M')) = \mathbf{NS}(S'^{\perp}_{\max }),
\end{align} where $S'^{\perp}_{\max }$  is the largest sub-scenario of $S'^{\perp}$ over which the set $\mathbf{NS}(S'^{\perp}_{\max })$ is not composable. We do not claim that this definition is useful for the study of composable sets in general. Given Corollary \ref{corollary:NSisSolidoraPDP}, however, this can be identified precisely as the relevant structure to distinguish the different equivalence classes of partially deterministic polytopes. 

Indeed, by construction $\mathbf{PD}(S,M') = \mathbf{PD}(S,M'')$ if and only if both polytopes have no solid fragments, in which case they both equal the Bell polytope,  or 
\begin{align}
    \mathrm{MFS}(\mathbf{PD}(S,M')) = \mathrm{MFS}(\mathbf{PD}(S,M'')).
\end{align}
    Note that if the maximal solid fragments are not equal, then the sets $\mathrm{MSF}(\mathbf{PD}(S,M'))$ and $\mathrm{MFS}(\mathbf{PD}(S,M''))$ are strictly speaking defined over different substructures $S_{\max}'^{\perp}$ and $S_{\max}''^{\perp}$ and therefore, some care needs to be exhibited when discussing their relations.  
    
    We may, however, endow the space of scenarios with a  partial order $\preceq$ which is read as 'is a subscenario of'. In particular, if a scenario $S$ is fixed, then there is a natural (unique) maximal element in the set of all scenarios $S_{|M'}'$ which are some restrictions of $S$, with respect to the relation $\prec$, namely the scenario $S$ itself so that $S_{|M'}' \preceq S $ for all $M'$ with equality if and only if $M_i' = \emptyset$ for all $i\in I$. Note that this is really a partial order, since some scenarios, such as $S'$ and $S'^{\perp}$ do not share inputs, and hence neither scenario is a subscenario of the other. 
    
    Owing to the partial ordering of the scenarios and the properties for the image $R_{|M'}(\mathbf{NS}(S)) = \mathbf{NS}(S')$, of the no-signalling set under any restriction map $R_{|M'}$, it follows however that the 'subscenario' relation $'\preceq'$ can be extended to a a partial ordering relation between fragments (which we denote by the same symbol), so  that an expression like 
\begin{align}
    \mathrm{MFS}(\mathbf{PD}(S,M')) \preceq \mathrm{MFS}(\mathbf{PD}(S,M'')),
\end{align} 
is well defined and can be understood to indicate that either their solid fragments are equal, or $  S_{\max}'^{\perp} \prec S''^{\perp}_{\max}$. In the latter case, it follows in particular that if $M'^{\min}$ is a collection like the one in Definition \ref{Definition:solidfragmentofPDP} then
\begin{align}
    &R_{|M'^{\min}}(\mathrm{MSF}(\mathbf{PD}(S,M'')) = R_{|M'^{\min}}(\mathbf{NS}(S_{\max}''^{\perp})) \\
    &= \mathbf{NS}(S_{\max}'^{\perp}) =\mathrm{MSF}(\mathbf{PD}(S,M')),
\end{align}
which is to say that a restriction of the solid fragment $\mathrm{MSF}(\mathbf{PD}(S,M''))$  exists which recovers the fragment $\mathrm{MSF}(\mathbf{PD}(S,M'))$. Clearly, the relation can not be inverted, since the scenario $S_{\max}''^{\perp}$ has, by definition, some inputs not present in $S_{\max}'^{\perp}$. Therefore,  we may also say that $\mathrm{MFS}(\mathbf{PD}(S,M'))$ is strictly contained in  $\mathrm{MFS}(\mathbf{PD}(S,M''))$.

Conventionally, we take $\bot$ to represent a minimum element in the partially ordered set of solid fragments, so that 
\begin{align}
    \bot \preceq \mathrm{MSF}(\mathbf{PD}(S,M'))
\end{align}
holds for all polytopes $\mathbf{PD}(S,M')$, with equality attained if and only if $\mathrm{MSF}(\mathbf{PD}(S,M')) = \bot$.

\theorem{\label{Theorem:EquivalenceClassesofPDP's}Let $S = (I, M, O)$ be a scenario and $M', M''$ be arbitrary collections of subsets of inputs $M_i', M_i'' \subset M_i$ for all $i\in I$ Then exactly one of the  following statements holds.
\begin{enumerate}[label=(\roman*), ref=\ref{Theorem:EquivalenceClassesofPDP's}.\roman*]

    \item \label{Theorem:EquivalenceClassesofPDP's1BELL} $ \mathrm{MSF}(\mathbf{PD}(S,M')) =  \mathrm{MSF}(\mathbf{PD}(S,M'')) = \bot$ and $\mathbf{PD}(S,M')=\mathbf{PD}(S,M'') = \mathbf{B}(S)$. 
    
    \item \label{Theorem:EquivalenceClassesofPDP's2NOTBELLEQUIV}  $\mathrm{MSF}(\mathbf{PD}(S,M')) =  \mathrm{MSF}(\mathbf{PD}(S,M'')) \neq \bot $ and  $\mathbf{PD}(S,M') =  \mathbf{PD}(S,M'') \neq \mathbf{B}(S)$. 
    
    \item \label{Theorem:EquivalenceClassesofPDP's3STRICTRELATION} $ \mathrm{MSF}(\mathbf{PD}(S,M')) \prec  \mathrm{MSF}(\mathbf{PD}(S,M''))$ and $\mathbf{PD}(S,M') \subsetneq \mathbf{PD}(S,M'')$

 \item \label{Theorem:EquivalenceClassesofPDP's3STRICTRELATIONOtherway} $ \mathrm{MSF}(\mathbf{PD}(S,M')) \succ \mathrm{MSF}(\mathbf{PD}(S,M''))$ and $\mathbf{PD}(S,M') \supsetneq \mathbf{PD}(S,M'')$

 \item \label{Theorem:EquivalenceClassesofPDP's4STRICTLYDISTINCT} Both $ \mathrm{MSF}(\mathbf{PD}(S,M')) \nprec  \mathrm{MSF}(\mathbf{PD}(S,M''))$ and $ \mathrm{MSF}(\mathbf{PD}(S,M')) \nsucc \mathrm{MSF}(\mathbf{PD}(S,M''))$ in which case $\mathbf{PD}(S,M') \not\subset \mathbf{PD}(S,M'') $ and $\mathbf{PD}(S,M') \not\supset \mathbf{PD}(S,M'') $

\end{enumerate}
}
\begin{proof}
These are straightforward consequences of results and definitions presented so far. That cases $(i)$ and $(ii)$ hold,  in particular, was already pointed out in the discussion following Definition \ref{Definition:solidfragmentofPDP}. For completeness, we provide a sketch of how the validity of $(iii) $ and $(v)$ can be established.

Case $(iii)$: By definitions
\begin{align}
    \mathbf{PD}(S,M') = \mathrm{conv}[\mathbf{B}(S')\odot \mathbf{NS}(S'^{\perp})]
\end{align}
and either $\mathrm{MSF}(\mathbf{PD}(S,M')) = \bot$, in which case $\mathbf{PD}(S,M') = \mathbf{B}(S)$, by Theorem \ref{Theorem:EquivalenceClassesofPDP's1BELL} or $\mathrm{MSF}(\mathbf{PD}(S,M')) = \mathbf{NS}(S'^{\perp}_{\max}) $. In the first case, the claim follows from $\mathrm{MSF}(\mathbf{PD}(S,M'')) \succ  \bot$ which implies that $\mathbf{PD}(S,M'') \neq \mathbf{B}(S)$ and the fact that $\mathbf{B}(S) \subset \mathbf{PD}(S,M'')$ which follows from Theorem \ref{TheoremPDpolytopeBasicInclusionTHRM}. Suppose then that $\mathrm{MSF}(\mathbf{PD}(S,M')) \neq \bot$. In that case,  $\mathrm{MSF}(\mathbf{PD}(S,M')) = \mathbf{NS}(S'^{\perp}_{\max})$ and from \ref{Theorem:EquivalenceClassesofPDP's2NOTBELLEQUIV}
\begin{align}
    \mathbf{PD}(S,M')  &= \mathrm{conv}[\mathbf{B}(S')\odot \mathbf{NS}(S'^{\perp})] \\
    &= \mathrm{conv}[\mathbf{B}(S'_{\max})\odot \mathbf{NS}(S_{\max}'^{\perp})]
\end{align}
and similarly
\begin{align}
    \mathbf{PD}(S,M'') &= \mathrm{conv}[\mathbf{B}(S'')\odot \mathbf{NS}(S''^{\perp})] \\
    &= \mathrm{conv}[\mathbf{B}(S_{\max}')\odot \mathbf{NS}(S_{\max}'^{\perp})].
\end{align}
Since by assumption $\mathrm{MSF}(\mathbf{PD}(S,M') \prec \mathrm{MSF}(\mathbf{PD}(S,M''))$ it follows, by writing $M^* = M'^{\min}\cap M''^{\min}$,   in particular that
\begin{align}
\begin{split}
  & R_{|M''^{\min}}(\mathbf{PD}(S,M')) \\
    &= \mathrm{conv}[R_{|M^*} (\mathbf{B}(S'_{\max})) \odot R_{|M^*}(\mathbf{NS}(S'^{\perp}_{\max}))]
    \end{split}
    \\[2ex]
    &\neq R_{|M''^{\min}}(\mathbf{PD}(S,M'')) = \mathrm{MSF}(\mathbf{PD}(S,M'')),
\end{align}
where the inequivalence is a consequence of $S'^{\perp}_{\max} \prec S''^{\perp}_{\max}$, which entails that $M^*\neq M''^{\min}$. On the other hand, $S'^{\perp}_{\max} \prec S'^{\perp}_{\max}$ implies conversely that $S''_{\max}\prec S'_{\max}$, since they correspond to two distinct bipartitions of $S$. From Theorem \ref{TheoremPDpolytopeBasicInclusionTHRM} it then follows that $\mathbf{PD}(S,M') \subset \mathbf{PD}(S,M'')$ which establishes the claim of $\mathbf{PD}(S,M') \subsetneq \mathbf{PD}(S,M'')$. Case $(iv)$ can be shown similarly by swapping the roles of the polytopes. 

Case $(v)$: This is the only remaining case given the results \ref{Theorem:EquivalenceClassesofPDP's1BELL}-\ref{Theorem:EquivalenceClassesofPDP's3STRICTRELATIONOtherway} and therefore holds trivially.  We can demonstrate this, however, without assuming cases $(iii)$ or $(iv)$, which we do, for the sake of illustration. 

Suppose then that $\mathrm{MSF}(\mathbf{PD}(S,M')) \nprec \mathrm{MSF}(\mathbf{PD}(S,M''))  $ and $\mathrm{MSF}(\mathbf{PD}(S,M'')) \nprec \mathrm{MSF}(\mathbf{PD}(S,M'))$. It immediately follows, from \ref{Theorem:EquivalenceClassesofPDP's1BELL}- \ref{Theorem:EquivalenceClassesofPDP's2NOTBELLEQUIV} that $\mathbf{PD}(S,M') \neq \mathbf{PD}(S,M'')$. Consider then the action of the restrictions $R_{|M'^{\min}}$ and $R_{|M''^{\min}}$ on the distinct polytopes. By definitions
\begin{align}
    R_{|M'^{\min}}(\mathbf{PD})(S,M')) &= \mathrm{MSF}(\mathbf{PD}(S,M')) 
    \end{align}
but, by writing $M^* = M''^{\min}\cap M'^{\min}$, 
    \begin{align}
    \begin{split}
  & R_{|M'^{\min}}(\mathbf{PD}(S,M'')) \\
      & = \mathrm{conv}[R_{|M'^*}(\mathbf{B}(S_{\max}'')) \odot R_{|M^{*}}\mathbf{NS}(S''^{\perp}_{\max})].
       \end{split}
\end{align}
Again, the images under the same restriction map are distinct, as from $S'^{\perp}_{\max} \neq S''^{\perp}_{\max}$ it follows that $M^* \neq M''$. We can also see from above that in fact $\mathbf{PD}(S,M'') \not\subset \mathbf{PD}(S,M')$, since if the polytopes did obey a subset relation, then, in particular the maximal solid fragment of $\mathbf{PD}(S,M')$ would be necessarily contained in $\mathbf{PD}(S,M'')$, which can not be the case, since the image of the map $R_{|M'^{\min}}$ which defines the fragment is composable for the polytope $\mathbf{PD}(S,M'')$. The converse, that $\mathbf{PD}(S,M'') \not\subset \mathbf{PD}(S,M')$ can be established by considering the map $R_{|M''^{\min}}$ instead. 

\end{proof}
    \normalfont

Theorem \ref{Theorem:EquivalenceClassesofPDP's} completely classifies the equivalence classes of partially deterministic polytopes. As immediate corollaries, we can state a sharpening of Theorem \ref{Theorem:BellandNSspecialcasesofPDP} and a general result on the relationship of the quantum set $\mathbf{Q}(S)$ and a arbitrary partially deterministic polytope. 

\corollary{\label{Corollary:PDPisaNSpolytopeIFFMSFisimageofNS}Let $S = (I,M,O)$ be any (possibly trivial) scenario and $M'$ a collection of subsets $M_i'\subset M_i$ for all $i\in I$. Suppose the scenario has at least two parties with at least two inputs each, so that $ \mathbf{NS}(S) \neq \mathbf{B}(S).$  Then $\mathbf{PD}(S,M') = \mathbf{NS}(S)$ if and only if $\mathrm{MSF}(\mathbf{PD}(S,M')) = R_{|M'^{\min}}(\mathbf{NS}(S))$. } 
\begin{proof}
    This statement  follows from Corollary \ref{corollary:NSisSolidoraPDP} by which $\mathbf{NS}(S)$ is composable and $\mathbf{NS}(S) = \mathbf{PD}(S,M'^{V})$ if and only if $V = \{i\in I: |M_i| = 1 \}\neq \emptyset$, where $M'^{V}$ is a collection with $M_i'^{V} = M_i$ if $i\in V$ and $M_i^{V} = \emptyset$ otherwise, which when combined with the Definition \ref{Definition:solidfragmentofPDP} of the maximal solid fragment yields the claim.  
 \end{proof}
\normalfont
The relevant consequence of Corollary \ref{Corollary:PDPisaNSpolytopeIFFMSFisimageofNS} to be highlighted is that if $S= (I,M,O)$ is a nontrivial scenario, that is has $|M_i|\geq 2$ for all $i\in I$ then $\mathbf{PD}(S,M') = \mathbf{NS}(S)$ if and only if $M_i' = \emptyset $ for all $i\in I$ as in that case. This is to say, that the equivalence class corresponding to the no-signalling polytope always contains exactly one element (which is the no-signalling polytope) in a nontrivial scenario. 

\theorem{\label{Theorem:QuantumsetandPDP}Let $S= (I,M,O)$ be a scenario, and $M'$ some collection of subsets of inputs $M_i'\subset M_i$ for all $i\in I$. Then exactly one of the following three statements holds for the  polytope $\mathbf{PD}(S,M')$ and the quantum set $\mathbf{Q}(S)$.
\begin{enumerate}[label=(\roman*), ref=\ref{Theorem:QuantumsetandPDP}.\roman*]

    \item \label{Theorem:QuantumsetandPDP1:QissupersetofB}  $\mathbf{PD}(S,M') = \mathbf{B}(S) \subsetneq \mathbf{Q}(S)$
    
    \item \label{Theorem:QuantumsetandPDP2:QissubsetofNS} $\mathbf{PD}(S,M') = \mathbf{NS}(S) \supsetneq \mathbf{Q}(S)$
    
    \item \label{Theorem:QuantumsetandPDP3:NeitherisSubsetOfother} $ \mathbf{B}(S) \neq \mathbf{PD}(S,M') \neq \mathbf{NS}(S)$ and both $\mathbf{Q}(S) \not\subset \mathbf{PD}(S,M')$ and $\mathbf{Q}(S) \not\supset\mathbf{PD}(S,M').$
\end{enumerate}
}
\begin{proof}
    Cases $(i)$ and $(ii)$ follow from Proposition \ref{Proposition:BellisStrictSubsetofQisstrictofNS} by which $\mathbf{B}(S) \subsetneq \mathbf{Q}(S) \subsetneq \mathbf{NS}(S)$, and the fact that there are partially deterministic polytopes $\mathbf{PD}(S,M')$ which equal $\mathbf{B}(S)$ and $\mathbf{NS}(S)$. Sufficient examples are extreme cases where $M_i' = M_i $ and $M_i' = \emptyset$ for all $i\in I$ shown in Theorem \ref{Theorem:BellandNSspecialcasesofPDP} while Theorem \ref{Theorem:EquivalenceClassesofPDP's} can be used to completely characterise the situations where the strict inclusions hold. 
    
    Case $(i)$ holds if and only if $\mathrm{MSF}(\mathbf{PD}(S,M')) = \bot$, by Theorem \ref{Theorem:EquivalenceClassesofPDP's1BELL}.

    Case $(ii)$ holds if and only if $\mathrm{MSF}(\mathbf{PD}(S,M')) = R_{M'^{\min}}(\mathbf{NS}(S))$, by Corollary \ref{Corollary:PDPisaNSpolytopeIFFMSFisimageofNS}.

        It remains to show case $(iii)$, that in every other situation the quantum set neither contains or is contained within the partially deterministic polytope $\mathbf{PD}(S,M')$. Suppose then that $\mathbf{B}(S) \neq \mathbf{PD}(S,M') \neq \mathbf{NS}(S)$, which implies  there is a nonredundant bipartition of $S$ such that 
        \begin{align}
            \mathbf{PD}(S,M') = \mathrm{conv}[\mathbf{B}(S') \odot \mathbf{NS}(S'^{\perp})],
        \end{align}
        with $\mathrm{MSF}(\mathbf{PD}(S,M')) = \mathbf{NS}(S_{\max}'^{\perp}) \neq \mathbf{B}(S'^{\perp})$, in particular. The relation $\mathbf{Q}(S)\not\subset \mathbf{PD}(S,M')$ then follows immediately from Theorem \ref{Theorem:strictcomposabilityresultofQ(S)}, by which $\mathbf{Q}(S)$ is solid over $S$.  
        
        The converse is seen by considering the images of the sets under the restriction $R_{|M'^{\perp}}$ defined by the collection $M'^{\perp}_i = M_i\setminus M'_i$ for all $i\in I$. By Proposition \ref{Proposition:restrictionmapidentities} and the fact that $\mathrm{MSF}(\mathbf{PD}(S,M'))\succ \bot$
        \begin{align}
            R_{|M'^{\perp}}(\mathbf{Q}(S)) = \mathbf{Q}(S'^{\perp}) \subsetneq \mathbf{NS}(S'^{\perp}) = R_{|M'^{\perp}}(\mathbf{NS}(S)),
        \end{align}
        and so it must be that  $\mathbf{PD}(S,M') \not\subset \mathbf{Q}(S)$ as claimed.
\end{proof}
\normalfont
We will finish this section with a few examples.

\example{\label{EXAMPLE:PartialDeterminismInBipartitetHREE-INPUTscenario}(Bipartite three-input scenario)\\ Let $S = (I,M,O)$ be a bipartite scenario with $|M_i| = 3$ for both parties.  There are $(\prod_{i\in I}(|2^{|M_i|}|) -2= 62$ different ways to impose partial determinism in a nontrivial way on $S$, i.e. to choose a collection $M'$ of subsets $M_i'\subset M_i$ such that $M_i' \neq \emptyset $ or $M_i' \neq M_i$ for at least for some $i\in I$. The equivalence classes can be characterized by their solid fragments, of which the scenario $S$ supports only two distinct forms. One form, where the subscenario $S'^{\perp}$ has two inputs on one side and three inputs on the other,  corresponding to the six distinct ways to have only one deterministic input corresponding to collections $M'^{\{x_i^j\}}$ defined by $M_i'^{\{x_i^j\}} = \{x_i^j \}$ for exactly one agent $i$ and $M_i'^{\{x_i^j\}} = \emptyset$ otherwise. The other form consists of those where the subscenario $S'^{\perp}$ has two inputs per site,  which correspond to the 9 different ways to choose one deterministic input from one of the agents and one from the other, corresponding to collections $M'^{\{x^j_i, x_{i'}^{j'}\}}$. These polytopes are all distinct by virtue of Theorem \ref{Theorem:EquivalenceClassesofPDP's}, since their maximal solid fragments are defined over different substructures $S'^{\perp}$. In every other of the remaining 47 cases, at least one of the parties in $I$ has at most one non-deterministic measurement, meaning that $\mathrm{MSF}(\mathbf{PD}(S,M')) = \bot$ and so,  Theorem \ref{Theorem:EquivalenceClassesofPDP's} $\mathbf{PD}(S,M') = \mathbf{B}(S)$. Thus, including the trivial cases where $M_i' = \emptyset$ for both $i\in I$ and $M' = M$, there are 17 distinct equivalence classes of partially deterministic polytopes in this scenario. They obey the type of relations illustrated in Fig.~\ref{fig:bipartiteThreeInputScenario}. }\\
\normalfont

\begin{figure}
    \centering
    \includegraphics[width=\linewidth]{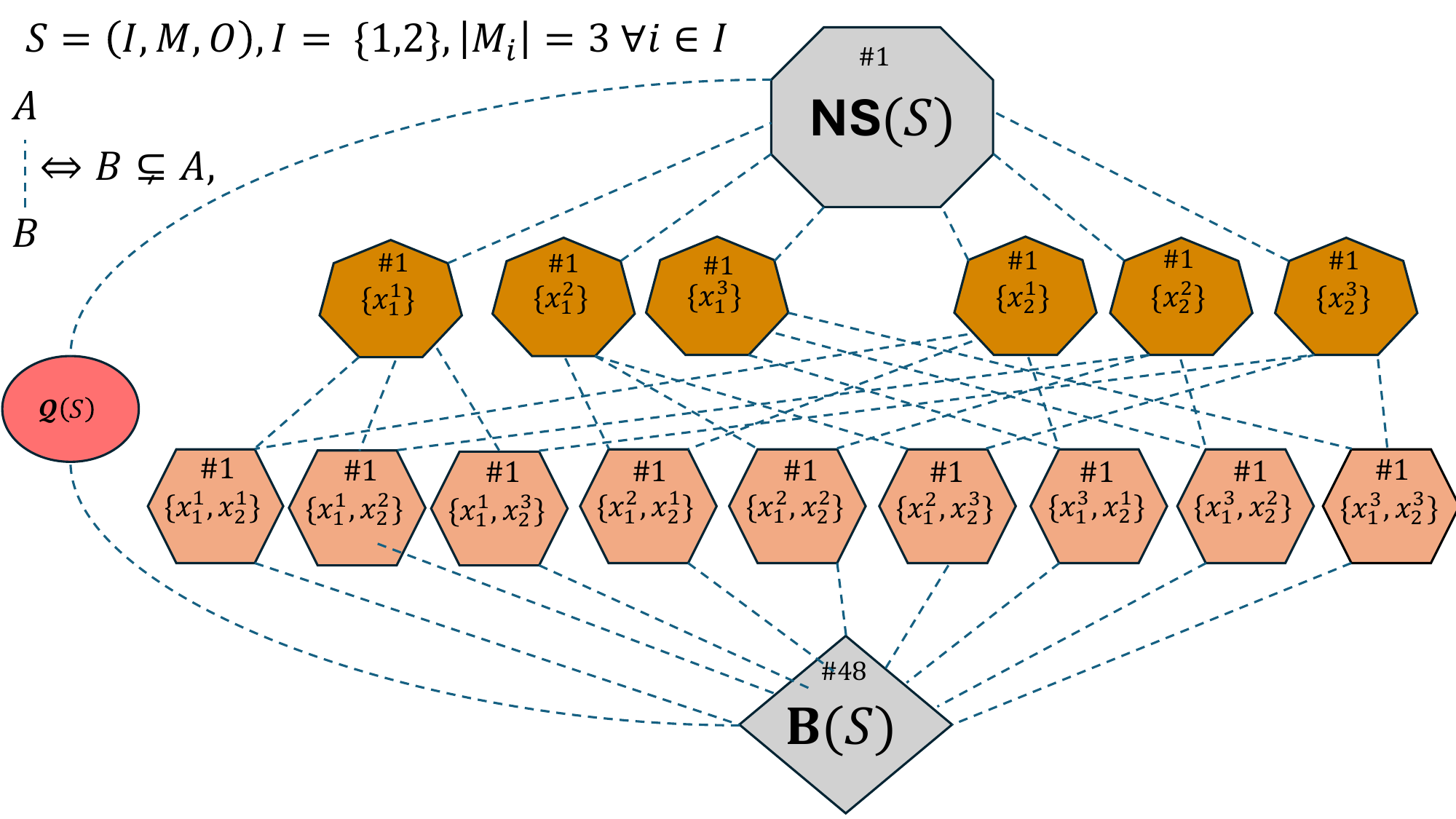}
    \caption{Illustration of the equivalence classes  of all the partially deterministic polytopes of Example \ref{EXAMPLE:PartialDeterminismInBipartitetHREE-INPUTscenario}. The bipartite scenario $S=(I,M,O)$, $I=\{1,2\}$ with three inputs per site allows 17 distinct equivalence classes of partially deterministic polytopes to be defined.  Six of those classes correspond to situations where only one input $x_i^j, i\in I, 1\leq j \leq |M_i|$ is deterministic, with each such class represented as a brown polyhedron with $\{x_i^j\}$ labelling the deterministic input. There are also 9 classes of partially deterministic polytopes for which pairs $x_i^j, x^{j'}_{i'}$, $i,i'\in I, i \neq i'$ of inputs are deterministic, represented in a similar way on a row above the class of Bell polytopes $\mathbf{B}(S)$. Each partially deterministic polytope with two deterministic inputs, is a strict subset of exactly two partially deterministic polytopes with one input. In the figure, if a path (sequence of dashed lines) exists from bottom to top (or top to bottom), then every object in that path can be ordered by the strict subset (or superset) relation. If two objects are not connected by such a path, then neither set is fully contained in the other. In particular, the quantum set $\mathbf{Q}(S)$ obeys $\mathbf{B}(S)\subsetneq \mathbf{Q}(S) \subsetneq \mathbf{NS}(S)$, but neither contains or is contained in any of the other partially deterministic polytopes.   Every equivalence class  has exactly one element, except for the class corresponding to the Bell-polytope which has 48 elements including the trivial fully deterministic case $M'=M$.   }
    \label{fig:bipartiteThreeInputScenario}
\end{figure}
\normalfont

In Ref.~\cite{Bong2020}, in the context of the Local Friendliness no-go theorem,  a complete list of facet defining inequalities was solved for one of the nine polytopes with two deterministic inputs in Fig.~\ref{fig:bipartiteThreeInputScenario}.

\example{\label{Example:partialdeterminismInThreepartiteTwo-inputScenario}(Tripartite binary-input scenario)\\Let $S= (I,M,O)$ be a three-partite scenario with $|M_i| = 2$ for all $i\in I$ and $M'$ a collection of subsets $M_i' \subset M_i $. There are $\prod_{i\in I}2^{|M_i|}-2 = 62$ ways to impose non-redundant partial determinism in this scenario. The partially deterministic polytopes on this scenario admit only one form of solid fragments, namely the form where the subscenario $S'^{\perp}$ simply omits the inputs of one party $i$. Each one of these classes can be realised by two distinct  types of collections $M'^{\{x_j^k\}}$,$M'^{\{j\}}$ of deterministic inputs, two of the cases correspond to the case where $M_i'^{\{x_j^k\}} = \{x^j_i\}$ for exactly one $j\in I$ and input $x_j^k \in M_i$, and $M_i'^{\{x_j^k\}} = \emptyset$ otherwise, and the last one $M_i'^{j} = M_j$ for exactly one $j \in I$ and $M_i' = \emptyset$ otherwise, respectively. In every other of the remaining  53 cases $\mathbf{PD}(S,M') = \mathbf{B}(S)$. Thus, including trivial cases, there are in total 5 distinct equivalence classes of partially deterministic behaviour which obey  the kind of relations illustrated in Fig.~\ref{fig:TripartiteEquivalenceClasses}. }\\
\normalfont

Example \ref{Example:partialdeterminismInThreepartiteTwo-inputScenario} complements the results obtained in Ref.~\cite{Haddara2025} which considered partially deterministic polytopes with one deterministic input per site. Namely,in this case the same polytope is obtained independently of which nonempty set of deterministic inputs is imposed at a given site. In Section \ref{Section:Applications} we will provide examples which concern the three distinct classes of tripartite polytopes portrayed in Fig.~\ref{fig:TripartiteEquivalenceClasses}.

\normalfont

 \begin{figure}
     \centering
     \includegraphics[width=\linewidth]{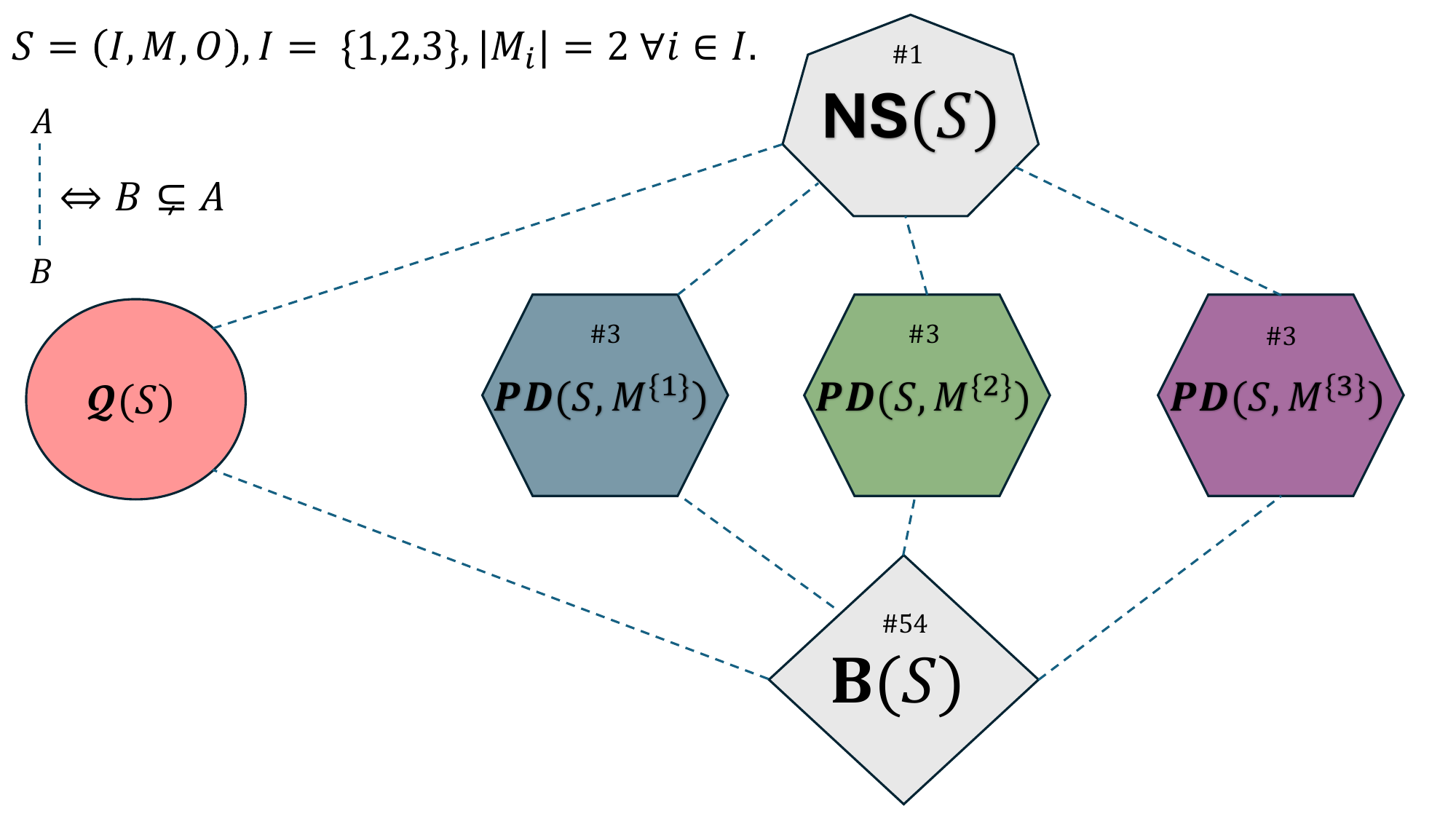}
     \caption{Illustration of all the equivalence classes of partially deterministic polytopes in the case of Example \ref{Example:partialdeterminismInThreepartiteTwo-inputScenario}. Any way to define partial determinism on the tripartite scenario $S=(I,M,O), I = \{1,2,3\}$, $|M_i|=2$ for all $i\in I$ leads to one of the five equivalence classes depicted as convex polyhedra. Three non-trivial classes of partially deterministic polytopes can be identified. Each of those equivalence classes has three elements,  two where one of the two inputs of one agent $i\in I$ is deterministic, and the third one,   which has both the inputs of one agent $i\in I$ deterministic. The latter has been used as representatives of this class. The distinct polytopes $\mathbf{PD}(S,M^{\{i\}})$ are not fully contained nor do they fully contain eachother, with a similar result holding in relation to the quantum set $\mathbf{Q}(S)$. The equivalence class of Bell polytopes $\mathbf{B}(S)$ contains 54 elements, including the trivial fully deterministic case with $M'=M$. }
     \label{fig:TripartiteEquivalenceClasses}
 \end{figure}

\subsection{\label{Section:PartialseparabilityandFactorizable}Partial Uncorrelatedness and Factorizability}
In Section \ref{section:PartialDeterminismSubsection} we constructed partially deterministic polytopes as the convex hull of partially predictable behaviour and showed, that equivalently a partially deterministic polytope is the $c$-composition of a Bell-local set over the substructure consisting of the 'deterministic measurements' in the collection $M'$, and a no-signalling polytope over the complementary substructure.

Given Fine's equivalence theorem (Theorem \ref{FinesTHRM}), by which the set $\mathbf{B}(S)$ can be expressed equivalently as the convex hull of predictable no-signalling behaviour $\mathbf{P}_{\mathbf{NS}}(S)$ or uncorrelated no-signalling behaviour $\mathbf{U}_{\mathbf{NS}}(S)$ it is immediately evident that the seemingly more general notion of uncorrelatedness, suitably generalized,  could be used in place of partial predictability, in order to get to the polytope $\mathbf{PD}(S,M')$. In this section, we show that this is indeed the case, by  providing generalizations of the the notions of uncorrelatedness and LHV-modelability and prove a result which equates the distinctly constructed sets with the partially deterministic polytope $\mathbf{PD}(S,M')$. This result leads to what may be viewed as a 'partial analogue' of Fine's theorem, which refers to the polytopes $\mathbf{PD}(S,M')$ instead of the Bell-polytope $\mathbf{B}(S)$. 

We also draw attention to the fact that the equivalence class of partially deterministic polytopes which equal the Bell-polytope established in Section \ref{section:PartialDeterminismSubsection} itself provides a strict refinement of the usual way in which  Fine's theorem for the LHV-modelability of behaviour is stated, and articulate this observation as a strengthened Fine's theorem.

    \definition{(Partially uncorrelated  behaviour)\label{DefinitionPartiallySeparableBehaviour}\\
   A behaviour $\wp \in \mathbf{E}(S)$ is called partially uncorrelated with respect to the set $M'$ if and only if $\forall \vec{a}, \vec{x}$
   \begin{align}
       \wp(\vec{a}|\vec{x}) = \prod_{i\in F_{\vec{x}}} \wp(a_i|\vec{x})\cdot \wp(\vec{a}_{(I\setminus F_{\vec{x}})})|\vec{x})
   \end{align}
   for all $\Vec{x}$, ${F_{\Vec{x}}} = \{ i \in I | x_i \in M_i'  \}$.  The set  behaviours partially uncorrelated  with respect to $M'$ is denoted by $\mathbf{PU}(S,M')$
}
\normalfont\\

\normalfont

Clearly, $ \mathbf{PP}(S,M') \subset \mathbf{PU}(S,M') $ and therefore also $\mathrm{conv}[\mathbf{PU}(S,M')] = \mathbf{E}(S)$. The convex hull becomes more interesting when the no-signalling constraints of Eq.~\eqref{no-signalling} are taken into account. 

\definition{\label{Definition:PartiallySeparableNSbehaviour}(Partially uncorrelated no-signalling behaviour)\\
The set $\mathbf{PU}_{\mathbf{NS}}(S,M')$ of partially uncorrelated no-signalling behaviour is defined as the intersection 
\begin{align}
  \mathbf{PU}_{\mathbf{NS}}(S,M') :=  \mathbf{PU}(S,M') \cap \mathbf{NS}(S).
\end{align}}\\

\normalfont

From  $ \mathbf{PP}(S,M') \subset \mathbf{PU}(S,M') $ it follows that  $\mathbf{PP}_{\mathbf{NS}}(S,M') \subset \mathbf{PU}_{\mathbf{NS}}(S,M')$.  Let $F_{\vec{x}} = \{i \in I | x_i \in M_i' \}$, as before. Owing to no-signalling,   in analogy to Eq.~\eqref{partialPredictableNosignallingequation} in the case of partially predictable no-signalling behaviour, it is seen that any $\wp \in \mathbf{PU}_{\mathbf{NS}}(S,M') $ decomposes as
\begin{align}
    \wp(\vec{a}|\vec{x}) = \prod_{i\in F_{\vec{x}}}\wp(a_i|x_i)\cdot \wp(\vec{a}_{I\setminus F_{\vec{x}}}|\vec{x}_{I\setminus F_{\vec{x}}})\label{EquationPartiallySeparableNosignallingEq.}.
\end{align}

Exactly the same way as in the case of partially predictable behaviour, we can show by use of the no-signalling conditions that the whole behaviour $\wp$ can be entirely characterized from a subset of the distributions. Namely since the distributions $\prod_{i\in F_{\vec{x}}}\wp(a_i|x_i)$ are to be recovered from a distribution in any context $\vec{x}$ where the substring $F_{\vec{x}}$ is present, it has to be the case that they can be recovered, in particular, from those in the set of maximal contexts $D = \{ \vec{x} : |F_{\vec{x}}| = \max |F_{\vec{x}'}| \}$ so that 
\begin{align}
    \prod_{i\in F_{\vec{x}}}\wp(a_i|x_i) = \sum_{i \in I\setminus F_{\vec{x}}} \wp^{|D}(\vec{a}|\vec{x}) \label{Equation:Separable|Dequation}
\end{align}
where the superscript $|D$ denotes that the $\vec{x}\in D$  in the distribution $\wp^{|D}(\vec{a}|\vec{x})$. Similarly, as was shown in Eq.~\eqref{Equation:marginalfromNDequation}, the distributions $\wp(\vec{a}_{I\setminus F_{\vec{x}}}|\vec{x}_{I\setminus F_{\vec{x}}})$ are to be recovered from any distribution in where $\vec{x}$ is in the set $ND = \{\vec{x} : |I \setminus F_{\vec{x}}| = \max |I\setminus F_{\vec{x}'}| \}$ meaning that
\begin{align}
    \wp(\vec{a}_{I\setminus F_{\vec{x}}}|\vec{x}_{I\setminus F_{\vec{x}}}) = \wp^{|ND}(\vec{a}_{I\setminus F_{\vec{x}}}|\vec{x}_{I\setminus F_{\vec{x}}}).
\end{align}
The collections of distributions $\wp^{|D}$ and $\wp^{|ND}$ may still contain redundant information for specification of a behaviour $\wp$, as any $i\in I$ for which $M_i' = \emptyset$ gets summed over anyway in the distributions $\wp^{|D}(\vec{a}_{F_{\vec{x}}}|\vec{x}_{F_{\vec{x}}})$ and conversely, any $i\in I$ for which $M_i' = M_i$ gets summed over in any distribution $\wp^{|ND}(\vec{a}_{I \setminus F_{\vec{x}}}|\vec{x}_{I \setminus F_{\vec{x}}})$. Following this chain through, and writing $I_{M'} = \{i \in I : M_i' \neq \emptyset \}$ it is seen that equivalently, 

\begin{align}
   \prod_{i \in  I_{M'}} \wp^{|D}(a_i|x_i) = R_{|M'}(\wp) = \wp^{S'}(\vec{a}'= \vec{a}_{I_{M'}}| \vec{x}'= \vec{x}_{I_{M'}} ) \label{Equation:partialseparabilitySubstructure}
\end{align}
and so, the marginals $\wp(\vec{a}_{F_{\vec{x}}}|\vec{x}_{F_{\vec{x}}})$ onto any $F_{\vec{x}}$ can be extracted from uncorrelated behaviour $\mathbf{U}_{\mathbf{NS}}(S')$ defined over the substructure $S'$, namely, from the image $R_{|M'}(\wp)$ of the behaviour under that restriction map. 

That a similar situation holds for the contexts $ND$, and the marginals $\wp(\vec{a}_{I\setminus F_{\vec{x}}}| \vec{x}_{ I\setminus F_{\vec{x}}})$ are completely characterized by considering the no-signalling polytope $\mathbf{NS}(S'^{\perp})$, in particular the image $R_{|M'^{\perp}}(\wp)$, was already shown in the study of partial predictability between Eqs.~\eqref{Equation:marginalfromNDequation}-\eqref{Equation:PartiallypredictableExpandedeXTREMEOFns}.

The following proposition, which also provides an alternative characterization of the set $\mathbf{PU}_{\mathbf{NS}}(S,M')$ is immediate.

\proposition{{\label{Proposition:PartialseparabilitydefinedasBehProduct}Let $S = (I,M,O)$ be a scenario and $M', M'^{\perp}$ define a bipartition of $S$ into $S', S'^{\perp}$. Then 
\begin{align}
    \mathbf{U}_{\mathbf{NS}}(S') \odot \mathbf{NS}(S'^{\perp}) = \mathbf{PU}_{\mathbf{NS}}(S,M')
\end{align}
}
\begin{proof} 
Omitted, follows from Definitions \ref{BehaviourProductDefinition} and \ref{DefinitionPartiallySeparableBehaviour} of the behaviour product and partial uncorrelatedness, paired with the observations made following Eq.~\eqref{Equation:Separable|Dequation}.
\end{proof}
\normalfont

Instead of considering the convex hull of $\mathbf{PU}_{\mathbf{NS}}(S,M')$ straight away, we give another closely related definition which is essentially a partial analogue  of LHV-modelability (local factorizability) of Definition \ref{LHVBehaviour}.

\definition{\label{Definition:PartialFactorizability}(Partially factorizable behaviour)\\
A behaviour $\wp\in \mathbf{E}(S)$ is termed partially factorizable with respect to the collection $M'$  if and only if it admits a model of the form
\begin{align}
    \wp(\vec{a}|\vec{x}) &= \int_{\Lambda} P(\vec{a}|\vec{x},\lambda) p(\lambda)\\
   &= \int_{\Lambda} \prod_{i\in F_{\vec{x}}} P(a_i| x_i, \lambda ) P(\vec{a}_{I\setminus F_{\vec{x}}}|\vec{x}_{I\setminus F_{\vec{x}}}, \lambda) p(\lambda),
\end{align}
with  $F_{\vec{x}} = \{i\in I: x_i \in M_i' \}$ and a measurable set $\Lambda$, $p(\lambda) \geq 0, \int_{\Lambda}p(\lambda) = 1$. The distributions $P(\vec{a}|\vec{x},\lambda)$ satisfy $P(\vec{a}_V|\vec{x},\lambda) = P(\vec{a}_V|\vec{x}_V, \lambda)$ for all subsets $V\subset I$, and $\vec{a}, \vec{x}$.  The set of partially factorizable behaviour is denoted by  $\mathbf{PF}(S,M')$.
}\\
\normalfont

Clearly $\mathrm{conv}[\mathbf{PU}_{\mathbf{NS}}(S,M')]\subset \mathbf{PF}(S,M')$. We will show that $\mathbf{PF}(S,M')$ is a convex polytope and, more precisely, equal to the set $\mathbf{PD}(S,M')$ thus establishing the equivalence between all three sets.

\theorem{\label{Theorem:PartiallyFactorizableEqualsPartiallyDeterministicPolytope} $\mathbf{PF}(S,M') = \mathrm{conv}[\mathbf{PU}_{\mathbf{NS}}(S,M')] = \mathbf{PD}(S,M')$}
\begin{proof} Since 
\begin{align}
    \mathbf{PD}(S,M') \subset \mathrm{conv}[\mathbf{PU}(S,M')] \subset \mathbf{PF}(S,M'),
\end{align} holds by definitions it is sufficient to prove the inclusion $\mathbf{PF}(S,M') \subset \mathbf{PD}(S,M')$. 

To prove this, we adapt the insight that lead to Proposition \ref{Proposition:PartialseparabilitydefinedasBehProduct} on the level of models. 

Suppose then, that $\wp \in \mathbf{PF}(S,M')$ is an arbitrary partially factorizable behaviour with a model 
\begin{align}
    \wp(\vec{a}|\vec{x}) = \int_{\Lambda} \prod_{i\in F_{\vec{x}}} P(a_i| x_i, \lambda ) P(\vec{a}_{I\setminus F_{\vec{x}}}|\vec{x}_{I\setminus F_{\vec{x}}}, \lambda) p(\lambda). \label{Equation:ModelForPFbehaviour}
\end{align}
Recall the sets $D = \{\vec{x}: |F_{\vec{x}}| = \max |F_{\vec{x}'}|\}$ and $ND = \{\vec{x} : |I\setminus F_{\vec{x}}| = \max |I\setminus F_{\vec{x}'}| \}$ and note that by definition the distributions $P(\vec{a}|\vec{x},\lambda)$ satisfy 

\begin{align}
    P(\vec{a}|\vec{x},\lambda) = P^{|D}(\vec{a}_{F_{\vec{x}}}|\vec{x}_{F_{\vec{x}}},\lambda)P^{|ND}(\vec{a}_{I\setminus F_{\vec{x}}}|\vec{x}_{I\setminus F_{\vec{x}}},\lambda). \label{Equation:PFmodelwithLambda}
\end{align}
For each $\lambda$, the collection $P_{\lambda}^{S}$ of distributions obeying Eq.~\eqref{Equation:PFmodelwithLambda} may be viewed as a point in the partially uncorrelated set $\mathbf{PU}_{\mathbf{NS}}(S,M')$,  and hence it immediately follows, by virtue of Proposition \ref{Proposition:PartialseparabilitydefinedasBehProduct} that equivalently
\begin{align}
P_{\lambda}^{S} = P_{\lambda}^{S'} \odot P_{\lambda}^{S'^{\perp}}
\end{align}
where $P_{\lambda}^{S'}, P_{\lambda}^{S'^{\perp}}$ may be viewed as points in the sets $\mathbf{U}_{\mathbf{NS}}(S')$ and $\mathbf{NS}(S'^{\perp})$  obtainable as restrictions $R_{|M'}$, $R_{|M'^{\perp}}$ of $P_{\lambda}^{S}$, respectively.  Now since $\mathbf{U}_{\mathbf{NS}}(S')\subset \mathrm{conv}[\mathbf{U}_{\mathbf{NS}}(S')] = \mathbf{B}(S')$, any point $P_{\lambda}^{S} \in \mathbf{U}_{\mathbf{NS}}(S')$ in particular can be represented as a convex combination $P_{\lambda}^{S'} = \sum_{r} s_r(\lambda) D^{S'}_{r}$ of the distinct deterministic extreme points of $\mathbf{B}(S)$  labelled by $r$. Here, the $s_r(\lambda)$ give the convex weight of each extreme point in the distribution conditioned on $\lambda$  so that  $s_r(\lambda)\geq 0$ for all $r,\lambda$ and $ \sum_{r} s_r(\lambda) =1$. Similarly, any point $P_\lambda^{S'^{\perp}} \in \mathbf{NS}(S'^{\perp})$ admits a convex decomposition $P_{\lambda}^{S'^{\perp}} = \sum_k l_k(\lambda) P_{k}^{S'^{\perp}}$ with $P_k^{S'^{\perp}} \in \mathrm{Ext}(\mathbf{NS}(S'^{\perp}))$  and $l_{k}(\lambda) \geq 0$ for all $k$, $\sum_k l_k(\lambda) =1$. Therefore, by inspection of Eq.~\eqref{Equation:ModelForPFbehaviour}, the assumption $\wp \in \mathbf{PF}(S,M')$ is equivalent to the behaviour admitting a model 
\begin{align}
    \wp(\vec{a}|\vec{x}) &= \int_{\Lambda} \prod_{i\in F_{\vec{x}}} P^{S'}(a_i| x_i, \lambda ) P^{S'^{\perp}}(\vec{a}_{I\setminus F_{\vec{x}}}|\vec{x}_{I\setminus F_{\vec{x}}}, \lambda) p(\lambda) \\
    \begin{split}
     &= \int_{\lambda} \sum_{r,k}s_r(\lambda)l_k(\lambda)p(\lambda) D^{S'}_{r}(\vec{a}_{F_{\vec{x}}}|\vec{x}_{F_{\vec{x}}}) \\
     & \hspace{1cm} P_k^{S'^{\perp}}(\vec{a}_{I \setminus F_{\vec{x}}}| \vec{x}_{I \setminus F_{\vec{x}}})
       \end{split}
       \\[2ex]
    &= \sum_{r,k} t_{rk} D^{S'}_{r}(\vec{a}_{F_{\vec{x}}}|\vec{x}_{F_{\vec{x}}})P_k^{S'^{\perp}}(\vec{a}_{I \setminus F_{\vec{x}}}| \vec{x}_{I \setminus F_{\vec{x}}}), \label{Equation:PFispdFINALEQ}
\end{align}
where we have defined $t_{rk} = \int_{\Lambda}s_{r}(\lambda)l_k(\lambda)p(\lambda)$. Equation \eqref{Equation:PFispdFINALEQ} establishes that also $\wp \in \mathbf{PD}(S,M')$, as claimed. 
\end{proof}
\normalfont

This result naturally leads us to state  a partial analogue of Fine's theorem, which can be thought of as a strict generalization of it appealing to properties of the partially deterministic polytopes. 

\theorem{\label{Theorem:PartialAnalogueofFine}(A partial analogue of Fine's theorem)\\ Let $S = (I,M,O)$ be a scenario, and $M'$ some collection of subsets $M_i' \subset M_i$ for all $i\in I$. The following statements are equivalent 

\begin{enumerate}[label=(\roman*), ref=\ref{Theorem:PartialAnalogueofFine}.\roman*]

    \item \label{Theorem:PartialAnalogueofFinePD}  The behaviour $\wp$ is partially deterministic with respect to $M'$, that is, it admits a model of the form 

    \begin{align}
         \wp(\Vec{a}|\Vec{x}) &= \sum_{i} \lambda_i P_i(\vec{a}|\vec{x}) \\ &= \sum_{i} \lambda_i D_i(\vec{a}_{F_{\vec{x}}}|\vec{x}_{F_{\vec{x}}}) P_i(\vec{a}_{I \setminus F_{\vec{x}}}|\vec{x}_{I \setminus F_{\vec{x}}}),
    \end{align}
    with, $P_i(\vec{a}_{V}|\vec{x}) = P_i(\vec{a}_V|\vec{x}_V)$ for all $V \subset I, \vec{a}, \vec{x}, i$,  $D_i(\vec{a}_{F_{\vec{x}}}|\vec{x}_{F_{\vec{x}}}) \in \{0,1 \}$ for all $i, \vec{a}_{F_{\vec{x}}}, \vec{x}$ and $\lambda_i \geq 0, \sum_i \lambda_i = 1$.
    
    \item \label{Theorem:PartialAnalogueofFineStrechedDistribution}  There exists a collection $P^{S'_{\infty} }$ of  no-signalling distributions $   P^{S'_{\infty}}(\boldsymbol{{\alpha}},\vec{\beta}_{\vec{x}^{\perp}}|\vec{x}^{\perp})$, giving a joint distribution over all the possible outcomes of all the possible measurements for all agents of the scenario $S' = (I',M',O')$ denoted by $\boldsymbol{\alpha}$, and  over the outcomes $\vec{\beta}_{\vec{x}^{\perp}} := \vec{\beta}$ of the agents in scenario $S'^{\perp} = (I'^{\perp},M'^{\perp}, O'^{\perp})$ in the context $\vec{x}^{\perp}$, which can be used to recover the distributions $\wp(\vec{a}|\vec{x})$ in the original behaviour $\wp$ by appropriate marginalization. Formally, 

    \begin{align}
    \begin{split}
   & \hspace{0.8cm}    \boldsymbol{\alpha} := (\alpha_{x_1=x_1^1}, \ldots , \alpha_{x_1 = x_1^{|M_1|}}, \ldots, \alpha_{x_{|I'|} = x^{|M_{|I'|}}_{|I'|} }) \\
 &\hspace{0.8cm}       \in O_{x_1 = x_1^1}  \times \ldots \times O_{x_{1}=x_1^{|M_1|}} \times \ldots \times O_{x_{|I'|} = x^{|M_{|I'|}}_{|I'|}}\\
&\hspace{0.8cm} := O^{S'_{\infty}} 
\end{split}
    \end{align}
    and 
    \begin{align}
    \begin{split}
      & \hspace{0.8cm}  \vec{\beta} = (\beta_{x^{\perp}_1}, \beta_{x^{\perp}_2}, \ldots , \beta_{x_{I'^{\perp}}}) \\
       & \hspace{0.8cm} \in O_{x^{\perp}_1} \times O_{x^{\perp}_2} \times \ldots \times O_{x_{|I_{M'^{\perp}}|}} = O_{\vec{x}^{\perp}},
        \end{split}
    \end{align}
   so that  $(\boldsymbol{\alpha},\vec{\beta}) \in O^{S'_{\infty}} \times O_{\vec{x}^{\perp}}$ 
    with 
    \begin{align}
        \wp(\vec{a}|\vec{x}) &= P^{S'_{\infty}}(\vec{\alpha}_{F_{\vec{x}}} = \vec{a}_{F_{\vec{x}}}, \vec{\beta}_{I\setminus F_{\vec{x}}}=  \vec{a}_{I\setminus F_{\vec{x}}} |\vec{x}^{\perp}_{I\setminus F_{\vec{x}}} = \vec{x}_{I\setminus F_{\vec{x}}})
        \end{align}
where 
        \begin{align}
    P^{S'_{\infty}}(\boldsymbol{\alpha}_{F_{\vec{x}}}, \vec{\beta}_{I\setminus F_{\vec{x}}}|\vec{x}^{\perp}_{I \setminus F_{\vec{x}}})   &= \sum_{\mathclap{\substack{\alpha_{x^{j}_{i'}}: x^j_{i'} \neq x_i \\
   \beta_i : i \in F_{\vec{x}} }}}P^{S'_{\infty}} (\boldsymbol{\alpha}, \vec{\beta}|\vec{x}^{\perp}).
    \end{align}

    \item \label{Theorem:PartialAnalogueofFineINEQUALITIES} The behaviour $\wp$ satisfies all the inequalities valid for the polytope $\mathbf{PD}(S,M')$. Specifically $\wp$ satisfies any finite set of inequalitites which form an $H$-representation of $\mathbf{PD}(S,M')$.

    \item \label{Theorem:PartialAnalogueofFineFACTORIZABILITY} The behaviour $\wp$ is partially factorizable with respect to $M'$, that is, it admits a model of the form 
    \begin{align}
        \wp(\vec{a}|\vec{x}) = \int_{\Lambda} \prod_{i\in F_{\vec{x}}} P(a_i| x_i, \lambda ) P(\vec{a}_{I\setminus F_{\vec{x}}}|\vec{x}_{I\setminus F_{\vec{x}}}, \lambda) p(\lambda),
    \end{align}
    with $p(\lambda) \geq 0$ and $\int_{\Lambda}p(\lambda) = 1$.

    \item \label{Theorem:PartialAnalogueofFineEQUIVALECECLASS} Any of the statements \ref{Theorem:PartialAnalogueofFinePD}-\ref{Theorem:PartialAnalogueofFineFACTORIZABILITY} hold with $M'$ replaced by any collection $M''$ such that  $\mathrm{MSF}(\mathbf{PD}(S,M')) = \mathrm{MSF}(\mathbf{PD}(S,M''))$.
\end{enumerate}
}

\begin{proof}
    The equivalences between \ref{Theorem:PartialAnalogueofFinePD}, \ref{Theorem:PartialAnalogueofFineINEQUALITIES} and \ref{Theorem:PartialAnalogueofFineFACTORIZABILITY} have already been established: they are immediate corollaries of Theorem \ref{convPPNSisapolytopeTheorem} and Definition \ref{PartiallyDeterministicPolytopeDefinition} by which  partially deterministic behaviour  form the polytope $\mathbf{PD}(S,M')$, and Theorem \ref{Theorem:PartiallyFactorizableEqualsPartiallyDeterministicPolytope} which showed the equivalence $\mathbf{PD}(S,M') = \mathbf{PF}(S,M')$ with partial factorizability. The statement \ref{Theorem:PartialAnalogueofFineEQUIVALECECLASS}, holds by appeal to Theorem \ref{Theorem:EquivalenceClassesofPDP's} where partially deterministic polytopes having the same solid fragment was pointed out necessary and sufficient for the equivalence of those polytopes. 
    
    It is therefore sufficient to show that any of those conditions is equivalent to \ref{Theorem:PartialAnalogueofFineStrechedDistribution}. We will prove that \ref{Theorem:PartialAnalogueofFinePD} $\Leftrightarrow \ref{Theorem:PartialAnalogueofFineStrechedDistribution}$.

    Suppose then that $\wp$ admits a partially deterministic model in the sense of \ref{Theorem:PartialAnalogueofFinePD}, which means that $\wp \in \mathbf{PD}(S,M')$. Let $S', S'^{\perp}$ be a bipartition of $S$ as defined by $M'$ and the complementary collection $M'^{\perp}$ with $M_i'^{\perp} = M_i\setminus M_i'$ for all $i\in I$. By Theorem \ref{TheoremPDPProduct} 
    \begin{align}
        \mathbf{PD}(S,M') = \mathrm{conv}[\mathbf{B}(S') \odot \mathbf{NS}(S'^{\perp})],
    \end{align}
    so that 
    \begin{align}
        \wp = \sum_{rk} t_{rk} D_r^{S'} \odot P_k^{S'^{\perp}}, \label{eQUATION:bEHAVIOURPRODUCTinProofofaNALOGUEfiNE}
    \end{align}
    with $D_r^{S'} \in \mathrm{Ext}( \mathbf{B}(S'))$ and $P^{S'^{\perp}} \in \mathrm{Ext}(\mathbf{NS}(S'^{\perp}))  \forall r,k $. In particular, by Fine's theorem (Theorem \ref{FinesTHRM}), each $D_r^{S'}$ can be recovered from a single distribution $\widetilde{P}_r^{S'_{\infty}}(\boldsymbol{\alpha})$, with $\boldsymbol{\alpha} \in O^{S'_\infty}$  the outcomes in the trivialized scenario $S'_{\infty}$, via
    \begin{align}
        D^{S'}_r(\vec{a}'|\vec{x}') = \widetilde{P}_r^{S'_{\infty}}(\vec{\alpha}_{\vec{x}'} = \vec{a}'), 
    \end{align}
    where 
    \begin{align}
        \widetilde{P}_r^{S'_{\infty}}(\boldsymbol{\alpha}_{\vec{x}'} = \vec{a}') = \sum_{\alpha_{x_i'^j} : x_i' \neq x_i'^j} \widetilde{P}^{S'_{\infty}}(\boldsymbol{\alpha}).
    \end{align}
Putting these identitites to Eq.~\eqref{eQUATION:bEHAVIOURPRODUCTinProofofaNALOGUEfiNE} it follows, by the definition of the behaviour product, that 

\begin{align}
\begin{split}
    \wp(\vec{a}|\vec{x}) &= \sum_{rk} t_{rk} \widetilde{P}_r^{S'_{\infty}}(\boldsymbol{\alpha}_{\vec{F}} = \vec{a}_{F_{\vec{x}}}) \\
     &\times P_k^{S'^\perp}(\vec{\beta}_{I \setminus F_{\vec{x}}} = \vec{a}_{I\setminus F_{\vec{x}}}|\vec{x}'^{\perp}_{I\setminus F_{\vec{x}}} = \vec{x}_{I\setminus F_{\vec{x}}}),
    \end{split}
\end{align}
or equivalently,  that the collection $P^{S'_{\infty}}$ of distributions defined by
\begin{align}
    P^{S'_{\infty}}(\boldsymbol{\alpha}, \vec{\beta}|\vec{x}'^{\perp}) = \sum_{rk}t_{rk} \widetilde{P}^{S'_{\infty}}_r(\boldsymbol{\alpha}) P_k^{S'^{\perp}}(\vec{\beta} |\vec{x}'^{\perp})
\end{align}
has the desired properties, proving the claim. 

For the converse, note that one can get a partially deterministic model for $\wp$ from $P^{S'_{\infty}}$ simply from the identity
\begin{align}
    P^{S'_{\infty}}(\boldsymbol{\alpha}, \vec{\beta}|\vec{x}'^{\perp}) &= P^{S'_\infty}(\vec{\beta}|\vec{x}'^{\perp}, \boldsymbol{\alpha})P^{S'_{\infty}}(\boldsymbol{\alpha}|\vec{x}'^{\perp})\\
    &=P^{S'_\infty}(\vec{\beta}|\vec{x}'^{\perp}, \boldsymbol{\alpha})P^{S'_{\infty}}(\boldsymbol{\alpha})
\end{align}
which immediately implies that a partially deterministic model of the form
\begin{align}
\begin{split}
    \wp(\vec{a}|\vec{x}) 
    &= \sum_{\boldsymbol{\alpha}} D^{S'}(\vec{a}_{F_{\vec{x}}} = \boldsymbol{\alpha}_{F_{\vec{x}}}|\boldsymbol{\alpha}) \\
    & \times P^{S'_{\infty}}(\vec{a}_{I\setminus F_{\vec
    x}} = \beta | \vec{x}_{I\setminus F_{\vec
    x}} = \vec{x}'^{\perp}_{I\setminus F_{\vec{x}}}, \boldsymbol{\alpha}) P^{S'_{\infty}}(\boldsymbol{\alpha})
    \end{split}
\end{align}
can be constructed. Here, $P^{S'_{\infty}}(\boldsymbol{\alpha})$ is essentially the convex weight of a parameter $\boldsymbol{{\alpha}}$ which determines outcomes for the inputs in $M'$, and hence $\wp\in \mathbf{PD}(S,M')$ as claimed. 
\end{proof}

\normalfont

We have formatted Theorem \ref{Theorem:PartialAnalogueofFine} as an 'analogue' of Fine's theorem, owing to the evident similarities in with respect to the formulation of Fine's theorem presented in our Theorem \ref{FinesTHRM} (see also Remark \ref{Remark:FinestheoremSatisfiesBell}). We emphasize, however, that the result  may, in fact, be viewed as a strict generalization of the original result for more than one reason. Firstly, omitting statement \ref{Theorem:PartialAnalogueofFineEQUIVALECECLASS}, the formulation of Fine's theorem (in the sense of our Theorem \ref{FinesTHRM}) can be recovered exactly for the special case $M' = M$. Thus, the result of Theorem \ref{Theorem:PartialAnalogueofFine} subsumes it. Secondly, when statement \ref{Theorem:PartialAnalogueofFineEQUIVALECECLASS} which takes care of the equivalence classes of the partially deterministic polytopes is taken into account, novel nontrivial equivalences can be added to the list of statements in Theorem \ref{FinesTHRM}. We can therefore extract a strengthened result, which we state below.

\theorem{\label{theorem:AstrengthenedFine'sTheorem}(A strengthened Fine's theorem) \\
Let $S = (I,M,O)$ be a scenario, and $M'$ a collection of subsets $M_i' \subset M_i$ for all $i\in I$. The following statements are equivalent.

\begin{enumerate}[label=(\roman*), ref=\ref{theorem:AstrengthenedFine'sTheorem}.\roman*]
\item \label{Theorem:Strengthenedfine,NORMALFINE} All the statements of Fine's theorem (Theorem \ref{FinesTHRM}) are valid for the behaviour $\wp$.

\item \label{StrenghtenedFinePDmodel} The behaviour $\wp$ admits a (genuine) partially deterministic model of the form in Definition \ref{PartiallyDeterministicPolytopeDefinition} with respect to any collection $M'$ of subsets $M_i' \subset M_i$ so as long as $|M_i'| \geq 2$ for at most one $i \in I$. 

\item \label{StrengthenedFinePartiallyjoint} There exists a collection of distributions $P^{S'_{\infty}}(\boldsymbol{\alpha}, \vec{\beta}|\vec{x}^{\perp})$, with each distribution being over all the outcomes for all possible inputs and agents in the subscenario $S' = (I',M',O')$ encoded in $\boldsymbol{\alpha}$, and over  outcomes $\vec{\beta}$ of the agents in scenario $S'^{\perp} = (I'^{\perp},M'^{\perp}, O'^{\perp})$ in the context $\vec{x}'^{\perp}$ with at most one agent $i\in I'^{\perp}$ in the scenario $S'^{\perp}$ having $|M_i'^{\perp}|\geq 2$. This collection can be used to recover the entire behaviour $\wp$ by marginalization as in 
 \begin{align}
        \wp(\vec{a}|\vec{x}) &= P^{S'_{\infty}}(\vec{\alpha}_{F_{\vec{x}}} = \vec{a}_{F_{\vec{x}}}, \vec{\beta}_{I\setminus F_{\vec{x}}}=  \vec{a}_{I\setminus F_{\vec{x}}} |\vec{x}^{\perp}_{I\setminus F_{\vec{x}}} = \vec{x}_{I\setminus F_{\vec{x}}}).
        \end{align}

\item \label{StrengthenedFinePartiallyFactorizable} The behaviour $\wp$ admits a (genuine) partially factorizable model of the form in Definition \ref{Definition:PartialFactorizability} with respect to any collection $M'$ of subsets $M_i' \subset M_i$ so as long as  as $|M_i'| \geq 2$ for at most one $i \in I$.

\end{enumerate}
}
\begin{proof}
    This set of equivalences is obtained as a special case of Theorem \ref{Theorem:PartialAnalogueofFine}, by application of Theorem \ref{TheoremMainPDequalBellTHRM}, where the class of partially deterministic polytopes corresponding to the Bell polytope was derived.
\end{proof}

\normalfont

We now draw particular attention to equivalence of Theorem \ref{FinesTHRM} and statement \ref{StrengthenedFinePartiallyjoint}, which establishes the existence of a certain family of 'partially context-dependent' distributions as interchangeable with the behaviour admitting a Bell-local model. It is easily seen that this statement recovers the condition for sufficiency of the existence of distributions over outcomes of triples of inputs established by Fine \cite{Fine1982} for the CHSH-experiment, as stated in Remark \ref{Remark:RemarkonFinesThrm}. Thus, our statement subsumes that condition of Fine's original theorem, and generalizes it to arbitrary numbers of inputs and parties in a novel form. We can immediately point out a few reasons of why this is  significant. Firstly, it is clear that in more general bipartite scenarios, that is in scenarios with $M_i \geq 3$ for both parties,  the statement \ref{StrengthenedFinePartiallyjoint} demands more than simply the existence of distributions over triples of observables (for any $M'$). Therefore we obtain an alternative, structurally based,  perspective to the point made by  Garg and Mermin \cite{Garg1982, Garg1982b} that the said condition should not, and does not generalize to more general bipartite scenarios. On the other hand, the equivalence of existence of triples in  Fine's original construction \cite{Fine1982} has been used also in proofs that the CHSH-inequalities (for all pairs of pairs of inputs) are necessary and sufficient in any bipartite scenario with two inputs on one side, and $\mathbb{N} \ni |M_i| $-inputs on the other. Since our result applies to arbitrary numbers of parties, it is conceivable that  our generalization could provide new insight into the question as to when novel forms of Bell-inequalities arise in more general scenarios, or whether it is sufficient to consider such 'permutations' of ones derived in a strictly smaller scenario. 

From a purely mathematical perspective, Theorems \ref{Theorem:PartialAnalogueofFine} and \ref{theorem:AstrengthenedFine'sTheorem} may be viewed as being related to forms of (marginal) probability extension problems which have been independently studied  in fields of science not related to quantum physics see e.g. Ref.~\cite{Vorobeev1962}. We therefore believe that these results, in their purely mathematical form,  may be of significance for broad classes of problems which concern joint probability distributions with marginal compatibility constraints of certain kinds. 

From a physical perspective, the existence of joint distributions over different inputs of the same party is  closely related to the question of joint measurability (see e.g. \cite{Fine1982, Fine1982b, Wolf2009}) of the measurements representing those inputs. From Theorem \ref{Theorem:PartialAnalogueofFine} it is possible to extract conditions for a given subset of inputs to be compatible with  joint measurability, corresponding to bounds on a certain kind of partially deterministic polytope, while staying agnostic about joint measurability of the others. This suggests a widely applicable approach for probing joint measurability of subsets of observables in a device independent way. This point is closely related to the work of  Ref.~\cite{Quintino2019} which begins their investigation  of device independent tests of 'structures of incompatibility' by considering precisely the partial uncorrelatedness conditions (which they term partially local sets) which form partially deterministic polytopes when the nondeterministic inputs are merely constrained by no-signalling. Our results could provide some general structural insight into the study of those kinds of questions, though we have not provided any general constructions for inequalities.

Given the significance of Fine's theorem  in some approches to the study of the notions of  noncontextuality and nonlocality (see, for example Refs.~\cite{Abramsky2011, Santos2021, BudroniContextualityreview2022}) specifically we believe that these results may find use in those areas of quantum foundational investigations too. Since not all marginal compatibility conditions in non-contextuality scenarios may be realised in Bell-scenarios however \cite{BudroniContextualityreview2022}, there may be room for further generalization of our results to some 'partially noncontextual' analogue equivalent to the notion of full non-contextuality, as well. We leave a more thorough investigation of the possibility of this kind of extension, and its possible physical significance,  for future work.

In Section \ref{Section:Applications}, we will demonstrate the usefulness of these results by presenting classes of situations with seemingly completely different physical motivations where partially deterministic polytopes, in the various equivalent forms it may be defined, arise.

\section{Example Applications \label{Section:Applications}}

The partially deterministic polytopes investigated in Section \ref{Section:PartialDeterministicMathsSection} generalize the bipartite structures introduced by Woodhead \cite{Woodhead2014} under the same name.  Woodhead discussed the significance of those  bipartite polytopes from the perspective of device-independent randomness certification in the presence of an adversary constrained by no-signalling. 

 More recent works have explicitly shown that violation of a Bell inequality is generally necessary but not sufficient for device-independent certification of randomness in multi-input scenarios  \cite{Zhang2025, ramanathan2024maximumquantumnonlocalitysufficient}.  Ref.~\cite{ramanathan2024maximumquantumnonlocalitysufficient} in particular (see the supplemental material E) made use of the bipartite partially deterministic polytopes of Woodhead \cite{Woodhead2014} and  the complete solution of the bipartite three-input Local Friendliness polytope of Ref.~\cite{Bong2020} in their analysis. Ref.~\cite{Zhang2025} also essentially writes down, in the supplemental material,  a partially deterministic model for the case where no randomness is certified. The fact that violations of the weaker notion of partial determinism (as opposed to Bell-locality) is relevant for  certifying randomness has also been acknowledged before see e.g. \cite{Barret2006}, but in cases (the bipartite two-input case) where the distinction disappears. Of course, a distinction also arises whether one assumes that the distributions are furthermore quantum modelable, or just compatible with no-signalling. Nonetheless,  we believe the general results of this work could provide some further insight into questions as to under what constraints randomness may be certified.

The bipartite  polytopes of Woodhead \cite{Woodhead2014} were recognized to be closely related to objects termed 'Local Friendliness polytopes' \cite{Bong2020}, which arose independently in investigations of extensions of Wigner's friend scenario \cite{Bong2020} (see acknowledgements of Ref.~\cite{Bong2020}) thus establishing a kind of common mathematical structure underlying them. Other more recent works, such as \cite{Utreras-Alarcon2024, Haddara2025} which considered generalizations of the Local Friendliness result, have also pointed out that the Local Friendliness polytopes can be viewed as instances of, or generalizations of, Woodheads partially deterministic polytopes. 

In summary, representative cases of different situations where some partially deterministic polytopes, or specific features of them,  have been identified as being of physical interest can therefore already be, and have already been, pointed out in the literature. 

In this section, we present (or propose) another variety of different physically motivated scenarios where classes of partially deterministic polytopes arise as the objects of interest and therefore posit, that partial deterministic polytopes provide a  unifying mathematical language to talk about the properties of these seemingly distinct forms of physical behaviour.   This identification is made possible by the fully general approach taken in Section  \ref{Section:PartialDeterministicMathsSection}, which applies to arbitrary number of parties and subsets of deterministic measurements in particular, which was not established in previous works, such as \cite{Woodhead2014,Bong2020, Haddara2025}.

In Section \ref{Section:device-independentinseparabilitywitnesses}, we show that a class of partially deterministic polytopes form device independent multipartite-inseparability witnesses of specific types. We also show, how unions, and convex-hulls of the unions of these partially deterministic polytopes attain relevance for entanglement certification tasks. In Section \ref{Section:BroadcastLocality} we show that the \emph{broadcast-local} sets of Refs.~\cite{Bowles2021, Boghiu2023} in certain types of scenarios are precisely identifiable as partially deterministic polytopes with certain structure. We also identify a limit to the applicability of partially deterministic polytopes, in the contexts of entanglement certification and broadcast non-locality, by  provididing example cases where more general notions of composability become relevant.  Finally, in Section \ref{Section:GeneralSequentialWignerSection} we find a one-to-one correspondence between partially deterministic polytopes and Local Friendliness scenarios by extending the investigation of the sequential Wigner's friend scenario introduced in Ref.~\cite{Utreras-Alarcon2024} to the multiparite-multi-input case. Therefore, for the sequential Wigner's friend scenario, the sets of interest from the perspective of the Local Friendliness no-go result \cite{Bong2020} are, in a sense, completely characterized by our results.

\subsection{\label{Section:device-independentinseparabilitywitnesses}Device-independent inseparability witnesses}
Let $\mathbb{S}(\mathcal{H})$ denote the quantum state space over the Hilbert space $\mathcal{H}$, that is, the set of positive trace-one linear operators acting on $\mathcal{H}$.

A quantum state $\rho \in \mathbb{S}(\mathcal{H}_{\vec{A}})$ on a tensor product Hilbert space $\mathcal{H}_{\vec{A}} = H_{A_1} \otimes \mathcal{H}_{A_2} \otimes \ldots \otimes \mathcal{H}_{A_{|I|}}$ is termed separable (or sometimes emphatically fully separable \cite{Horodecki2009}), if and only if it decomposes into a mixture\footnote{To be precise, the condition that the state $\rho$ can be \emph{approximated} by a convex sum of product states is sometimes stated instead \cite{Werner1989}.  We omit this detail in order to keep the presentation as simple as possible.} of product states as in 

\begin{align}
    \rho =  \sum_{k \in K} \eta_k\rho^{A_1}_k \otimes \rho^{A_2}_k \otimes \ldots \otimes\rho^{A_{|I|}}_k, \label{Equation:SeparablestatesEquation}
\end{align}
with each $\rho_k^{A_i}\in \mathbb{S}(\mathcal{H}_{A_i})$, $\eta_k \geq 0$ for all $k\in K$ and $\sum_{k\in K} \eta_{k}=1$. The set of fully separable states will be denoted by $\mathbf{SEP}(\mathcal{H}_{\vec{A}})$. If a fully separable decomposition is not admissible, the state $\rho$ is entangled. It is well known that entanglement is necessary for a violation of a Bell-inequality \cite{Brunner2014}. Indeed, if the state $\rho$ were separable, then a quantum behaviour $\wp \in \mathbf{Q}(S)$ would decompose as 
\begin{align}
    \wp(\vec{a}|\vec{x}) &= \mathrm{tr}[\rho M_{a_1|x_1} \otimes M_{a_2|x_2} \otimes \ldots \otimes M_{a_{|I|} |x_{|I|}}] \\
    &= \sum_{k\in K} \eta_k \left[\prod_{i \in I} \mathrm{tr}[\rho^{A_i}_k M_{a_i|x_i}]\right] \\
    &= \sum_{k\in K} \eta_k \left[ \prod_{i \in I} P_k(a_i|x_i) \right] \in \mathbf{B}(S).
\end{align}

 For this reason, separable states have also been termed ``classically correlated states'' \cite{Werner1989}. It is important to note that the converse is significantly more subtle.  While it is generally true that for pure entangled states one can indeed find measurements that lead to a Bell-inequality violation \cite{Gisin1992, Popescu1992, Gachechiladze2017} there are examples of mixed states which are entangled, but cannot be used to violate any Bell inequality \cite{Werner1989,Barret2002}. In this sense, entanglement should not be taken to be a synonym for 'Bell-nonlocality' and similarly, even if a the behaviour assumed to be obtainable from a quantum state were demonstrably LHV-modelable, one ought not to assume that the underlying quantum state is necessarily separable. 

From a violation of a Bell inequality one can, however, rightfully infer the presence of entanglement, irrespectively of whether the dimensions of the underlying quantum systems or details of the measurements are known or not. Therefore, a Bell inequality can be understood as an example of a device independent quantum state inseparability witness. 
 
 In multipartite $|I| \geq 3$ cases the concepts of separability and entanglement attain further nyance. Namely, while an entangled multipartite quantum state may not admit a decomposition of the form of Eq.~\eqref{Equation:SeparablestatesEquation}, it may be the case that the entanglement is distributed only over some disjoint partitions of parties. More generally, at each round of the experiment a state with entanglement distributed differently could be sent by the source, leading to mixtures of distinctly entangled states. Evidently, the violation of a Bell inequality alone does not provide information about the inseparability structure of the states in such cases.

 Our purpose in this work is not to approach the issue of inseparability certification in full generality, nor to comprehensively discuss the various forms in which inseparability may manifest in a quantum system. Rather we point out, that certain classes of partially deterministic polytopes, or more precisely the inequalities bounding them, can be understood to witness specific forms of inseparability. These objects can therefore be  utilised to provide information about the structure of the entangled states in multipartite Bell scenarios which can not be obtained by merely observing a violation of a Bell inequality.

Let us then introduce a specific sub-class of so-called \emph{partially separable} quantum states \cite{Horodecki2009}.

\definition{(Separability with respect to a subset)\label{Definition:SeparablewithrespectToAsubsetofPartiesQstate}
Let $\rho \in \mathrm{S}(\mathcal{H}_{\vec{A}})$ be an $|I|$-partite quantum state.  Then $\rho$ is termed separable with respect to a subset $I' \subset I$ if and only if it can be decomposed as
\begin{align}
     \rho = \sum_{k\in K} \eta_k [\bigotimes_{i \in I' }\rho^{A_i}_k] \otimes \rho^{\vec{A}_{ I \setminus I'}}_k,
 \end{align}
 with $\eta_k \geq 0$ $\forall k$, $\sum_{k\in K}\eta_{k}= 1 $, $\rho_k^{A_i} \in \mathbb{S}(\mathcal{H}_{A_i})$  and $ \rho_k^{\vec{A}_{I\setminus I'}} \in \mathbb{S}(\mathcal{H}_{\vec{A}_{I\setminus I'}}) = \mathbb{S}(\bigotimes_{i \in I\setminus I'} \mathcal{H}_{A_i})$ for each $k\in K$. The set of $I'$-separable states is denoted by $\mathbf{SEP}(\mathcal{H}_{\vec{A}},I')$. 
}\\

\normalfont
A state $\rho \notin \mathbf{SEP}(\mathcal{H}_{\vec{A}},I')$  shall be termed an $I'$ -inseparable state. Since the condition $I' = \emptyset$ in Definition \ref{Definition:SeparablewithrespectToAsubsetofPartiesQstate} imposes no constraints, we rule this case out of the discussion as trivial. Furthermore if $|I| = 2$, then  any nontrivial $ \emptyset \neq I'\subset I$ reduces to the concept of full separability of Eq.~\eqref{Equation:SeparablestatesEquation}. Therefore we focus the exploration of $I'$-separability to multipartite $|I|\geq 3$ scenarios in particular.

From Definition \ref{Definition:SeparablewithrespectToAsubsetofPartiesQstate} it also follows that if $I'' \subset I'$ then $\mathbf{SEP}(\mathcal{H}_{\vec{A}}, I') \subset \mathbf{SEP}(\mathcal{H}_{\vec{A}}, I'')$. This relation is naturally obeyed on the level of behaviour obtainable from  quantum states in the respective sets as well. In fact, the following theorem concerning the type of behaviour obtainable by local measurements on an $I'$-separable state is immediately established.

\theorem{\label{Theorem;I'separableQstateisinBstimesQs}Let $S=(I,M,O)$ be a Bell scenario and $\rho \in \mathbb{S}(\mathcal{H}_{\vec{A}})$ be separable with respect to $I'$. Then a quantum behaviour $\wp^{\rho}$ obtainable from the state $\rho$ in a scenario $S$ is an element in the partially factorizable quantum set $\mathrm{conv}[\mathbf{B}(S^{I'}) \odot \mathbf{Q}(S^{I'^{\perp}})]$, where $S^{I'}, S^{I'^{\perp}}$ is a bipartition of $S$ defined by the collections $M^{I'}, M^{I'^{\perp}}$ with $M_i^{I'} = M_i$ for all $i\in I'$ and $M_i^{I'}=\emptyset$ otherwise. }
\begin{proof}
    Immediate. If $\rho$ is separable with respect to $I'\subset I$, then $\wp^{\rho}$ decomposes as
    \begin{align}
      \wp^{\rho}(\vec{a}|\vec{x}) &=   \mathrm{tr}\left[\sum_{k}\eta_k(\prod_{i\in I'} M_{a_i|x_i}\rho^{A_i}_k )\otimes  M_{a_{I\setminus I'}|x_{I\setminus I'}}\rho^{A_{I\setminus I'}}_k \right]\\
      &= \sum_{k\in K}\eta_{k}\prod_{i\in I'}P_k(a_i|x_i) \cdot P_k^{Q(S_{|M^{I'^{\perp}}})}(a_{I\setminus I'}|x_{I\setminus I'})\\
      &= \sum_{k\in K}\eta_k \prod_{i\in F_{\vec{x}}} P_{k}(a_i|x_i) \cdot P_k^{Q(S^{I'^{\perp}})}(\vec{a}_{I\setminus F_{\vec{x}}}|\vec{x}_{I\setminus F_{\vec{x}}})\label{Equation:EqbellproductQuantuminQuantumSeparability}
      \end{align}
      where on the last line $F_{\vec{x}}=\{i\in I| x_i\in M_i^{I'}\}$, as before. 
      Equation \eqref{Equation:EqbellproductQuantuminQuantumSeparability} means that 
      \begin{align}
          \wp^{\rho} \in  \mathrm{conv}[\mathbf{B}(S^{{I'}}) \odot \mathbf{Q}(S^{{I'^{\perp}}})],
    \end{align}
as claimed.
\end{proof}

\normalfont

Theorem \ref{Theorem;I'separableQstateisinBstimesQs} can be relaxed to concern the appropriate partially deterministic polytopes instead, as in the corollary below. 

\corollary{\label{Corollary:PDPwitnessesI'inseparability}Let $S=(I,M,O)$ be a Bell scenario, $I' \subset I$ and $\mathbf{PD}(S,M^{I'})$ a partially deterministic polytope defined by $M_i^{I'} = M_i$ for all $i\in I'$ and $M_i^{I'}=\emptyset$ otherwise. If $\wp \notin \mathbf{PD}(S,M^{I'})$ then $\wp$ cannot be obtained by performing local measurements on a quantum state separable with respect to $I'$. }
\begin{proof}
    Follows straightforwardly from definitions. For completeness, let $S^{I'}, S^{I'^{\perp}}$ be a bipartition of $S$ defined by the collections $M^{I'}$ and $M^{I'^{\perp}}$. Since always $\mathbf{Q}(S)\subset \mathbf{NS}(S)$ one gets by Lemma \ref{LemmaBheaviourproductsubsetlemma} and Theorem \ref{TheoremPDPProduct}
    \begin{align}
        \mathrm{conv}[\mathbf{B}(S^{I'}) \odot \mathbf{Q}(S^{I'^{\perp}})] \subset \mathbf{PD}(S,M^{I'}).\label{EqautionConvBprodQisSubsetofPD}
    \end{align}
    On the other hand, by Theorem \ref{Theorem;I'separableQstateisinBstimesQs}, if $\rho$ is separable with respect to $I'$ then any behaviour $\wp^{\rho}$ obtainable from it by local measurements is in $\mathrm{conv}[\mathbf{B}(S^{I'}) \odot \mathbf{Q}(S^{I'^{\perp}})]$. Therefore if $\wp \notin \mathbf{PD}(S,M^{I'})$, it follows by inspection of Eq.~\eqref{EqautionConvBprodQisSubsetofPD} that the quantum state underlying $\wp$ cannot be separable with respect to $I'$.
\end{proof}

\normalfont

Corollary \ref{Corollary:PDPwitnessesI'inseparability}  establishes that the inequalities bounding the partially deterministic polytope $\mathbf{PD}(S,M'^{I'})$ witness $I'$-inseparability of a quantum state in a device-independent manner. Theorem \ref{Theorem:QuantumsetandPDP} shows that, omitting the extreme case where $I' = \emptyset$, there are always at least some quantum behaviour outside of $\mathbf{PD}(S,M^{I'})$ and hence the polytope provides at least some nontrivial witnesses. 

We may also consider situations where it is desired that the state $\rho \in \mathbb{S}(\mathcal{H}_{\vec{A}})$ is not $I'$ separable with respect to multiple distinct subsets $I'_1, \ldots, I_n'$ simultaneously.

\definition{\label{Definition:InseparabilityWRTcollectionofSubsets}(Strong separability with respect to a collection of subsets)\\
Let $\rho \in \mathbb{S}(\mathcal{H}_{\vec{A}})$ be an $|I|$-partite quantum state and  $\mathcal{I} \subset 2^{I}$ be some collection of subsets of $I$. Then $\rho$  is termed strongly separable with respect to the collection $\mathcal{I}$, if and only if 
\begin{align}
    \rho \in \mathbf{SEP}(\mathcal{H}_{\vec{A}}, I') \hspace{0.2cm} \forall I' \in \mathcal{I}, 
\end{align}
or equivalently 
\begin{align}
    \rho \in \bigcap\limits_{I' \in \mathcal{I}} \mathbf{SEP}(\mathcal{H}_{\vec{A}}, I').
\end{align}
The set of strongly $\mathcal{I}$-separable states is denoted by $\mathbf{SSEP}(\mathcal{H}_{\vec{A}}, \mathcal{I}).$ A state $\rho \notin \mathbf{SSEP}(\mathcal{H}_{\vec{A}}, \mathcal{I})$ is termed a weakly $\mathcal{I}$-inseparable. 
}\\
\normalfont

Note that if $|\mathcal{I}| = 1$ then strong $\mathcal{I}$-separability reduces to the notion of $I'$-separability of Definition \ref{Definition:SeparablewithrespectToAsubsetofPartiesQstate}, which is nontrivial only if $I' \neq \emptyset$. 

It is straightforward to extend Theorem \ref{Theorem;I'separableQstateisinBstimesQs} and Corollary \ref{Corollary:PDPwitnessesI'inseparability} to provide conditions for device independent weak $\mathcal{I}$-inseparability witnesses. Indeed, if $\rho \in \mathbf{SSEP}(\mathcal{H}_{\vec{A}}, \mathcal{I})$ and $\wp^{\rho}$ is a behaviour obtainable by performing local measurements on $\rho$, then from Theorem \ref{Theorem;I'separableQstateisinBstimesQs} and analogous arguments as in Corollary \ref{Corollary:PDPwitnessesI'inseparability}
\begin{align}
    \wp^{\rho}\in \bigcap_{I' \in \mathcal{I}} \mathrm{conv}[\mathbf{B}(S^{I'}) \odot \mathbf{NS}(S^{I'^{\perp}})] \subset \bigcap_{I' \in \mathcal{I}} \mathbf{PD}(S,M^{I'}). \label{Equation;IntersectionfofI'polytopes}
\end{align}

The intersection $\bigcap_{I' \in \mathcal{I}} \mathbf{PD}(S,M^{I'})$ in Eq.~\eqref{Equation;IntersectionfofI'polytopes}, as the intersection of convex polytopes, is itself a convex polytope. From Theorem \ref{TheoremPDpolytopeBasicInclusionTHRM} it follows that  $\mathbf{B}(S) \subset \bigcap_{I' \in \mathcal{I}} \mathbf{PD}(S,M^{I'}) $ for any $\mathcal{I}$, as each $\mathbf{PD}(S,M^{I'})$ always contains $\mathbf{B}(S)$. Furthermore it can be established that 
\begin{align}
    \mathrm{\mathrm{Ext}}(\mathbf{B}(S)) \subset \mathrm{\mathrm{Ext}}(\bigcap_{I' \in \mathcal{I}} \mathbf{PD}(S,M^{I'})).\label{Equation:extremepointsofBSarecontainedintheIntersectionofPD's}
\end{align}

This can be seen by use of the implication $\mathbf{B}(S) \subset \bigcap_{I' \in \mathcal{I}} \mathbf{PD}(S,M^{I'}) \subset \mathbf{NS}(S) $ obtainable from Theorem \ref{TheoremPDpolytopeBasicInclusionTHRM} and the fact that $\mathrm{\mathrm{Ext}}(\mathbf{B}(S)) \subset \mathrm{\mathrm{Ext}}(\mathbf{NS}(S))$ proved in Theorem \ref{TheoremExtremePointsInclusionBellPDPNS}. Equation \eqref{Equation:extremepointsofBSarecontainedintheIntersectionofPD's} then follows by appeal to Lemma \ref{LemmaExtremepointsimplication}.

    By definition, weak $\mathcal{I}$-inseparability can be witnessed by violation of any of the inequalities bounding some $\mathbf{PD}(S,M^{I'})$ where $I' \in \mathcal{I}$. In this sense, $I'$-inseparability of some $I' \in \mathcal{I}$ is sufficient to demonstrate that $\rho$ is not weakly separable with respect to $\mathcal{I}$. A stronger notion of inseparability, where the state $\rho$ is not separable with respect to any subset $I' \in \mathcal{I}$, can immediately be identified. 

\definition{($\mathcal{I}$-inseparability)\label{Definition;StronginseparabilityWRTasubsetofsets}\\
Let $\rho \in \mathbb{S}(\mathcal{H}_{\vec{A}})$ be an $|I|$-partite quantum state and  $\mathcal{I} \subset 2^{I}$ be some collection of subsets of $I$. Then $\rho$  is termed  inseparable with respect to the collection $\mathcal{I}$, or  $\mathcal{I}$-inseparable, if and only if
\begin{align}
    \rho \notin \mathbf{SEP}(\mathcal{H}_{\vec{A}}, I') \hspace{0.2cm} \forall I' \in \mathcal{I}, 
\end{align}
or equivalently that
\begin{align}
    \rho \notin \bigcup_{I' \in \mathcal{I}} \mathbf{SEP}(\mathcal{H}_{\vec{A}}, I').
\end{align}
}
\normalfont

Similar arguments as before can now be employed to establish classes of device-independent  $\mathcal{I}$-inseparability witnesses of a state $\rho$.  Namely,  if $\wp \in \mathbf{NS}(S)$ is some behaviour with the property that
\begin{align}
    \wp \notin \bigcup_{I' \in \mathcal{I}}\mathbf{PD}(S,M^{I'})
\end{align} 
then, as a consequence of what was shown in Corollary \ref{Corollary:PDPwitnessesI'inseparability},
\begin{align}
\wp \notin \bigcup_{I' \in \mathcal{I}} \mathrm{conv}[\mathbf{B}(S^{I'}) \odot \mathbf{Q}(S^{I'^\perp})]
\end{align}
and hence, by appeal to Theorem \ref{Theorem;I'separableQstateisinBstimesQs}, the behaviour $\wp$ can not be obtainable from a $\rho \in \bigcup_{I'\in \mathcal{I}} \mathbf{SEP}(\mathcal{H}_{\vec{A}},I')$. 

In order to demonstrate  $\mathcal{I}$-inseparability of the state underlying a behaviour $\wp$, it is sufficient to check that a violation of each individual polytope $\mathbf{PD}(S,M^{I'})$ for every $I' \in \mathcal{I}$ occurs.

 In general, the union of a family of convex polytopes need not be a convex polytope. Furthermore, it is worth noting that the convex hull of the union $\bigcup_{I' \in \mathcal{I}}\mathbf{PD}(S,M^{I'})$  is an object worthy of attention in itself. To illustrate this, let us identify another form of separability, which corresponds to the situation where entangled states of different structures are mixed in the experiment.

 \definition{(Weak $\mathcal{I}$-separability\label{Definition;WeakSeparabilityofQstate})\\
 Let $\rho \in \mathrm{S}(\mathcal{H}_{\vec{A}})$ be an $|I|$-partite quantum state and $\mathcal{I}\subset 2^{I}$ a collection of subsets of $I$.  Then $\rho$ is termed weakly separable with respect $\mathcal{I}$ if and only if 
\begin{align}
    \rho \in \mathrm{conv}\left( \bigcup_{I' \in \mathcal{I}} \mathbf{SEP}(\mathcal{H}_{\vec{A}}, I') \right).
\end{align}
The set of weakly $\mathcal{I}$-separable states is denoted by $\mathbf{WSEP}(\mathcal{H}_{\vec{A}}, \mathcal{I})$.
 }\\

 \normalfont

Intuitively, the set $\mathbf{SEP}(\mathcal{H}_{\vec{A}}, \mathcal{I})$ allows arbitrary quantum correlations between any parties in for at least some $I\setminus I' $, where $I'\in  \mathcal{I}$ while all the parties $I'$ are uncorrelated. The set $\mathbf{WSEP}(\mathcal{H}_{\vec{A}}, \mathcal{I}$) of weakly separable states on the other hand further allow for classical correlations between the parties.
 
 A weakly $\mathcal{I}$-separable quantum state $\rho$ can be decomposed as
 \begin{align}
     \rho =\sum_{I' \in \mathcal{I}} \eta^{I'} \rho^{I'},
 \end{align}
 with $\sum_{I' \in \mathcal{I}}\eta^{I'} = 1$, $\eta^{I' \geq 0}$ $\forall I'$ and $\rho^{I'} \in \mathbf{SEP}(\mathcal{H}_{\vec{\mathcal{A}}}, I')$ $\forall I' \in 
 \mathcal{I}$. Clearly for $I' \in \mathcal{I}$
 \begin{align}
   \mathbf{SEP}(\mathcal{H}_{\vec{A}}, I') \subset   \mathbf{SEP}(\mathcal{H}_{\vec{A}}, \mathcal{I}) \subset \mathbf{WSEP}(\mathcal{H}_{\vec{A}}, \mathcal{I}),
 \end{align}
and consequently any $\rho \notin \mathbf{WSEP}(\mathcal{H}_{\vec{A}}, \mathcal{I})$ is also strongly $\mathcal{I}$-inseparable, and inseparable with respect to any $I' \in \mathcal{I}$ so that a hierarchy of inseparability classes may be established. It is useful to give this form of inseparability its own name. 
 
\definition{(Strong $\mathcal{I}$-inseparability\label{Definition:GenuineInseparability})\\
Let $\rho \in \mathrm{S}(\mathcal{H}_{\vec{A}})$ be an $|I|$-partite quantum state and $\mathcal{I}\subset 2^{I}$ a collection of subsets of $I$.  Then $\rho$ is termed strongly inseparable with respect to $\mathcal{I}$, or strongly $\mathcal{
I}$-inseparable,  if and only if 
\begin{align}
    \rho \notin \mathbf{WSEP}(\mathcal{H}_{\vec{A}}, \mathcal{I}).
\end{align}}\\
\normalfont

Table~\ref{Table:SeparabilityandInseparabilityDefinitions} summarizes the different definitions of separability and inseparability considered in this section.

\tikzset{myarrow/.style={-stealth,shorten >=\spc, shorten <=2.5*\spc}}

\begin{table*}[ht!]
    \begin{tabular}{|c|c|c||c|c|}\hline
        Separability property& Terminology & Associated set & Converse & Terminology  \\\hline
        $\rho = \sum_k \eta_k\bigotimes_{i\in I}\rho^{A_i}_k $ & Fully separable & $\mathbf{SEP}(\mathcal{H}_{\vec{A}})$ & $\rho \notin \mathbf{SEP}(\mathcal{H}_{\vec{A}})$& Entangled\\\hline
         $\rho = \sum_{k\in K} \eta_k [\bigotimes_{i \in I' }\rho^{A_i}_k] \otimes \rho^{\vec{A}_{ I \setminus I'}}_k$ &  $I'$-separable& $\mathbf{SEP}(\mathcal{H}_{\vec
         A}, I')$ & $\rho \notin \mathbf{SEP}(\mathcal{H}_{\vec{A}}, I')$& $I'$-inseparable\\ \hline
         $\rho \in \bigcap\limits_{I' \in \mathcal{I}}\mathbf{SEP}(\mathcal{H}_{\vec{A}},I') $ & Strongly  $\mathcal{I}$-separable & $\mathbf{SSEP}(\mathcal{H_{\vec{A}}},\mathcal{I})$)& $\rho \notin \mathbf{SSEP}(\mathcal{H}_{\vec{A}},\mathcal{I})$& Weakly $\mathcal{I}$-inseparable\\\hline
         $\ \rho \in \bigcup\limits_{I' \in \mathcal{I}}\mathbf{SEP}(\mathcal{H}_{\vec{A}}, I')$ & $\mathcal{I}$-separable& $\mathbf{SEP}(\mathcal{H}_{\vec{A}}, \mathcal{I})$& $\rho \notin \mathbf{SEP}(\mathcal{H}_{\vec{A}}, \mathcal{I})$&  $\mathcal{I}$-inseparable \\\hline
         $\rho \in \mathrm{conv}[\mathbf{SEP}(\mathcal{H}_{\vec{A}}, \mathcal{I})]$& Weakly $\mathcal{I}-$separable & $\mathbf{WSEP}(\mathcal{H}_{\vec{A}}, \mathcal{I})$& $\rho \notin \mathbf{WSEP}(\mathcal{H}_{\vec{A}}, \mathcal{I})$& Strongly $\mathcal{I}-$inseparable \\\hline
    \end{tabular}
    \caption{A table summarizing the various quantum state separability and inseparability notions discussed in this section in the order of appearance. Here $I'\subset I$ is some fixed but arbitrary subset of agents and $\mathcal{I}\subset 2^{I}$ is some fixed but arbitrary collection of subsets of agents. If $I' \in \mathcal{I}$ above, then strong $\mathcal{I}$-separability implies $I'$-separability, in which case the separability notions may be reordered with respect to set inclusion by swapping the second and third rows. Similarly, a hierarchy of inseparability notions is obtained with strong $\mathcal{I}$-inseparability in particular implying all other forms of inseparability of the state. We focus our attention to device independent witnesses of $I'$-inseparability,  $\mathcal{I}$-inseparability and strong $\mathcal{I}$-inseparability, and relate them to partially deterministic polytopes of certain types. 
   }
   \label{Table:SeparabilityandInseparabilityDefinitions}
\end{table*}

Again, as in the cases of $I'$, weak $\mathcal{I}$ and $\mathcal{I}$-inseparability witnesses, it can be established by appeal to Theorem \ref{Theorem;I'separableQstateisinBstimesQs} and Corollary \ref{Corollary:PDPwitnessesI'inseparability} that 
$\rho \in  \mathbf{WSEP}(\mathcal{H}_{\vec{A}}, \mathcal{I})$ implies that a behaviour $\wp^{\rho} $ obtainable from the state $\rho$ with local measurements satisfies 

\begin{align}
    \wp^{\rho} &\in \mathrm{conv}\left[ \bigcup_{I' \in \mathcal{I}} (\mathrm{conv}(\mathbf{B}(S^{I'}) \odot \mathbf{Q}(S^{I'{^\perp}}) \right] \\ &\subset \mathrm{conv}\left[ \bigcup_{I' \in \mathcal{I}} \mathbf{PD}(S, M^{I'})) \right].
\end{align}
Therefore, the convex set $\mathrm{conv}\left[ \bigcup_{I' \in \mathcal{I}} \mathbf{PD}(S, M^{I'}) \right]$ provides a device-independent strong $\mathcal{I}$-inseparability witness in the sense that $\wp \notin \mathrm{conv}\left[ \bigcup_{I' \in \mathcal{I}} \mathbf{PD}(S, M^{I'}) \right] $ implies the behaviour $\wp$ is not obtainable from a weakly $\mathcal{I}$-separable quantum state. 

The following result is immediate.

\theorem{\label{Theorem:convexHullofUnionofI'PDP'sisapolytope}Let $S= (I,M,O)$ be a scenario and $\mathcal{I}\subset 2^{I}$ an arbitrary collection of subsets $I'$. Then, the convex set 
\begin{align}
    \mathrm{conv}[\bigcup_{I' \in \mathcal{I}} \mathbf{PD}(S,M^{I'})]
\end{align}
is a convex polytope. Furthermore, the set of is extreme points is exactly the union $\bigcup_{I'\in \mathcal{I}}\mathrm{Ext}(\mathbf{PD}(S,M'))$ of the extreme points of each partially deterministic polytope $\mathbf{PD}(S,M')$.
}\begin{proof}
    Note first that for any collection of sets $U_k$ it always holds that $
        \bigcup_{k}\mathrm{conv}[ U_k ]\subset \mathrm{conv}\left[\bigcup_k U_k \right]$
    and hence also
    \begin{align}
        \bigcup_{I' \in \mathcal{I}} \mathbf{PD}(S,M') &\subset \mathrm{conv}\left[\bigcup_{I' \in \mathcal{I}} \mathrm{Ext}(\mathbf{PD}(S,M^{I'}))\right]\\
        & \subset \mathrm{conv}\left[\bigcup_{I' \in \mathcal{I}} \mathbf{PD}(S,M') \right],
    \end{align}
    which implies, by taking the convex hull of both sides of the relations above that $
        \mathrm{conv}\left[\bigcup_{I' \in \mathcal{I}} \mathbf{PD}(S,M') \right] = \mathrm{conv}\left[\bigcup_{I' \in \mathcal{I}} \mathrm{Ext}(\mathbf{PD}(S,M^{I'}))\right].$

    Since a finite union of finite sets is itself a finite set, this already establishes that the set $\mathrm{conv}\left[\bigcup_{I' \in \mathcal{I}} \mathbf{PD}(S,M^{I'})\right] $ is a convex polytope, the extreme points of which are contained in the union of the extreme points of each polytope $\mathbf{PD}(S,M^{I'})$, that is 
    \begin{align}
        \mathrm{Ext}(\bigcup_{I' \in \mathcal{I}} \mathbf{PD}(S,M^{I'})) \subset \bigcup_{I'\in \mathcal{I}}\mathrm{Ext}(\mathbf{PD}(S,M^{I'})).
    \end{align}
     To see that in fact every point in the set $\bigcup_{I' \in \mathcal{I}} \mathrm{Ext}(\mathbf{PD}(S,M^{I'}))$ is extremal, it is sufficient to note that 
    \begin{align}
        \mathbf{PD}(S,M^{I'}) \subset \mathrm{conv}\left[\bigcup_{I' \in \mathcal{I}} \mathbf{PD}(S,M^{I'})\right] \subset \mathbf{NS}(S)
    \end{align}
   holds for every $I' \in \mathcal{I}$. And since by Theorem \ref{TheoremExtremePointsInclusionBellPDPNS} the subset relation $\mathrm{Ext}(\mathbf{PD}(S,M^{I'}))\subset \mathrm{Ext}(\mathbf{NS}(S))$ holds for all $I'\in \mathcal{I}$, an application of Lemma \ref{LemmaExtremepointsimplication} yields $\wp\in \mathrm{Ext}(\mathbf{PD}(S,M^{I'})) \Rightarrow \wp \in \bigcup_{I' \in \mathcal{I}} \mathrm{Ext}(\mathbf{PD}(S,M^{I'}))$  for every $I' \in \mathcal{I}$, or equivalently that 
   \begin{align}
      \bigcup_{I'\in \mathcal{I}}\mathrm{Ext}(\mathbf{PD}(S,M^{I'}))  \subset \mathrm{Ext}(\bigcup_{I' \in \mathcal{I}} \mathbf{PD}(S,M^{I'})),
   \end{align}
   completing the proof. 
\end{proof}
\normalfont

We can also demonstrate by means of specific scenario that a distinction between $\mathcal{I}$ and strong $\mathcal{I}$-inseparability witnesses phrased relative to the partial deterministic polytopes does indeed emerge. 

\lemma{\label{Lemma:ThreepartyPDPI'lemma}Let $S= (I,M,O)$ be a scenario with $|I|=3$. Let $\emptyset \neq I' \subset I$ be a non-empty subset of the parties. Then either  
\begin{enumerate}
    \item $|I'|=1$ and $\mathbf{PD}(S,M^{I'})$ is one of the non-Bell partially deterministic polytopes $\mathbf{PD}(S,M^{\{i'\}})$ which has all the inputs for one party $i'\in I$ deterministic.  \\
    \item $\mathbf{PD}(S,M^{I'}) = \mathbf{B}(S)$.
\end{enumerate}
}
\begin{proof}
    Omitted. This is a special case of Theorem \ref{TheoremMainPDequalBellTHRM}. See also Example \ref{Example:partialdeterminismInThreepartiteTwo-inputScenario} and Fig.~\ref{fig:TripartiteEquivalenceClasses}. 
\end{proof}

\theorem{\label{Theorem:unionof3partypdp'SiSNOTCONVEX}
Let $S= (I, M, O)$ be a scenario with $|I| = 3.$ Let $\mathcal{I} \subset 2^{I}$ be a collection of subsets of $I$ such that  $|I'| \neq 0$ for all $I' \in \mathcal{I}$. Then the union  $\bigcup_{I' \in \mathcal{I}}\mathbf{PD}(S,M^{I'})$ is not convex if and only if  $|I|-|I'| =2 $ for at least two $I' \in \mathcal{I}$.
}
\begin{proof}[Proof sketch.]
    We provide a detailed proof in Appendix \ref{appendixSection:ProofofProp}. Here we merely provide an outline of how the proof is constructed.

    The 'only if' is an immediate consequence of Lemma \ref{Lemma:ThreepartyPDPI'lemma}, by which the condition $|I|-|I'| = 2$ for at least two $I' \in \mathcal{I}$ is necessary for two distinct non-Bell partially deterministic polytopes to be included in the union. Since by definition every partially deterministic polytope is individually convex  and by Theorem \ref{TheoremPDpolytopeBasicInclusionTHRM} contains the Bell polytope, this is necessary for the possibility of having non-convexity of the union.
    
    To prove the 'if' claim, we show  that for any collection $\mathcal{I}$ compatible with the assumption in the three-partite case,  one can find a class of Bell-inequalities which are valid for one non-Bell polytope  polytopes $\mathbf{PD}(S,M^{I'})$ defined by the $I'\in \mathcal{I}$ with $|I|-|I'| = 2$, but is not valid for any of the other ones. Then, by taking suitable convex combinations of points of the different partially deterministic polytopes, one finds that their mixtures violate every such inequality for each polytope, witnessing that the convex combination of said points cannot be contained in any single one of the polytopes. 
\end{proof}

\normalfont
 
 We conjecture that the claim of Theorem \ref{Theorem:unionof3partypdp'SiSNOTCONVEX} generalises to cases with $I>3$ so that in general one may expect the sets $\bigcup_{I' \in \mathcal{I}} \mathbf{PD}(S,M^{I'})$ and its convex hull $\mathrm{conv}[\bigcup_{I' \in \mathcal{I}} \mathbf{PD}(S,M^{I'})]$ to provide different kinds of witnesses. A proof of that claim is, however, beyond the scope of this work.  

As stated before, by appeal to Theorem \ref{Theorem:QuantumsetandPDP} the set $\mathbf{PD}(S,M^{I'})$ provides witnesses $I'$-inseparability for at least some quantum states whenever $I' \neq \emptyset$. So far, however, we have not shown that the (at least sometimes strictly) larger sets $\bigcup_{I'\in \mathcal{I}}\mathbf{PD}(S,M^{I'})$ and $\mathrm{conv}[\cup_{I'\in \mathcal{I}}\mathbf{PD}(S,M^{I'})]$ witness any distinct forms of quantum nonseparability. For completeness, we therefore demonstrate by means of an example in the 3-party case that the objects  constructed in this manner may indeed witness inseparability of  distinct subsets of quantum states. The following definition is introduced for convenience. 

\definition{\label{Definition:maximalUnion}(Maximal collections and unions)\\
Let $S= I,M,O$ be an $|I|\geq 3$-partite scenario and $\mathcal{I} \subset 2^{I}$ a collection of subsets of the parties such that $I' \neq \emptyset$ for all $I' \in \mathcal{I}$. The collection $\mathcal{I}$ is referred to as a maximal collection if and only if it contains all the, and only the $|I|$ distinct singleton sets $I' = \{i \}$ with $i\in I$. The notation $\mathcal{I}^{\max}$ is used for a maximal collection. 

The set $ \bigcup_{I' \in \mathcal{I}} \mathbf{PD}(S,M^{I'})$} is referred to as the maximal union of  partially deterministic polytopes if and only if
\begin{align}
    \bigcup_{I' \in \mathcal{I}} \mathbf{PD}(S,M^{I'}) = \bigcup_{I'' \in 2^{I}\setminus \{ \emptyset \}} \mathbf{PD}(S,M^{I''}).\label{Equation:maximaluniuonofPDP's}
\end{align}
 Similarly, the sets  composed of partially separable quantum states $\bigcup_{I' \in \mathcal{I}}\mathbf{SEP}(\mathcal{H}_{\vec{A}}, I')$ and  of partially factorizable quantum behaviour $ \bigcup_{I'\in \mathcal{I}} [\mathrm{conv}[\mathbf{B}(S^{I'}) \odot \mathbf{Q}(S^{I'^{\perp}})]]$ are termed maximal unions, when they satisfy their respective analogue of Eq.~\eqref{Equation:maximaluniuonofPDP's}.
}\\
\normalfont

The relationship of maximality of a collection $\mathcal{I}^{\max}$ and of a union of sets in Definition \ref{Definition:maximalUnion} can be illustrated in the three-partite case; here, for example,  the union $\bigcup_{I'\in \mathcal{I}} \mathbf{PD}(S,M^{I'})$ is maximal if it contains the three polytopes which are deterministic with respect to only one party, since by Lemma \ref{Lemma:ThreepartyPDPI'lemma} every other polytope deterministic with respect to $I'$ is equal to the Bell-polytope. It is easy to see that the argument generalizes irrespectively of that fact, so that the union $\bigcup_{I'\in \mathcal{I}} \mathbf{PD}(S,M^{I'})$ is maximal, whenever the set $\mathcal{I}$ is maximal in the sense of containing the $|I|$ distinct singleton sets. Namely if $I' \subset I$ is a subset with $|I|-|I'| > |I|-1$, then by definition of maximality of the union, there exists at least one subset $\widetilde{I}' \in \mathcal{I}$ such that $I' \subsetneq \widetilde{I}'$ and hence by Theorem \ref{TheoremPDpolytopeBasicInclusionTHRM} at least one polytope $\mathbf{PD}(S,M^{\widetilde{I}'})$ such that $\mathbf{PD}(S,M^{I'}) \subset \mathbf{PD}(S,M^{\widetilde{I'}})$. Thus, in particular,  it is sufficient to only take the $|I|$ distinct polytopes which have the inputs of $i\in I$ deterministic in the union. Owing to Definition \ref{Definition:SeparablewithrespectToAsubsetofPartiesQstate} of $I'$-separability, similar conclusions hold for the sets comprised of partially separable quantum states and partially factorizable quantum behaviours. 

Before proceeding further with our discussion,  let us mention  results which apply to the  intersections $\mathbf{SSEP}(\mathcal{H}_{\vec{A}}, \mathcal{I}^{\max})$ and $\bigcap\limits_{I' \in \mathcal{I}^{\max}}\mathbf{PD}(S,M^{I'})$ obtained in Refs.~\cite{ Vertesi2012, Liang2020}. Namely, as a special case of a result in Ref.~\cite{Liang2020}, there are quantum-realizable behaviours $\wp^Q  \in \bigcap\limits_{I' \in \mathcal{I}^{\max}}\mathbf{PD}(S,M^{I'})$ in a tripartite scenario  $S$ with two dichotomic inputs per site, which violate a triparite Bell inequality. Hence $\mathbf{B}(S) \subsetneq \bigcap\limits_{I' \in \mathcal{I}^{\max}}\mathbf{PD}(S,M^{I'})$ in that scenario. This phenomenon, where a behaviour is 'separable with respect to any bipartition' but not fully separable was termed 'anonymous nonlocality' in \cite{Liang2020}, and it may be considered as a device-independent analogue of a related phenomenon on the level of quantum states called 'delocalized entanglement' \cite{diVincenzo2003}, where a quantum state is entangled, but is separable with respect to every nontrivial partition. Note that this tripartite scenario is within the scope of Example \ref{Example:partialdeterminismInThreepartiteTwo-inputScenario} illustrated in Fig.~\ref{fig:TripartiteEquivalenceClasses}, where three distinct equivalence classes of polytopes are included in the intersection\footnote{Interestingly  in the appendix of Ref.~\cite{ramanathan2024maximumquantumnonlocalitysufficient} it was shown that the intersection of the nine different bipartite partially deterministic polytopes with one deterministic input at each site portrayed in Fig.~\ref{fig:bipartiteThreeInputScenario} was shown to be the Bell-polytope, in contrast to the intersection of the tripartite two-input polytopes. There, the intersection of the bipartite polytopes was put forward as a witness of 'randomness in at least some pair of inputs'. } with respect to the collection $\mathcal{I}^{\max}$. 

Ref.~\cite{Vertesi2012} explicitly  constructed a three-qubit state $\rho^{VB}$\footnote{Here the letters VB stand for $Vertesi-Brunner$ --the initials of the authors of Ref.~\cite{Vertesi2012}. } which is separable with respect to any bipartition, and hence an element in $\mathbf{SSEP}(\mathcal{H}_{\vec{A}}, \mathcal{I}^{\max})$, but which can be used to violate a Bell-inequality.  Together the results of \cite{Vertesi2012, Liang2020} indicate that (at least in the tripartite scenario) there are indeed strongly $\mathcal{I}^{\max}$ separable states for which the behaviour $\wp^\rho$ obtainable by local measurements are always contained in the intersection $\bigcap\limits_{I' \in \mathcal{I}^{\max}}\mathbf{PD}(S,M^{I'})$, but which can nonethless be witnessed to be entangled in a device-independent way.

As stated before, a violation of any bounds on a single partially deterministic polytope $\mathbf{PD}(S,M^{I'})$, $I'\in \mathcal{I}$ is sufficient to demonstrate weak $\mathcal{I}$-inseparability. It is easy to show examples of behaviours $\wp^\rho \in \mathbf{PD}(S,M^{I'})$ for some $I'\in \mathcal{I}$ but $\wp^{\rho}\notin \bigcap\limits_{I' \in \mathcal{I}}\mathbf{PD}(S,M^{I'})$.

\example{\label{Example:ThreequbitQstatewhichisNotImaxSeP}(A three-qubit quantum state $\rho \notin \mathbf{SEP}(\mathcal{H}_{\vec{A}}, I') $ for an $I' \in \mathcal{I}$ but $\rho \in \mathbf{SEP}(\mathcal{H}_{\vec{A}}, \mathcal{I}^{\max}))$\\

Let $\mathcal{H}_{\vec{A}} \simeq \mathbb{C}^{2}\otimes \mathbb{C}^{2}\otimes \mathbb{C}^{2}$ be the Hilbert space of three qubits distributed to  parties $A, B$ and $C$, respectively.  Consider a pure state $\rho_{A}^{\Psi^{-}}$ defined by 
\begin{align}
    \label{Equation:DefinitionofPartiallyI'Three-partySingletState}\rho_{A}^{\Psi^{-}} = \ket{1_A}\bra{1_A} \otimes \ket{\Psi^{-}_{BC}}\bra{\Psi^{-}_{BC}},
\end{align}
where $\ket{\Psi^{-}_{BC}} = \dfrac{1}{\sqrt{2}}(\ket{0_B}\otimes \ket{1_C}- \ket{1_B}\otimes \ket{0_{C}})$ is a singlet state, and analogous states $\rho^{\Psi^{-}}_{B}$, $\rho^{\Psi^{-}}_{C}$ obtainable by swapping suitable labels in Eq.~\eqref{Equation:DefinitionofPartiallyI'Three-partySingletState}. Clearly $\rho \in \mathbf{SEP}(\mathcal{H}_{\vec{A}}, \{A\})$ and hence $\rho_{A}^{\Psi^{-}}\in \mathbf{SEP}(\mathcal{H}_{\vec{A}}, \mathcal{I^{\max}}))$.  

We will show that $\rho^{\Psi^{-}}_{A} \notin \mathbf{SEP}(\mathcal{H_{\vec{A}}}, \{ B \})$ can be witnessed by $\mathbf{PD}(S,M^{\{B\}})$. Suppose then, for the sake of deriving a contradiction, that $\rho^{\Psi^{-}}_{A}\in \mathbf{SEP}(\mathcal{H_{\vec{A}}}, \{ B \}).$ Let $S= (I,M,O)$ be a tripartite scenario with two binary inputs per site. Then, from Theorem \ref{Theorem;I'separableQstateisinBstimesQs} and Corollary \ref{Corollary:PDPwitnessesI'inseparability} it follows that any behaviour $\wp^{\rho^{\Psi^{-}}_{A}}$ obtainable from the state $\rho^{\Psi^{-}}_{A}$ by local measurements by the parties satisfies $\wp^{\rho^{\Psi^{-}}_{A}} \in \mathbf{PD}(S,M^{\{B\}})$. Let $R_{|\{BC\}}$ denote a restriction map in the sense of Definition \ref{Definition:RestrictionofaBehaviour} to a bipartite scenario $S_{|\{BC\}}$ which includes only the inputs of $B$ and $C$. Then by Proposition \ref{Proposition:ImageRestrictionofBehproductistheDomain}, Theorems \ref{TheoremPDPProduct}, \ref{TheoremMainPDequalBellTHRM} and Lemma \ref{Lemma:convexityofRestrictionmap}
\begin{align}
    R_{|BC}(\mathbf{PD}(S,M^{\{B\}}) = \mathbf{PD}(S_{|BC}, M^{\{B\}}) = \mathbf{B}(S_{|BC}).
\end{align}
On the other hand, by Proposition \ref{Proposition:ImageRestrictionofBehproductistheDomain} the relevant constraint on the image of $\wp^{\rho^{\Psi^{-}}_A}$ is that
\begin{align}
    R_{|BC}(\wp^{\rho^{\Psi^{-}_A}}) \in \mathbf{Q}(S_{|BC}) \supsetneq \mathbf{B}(S_{|BC}).
\end{align}
The relation $R_{|BC}(\wp^{\rho^{\Psi^{-}_A}}) \notin \mathbf{B}(S_{|BC})$ can be demonstrated, for example, with $B$ and $C$ having measurements which violate the relevant CHSH-inequality, for example. This is the desired contradiction, as we have shown that there are at least some behaviours $\wp^{\Psi^{-}_{A}}\notin \mathbf{PD}(S,M^{\{B\}})$ obtainable from $\rho^{\Psi^{-}_{A}}$ by local measurements by the parties which is sufficient to demonstrate that $\rho^{\Psi^{-}_{A}}\notin \mathbf{SEP}(\mathcal{H}_{\vec{A}}, \{B\})$. An analogous argument can be used to show that also $\rho^{\Psi^{-}_{A}}\notin \mathbf{SEP}(\mathcal{H}_{\vec{A}}, \{C\})$. 
}\\
\normalfont

Of course, to see that $\rho_A^{\Psi_-} \notin \mathbf{SEP}(\mathcal{H}_{\vec{A}}, \{B\})$ in Example \ref{Example:ThreequbitQstatewhichisNotImaxSeP} one could simply look at the form of the marginal states $tr_A[\rho]$ for $\rho \in \mathbf{SEP}(\mathcal{H}_{\vec{A}, \{B\}})$ and observe a contradiction. We reiterate the point which is  that  the polytope $\mathbf{PD}(S,M^{\{B\}})$ provides a device-independent witness of this form of inseparability, and hence knowledge of the initial state is not strictly required. Example \ref{Example:ThreequbitQstatewhichisNotImaxSeP} simply demonstrates, that there are some states, such as $\rho_A^{\Psi-}$ for which $I'$-inseparability can be witnessed for some $I'$, but which satisfies a weaker notion of $\mathcal{I}^{\max}$-separability.

We will also show that there are quantum states which are not weakly separable with respect to the maximal collection $\mathcal{I}^{\max}$. At this point, however, let us connect this discussion to the notions of genuine multipartite nonlocality which have been previously investigated in the literature for example, in Refs.~\cite{Svetlichny1987, Gallego2012, Bancal2013}.

\definition{\label{Definition:Svetlichny3-partynonlocality}(Svetlichny's genuine three-partite nonlocality \cite{Svetlichny1987})\\
Let $S = (I,M,O)$, $|I|=3$ be a three-partite scenario. A behaviour $\wp \in \mathbf{E}(S)$ admits a Svetlichny-type model\footnote{The terms ``hybrid local-nonlocal model'' \cite{Collins2002} and ``bilocal model'' \cite{Gallego2012}  have also been used in the literature for behaviour of the Svetlichny-type. We choose this terminology since the term 'hybrid' is not, in our view, descriptive enough in the present context, while the term 'bilocal' has also been used to refer to different correlation sets in network scenarios with multiple independent sources \cite{Branciard2010}. } if and only if the behaviour $\wp$ can be decomposed as
\begin{align}
\begin{split}
    \wp(\vec{a}|\vec{x}) = \int_{\mathscr{M}_1}P(a_1|x_1,\mu_1)P(a_2a_3|x_2x_3, \mu_1)p_1(\mu_1) d\mu_1 \\
    + \int_{\mathscr{M}_2} P(a_2|x_2,\mu_2)P(a_1a_3|x_1x_3,\mu_2)p_2(\mu_2)d\mu_2 \\
    + \int_{\mathscr{M}_3}P(a_3|x_3,\mu_3)P(a_1a_2|x_1x_2,\mu_3)p_3(\mu_3)d\mu_3,
    \end{split}
\end{align}
with $\mathscr{M}_i$ a measurable set, $p_i(\mu_i) \geq 0$ for all $i\in I$ and   
\begin{align}
    \int_{\mathscr{M}_1}p_1(\mu_1)d\mu_1 + \int_{\mathscr{M}_2}p_2(\mu_2)d\mu_2 + \int_{\mathscr{M}_3}p_3(\mu_3)d\mu_3 =1. \label{Equation:normalizationofSvetlichny}
\end{align}
The set of Svetlichny-modelable, or Svetlichny two-way-local behaviour in the three-partite scenario is denoted by $\mathbf{SL}(S_{|I|=3})$. If $\wp \notin \mathbf{SL}(S_{|I|=3})$, the behaviour $\wp$ is termed genuinely Svetlichny three-partite nonlocal. 
}\\
\normalfont

Note that in Definition \ref{Definition:Svetlichny3-partynonlocality} the marginals \begin{align}
    P(a_ja_k|x_jx_k) = \int_{\mathscr{M}_i} P(a_ja_k|x_jx_k,\mu_i)p_i(\mu_i)d\mu_{i},
\end{align}
with $i,j,k \in I$, $i\neq j \neq k$, are not enforced to satisfy no-signalling, in contrast to distributions in $ \bigcup_{I'\in \mathcal{I^{\max}}}\mathbf{PD}(S,M^{I'}).$ Due to the observations made in Remark \ref{remark:behaviourProductOFsignALLINGsubsets}, the behaviour product of Definition \ref{BehaviourProductDefinition} can nonetheless in this situation be  utilized to give an alternative characterization of the set $\mathbf{SL}(S_{|I|=3})$.

\proposition{\label{Proposition:Svetlichny-3partiteInTermsofBehaviourproduct}$\mathbf{SL}(S_{|I|=3}) = \bigcup_{I' \in \mathcal{I^{\max}}} \mathrm{conv}[\mathbf{B}(S^{I'})\odot \mathbf{E}(S^{I'^{\perp}})]$}
\begin{proof}
   We omit a detailed proof. Note that if $|I'|\geq 2$ then the set $\mathrm{conv}[\mathbf{B}(S^{I'})\odot \mathbf{E}(S^{I'^{\perp}})]= \mathbf{B}(S)$. Hence it is sufficient to consider only cases with where $I' =\{i\}$, $i\in I$, since those products always contain $\mathbf{B}(S)$. 
   
   The claim can then be established by inspection following, for example,  the argument made with respect to equality of the sets $\mathbf{PF}(S,M')$ and $\mathbf{PD}(S,M')$ in Theorem \ref{Theorem:PartiallyFactorizableEqualsPartiallyDeterministicPolytope} which allows the integrals be replaced with sums for each $i\in \{1,2,3\}$, and noting that a behaviour $\wp\in \mathbf{SL}(S_{|I|=3})$ is a convex mixture of such of behaviour in the composable sets $\mathrm{conv}[\mathbf{B}(S^{I'})\odot \mathbf{E}(S^{I'^{\perp}})]$ once the normalization in Eq.~\eqref{Equation:normalizationofSvetlichny} is taken into account. 
\end{proof}
\normalfont

It may be worth pointing out here, that an arbitrary restriction $R_{|V}$ may not be well defined on the whole set $\mathbf{SL}(S_{|I|=3})$. Nonetheless, the claim of Lemma \ref{LemmaBheaviourproductsubsetlemma} still holds, and one may define restriction maps with respect to the collections $V = M^{\{i\}}$ and $V=M^{\{i\}^{\perp}}$ so that Eqs.~\eqref{Equation:restrictiontoM'isKS'} and \eqref{Equation:restrictiontoM'isKS'perp}  in particular are still valid for each individual composition. From Proposition \ref{Proposition:Svetlichny-3partiteInTermsofBehaviourproduct} it is easy to see that in a three three-partite scenario
\begin{align}
    \mathrm{conv}[\bigcup_{I' \in \mathcal{I}^{\max}} \mathbf{PD}(S,M^{I'})] \subsetneq \mathbf{SL}(S_{|I|=3}).
\end{align} since for each element in the union, the relation $\mathbf{NS}(S^{I'})\subsetneq \mathbf{E}(S^{I'^{\perp}})$ holds. In particular, $\mathbf{SN}(S_{I_{|I|=3}})$ contains some behaviour which are outside the no-signalling polytope $\mathbf{NS}(S)$.

\example{ \label{Example:ThreeQubitNotWeaklySeparable}(A strongly $\mathcal{I}^{\max}$-inseparable three-qubit quantum state $\rho \notin \mathbf{WSEP}(\mathcal{H}_{\vec{A}}, \mathcal{I^{\max}})$.\\ 

Svetlichny derived an inequality valid for the set $\mathbf{SL}(S_{|I|=3})$ \cite{Svetlichny1987}, and showed that it is violated by a behaviour obtained by suitable measurements on the tripartite state $ \rho^G = \ket{\Phi}\bra{\Phi}, \ket{\Phi} \in \mathcal{H}_{\vec{A}} \simeq \mathbb{C}^{2}\otimes \mathbb{C}^2 \otimes \mathbb{C}^2$ with
\begin{align}
    \ket{\Phi} = \frac{1}{2}(\ket{100}+\ket{010}+\ket{001}-\ket{111}).
\end{align}
Since $ \mathrm{conv}[\bigcup_{I' \in \mathcal{I}^{\max}} \mathbf{PD}(S,M^{I'})] \subsetneq \mathbf{SL}(S_{|I|=3})$, it is also the case that $\rho^G \notin \mathbf{WSEP}(\mathcal{H}_{\vec{A}}, \mathcal{I^{\max}})$, since any behaviour $\wp^{\rho}$ was shown to be an element in $ \mathrm{conv}[\bigcup_{I' \in \mathcal{I}^{\max}} \mathbf{PD}(S,M^{I'})] $ for any weakly $\mathcal{I}$-separable state $\rho$.}\\

\normalfont
For other investigations of violations of Svetlichny's tripartite conditions, see for example \cite{Mitchell2004}.

Owing to the strict inclusion relative to $\mathbf{SL}_{|I|=3}$, the polytope $\mathrm{conv}[\bigcup_{I' \in \mathcal{I}^{\max}}\mathbf{PD}(S,I')]$  witnesses an alternate form of 'genuine multipartite nonlocality' , distinct from Svetlichny's form of Definition \ref{Definition:Svetlichny3-partynonlocality}. In Ref.~\cite{Bancal2013} this form of tripartite nonlocality was denoted as $\mathbf{NS}_2(S)$-locality, where the subscript ``$2$'' refers to the fact that the expansion consists of mixtures of terms which are local with respect to a bipartition of the parties. Thus, we may say that behaviour of this form have the property of  ``biseparability''  in analogy with terminology used in the study of  quantum state separability properties in multipartite scenarios \cite{Seevinck2008} (in fact in Ref.~\cite{Gallego2012} the term ``no-signalling bilocality'' was used for behaviour in $\mathbf{NS}_2(S))$.  These are, of course, not the only interesting notions of nonclassical correlations one may consider in multipartite scenarios as has been pointed out in \cite{Bancal2013, Gallego2012}, for example. It seems that our general approach could provide alternative  definitions of other notions in terms of composability as well.  

In the tripartite case, the relevant structure in the union 
$\mathrm{conv}[\bigcup_{I'\in \mathcal{I}^{\max}}\mathbf{PD}(S,M^{I'})]$ is indeed precisely that of biseparability, as the polytopes $\mathbf{PD}(S,M^{I'})$ are distinct from the Bell-polytope if and only if $I'$ is a singleton, in which case the set of parties $I$ is split into two subsets $I'$ and $I'^{\perp}$. The polytope $\mathbf{PD}(S,M^{I'})$ for $I' = \{i\}$ for some $i\in I$ thus witnesses whether the behaviour admits a biseparable model with respect to a \emph{specific} partition $I',I'^{\perp}$. Similar statements hold on the level of quantum states, for example, the set $\mathbf{WSEP}(\mathcal{H}_{\vec{A}},\mathcal{I}^{\max})$ of weakly $\mathcal{I}$-separable states in the three-party scenario can be identified with the set of biseparable quantum states, the set $\mathbf{SEP}(\mathcal{H}_{\vec{A}},\mathcal{I})$ consists of the states which are separable with respect to some $I'\in \mathcal{I}$, while the set $\mathbf{SEP}(\mathcal{H}_{A},I')$ is seen to match the set of states which are biseparable with respect to the specific partition $I', I'^{\perp}$; all in the sense of definitions found for example, in Ref.~\cite{Seevinck2008}. 

We emphasize, however, that the connection between  the sets $\mathbf{SEP}(\mathcal{H}_{\vec{A}}, I'),$ $ \mathbf{SEP}(\mathcal{H}_{\vec{A}}, \mathcal{I}^{\max}),$  $\mathbf{WSEP}(\mathcal{H}_{\vec{A}}, \mathcal{I}^{\max})$ and their associated device-independent witnesses phrased in terms of $I'$-separability of Definition \ref{Definition:SeparablewithrespectToAsubsetofPartiesQstate} to the notion of biseparability \emph{does not} extend\footnote{It does seem though, that some kind of relation to ``$k$ -separability'', which demands that a partition into $k$-subsets of parties is admissible \cite{Seevinck2008},  for $k= |I|-1$ can perhaps be identified. This is because an $I'$-separable quantum state is automatically $|I'|+1$-separable with respect to the partition $\{i'\}_{i'\in I'}, I'^{\perp}$, while the polytope $\mathbf{PD}(S,M^{I'})$ is distinct from the Bell polytope only if $|I|-|I'|\geq 2$. We leave a more thorough investigation  of this issue for future work. } to $|I|>3$-partite states or scenarios. This is because the notion of biseparability for a quantum state allows for states of the form, for $I'\subset I$ with $|I'|, |I \setminus I'|\geq 2$
\begin{align}
   \rho = \sum_{k}\eta_k \rho_k^{A_{I'}}\otimes \rho_k^{A_{I\setminus I'}}\label{Equation:notpartiallydeterministicQstate}
\end{align}
to be included in the sum. When considering separability with  respect to a particular partition with those properties, it is easy to see, following arguments similar to proof of Theorem \ref{Theorem;I'separableQstateisinBstimesQs} that any behaviour $\wp^{\rho}$ obtainable from a state of the form in Eq.~\eqref{Equation:notpartiallydeterministicQstate} would satisfy
\begin{align}
    \wp^{\rho} \in \mathrm{conv}[\mathbf{Q}(S^{I'}) \odot \mathbf{Q}(S^{I'^{\perp}}) ],
\end{align}
which is not a partially deterministic quantum set of the type in Definition \ref{Theorem;I'separableQstateisinBstimesQs} for any $I' \subset I$. Similarly, by relaxing the quantum sets above to the no-signalling sets as was done in the case of $I'$-inseparability witnesses in terms of partially deterministic polytopes $\mathbf{PD}(S,M^{I'})$, one arrives at 
\begin{align}
    \wp^{\rho}\in \mathrm{conv}[\mathbf{NS}(S^{I'})\odot \mathbf{NS}(S^{I'^{\perp}})],
\end{align}
which is not a partially deterministic polytope for any $I'$. Intersections, unions and relationship of objects biseparable with respect to different partitions for $|I|>3$ partite quantum states and scenarios is therefore not fully contained in the study of partial determinism in the sense explored here. Nonetheless,   we have seen that certain classes of partially deterministic polytopes, or more precisely, sets that can be derived from them by suitable set-theoretic operations, have been  studied under different names in the literature before, and that the study of partial determinism is useful for the study of certain kinds of multipartite behaviour. In the next section, we will also demonstrate that the the polytopes $\mathbf{PD}(S,M^{I'})$ are quite generally the objects of interest in a class of physically relevant experimental scenarios.

Finally, while we have mostly focused on the role of partially deterministic polytopes as device-independent witnesses of quantum state inseparability properties, Theorem \ref{Theorem;I'separableQstateisinBstimesQs} suggests that an alternative approach to derive $I'$-inseparability witnesses is to bound the set partially deterministic quantum set $\mathrm{conv}[\mathbf{B}(S^{I'})\odot \mathcal{\mathbf{Q}(S^{I'})}]$ directly instead. Evidently the bounds for this set are easier to violate and therefore, presumably, allow witnessing inseparability properties for larger subsets of quantum states in a device-independent manner as opposed to the polytopes $\mathbf{PD}(S,M^{I'})$.   Whether the composability-structure which we have identified in this work could provide useful insight to bound these types of sets is, in our view, worth future investigation.

We summarise the forms of inseparability discussed in this section and the related inseparability witnesses related to appropriate partially deterministic polytopes in Table \ref{table:InseparabilitywitnessesTable}.

\tikzset{myarrow/.style={-stealth,shorten >=\spc, shorten <=2.5*\spc}}

\begin{table*}[!ht]
\centering
\begin{tikzpicture}
    \node[inner sep=\spc] (t)
        {
         \centering 
         $\rotatebox[origin=c]{90}{Implies} \left\uparrow \hspace{0.2cm} 
    \begin{tabular}{|c|c|c|}\hline
      Property of behaviour   & Tighter device-independent property& Property of quantum state \\\hline 
       $\wp^{\rho} \notin \mathbf{PD}(S,M^{I'})$ & $\wp^{\rho} \notin \mathrm{conv}[\mathbf{B}(S^{I'}) \odot \mathbf{Q}^{I'^{\perp}}]$ & $\rho \notin\mathbf{SEP}(\mathcal{H}_{\vec{A}}, I')$ \\\hline
         $\wp^{\rho} \notin \bigcup\limits_{I' \in \mathcal{I}}\mathbf{PD}(S,M^{I'})$    & $\wp^{\rho} \notin \bigcup\limits_{I' \in \mathcal{I}}\mathrm{conv}[\mathbf{B}(S^{I'}) \odot \mathbf{Q}^{I'^{\perp}}]$ & $ \rho \notin \mathbf{SEP}(\mathcal{H}_{\vec{A}}, \mathcal{I})$ \\\hline
         $\wp^{\rho} \notin \mathrm{conv}[\bigcup\limits_{I' \in \mathcal{I}}\mathbf{PD}(S,M^{I'})]$ & $\wp^{\rho} \notin \mathrm{conv}[\bigcup\limits_{I' \in \mathcal{I}}(\mathrm{conv}[\mathbf{B}(S^{I'}) \odot \mathbf{Q}^{I'^{\perp}}])]$ & $ \rho \notin\mathbf{WSEP}(\mathcal{H}_{\vec{A}}, \mathcal{I})$ \\\hline
    \end{tabular}
    \right.$
        };
    
    \draw[myarrow] (t.north west) -- (t.north east) node[midway,above] {Implies};
\end{tikzpicture}
  \caption{A tabular representation of device-independent witnesses for the forms of $|I|-partite$ quantum state inseparability of Definitions \ref{Definition:SeparablewithrespectToAsubsetofPartiesQstate},  \ref{Definition;StronginseparabilityWRTasubsetofsets} and \ref{Definition:GenuineInseparability}. The objects in the left-most column are nontrivial, i.e. can be used to witness  inseparability of at least some quantum states, whenever $I' \neq \emptyset$ for all $I'\in \mathcal{I} \subset 2^{I}$.  Furthermore, neither of the two bottom-most objects in the left-most column are equal to the Bell or no-signalling polytopes whenever $|I| \geq 3,$ \ and $|I|-|I'|\geq 2$ for at least two subsets $ \mathcal{I} \ni I' \subset I$ when $\mathcal{I} \subset 2^{I}$ is a collection of subsets of the parties for which none equal the empty set. Theorem \ref{Theorem:unionof3partypdp'SiSNOTCONVEX} shows that in such situations one may expect a distinction between all the sets in the left-most column. The inseparability properties are preserved moving up and to the right in the table, so that the bottom-leftmost cell provides witnesses which imply all other forms of inseparability. 
  }
   \label{table:InseparabilitywitnessesTable}
\end{table*}

\subsection{broadcast-locality\label{Section:BroadcastLocality}}

The works of Werner \cite{Werner1989}, which considered projective measurements,  and Barret \cite{Barret2002}, which extended the result to include generic POVM's, established that there are bipartite entangled states which cannot be used to violate any Bell inequality in a standard Bell scenario by explicitly showing that in such cases an LHV model can be constructed which reproduces all quantum behaviours. If the parties can utilize more sophisticated methods, such as joint measurements on multiple copies of the original state \cite{peres1996} or sequential measurements essentially amounting to preprocessing of the joint mixed state \cite{Popescu1995,Gisin1996}, a state which in the original situation could only produce LHV-modelable behaviour may be utilized for a violation of a Bell inequality. These scenarios  broadly speaking provide examples of 'nonlocality activation' where an ``LHV-modelable'' entangled state demonstrates its nonclassical properties when subject to an alternative scheme. For a snapshot of recent works which explore the questions of when and how nonlocality can be activated, see for example Refs.~\cite{Aditi2005, Masanes2006, Navascues2011, Cavalcanti2011, Palazuelos2012, Cavalcanti2013, Hirsch2013, Gallego2014, Tendick2020, Bowles2021, Boghiu2023, Villegas-Aguilar2024}.

Recently \cite{Bowles2021, Boghiu2023}, a  proposal based on broadcasting was put forward  which succeeds at the task of activation without the need for multiple copies. Due to the details of the protocol however \cite{Bowles2021, Boghiu2023}, as we shall also discuss, a violation of a facet Bell inequality is not in general sufficient to demonstrate activation, instead Bell-type 'broadcasting-inequalities' were derived, based on physical assumptions which are more appropriate for the situation, the violation of which suffices at the task.  We will demonstrate that the sets of broadcast-local correlations introduced in those scenarios can in fact be identified  precisely as the partially deterministic polytopes $\mathbf{PD}(S,M^{I'})$ investigated at length in Section \ref{Section:device-independentinseparabilitywitnesses} in a class of broadcasting scenarios, namely ones that resemble the case (c) in Fig.~1 of Ref.~\cite{Boghiu2023}, which we term 'basic broadcasting scenarios'. We also show, that in a wide class of more complicated scenarios which include multiple broadcasters, the relevant sets are composable.

A basic broadcasting scenario consists of a set of space-like separated agents each with a set of inputs $M_i, i\in I$. A subset $I_\mathcal{L} \subset I$ receive a system directly from a source, while on the other side a system is sent through a broadcasting machine. The broadcasting machine outputs $I_\mathbf{B} = I\setminus I_\mathcal{L}$-systems which are distributed between the remaining parties $i\in I_\mathbf{B}$. A broadcasting scenario $S_{\mathbf{B}}$ of this type can then be understood to be completely specified by the sets $S_{\mathbf{B}} = (I,M,O,I_{\mathcal{L}}) = (S,I_{\mathcal{L}})$. An example of a basic broadcasting scenario is provided in Fig.~\ref{fig:BroadcastingScenarioExample}.

\begin{figure}
    \centering
    \includegraphics[width=\columnwidth]{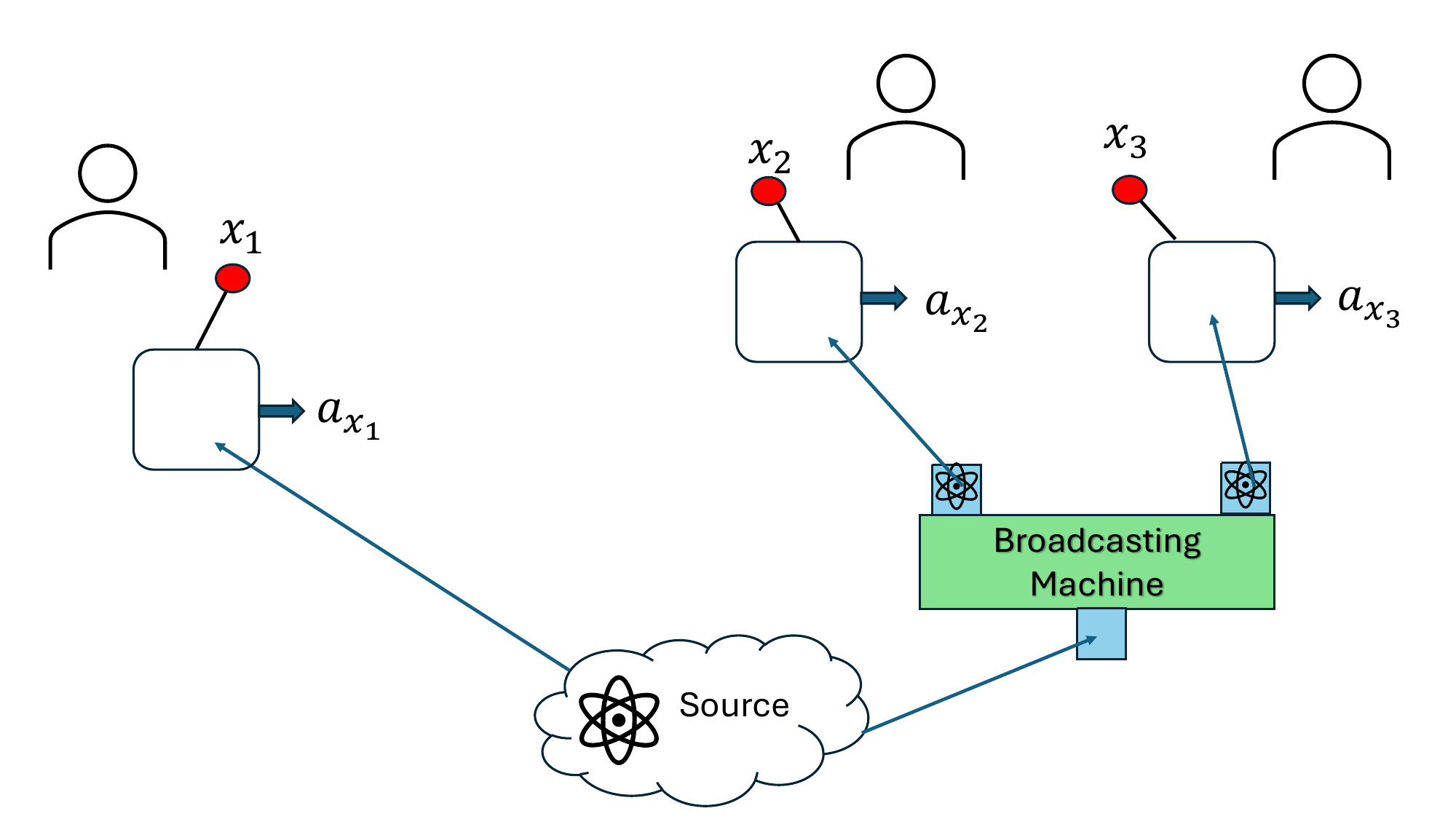}
    \caption{An example of a basic broadcasting scenario $S_\mathbf{B}$ with three parties of the kind investigated at length in Refs.~\cite{Bowles2021,Boghiu2023}.  Here the first agent $i\in I_{\mathcal{L}} = \{1\}$ receives the system directly from the source, while the remaining agents in the set $I_{\mathbf{B}} = \{2,3\}$ receive their systems from a broadcaster. }
    \label{fig:BroadcastingScenarioExample}
\end{figure}

Now suppose the source were such that without the presence of the broadcaster the $|I_{\mathcal{L}}|+1$ agents receiving a system directly would not be violate any Bell inequality, as is the case in nonlocality activation. Clearly a violation of a Bell inequality in the $|I| = |I_{\mathcal{L}}|+|I_{\mathbf{B}}|$ partite scenario $S_{\mathbf{B}}$ with the broadcaster \emph{would not} constitute evidence that some kind of nonclassicality originating from the source was activated in the experiment, as no constraints were placed on the functioning of the broadcaster.  For example,  it could as well  be the case, that the broadcaster simply discards the original system  it receives and produces an $|I_{\mathbf{B}}|$-partite entangled state, perhaps conditioned on the value of the classical variable, which can be used to violate a Bell inequality. A proper witness of nonlocality activation would have to take into account possible nonclassical correlations originating from the broadcasting machine itself.

If it is assumed, that the only constraint on the broadcaster is that correlations coming out of it satisfy no-signalling, one arrives, as in Refs.~\cite{Bowles2021,Boghiu2023} to the notion of \emph{broadcast-locality}.

\definition{(Broadcast-locality)\label{Definition:BroadcastLocality}\\
Let $S_{\mathbf{B}} = (S,I_{\mathbf{B}})$ be a basic broadcasting scenario. The set of broadcast-local correlations $\mathbf{BL}(S_B)$ in the scenario $S_\mathbf{B}$ consists of those behaviour $\wp$ which decompose as 
\begin{align}
    \wp(\vec{a}|\vec{x}) = \int p(\lambda)  \left[\prod_{i\in I_\mathcal{L}} P(a_i|x_i, \lambda)\right] \cdot P^{\mathbf{NS}}(\vec{a}_{I_\mathbf{B}}|\vec{x}_{I_{\mathbf{B}}}, \lambda) d\lambda \label{Equation:BroadcastLocalityDefiningeq}
\end{align}
corresponding to the hypothesis that the source is classical while the broadcaster may produce arbitrary no-signalling correlations. 
}\\

\normalfont

Equation \eqref{Equation:BroadcastLocalityDefiningeq} expresses that the set $\mathbf{BL}(S_{B})$ consists of partially factorizable behaviours $\mathbf{PF}(S,M^{I_{\mathcal{L}}})$ defined by the collection $M^{I_{\mathcal{L}}}_i = M_i$ if $i\in I_{\mathcal{L}}$ and $M^{I_\mathcal{L}}_i = \emptyset$ otherwise. By Theorem \ref{Theorem:PartiallyFactorizableEqualsPartiallyDeterministicPolytope}, $\mathbf{BL}(S_B) = \mathbf{PD}(S,M^{I_{\mathcal{L}}})$ and therefore all the results of section \ref{Section:device-independentinseparabilitywitnesses} apply, and vice versa. We report this result as a theorem below, for ease of reference.

\theorem{\label{Theorem:broadcastLocalCorrelationsEqualPDP}Let $S_{\mathbf{B}}=(S,I_{\mathbf{B}})$ be a basic broadcasting scenario. The set of broadcast-local correlations $\mathbf{BL}(S_{B})$ is precisely the partially deterministic polytope $\mathbf{PD}(S,M^{I_{\mathcal{L}}})$.}
\begin{proof}
    Immediate.
\end{proof}
\normalfont

 Theorem \ref{Theorem:broadcastLocalCorrelationsEqualPDP} is an example of potentially important mathematical correspondences of structures derived in fields which may be a-priori independent. Namely now, any inequalities valid for the broadcast-local set, such as those considered in Refs.~\cite{Bowles2021,Boghiu2023}, are also valid in any other  situation where a suitable partially deterministic polytope $\mathbf{PD}(S,M^{I'})$ arises. Note that the word 'suitable' here,  refers to any nontrivial polytope in the same equivalent class, which by Theorems \ref{Theorem:nonBellPDP'sEqualIFF} and \ref{Theorem:EquivalenceClassesofPDP's} contains more than one element unless $I' = \emptyset$.  

For the purpose of demonstrating activation in a basic broadcasting scenario, one needs an LHV-modelable entangled state in the source, which after broadcasting  violates  an inequality bounding the appropriate polytope $\mathbf{PD}(S,M^{I_{\mathcal{L}}})$. In Ref.~\cite{Bowles2021} a facet inequality for a tripartite scenario of the type illustrated in Fig.~\ref{fig:BroadcastingScenarioExample} with three two-outcome inputs at the first site and two two-outcome sites at the broadcasted sites was derived.    An experiment in which nonlocality activation for all PVM's was demonstrated by means of a violation of this facet inequality in that scenario was recently reported in Ref.~\cite{Villegas-Aguilar2024}.

Note that for the task of  nonlocality activation information about the initial state at the source is required and thus this task  \emph{can not} be done purely device-independently.   Nonetheless, it is a question of significant foundational importance, closely related to the issue of how entanglement relates  to theory independent notions of nonclassicality.   A number of open questions still remain in the field. for example, it is in general a hard problem to decide which entangled quantum states can only provide LHV-modelable behaviour for all possible measurements. For a review in this topic, see for example Ref.~\cite{Augusiak_2014}  and for some other recent works see e.g.~Refs.~\cite{Hirsch2017betterlocalhidden, Designolle2023}.

Aside from activation, the broadcasting scenario can also be used to witness entanglement in the source in a device independent way \cite{Bowles2021, Boghiu2023}. Indeed, by construction, a violation of the appropriate $\mathbf{PD}(S,M^{I_\mathcal{L}})$ can not be solely due to the broadcaster producing non-classical correlations and hence a violation constitutes evidence of entanglement in the source. Here, no knowledge of the inner workings of either the source nor the broadcaster is required. Still, the fact that activation is possible in the broadcasting scenario shows that by considering broadcasting machines distinct, potentially useful ways to certify entanglement which could not be certified without the inclusion of the broadcaster become possible.

For this discussion the action of the broadcasting machine needed not be specified. Assuming a quantum model for the experiment, however, the action of the broadcasting machine will of course have to obey quantum-dynamical rules and hence in particular also limitations such as  the no-broadcasting theorem \cite{Barnum1996}. Evidently quantum mechanics does not allow for arbitrary no-signalling correlations to be produced by the broadcasting machine, and hence the broadcast-local set of Definition \ref{Definition:BroadcastLocality} can be further constrained. We believe it is useful to provide a quantum mechanical description of the broadcasting protocol to support this discussion as an illustration. 

Let us then suppose that the basic broadcasting experiment admits a quantum model. If one furthermore insists that the source is classical, i.e. produces $|I_{\mathcal{L}}|+1$-partite separable states of the form 
\begin{align}
 \rho^S = \sum_{{k}\in K} \eta_k \bigotimes_{i \in  I_{\mathcal{L}}}\rho_{k}^{A_{i}}\otimes \rho^{A_{b}}_{k} =: \sum_{k}\eta_k \rho^{\vec{A}_{{I}_{\mathcal{L}}}}_k \otimes \rho^{b}_{k},
\end{align} where $\rho^{\vec{A}_{{I}_{\mathcal{L}}}}_k  \in S(\mathcal{H}_{\vec{A}})= \bigotimes_{i\in I_{\mathcal{L}}} S(\mathcal{H}_{A_i})$, $\rho^{b}_k \in S(\mathcal{H}_{b})$,  $\eta_k \geq 0$ for each $k\in K$ and $ \sum_{k\in K}\eta_k = 1$. Then, following a broadcasting map $B:S(\mathcal{H}_{b})\rightarrow S(\mathcal{H}_{\vec{A}_{I_{\mathbf{B}}}})$ acting on the system $b$, with $\mathcal{H}_{\vec{A}_{I_{\mathbf{B}}}} = \bigotimes_{i\in I_{\mathbf{B}}} \mathcal{H}_{A_i}$ one gets an $|I_{\mathcal{L}}|+|I_{\mathbf{B}}|$-partite state $\rho \in S(\mathcal{H}_{\vec{A}})$  which is partially separable with respect to $I_\mathcal{L}$:
\begin{align}
 \rho &=  (\mathrm{id}_{S(\mathcal{H}_{\vec{A}_{I_\mathcal{L}}})}) \otimes B)  (\sum_{k\in K}  \eta_k \rho_{k}^{\vec{A}}\otimes \rho^{A_{b}}_{k}) \\
 &= \sum_{k\in K}\eta_{k} \rho_k^{\vec{A}_{I_{\mathcal{L}}}} \otimes B(\rho^{b}_k) \in \mathbf{SEP}(\mathcal{H_{\vec{A}}, I_{\mathcal{L}}}).
\end{align} 
Therefore if $\rho \notin \mathbf{SEP}(\mathcal{H_{\vec{A}}, I_{\mathcal{L}}})$, it is necessarily also the case that $\rho^S \notin \mathbf{SEP}(\mathcal{H}_{\vec{A}_{\mathcal{L}}}\otimes \mathcal{H}_b)$ which is to say that to witness non-separability of the state $\rho^{S}$ produced by the source one may witness $I_{\mathcal{L}}$-inseparability of the state $\rho$ after broadcasting.   Clearly, if one holds on to the black-box description of the experiment,  the set of broadcast-local correlations of Definition \ref{Definition:BroadcastLocality} can be further constrained to the \emph{quantum broadcast-local} set \cite{Bowles2021,Boghiu2023}, corresponding to the situation where it is assumed that the broadcaster operates within the constraints of quantum theory.

\definition{(Quantum broadcast-locality)\label{Definition:quantumBroadcastlocality}\\
Let $S_{\mathbf{B}} = (S,I_{\mathbf{B}})$ be a basic broadcasting scenario. The set of quantum broadcast-local correlations $\mathbf{QBL}(S_{\mathbf{B}})$ consists of the behaviour $\wp^{\rho}$ which admit a model of the form
\begin{align}
    \wp^{\rho} = \mathrm{tr}[(\rho_{I_{\mathcal{L}}} \bigotimes_{i\in I}M_{a_i|x_i} ],
\end{align}
where $\rho_{I_{\mathcal{L}}}\in \mathbf{SEP}(\mathcal{H_{\vec{A}}, I_{\mathcal{L}}})$ is a state partially separable with respect to $I_{\mathcal{L}}$.  
} \\
\normalfont

The set $\mathbf{QBL}(S_{\mathbf{B}})$ can be immediately identified as the set of partially factorizable quantum behaviour of Theorem \ref{Theorem;I'separableQstateisinBstimesQs}.

\theorem{\label{Theorem:QBLisBprodQ}Let $S_{B} = (S,I_{\mathbf{B}})$ be a basic broadcasting scenario. Then the set of quantum broadcast-local correlations $\mathbf{QBL}(S_{\mathcal{L}})$ is precisely the set of partially factorizable quantum behaviour $\mathrm{conv}[\mathbf{B}(S^{I_{\mathcal{L}}})\odot \mathbf{Q}(S^{I_{\mathbf{B}}})]. $}
\begin{proof}
    Immediate. 
\end{proof}
\normalfont

A similar hierarchy as in Table \ref{table:InseparabilitywitnessesTable} of device independent witnesses of entanglement in the source can therefore be identified in the broadcasting scenario. Conceptually, we may also emphasize that a distinction between 'device-independence' and 'theory independence' can be made \cite{Villegas-Aguilar2024}. Namely, in these type of experiments, the partially deterministic polytope $\mathbf{PD}(S,M^{I_{\mathcal{L}}})$ provides non-classicality witnesses which do not presuppose that the broadcasting device operates under quantum mechanical laws, contrary to the strictly smaller set of partially factorizable quantum behaviour $\mathrm{conv}[\mathbf{B}(S^{I_{\mathcal{L}}})\odot \mathbf{Q}(S^{I_{\mathbf{B}}})]$. 

   We will next demonstrate that, analogously to the separability criterion discussed at the end of Section \ref{Section:device-independentinseparabilitywitnesses},  broader notions of composable sets play an important role in scenarios which contain multiple broadcasters as well. Such scenarios were also introduced and explored in Refs.~\cite{Bowles2021, Boghiu2023}. For the sake of completeness, we present a formal definition  of general broadcasting scenarios below. 

Suppose then that the source $S$ produces $|I^S| = |I_{\mathcal{L}}| + |I_{\vec{b}}|$ systems.  The set $I_\mathcal{L}$ denotes, as before, the parties $i\in I$ which receive a system directly from the source. The index set  $I_{\vec{b}}$ denotes the systems that are going to be broadcast in individual channels, each producing a disjoint set $I_{B_{b}}, b\in I_{\vec{b}}$ of systems, which are distributed to space-like separated parties. The resulting $|I|$-partite system is then probed in the scenario $S=I,M,O$ where $I = I_{\mathcal{L}}\cup(\bigcup_{b\in I_{\vec{b}}}{I_{B_b}})$  is a union of disjoint sets. A general broadcasting scenario $S_{\mathcal{GB}}$ of this type is then fully characterized by the scenario $S$ and the collection $\{I_{B_b} \}_{b\in I_{\vec{b}}}$ so that $S_{\mathcal{GB}} = (S,\{I_{B_{b}} \}_{b\in I_{\vec{b}}})$. A basic scenario corresponds to a general scenario with $|I_{\vec{b}}|=1$ and therefore, for the sake of removing redundancy in the following discussion, it is assumed that $|I_{\vec{b}}|\geq 2$.  An example of a general broadcasting scenario which does not fit into the definition of a basic scenario is provided in Fig.~\ref{fig:GeneralBroadcastingFigure}.

\begin{figure}
    \centering
    \includegraphics[width=\linewidth]{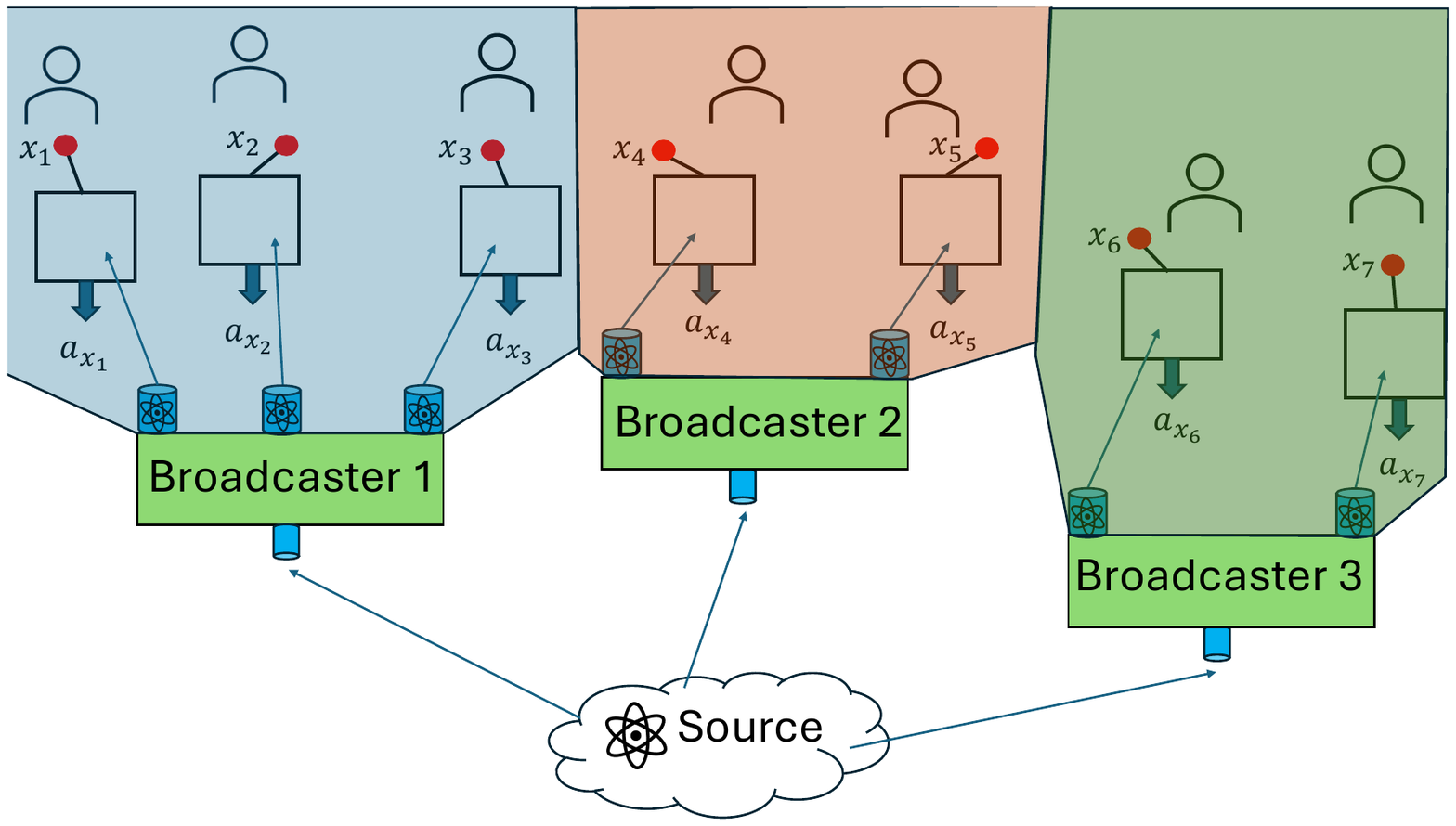}
    \caption{A general broadcasting scenario $S_{\mathcal{GB}} = (S, \{I_{B_b}\}_{b\in I_{\vec{b}}}$) with a total of seven agents performing the experiment. Here $I_{\vec{b}} = \{1,2,3\}$ is a set which indexes the three broadcasters, and the sets $I_{B_b} \subset I$ denote the subsets of the seven parties which receive a system from a given broadcaster. The subsets $I_{B_b}$ of parties have been colour coded for clarity, with $I_{B_1} = \{1,2,3\}, I_{B_2}=\{4,5\}$ and $I_{B_3}= \{6,7\}$.  In this scenario $I_{\mathcal{L}}=\emptyset$, so that no one receives a system directly from the original 3-partite source.  }
    \label{fig:GeneralBroadcastingFigure}
\end{figure}

Definition \ref{Definition:BroadcastLocality} is easily extended to these more general scenarios.  Namely,  the set of broadcast-local correlations $\mathbf{BL}(S_{\mathcal{GB}})$ in a general scenario $S_{\mathcal{GB}}$ is again defined as the set of correlations which are compatible with the assumption of a classical source emitting system to broadcast channels which can create arbitrary no-signalling correlations.  Following analogous arguments which lead to Definition \ref{Definition:BroadcastLocality} of broadcast-locality and Theorem \ref{Theorem:broadcastLocalCorrelationsEqualPDP} which identified the set of broadcast-local correlations constrained only by no-signalling in the basic scenario $S_{\mathbf{B}}$ as the partially deterministic polytope $\mathbf{PD}(S,M^{I_{\mathcal{L}}})$, it is easy to see that  the following proposition holds for the set $\mathbf{BL}(S_{\mathcal{GB}})$.

\proposition{\label{Proposition:GeneralBroadcastingScenarioblIscomposablesET}Let $S_{\mathcal{GB}} = (S,\{I_{B_b}\}_{b\in I_{\vec{b}}})$ be a general broadcasting scenario and let   $\mathbf{BL}(S_{\mathcal{GB}})$ denote the set of broadcast-local correlations  in $S_{\mathcal{GB}}$. Then  
\begin{align}
    \mathbf{BL}(S_{\mathcal{GB}}) = \mathrm{conv}[\mathbf{B}(S^{I_{\mathcal{L}}}) \odot (\bigodot_{b\in I_{\vec{b}}} \mathbf{NS}(S^{I_{B_b}}))].
\end{align}

}
\begin{proof}
    Omitted, follows from the definition of broadcast-locality and the definition of composability. 
\end{proof}
\normalfont

As an example, the broadcast-local set $\mathbf{BL}(S_{\mathcal{GB}})$ in the type of scenario depicted in Fig.~\ref{fig:GeneralBroadcastingFigure} would be composed as the product of three no-signalling sets $\mathbf{NS}(S^{I_{B_b}})$, as 

\begin{align}
    \mathbf{BL}(S_{\mathcal{SG}}) = \mathrm{conv}[\mathbf{NS}(S^{I_{B_1}}) \odot \mathbf{NS}(S^{I_{B_2}}) \odot \mathbf{NS}(S^{I_{B_3}})]. \label{Equation:GeneralBroadcastingFigureDecomposed}
\end{align}

The set in Eq.~\eqref{Equation:GeneralBroadcastingFigureDecomposed} is contained in the no-signalling set $\mathbf{NS}(S)$. It is easy to see that it is in fact strictly contained in it. This can be established, for example, by considering the image of a restriction $R_{|V}$ with $V = M^{\{j,k\}}$ where $j,k$ are labels of parties belonging to different subsets $I_{B_b}$. Namely, from Lemmas  \ref{Lemma:convexityofRestrictionmap}, \ref{PropositionNSisBell} and Proposition \ref{Proposition:ImageRestrictionofBehproductistheDomain} it then follows that 
\begin{align}
    R_{|V}(\mathbf{BL}(S_{\mathcal{GB}})) &= \mathrm{conv}[\bigodot_{b\in \{1,2,3\}} \label{Equation:ImageofBLsGBisBellLine1}\mathbf{NS}(S_{|M^{I_{B_b}}\cap V)})] \\
    &=  \mathrm{conv}[\mathbf{NS}(S^{\{j\}}) \odot \mathbf{NS}(S^{\{k\}})]\label{Equation:ImageofBLsGBisBellLine2} \\
    &= \mathbf{B}(S^{\{j,k\}}), \label{Equation:ImageofBLsGBisBellLine3}
\end{align}
whereas by Proposition \ref{Proposition:restrictionmapidentities} $R_{|V}(\mathbf{NS}(S)) = \mathbf{NS}(S^{\{j,k\}})\neq \mathbf{B}(S^{\{j,k\}})$.  On the other hand, the set $\mathbf{BL}(S_{\mathcal{SG}})$ in Eq.~\eqref{Equation:GeneralBroadcastingFigureDecomposed}) is not any partially deterministic polytope $\mathbf{PD}(S,M^{I'})$. This suggests that in order to investigate properties of broadcast-local sets in full generality a broader study of composable sets is in order.

If in a general broadcasting scenario $S_{\mathcal{GB}}$  the experiment is assumed to be  quantum modelable, then  the set $\mathbf{BL}(S_{\mathcal{GB}})$ is composable similarly to the case of Proposition \ref{Proposition:GeneralBroadcastingScenarioblIscomposablesET} with the no-signalling sets replaced by appropriate quantum sets instead.

\proposition{\label{Proposition:GeneralbroadcastscenarioIfquantumthenComposable} Let $S_{\mathcal{G}}$ be a general broadcasting scenario and let us denote by $\mathbf{QBL}(S_{\mathcal{GB}})$ the set of broadcast-local quantum correlations in that scenario. Then 
\begin{align}
    \mathbf{QBL}(S_{\mathcal{GB}}) = \mathrm{conv}[\mathbf{B}(S^{I_{\mathcal{L}}})\odot (\bigodot_{b\in I_{\vec{b}}} \mathbf{Q}(S^{I_{B_b}})].
\end{align}
}
\begin{proof}
    Omitted, follows the argument leading to Definition \ref{Definition:quantumBroadcastlocality} when generalized to $S_{\mathcal{GB}}$.
\end{proof}
\normalfont

The  set $\mathbf{QBL}(S_{\mathcal{GB}})$ is clearly strictly contained in the set $\mathbf{BL}(S_{\mathcal{GB}})$. A few similarities between these sets can be easily extracted from the property of composability. For example, analogously to the the three-partite example of Eqs.~\eqref{Equation:GeneralBroadcastingFigureDecomposed}-\eqref{Equation:ImageofBLsGBisBellLine3} it is seen that for a restriction onto any $V = \{j,k\}$, where $j,k$ belong to different subsets $I_{B_b}$, it is also the case that 

\begin{align}
    R_{|V}(\mathbf{QBL}(S_{\mathcal{GB}})) = \mathbf{B}(S^{j,k}).
\end{align}
Therefore, one may expect families of similar nontrivial constraints to be valid for both objects. Here the expression 'nontrivial' is used in the sense of the constraints not being necessarily satisfied by $\mathbf{NS}(S)$ or $\mathbf{Q}(S)$.

Note that in any general scenario $S_{\mathcal{GB}}$ the set of broadcast-local correlations $\mathbf{BL}(S_{\mathcal{GB}})$ is  always  a genuine subset of the partially deterministic polytope $\mathbf{PD}(S,M^{I_{\mathcal{L}}})$ as 
\begin{align}
\begin{split}
   \mathrm{conv}[\mathbf{B}(S^{I_{\mathcal{L}}}) \odot( \bigodot_{b\in I_{\vec{b}}} \mathbf{NS}(S^{I_{B_b}}))] \label{equation:nonemptyI_LBL_GB}\\
   \subsetneq  \mathrm{conv}[\mathbf{B}(S^{I_{\mathcal{L}}}) \odot \mathbf{NS}(S^{I\setminus I_{\mathcal{L}}})],
   \end{split}
\end{align}
which can be established, for example,  by a technique similar to the one used in Eqs.~\eqref{Equation:ImageofBLsGBisBellLine1}-\eqref{Equation:ImageofBLsGBisBellLine3} regarding the relationship of product of no-signalling sets $\mathbf{NS}(S^{I_{B_b}})$ and the no-signalling set $\mathbf{NS}(S)$ itself. Equation \eqref{equation:nonemptyI_LBL_GB} holds for all $I_{\mathcal{L}}$, but the right hand side provides no nontrivial constraints if $I_{\mathcal{L}}=\emptyset$, as then $\mathbf{PD}(S,M^{I_{\mathcal{L}}}) = \mathbf{NS}(S)$. If, however, $I_{\mathcal{L}}\neq \emptyset$ then the the partially deterministic polytope $\mathbf{PD}(S,M^{I_{\mathcal{L}}})$ may provide some relevant testable criteria of broadcast-locality, though the composable set\footnote{This composable set is, similarly to the polytope $\mathbf{PD}(S,M^{I_{\mathcal{L}}})$, partially deterministic with respect to the set $M^{I_{\mathcal{L}}}$ and a convex polytope. It could perhaps be referred to as an example of a  'partially deterministic sub-polytope' as it does not fit Definition \ref{DefinitionPartialDeterministicNosignallingbehaviours} of a partially deterministic polytope in full generality, but is nonetheless a polytope strictly contained in it.  } expressed on the left hand side of Eq.~\eqref{equation:nonemptyI_LBL_GB} is, as we have argued, the more appropriate witness.

A more comprehensive study of properties of general broadcast-local, or composable sets is beyond the scope of this work. Let us point out, however, that an example of the usefulness of studying behaviours in the more general scenarios has already been demonstrated in Sec.~5 of Ref.~\cite{Bowles2021}, where a symmetric general broadcast scenario with two broadcasters producing two systems at each arm distributing in total to four agents each with two two-outcome inputs and no one receiving a system directly from the source was considered. There \cite{Bowles2021}, they derived a threshold for device independent certifying of entanglement in a  class of mixed bipartite states (so-called isotropic states) which exceeded the previously known bound in visibility and suggested that it may even be possible to certify entanglement across the entire region of visibility where the state is entangled. This possibility received further affirmation in Ref.~\cite{Boghiu2023}, which pushed the visibility bound even lower, but not precisely to the entanglement bound,  with the aid of numerical methods in the same scenario. It seems then, that aside from providing biseparability witnesses mentioned at the end of Sec.~\ref{Section:device-independentinseparabilitywitnesses}, understanding the structure of these types of composable sets may provide insight into fundamental questions related to entanglement and quantum information processing. 

Finally, let us acknowledge that in Ref.~\cite{Boghiu2023} the relationship between the broadcast-local polytope in the kind of  three partite basic scenario illustrated in Fig.~\ref{fig:BroadcastingScenarioExample} and the $\mathbf{NS}_2$-local polytope (which we also pointed out is equal to $\mathrm{conv}[\bigcup_{I'\in \mathcal{I}^{\max}}\mathbf{PD}(S,M^{I'})$] in Sec.~\ref{Section:device-independentinseparabilitywitnesses}) was also recognized. Furthermore, they  provided an example of an LHV-modelable entangled state which can be activated following a broadcasting operation to a genuinely not $\mathbf{NS}_2-$local state, see Sec.4. of Ref.\cite{Boghiu2023}. Their example naturally provides an alternative to our Example \ref{Example:ThreeQubitNotWeaklySeparable} of a state which is not weakly separable with respect to to the maximal collection $\mathcal{I}^{\max}$. We believe recognizing the common mathematical structures underlying  device-independent behaviour sets derived in different contexts in this manner is of importance as it facilitates crosstalk between fields and therefore minimizes, among other things, computational resources spent in characterizing these kinds of sets.

\normalfont

\subsection{\label{Section:GeneralSequentialWignerSection}Local Friendliness in general sequential Wigner's friend scenarios}

In Sections \ref{Section:device-independentinseparabilitywitnesses} and \ref{Section:BroadcastLocality} we have shown that  specific classes of partially deterministic polytopes, namely ones for which all the inputs of one or more observers are deterministic, witness forms of quantum state inseparability and bound the behaviour compatible with the assumption of broadcast-locality in a variety of scenarios. In this section we establish a result which may be considered strictly more extensive: that \emph{essentially every kind} of partially deterministic polytope is of physical relevance. We demonstrate the result by generalizing the bipartite sequential Wigner's friend experiment which was introduced in Ref.~\cite{Utreras-Alarcon2024} as means to probe a recent result in quantum foundations known as the Local Friendliness no-go theorem \cite{Bong2020} to the multipartite $|I|>2$ case, and by allowing the agents more general choices of inputs. 

First, let us briefly introduce the core elements going into the Local Friendliness no-go theorem and review some of the relevant results already present in the literature.

\definition{(Local Friendliness)\label{definition:LocalFriendliness}\\
Local Friendliness is the compound condition formed of the following two assumptions 
\begin{enumerate}[label=\roman*), ref=\ref{definition:LocalFriendliness}.\roman*] 
    \item (Absoluteness of Observed Events) \label{definition:LocalfriendlinessAOE}
    Every observed event is a real single event, not relative to anything or anyone. 
    \item (Local Agency)\label{definition:LocalFriendlinessLA}
    Any intervention is uncorrelated with events outside its future light cone. 
\end{enumerate}
}
\normalfont
The assumptions in Def.~\ref{definition:LocalFriendliness} are strictly weaker than those that lead to Bell's theorem \cite{Wiseman2017, Cavalcanti2021}. In particular, Absoluteness of Observed Events (termed 'Macroreality' in  Ref.~\cite{Wiseman2017}) is a principle which has, especially in older literure, seldom been articulated as one of the assumptions used to derive Bell's theorem as the only events observed in a Bell scenario are exactly those recorded by the agents at the end of each round the statistics which form the behaviour $\wp$ and hence place no constraints.  Local Agency, on the other hand, is a natural assumption in a Bell-type setup which incorporates ideas of relativistic causal arrow and independent interventions as causes for correlations \cite{Cavalcanti2021}. In a general Bell scenario of the kind in Fig.~\ref{fig:correlationscenarioFig} Local Friendliness merely implies that the behaviour obey no-signalling \cite{Haddara2025} which is a condition not plausibly violated in such an experiment.

The key insight of Ref.~\cite{Bong2020} was that by considering 'extended Wigner's friend scenarios' of the kind examined by Brukner \cite{Brukner2017, Brukner2018} instead of the usual Bell scenario,  the implications of Local Friendliness become testable. These so-called 'canonical Local Friendliness scenarios' \cite{Haddara2025} include a set $I$ of space-like separated agents termed ``superobservers'', some subset $I_F \subset I$ of which have a local observer, also referred to as their ``friend'', over which the superobservers are assumed to have close to complete control over. In the first stage of the protocol, the friends measure a system recording an outcome. The  superobservers may then decide to read that outcome in which case the friends record determines the outcome of the measurement for the superobserver. Alternatively, they may perform an arbitrary possibly disruptive measurement on the friend. 

Intuitively, the reason why Absoluteness of Observed Events of Def.~\ref{definition:LocalfriendlinessAOE} places nontrivial constraints in extended Wigner's friend scenarios is because it demands  that the observation of the friend ought to have a well defined value whether or not the superobserver reads it or not, which combined with Local Agency of Def.~\ref{definition:LocalfriendlinessAOE} essentially leads to a specific kind of partially deterministic model for the resulting behaviour, namely one where exactly one input of each superobserver with a friend is deterministic. The Local Friendliness no-go result may then demonstrated by showing that there are quantum behaviour outside that partially deterministic polytope \cite{Bong2020}.  Ref.~\cite{Bong2020} also provided a proof of principle experimental demonstration of the violation of Local Friendliness with qubit systems taking the role of observers.

The original work \cite{Bong2020} investigated Local Friendliness in specific bipartite scenarios and completely solved for the facet-defining inequalities of the polytope in the scenario where two space-like separated agents each with a friend have three two-outcome inputs at each site. Those Local Friendliness polytopes were recognized  to be closely related to the bipartite partially deterministic polytopes investigated by Woodhead \cite{Woodhead2014} (see the acknowledgements of Ref.~\cite{Bong2020}), which allow also for the more general situations with more than one deterministic input per site.  In Ref.~\cite{Haddara2025} a comprehensive investigation of Local Friendliness polytopes in arbitrarily generalizations of the canonical scenarios was conducted while Ref.~\cite{Utreras-Alarcon2024} introduced a novel instance of a bipartite scenario with sequential measurements by the friend, where the assumption of Local Friendliness was shown to lead the Bell-polytope owing to results concerning more general (more than one deterministic input) bipartite partially deterministic polytopes derived by Woodhead \cite{Woodhead2014}. In summary,  the physical relevance of various classes of partially deterministic polytopes has been recognized in a number of previous works which explore Local Friendliness. The contribution we  make in  this section is twofold: first, we provide a multipartite generalization of the sequential Wigner's friend scenario of Ref.~\cite{Utreras-Alarcon2024} allowing for arbitrary numbers of sequential measurements and non-deterministic inputs and second, we show that the resulting correlation sets compatible with Local Friendliness in Def.~\ref{definition:LocalFriendliness} are in one-to-one correspondence with the partially determinstic polytopes. In particular, for every partially deterministic polytope $\mathbf{PD}(S,M')$ there is a sequential scenario\footnote{While preparing this manuscript we became aware of work done by other authors \cite{walleghem2024connectingextendedwignersfriend} which also considered a multipartite generalization of the scenario in Ref.~\cite{Utreras-Alarcon2024} and identified situations where the Local Friendliness polytope equals the Bell polytope $\mathbf{B}(S)$. Our presentation of the scenario differs from theirs in that we also allow for arbitrary numbers of non-deterministic measurements which is necessary to establish the general correspondence with the partially deterministic polytopes $\mathbf{PD}(S,M')$.} where the assumptions in Def.~\ref{definition:LocalFriendliness} imply it, and vice versa. 

While in this work we mostly focus on adding to the understanding of the mathematical constraints related to the assumptions that go into Local Friendliness, let us also acknowledge that other aspects of the  no-go result \cite{Bong2020} are of course of independent interest. A number of works have explored, for example,  the possible  implications \cite{Cavalcanti2021, Cavalcanti2021Foundphys, DiBiagio2021, Yang2022, Ying2024relatingwigners}, intimately related alternative no-go results \cite{Haddara_2023, Wiseman2023thoughtfullocal} and prospects of experimental tests of Local Friendliness \cite{Ding2023, zeng2024violationslocalfriendlinessquantum}. There are also attempts to dispense of the space-like separation \cite{walleghem2024extendedwignersfriendparadoxes,walleghem2025extendedwignersfriendnogo} usually assumed in this context. These are important avenues of investigation which we are not going to discuss here, though we believe that the general  results of this work could find use for those purposes as well.

\begin{figure}
    \centering
    \includegraphics[width= \linewidth]{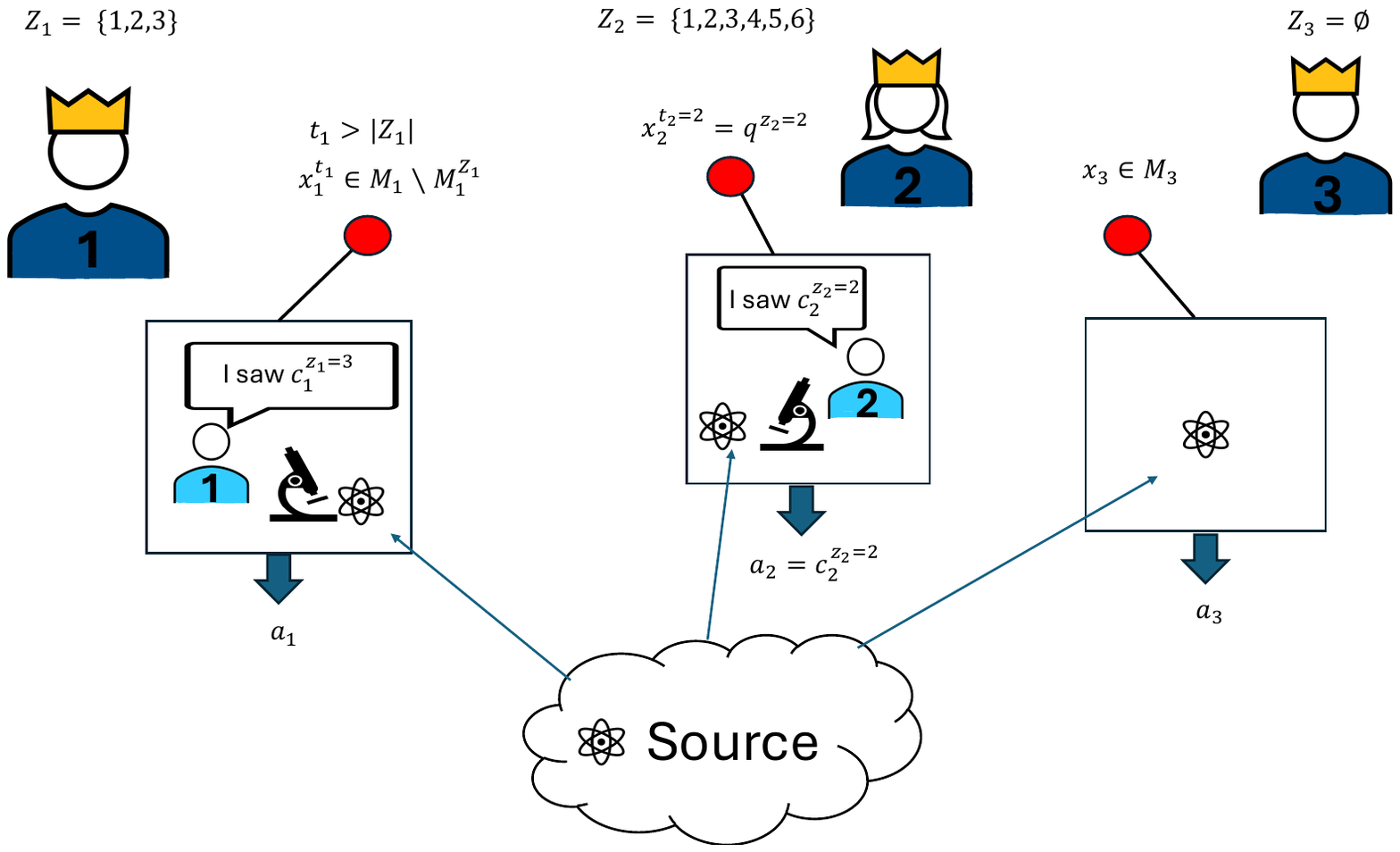}
    \caption{An illustration of a round of a sequential Wigner's friend experiment $S_{sW} = (I,M,O,Z)$, with three superobservers illustrated wearing crowns, and shirts indicating their labels. The first two superobservers have a friend. The first superobserver is allowed to query their friend on three occasions, since $|Z|=3$, but on this round, they have chosen to perform another arbitrary measurement as indicated by $x_1^{t_1}, t_1>3$ meaning that the friends  evolution got reversed twice in the protocol. The second superobserver is allowed a maximum of 6 queries, and on this round they chose to take the second chance with the setting $x_2^{t_2=2}$. This constraints the outcome $a_2$ to match the second observation $c_2^{z=2}$ made by their friend following reversal of the first record. No special constraints are imposed on the measurements of the third superobserver in any possible run of the experiment.   }
    \label{fig:SequentialWigner}
\end{figure}

\begin{figure}
    \centering
    \includegraphics[width=\linewidth]{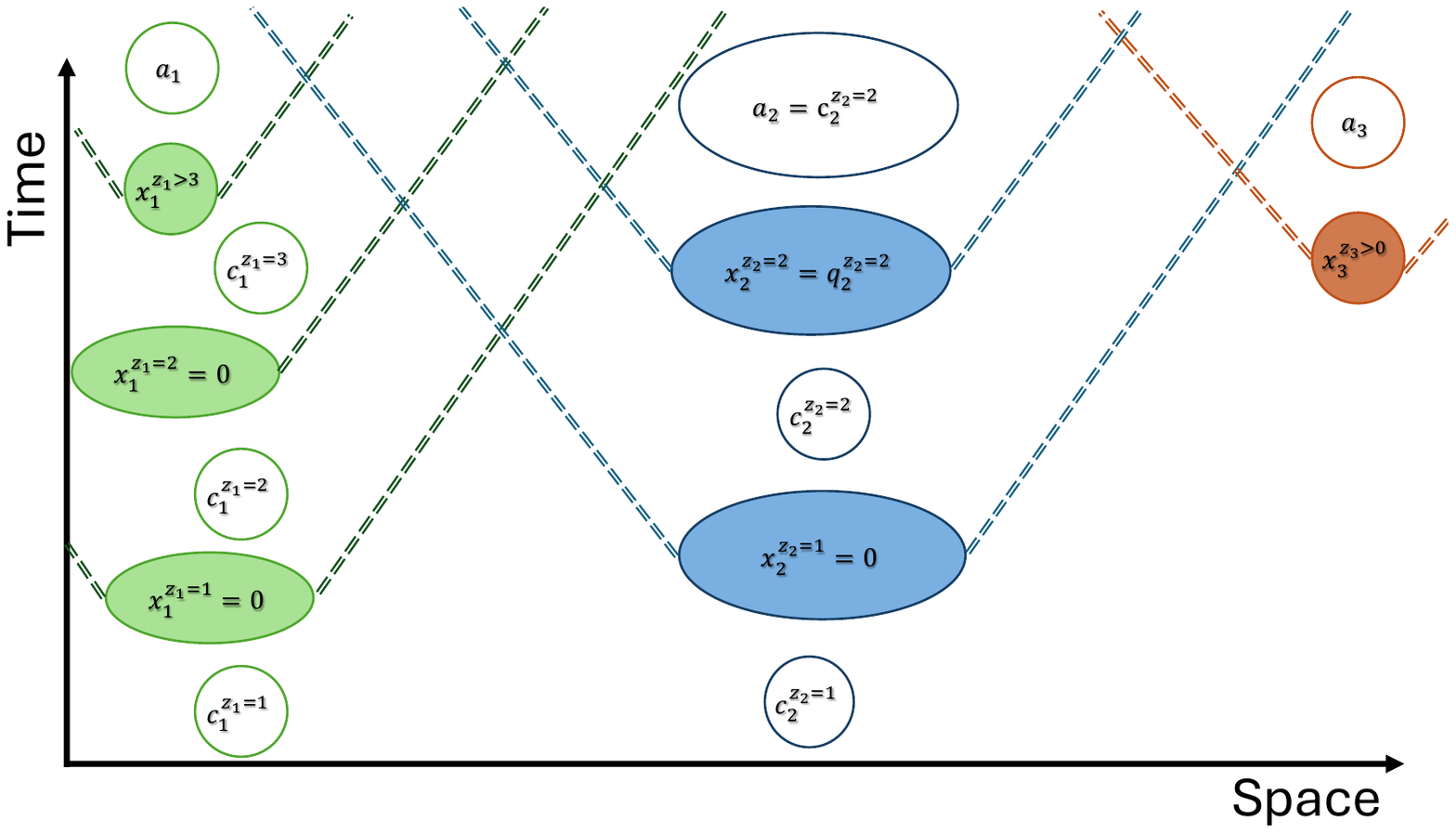}
    \caption{A spacetime diagram for the run of the tripartite sequential Wigner's friend experiment illustrated in Fig.~\ref{fig:SequentialWigner}. The friend of the first superobserver $i=1$ made a total of three observations (the maximum allowed), with the superobserver reversing the evolution of the friend following the first two, depicted with the setting $x_1^{t=1}=x_1^{t=2}=0$. Note that whether or not another reversal $x_1^{z_1=3}=0$ is made is the redundant, since the intervention $x_1^{z_1>3}$ is arbitrary, such a situation may always be  absorbed to it.  The values of $c_1^{z_1=t_1}$ are independent of interventions $x^{z_1\geq t_1}$ made in the future, and of any interventions made by other agents. The second superobserver reversed their friend once, and then measured the second observation made by the friend. The choice to query their friend ended the experiment on this round, which according to the protocol, could otherwise have continued up to 6 observations made by the friend.  The last superobserver always measures a system directly.  }
    \label{fig:SequentialWignerSpacetime}
\end{figure}

In what follows a multipartite-multi input generalization of the sequential protocol in \cite{Utreras-Alarcon2024} is described. We will first present the parameter sets of interest, followed by an account of the protocol. 

Let $I = \{1,2, \ldots ,|I| \}$ be a set of agents, or superobservers. In addition to the set $I$, let , $I_F \subset I$ denote the of the agents who have a friend in their laboratory. As before, each $i\in I$ has some set $M_i$ of inputs representing interventions they may perform on the contents of the laboratory, including on the possible friend. Each input $x_i \in M_i$ also has associated outcome sets $O_{x_i}$. In anticipation of discussing relevant behaviour, we will already drop the labels of the inputs from the corresponding outputs.  We will, however, make further demands and conventional choices on the input sets based on other parameters in the scenario.

Let us then introduce an index set $Z_i   = \{ 1 ,2,  \ldots, |Z_i|  \}$, with $Z_i \neq \emptyset$ for every $i \in I_F$ and $Z_i = \emptyset$ otherwise. Each $z_i \in Z_i$ indexes a possible measurement the friend performs during the experiment which is known to the superobserver.  The cardinality $|Z_i|$ therefore decides the maximum number of observation the friend makes during a round of the experiment\footnote{While our presentation of the sequential scenario is a direct generalization of that  in  Ref.~\cite{Utreras-Alarcon2024} we make some distinct notational conventions and changes to the representation of the scenario for the sake of added coherence with respect to other parts of this work.}.  

 For every $z_i \in Z_i$ there is a corresponding 'query measurement' $q^{z_i} \in M_i$ the agent can ask his friend, so that $|M_i| \geq |Z_i|$. Let us denote by $M_i^{Z} \subset M_i$ the subset of the inputs which are query measurements for the party $i$. Note that if $i\notin I_F$ then $M_i^{Z} = \emptyset$. For every $i\in I$ we also allow an arbitrary number of alternative measurements.  It will be useful to introduce an index $t_i: 1 \leq t_i \leq |M_i|$  for each $i\in I$ and adopt the convention where the superscript $t_i$ is added to the inputs $x_i\in M_i$, so that the input $x_i^{t_i}$ for $ 1\leq t_i \leq |Z_i|$ denotes that the query $q^{z_i = t_i}$ is performed.  For $  |Z_i| < t_i\leq |M_i|$, the reading is that $x^{t_i}_i\in M_i \setminus M^{Z_i}_i  =M_i^{{Z}^\perp}$, without the need for further specification.

The protocol for a round of the experiment goes as follows:

\begin{enumerate}
    \item A common source sends systems to each laboratory $i \in I$. If $ i\in I\cap I_F$ The friend measures the observable $z_i = 1$ as instructed by $i$ and records an outcome $c_i^{z_i =1}$. If $i\notin I_F$, the superobserver measures the system directly.
    
    \item The superobservers $i \in I_F$ have two choices. They may ask their friend for their observed outcome, corresponding to the input $x^{t_i=1}_i = q^{z_i = 1}$ , and assign their outcome $a_i = c_i^{z_i=1}$, in which case the experiment ends for them. They may alternatively choose $x_i^{t_i=1} =0$ in which case they reverse the evolution of the friend to what it was prior to observing $c_i^{z_i=1}$, and instruct them to measure the next observable in $Z_i$, in this case $z_i=2$. 

    \item If some observer $i$ reversed their friend in step 2, the experiment continues by giving them another choice; they may again ask for their record with the input $x_i^{t_i = 2}= q^{z_i =2}$ so that $a_i = c_i^{z_i = 2}$, which ends the experiment for them.  Alternatively, they may again reverse their friends evolution, by setting $x_i^{z_i=2} = 0$ and give them the next instruction. 

 \item The round of the experiment continues until all superobservers with a friend either asked their friend for their record $x_i^{t_i}=q^{z_i}$ and set $a_i = c_i^{r_i}$ or they reached their limit of $|Z_i|-1$ reversals. If the latter happened, they perform an arbitrary measurement $x_i^{t_i>|Z_i|} \in M_i$ on the friend/system pair and record the corresponding outcome $a_i$. 
\end{enumerate}

Following numerous rounds of the experiment, the statistics are gathered to estimate the behaviour $\wp$ with distributions $\wp(\vec{a}|\vec{x}^{\vec{t}})$, where the string notation has been extended to the superscript so that $\vec{x}^{\vec{t}} = x_1^{t_1}, x_2^{t_2}, \ldots x_{|I|}^{t_{|I|}}$. Note that it is not necessary to include the string of zeroes representing the choices to reverse the evolution of the friend in a given round, as that is implicit in any $x_i^{t_i}$. We will show later that emphasis for the reversal interventions will be irrelevant for the sake of specifying the model as well, owing to the assumption of Local Agency. For now, we merely stress that inclusion of the reversal therefore does not make a difference for specifying the behaviour, which is defined over the structure $S= (I,M,O)$, where as before, $O$ and $M$ represent the collections of $\{O_{x_i}\}_{i\in I, x^{t_i}_i \in M_i}$ and $\{M_i\}_{i\in I}$.  The term ``public scenario'' has also been used in the context of Wigner's friend-type scenarios \cite{Haddara2025} while referring to $S$, for the reason that the statistics available at the end of the experiment are defined over $S$. Further details about the set of friends and numbers of sequential measurements are needed  of course,  to specify the additional structure necessary to model the experiment to explore the implications of assumptions such as those in Def.~\ref{definition:LocalFriendliness}.

\definition{(Local Friendliness behaviours)\label{Definition:LocalFriendlinessBehaviours}\\
Let $S_{sW}$ be a general sequential Wigner's friend scenario. The set $\mathbf{LF}(S_{sW})$ of Local Friendliness behaviours in that scenario consists of those which are compatible with the assumptions of Absoluteness of Observed Events and Local Agency in Def.~\ref{definition:LocalFriendliness}.
}\\
\normalfont

 Note that in a sequential Wigner's friend scenario the set $I_F$ of agents who have a friend can be specified given the collection $ Z = \{Z_i \}_{i\in I}$  as $I_F = \{i\in I: Z_i \neq \emptyset \}$. We can therefore fully specify a sequential Wigner's friend scenario $S_{sW}$ of this type by the quadruple $S_{sW} = (I,M,O,Z) = (S,Z)$, respectively. Nonetheless, we shall keep referring to the set of friends $I_F$ as a conversational aide.

In general, information about the  observation $c_i^{z_i=t_i}$ does not persist at the end of the experiment, unless the superobserver performs the query $q^{z_i=t_i}$. Owing to the assumption Absoluteness of Observed Events, however, all the variables observed during a round of experiment ought to have a well defined value. Importantly, the number of observations $c_i^{z_i}$ of  now depends explicitly on the interventions the superobserver $i$ performs, and thus may change from round to round.  To account for this, it will be useful to  introduce the notation 
\begin{align}
    \mathbf{C}^{\vec{t}} = (c_1^{z_1=1}, c_1^{z_1 = 2}, \ldots c_1^{z_1 = t_1}, \ldots, c_{|I|}^{z_{|I|} = t_{|I|}}),
\end{align}
for the records of the observers where each superobserver $i\in I$ finished the experiment at the $t_i$'th opportunity. Furthermore, let $\mathbf{C}^{|\vec{Z}|}$ represent the records in a situation where each friend made the maximum number $|Z_i|$ of observations, or equivalently, the cases where the superobservers $i\in I$ finished their experiment with any input $x_i^{t_i \geq |Z_i|}$.  

In what follows, we give a complete mathematical characterization of Local Friendliness behaviours $\mathbf{LF}(S_{sW})$ in an arbitrary general Wigner's friend scenario $S_{sW}$.

First, from  Absoluteness of Observed Events it follows that one can assign a joint distribution  over all the variables observed in a round of the experiment specified by $\vec{x}^{\vec{t}}$. The distributions $\wp(\vec{a}|\vec{x}^{\vec{t}})$ in the behaviour ought to be recovered from those joint distributions via marginalization as in

\begin{align}
    \wp(\Vec{a}|\Vec{x}^{\Vec{t}}) &= \sum_{\mathbf{C}^{\Vec{t}}} P^{\mathcal{LA}}(\Vec{a}, \mathbf{C}^{\Vec{t}}|\Vec{x}^{\Vec{t}}), \label{Equation:BehaviourasMarginalinsWs}
\end{align}
where the where the superscript $\mathcal{LA}$ has been added to emphasize that the distributions $P^{\mathcal{LA}}(\Vec{a}, \mathbf{C}^{\Vec{t}}|\Vec{x}^{\Vec{t}})$ are expected to satisfy Local Agency of Def.~\ref{definition:LocalFriendlinessLA}.

Let us define two sets which are specified by the context $\vec{x}^{\vec{t}}$: $F(\vec{x},\vec{t}) = \{ i \in I| x_i^{t_i} \in M_i^{Z}$ \} and $J(\vec{x},\vec{t}) = \{ i \in I | x_i^{_i} \in M^{Z^{\perp}}_i \} = I \setminus F(\vec{x},\vec{t})$. These represent the sets of agents who choose to ask the $t_i$'th observation $z_i$ of their friend, and those who either have no friend or let their friends make the maximum number of observations and then choose to make an alternative measurement, respectively. Then, using the properties of conditional probabilities,  and the conditions for the $x^{t_i}_i \in M^{Z}_i$ it can be seen that the following context-dependent expansion is valid for the expression $P^{\mathcal{LA}}(\vec{a}|\vec{x}^{\vec{t}})$ on the right-hand side of Eq.~\eqref{Equation:BehaviourasMarginalinsWs}
\begin{align}
\begin{split}
  P^{\mathcal{LA}}(\Vec{a}, \mathbf{C}^{\Vec{t}}|\Vec{x}^{\Vec{t}}) &= P^{\mathcal{LA}}(\vec{a}_{F(\vec{x},\vec{t})}|\vec{x}^{\vec{t}}, \mathbf{C}^{\vec{t}}, \vec{a}_{J(\vec{x},\vec{t})}) \\  & \hspace{0.5cm} \times P^{\mathcal{LA}}(\vec{a}_{J(\vec{x},\vec{t})}, \mathbf{C}^{\vec{t}}|\vec{x}^{\vec{t}})
  \end{split}
  \\[2ex]
  &= \prod_{i \in F(\vec{x},\vec{t})} \delta_{a_i, c_i^{t_i}} P^{\mathcal{LA}}(\vec{a}_{J(\vec{x},\vec{t})}, \mathbf{C}^{\vec{t}}|\vec{x}^{\vec{t}}).
\end{align}
The next steps in our argument are more delicate. 
The term $P^{\mathcal{LA}}(\vec{a}_{J(\vec{x},\vec{t})}, \mathbf{C}^{\vec{t}}|\vec{x}^{\vec{t}})$ above could be further manipulated by the properties of conditional probabilities and application of Local Agency, but we find it easier to work with this form. Since by definition the parties in $J(\vec{x},\vec{t})$ allowed their friends to make the maximum number of observations, it is useful to extend this notation  by using the shorthand $\mathbf{C}^{\vec{t}} = \mathbf{C}^{\vec{t}}_F \mathbf{C}^{|\vec{Z}|}_J$ as means to indicate the observations done at the sites $i\in F(\vec{x},\vec{t})$ and $i\in J(\vec{x},\vec{t})$ respectively.   

Now, owing to the space-time relations of the experiment any value $c_i^{z_i = t_i}$ ought to be independent of interventions that are in the future or space-like separated from it, hence in particular of all the choices $x_i^{t_i'}$ for $t_i' \geq t_i$ and $x_j^{t_j}$ for $j \neq i$ for all $t_j$. Therefore the family of equalities with
\begin{align}
\begin{split}
  \label{Equation:T'equalsTequations}  & P^{\mathcal{LA}}(\vec{a}_{J(\vec{x},\vec{t})}, \mathbf{C}^{\vec{t}}_F \mathbf{C}^{|\vec{Z}|}_J|\vec{x}_{F(\vec{x},\vec{t})}^{\vec{t}}\vec{x}^{\vec{t}}_{J(\vec{x},\vec{t})} ) \\
     &= P^{\mathcal{LA}}(\Vec{a}_{J(\vec{x},\vec{t})}, \mathbf{C}_F^{\Vec{t}} \mathbf{C}^{|\vec{Z}|}_J|\Vec{x}^{\vec{t'}}_{F(\vec{x},\vec{t'})}, \vec{x}^{\vec{t}}_{J(\vec{x},\vec{t})}) 
     \end{split}
\end{align}
holds for all  $\vec{x}^{\vec{t}}_{J(\vec{x},\vec{t})}$ and $1 \leq t_i \leq t_i' \leq |Z_i|$  with $i\in F(\vec{x},\vec{t}) = F(\vec{x},\vec{t'})$. This is because all of the above interventions $x_i^{t_i'}$ are by construction in the future of every $c_i^{t_i}$ while space-like separated from the other sites $j\in I$ $j\neq i$, as portrayed in the example of Fig.~\ref{fig:SequentialWignerSpacetime}, and thus, by appeal to Local Agency, such interventions are uncorrelated with the $c^{t_i}_i$. Therefore, as mentioned before, emphasizing whether $x^{t_i{'}}_i=0 $ or $x^{t_i'}_i = q^{{z_i} =t_{i}'}$ plays no difference for the variables $c_i^{t_i}$ with $t_i \leq t_i'$. 

On the other hand, since by Absoluteness of Observed Events, a joint distribution $P^{\mathcal{LA}}(\vec{a}, \mathbf{C}^{\vec{t}}|\vec{x}^{\vec{t}})$ over all the observed variables is well defined in any context it must be that in the case of Eq.~\eqref{Equation:T'equalsTequations} where $F{(\vec{x},\vec{t})} = F(\vec{x},\vec{t}')$ one can define a distribution $P^{\mathcal{LA}}(\vec{a}, \mathbf{C}^{\vec{t'}}_F \mathbf{C}^{|\vec{Z}|}_J|\vec{x}_{F(\vec{x},\vec{t'})}^{\vec{t'}}\vec{x}^{\vec{t}}_{J(\vec{x},\vec{t})} )$ which includes the possible additional observations $c_i^{z_i=k_i}$ with $ t_i <k_i \leq t_i'$ as well. The distributions in Eq.~\eqref{Equation:T'equalsTequations} have to be compatible with this distribution meaning simply that 

\begin{align}
\begin{split}
    P^{\mathcal{LA}}(\Vec{a}_{J(\vec{x},\vec{t})}, \mathbf{C}_F^{\Vec{t}} \mathbf{C}^{|\vec{Z}|}_{J}|\Vec{x}_{F(\vec{x},\vec{t'})}^{\Vec{t'}}\vec{x}^{\vec{t}}_{J(\vec{x},\vec{t})})\\
    =   \sum_{c_i^{z_i = k_i}}  P^{\mathcal{LA}}(\Vec{a}_{J(\vec{x},\vec{t})}, \mathbf{C}_F^{\Vec{t}'}\mathbf{C}_J^{|\vec{Z}|}|\Vec{x}_{F(\vec{x},\vec{t'})}^{\Vec{t'}},\vec{x}^{\vec{t}}_{J(\vec{x},\vec{t})}),\label{Equation:summingoverSomeK_i's}
    \end{split}
\end{align}
where in the summation $k_i>t_i $ and $i\in F(\vec{x},\vec{t'})$. Note that in the right hand side we have already summed over the suitable $a_i$. Since Eq.~\eqref{Equation:summingoverSomeK_i's} holds for all $i\in F(\vec{x},\vec{t'})$, $t_i' >t_i$ it holds, in particular,  in the extreme cases where $t_i's$ reach their maximum value $|Z_i|$. Therefore  
\begin{align}
\begin{split}
    P^{\mathcal{LA}}(\Vec{a}_{J(\vec{x},\vec{t})}, \mathbf{C}_F^{\Vec{t}}\mathbf{C}_J^{|\vec{Z}|}|\Vec{x}_{F(\vec{x},\vec{t})}^{\Vec{t}}\vec{x}^{\vec{t}}_{J(\vec{x},\vec{t})}) \\
    = \sum_{c_i^{z_i=k_i }, k_i>t_i}  P^{\mathcal{LA}}(\Vec{a}_{J(\vec{x},\vec{t})}, \mathbf{C}^{|\vec{Z}|}|\Vec{x}_{F(\vec{x},\vec{t})}^{|\vec{Z}|}\vec{x}^{\vec{t}}_{J(\vec{x},\vec{t})}).
    \end{split}
\end{align}

Finally,  the term $P^{\mathcal{LA}}(\Vec{a}_{J(\vec{x},\vec{t})}, \mathbf{C}^{|\vec{Z}|}|\Vec{x}_{F(\vec{x},\vec{t})}^{|\vec{Z}|}\vec{x}^{\vec{t}}_{J(\vec{x},\vec{t})})$  can be decomposed by using the properties of conditional probabilities and the assumption of Local Agency into
\begin{align}
\begin{split}
   & P^{\mathcal{LA}}(\Vec{a}_{J(\vec{x},\vec{t})}, \mathbf{C}^{|\vec{Z}|}|\Vec{x}_{F(\vec{x},\vec{t})}^{|\vec{Z}|}\vec{x}^{t}_{J(\vec{x},\vec{t})})  \\
    &=  P^{\mathcal{LA}}(\Vec{a}_{J(\vec{x},\vec{t})}|\Vec{x}_{F(\vec{x},\vec{t})}^{|\vec{Z}|}\vec{x}^{\vec{t}}_{J(\vec{x},\vec{t})}, \mathbf{C}^{|\vec{Z}|}) \\
   & \hspace{0.5cm} \times P^{\mathcal{LA}}(\mathbf{C}^{|\vec{Z}|}|\Vec{x}_{F(\vec{x},\vec{t})}^{|\vec{Z}|}\vec{x}^{\vec{t}}_{J(\vec{x},\vec{t})}) \\
\end{split}\\
    &=P^{\mathcal{LA}}(\Vec{a}_{J(\vec{x},\vec{t})},|\Vec{x}_{J(\vec{x},\vec{r})}^{|\vec{Z}|}, \mathbf{C}^{|\vec{Z}|})\cdot P^{\mathcal{LA}}(\mathbf{C}^{|\vec{Z}|}),
\end{align}
where the second equality follows from the fact that the $a_j$'s with $j\in J(\vec{x},\vec{t})
$ are outside the light cones of the interventions $x_i^{t_i}$ with $i\in F(\vec{x},\vec{t})$ and the interventions $\vec{x}^{|\vec{Z}|}$ are in the future of all the records $\mathbf{C}^{|\vec{Z}|}$  have been used.

Putting everything together, we find that any behaviour $\wp$ in a sequential Wigner's friend scenario can be decomposed as

\begin{align}
\begin{split}
 \label{Equation:LocalFriendlinessmodelInSWS}   \wp(\vec{a}|\vec{x}^{\vec{t}}) = &\sum_{\mathbf{C}^{|\vec{Z}|}}\left[ \prod_{i \in F(\vec{x},\vec{t})}\delta_{a_i, c_i^{t_i}} \cdot P^{\mathcal{LA}}(\vec{a}_{J(\vec{x},{\vec{t}})}|\vec{x}^{\vec{t}}_{J(\vec{x},\vec{t})}, \mathbf{C}^{|\vec{Z}|} ) \right. \\
    & \hspace{0.5cm} \left. \times  P^{\mathcal{LA}}(\mathbf{C}^{|\vec{Z}|})\right]
    \end{split}
\end{align}
with $F(\vec{x},\vec{t}) = \{i\in I_| x_i^{t_i} \in M_i^{Z}\}$, $J(\vec{x},\vec{r})$ = $I \setminus F(\vec{x}, \vec{t})$. The properties of this model immediately entail the desired result. 

\begin{theorem}\label{Theorem:LFcorrelationsarePDP}
    Let $S_{sW} =  (S,Z)$ be an arbitrary sequential Wigner's friend scenario. Then, the set $\mathbf{LF}(S_{sW})$ of correlations compatible with the assumptions going into Local Friendliness form a partially deterministic polytope. More presicely,
    \begin{align}
        \mathbf{LF}(S_{sW}) = \mathbf{PD}(S, M^{Z}),
    \end{align}
    where the collection $M^{Z}$ consists of the query measurements.
\end{theorem}
\begin{proof}
    See the steps between Eqs.~\eqref{Equation:BehaviourasMarginalinsWs}-\eqref{Equation:LocalFriendlinessmodelInSWS}.
\end{proof}

Theorem \ref{Theorem:LFcorrelationsarePDP} establishes partially deterministic polytopes as the relevant mathematical objects for bounding Local Friendliness correlations in seqential scenarios. To establish a no-go result, one would need behaviours which are outside the polytope. In particular, by Theorem \ref{Theorem:QuantumsetandPDP}, there are always quantum behaviour outside every partially deterministic polytope, excluding the extreme case where there are no deterministic measurements. It is worth pointing out here that this fact is a priori not enough to state that any such scenario can be used to demonstrate a quantum no-go result, since the set  $\mathbf{Q}(S)$ of quantum correlations was defined over the ordinary Bell scenario depicted in Fig.~\ref{fig:correlationscenarioFig} and does not contain information of the additional steps in the protocol which would have to be included in the model. Similarly to Proposition 1 of Ref.~\cite{Haddara2025} in the case of canonical scenarios, we can however establish a general equivalence between the sets $\mathbf{Q}(S_{sW})$ of quantum behaviours that can be produced in a sequential Wigner's friend scenario and $\mathbf{Q}(S)$. In the case of the sequential scenario, we the claim can be seen to hold by generalizing the model in Ref.~\cite{Utreras-Alarcon2024}.   For completeness, we report this conclusion as a proposition. 

\proposition{Let $S_{sW}= (S,Z)$ be an arbitrary sequential Wigner's friend scenario. The set $Q(S_{sW})$ of correlations compatible with a quantum model in that scenario equals the set of correlations in an ordinary Bell scenario $S$. That is
\begin{align}
    \mathbf{Q}(S) = \mathbf{Q}(S_{sW}).
\end{align}
}
\begin{proof}[proof sketch.]
   We omit a detailed proof, as it is a straightforward generalization of the model considered in Ref.~\cite{Utreras-Alarcon2024}, and instead sketch the general idea. 

Let $\wp(\vec{a}|\vec{x}) = tr[(\bigotimes\limits_{i\in I}M_{a_i|x_i}) \rho_{S}]$ be a behaviour realised in an ordinary Bell scenario with local measurements on the state $\rho_{S}$ sent by the source. The claim is that this behaviour can be realised in a quantum description the sequential Wigner's friend protocol with a reading consistent with the demand that the inputs $x_i^{t_i}$, $1\leq t_i \leq |Z_i|$ correspondto querying the friend for their observation.

Note first that without loss of generality the operators $M_{a_i|x_i}$ may be taken to be projective, and the state $\rho_{S}$ to be a pure $|I|$-partite state. The claim can then be seen simply by considering a model where the evolution of the friend+system state due to observation by the friend is modelled unitarily and the non-query measurements first reverse the friends laboratory and measure the system directly with the original $M_{a_i|x_i}$. The correct reading for the query measurements can be guaranteed, for example, by restricting each friend to always make the same observation, and  having  all the query measurements match a projector onto the subspace of the friends memory where a given outcome is stored. 
\end{proof}
\normalfont

Finally let us draw attention to the fact that the equality of polytopes in Theorem \ref{Theorem:LFcorrelationsarePDP} can be viewed to be a consequence of the one-to-one correspondence  between the structure $S_{sW} = (S,Z)$ which under Local Friendliness imposes the collection $M^{Z}$ to be deterministic, and the structure $S$ with a collection $M'$ of deterministic inputs. The mapping of sequential Wigner's friend scenarios $S_{sW}$ and the corresponding Local Friendliness-polytopes in those scenarios is therefore many-to-one, owing to Theorem \ref{Theorem:EquivalenceClassesofPDP's}, where it was established that there are nontrivial classes where different collections of deterministic inputs lead to the same polytope.  In this context, these equivalence classes can be seen to attain considerable practical  significance, since two scenarios  $S_{sW} = (S,Z)$ and $S'_{sW}=(S,Z')$ correspond to physically different experiments, with possibly different numbers of observers and sequences of measurements. Thus, in particular, Theorem \ref{Theorem:EquivalenceClassesofPDP's} can be used to find the 'simplest' implementation of a sequential Wigner's friend scenario where the Local Friendliness polytope matches a desired partially deterministic polytope. This observation complements a similar result in the context of canonical scenarios \cite{Haddara2025} where it was recognized to  be relevant for  potential experimental implementations.  Indeed, from an experimental perspective in this setting,  there may be motivation to minimize the numbers of reversals and inclusion of additional physical systems to avoid inclusion of ancillary noise.

\section{\label{Section:conclusions}Conclusions}

In this work we have defined and thoroughly investigated  the properties of the  device-independent behaviour sets called \emph{partially deterministic polytopes}, following terminology introduced by Woodhead \cite{Woodhead2014}, special cases of which have been studied in the literature under different names before,  e.g. as Local Friendliness polytopes in Refs.~\cite{Bong2020, Haddara2025}. Loosely speaking, these polytopes consists of those behaviours which are compatible with a given \emph{partially deterministic model}, a  notion which generalizes that  of a local deterministic model by imposing that only a certain subset of inputs at each site are determined by a local variable.

We studied the families of polytopes which arise from the distinct ways partial determinism may be imposed in a given scenario. Among our results we found that the polytopes may be arranged into equivalence classes and completely classified them. The equivalence class corresponding to the no-signalling polytope always has a single element, corresponding to no deterministic inputs at any site, while at the other extreme, the Bell-polytope in particular can in general be represented in multiple different ways by imposing partial determinism. Nontrivial equivalence classes may also arise for polytopes which do not equal the Bell or the no-signalling sets. 

While every partially deterministic polytope contains the Bell polytope and is contained in the no-signalling polytope, two different partially deterministic polytopes belonging to different equivalence classes may either obey a strict subset relation, or it may be that neither polytope is contained in the other. We have also shown, that in every case where a partially deterministic polytope does not equal the Bell or the no-signalling polytopes, then the quantum set is not contained in, nor fully contains, that partially deterministic polytope. 

When two partially deterministic polytopes can be ordered with respect to the strict subset relation, so can the sets of their extreme points. As a corollary, every vertex of any partially deterministic polytope is also a vertex of the no-signalling polytope in particular. Furthermore if the vertex representation of a no-signalling polytope for a given scenario is known, it can be used to give vertex representations of any partially deterministic polytope for that scenario by excluding all the, and only the vertices which do not give a deterministic value for the inputs which are deterministic for that polytope.  

We also demonstrated that, in analogy with Fine's theorem \cite{Fine1982,Fine1982b} for the case of local determinism, the assumption that a behaviour has a partially deterministic model is equivalent to a seemingly more general notion, which we call ``partial factorizability''. By pursuing this analogy further, and considering statements about the existence of joint distributions with certain marginals,  we proved a theorem  that strictly generalizes Fine's theorem, recovering it as a special case. More generally, this theorem introduces new necessary and sufficient conditions for the existence of a partially deterministic model of a certain kind. In particular, the numerous  ways to impose partial determinism leading to the equivalence class corresponding to the Bell-polytope can be used to obtain novel equivalent ways to express necessary and sufficient conditions for LHV-modelability in general Bell scenarios, further sharpening the usual formulation of Fine's theorem.

On the technical side we provided a general method to construct new behaviour sets in a given scenario from behaviours defined over subscenarios of the original one. This type of  map, termed the `behaviour product', generalises the notion introduced in the bipartite case by Woodhead \cite{Woodhead2014}, in his definition of partially deterministic polytopes. Indeed, we find that generally, a partially deterministic polytope can be equivalently expressed using the behaviour product of two subscenarios, one scenario with local deterministic behaviours, and the other constrained to no-signalling correlations. We showed that not every set of behaviours -- including the quantum and no-signalling sets -- can be represented as a nontrivial combination of sets over smaller scenarios, and hence this property, which we termed ``composability'', is a nontrivial  structural property of some behaviour sets. Our formalism provides the basic tools to investigate the features of such sets. 

In previous works, some properties of partially deterministic polytopes have explicitly been studied in the context of randomness certification \cite{Woodhead2014,ramanathan2024maximumquantumnonlocalitysufficient} and the foundational Local Friendliness no-go theorem \cite{Bong2020, Utreras-Alarcon2024, Haddara2025}. We identified new applications of partially deterministic polytopes, such as device-independent (indeed theory-independent) quantum state inseparability witnesses, and as precisely the bounds for the assumption of  ``broadcast-locality'' \cite{Bowles2021, Boghiu2023} in a class of broadcasting scenarios. The intersections, unions, and convex hulls of the unions of partially deterministic polytopes were also shown to be of interest, and closely related to notions of anonymous nonlocality \cite{Liang2020}, separability with respect to a given partition, and  genuine multipartite nonlocality \cite{Svetlichny1987, Bancal2013, Gallego2012}. We also recognized a limitation of the applicability of  definitions of partial determinism for those purposes, suggesting that more general notions of composable sets, and the objects obtainable from them by set-theoretic relations,  should be investigated in order to capture bounds related to broadcast-locality in scenarios which include multiple broadcasters, and other separability notions in multipartite scenarios with more than three parties.

As another demonstration of applicability, we considered a generalization of the sequential extended Wigner's friend scenario introduced in Ref.~\cite{Utreras-Alarcon2024}, and showed that the polytopes which constrain the Local Friendliness correlations in those scenarios are in one-to-one correspondence with partially deterministic polytopes, so that in this sense, every partially deterministic polytope may be considered of physical relevance. In this case, the equivalence classes of polytopes also attain physical relevance, as their elements correspond to different experimental set  ups. Therefore, in particular,  given some partially deterministic polytope, it is always possible to decide the simplest experimental realisation of a sequential Wigner's friend scenario which leads to those constraints. This observation generalises that of Ref.~\cite{Haddara2025}, where a similar potentially useful consequence of the equivalence classes of polytopes was recognized in the context of so-called canonical Local Friendliness scenarios. 

As our framework is very general, we believe that classes of partially deterministic polytopes could be of significance in various other topics in device-independent quantum information, motivated by distinct information processing tasks or physical principles. Some candidates for further enquiry include explorations of (non)joint measurability structures \cite{Quintino2019}, and classes of experiments which aim to simultaneously probe features of Bell-nonlocality  and noncontextuality  \cite{Temistocles2019, Mazzari2023}.  The concept may also be useful in situations where non-classical properties emerge but for whatever reason one may expect some subsets of systems, or measurements on those systems,  to be compatible with a classical model \cite{Chiribella2024}.

Finally, while reliance on the precise causal structure of a Bell scenario, as our examples demonstrate, is not required to apply our results, we have essentially built our definitions on families of `marginal compatibility structures' which are constructed on top of the Bell-scenario. Given that not all compatibility structures may be realised in a Bell scenario \cite{BudroniContextualityreview2022}, it seems there is room for further generalization of our definitions and results. Such generalizations could be of relevance, for example, for attempts to prove Wigner's friend no-go theorems  without assuming space-like separation \cite{walleghem2024extendedwignersfriendparadoxes, walleghem2024connectingextendedwignersfriend, walleghem2025extendedwignersfriendnogo}. Another interesting potential generalization of our results concerns more general correlation scenarios with multiple independent sources  \cite{Branciard2010, Fritz_2012}. To what extent the techniques used in this work could be applied to those situations is in our view, worth future investigation.

\section*{Acknowledgements}
MH gratefully acknowledges Yeong-Cherng Liang for useful discussions and for pointing out the results Ref.~\cite{Liang2020} which showed that the intersection of three distinct tripartite partially deterministic polytopes is not equal to the Bell polytope. 
MH is also grateful to Elie Wolfe for pointing out the work in Ref.~\cite{ramanathan2024maximumquantumnonlocalitysufficient} which made use of partially deterministic polytopes, and to Luis-Villegas Aguilar for pointing out Refs.~\cite{Temistocles2019, Mazzari2023} which described objects closely related to the ones investigated here. While preparing this manuscript, we became aware of work in Ref.~\cite{walleghem2024connectingextendedwignersfriend} which also explored constraints arising from Local Friendliness in sequential Wigner's friend scenarios. This research was funded by the Australian Government through the Australian Research Council’s Centre of Excellence Program CE170100012, the Centre for Quantum Computation and Communication Technology (CQC2T), and an Australian Government Research Training Program (RTP) Scholarship. This work was also supported by ARC Grant No. DP210101651. 
\normalfont

\newpage
\bibliography{LF_PDPs}

\begin{thebibliography}{113}%
\makeatletter
\providecommand \@ifxundefined [1]{%
 \@ifx{#1\undefined}
}%
\providecommand \@ifnum [1]{%
 \ifnum #1\expandafter \@firstoftwo
 \else \expandafter \@secondoftwo
 \fi
}%
\providecommand \@ifx [1]{%
 \ifx #1\expandafter \@firstoftwo
 \else \expandafter \@secondoftwo
 \fi
}%
\providecommand \natexlab [1]{#1}%
\providecommand \enquote  [1]{``#1''}%
\providecommand \bibnamefont  [1]{#1}%
\providecommand \bibfnamefont [1]{#1}%
\providecommand \citenamefont [1]{#1}%
\providecommand \href@noop [0]{\@secondoftwo}%
\providecommand \href [0]{\begingroup \@sanitize@url \@href}%
\providecommand \@href[1]{\@@startlink{#1}\@@href}%
\providecommand \@@href[1]{\endgroup#1\@@endlink}%
\providecommand \@sanitize@url [0]{\catcode `\\12\catcode `\$12\catcode `\&12\catcode `\#12\catcode `\^12\catcode `\_12\catcode `\%12\relax}%
\providecommand \@@startlink[1]{}%
\providecommand \@@endlink[0]{}%
\providecommand \url  [0]{\begingroup\@sanitize@url \@url }%
\providecommand \@url [1]{\endgroup\@href {#1}{\urlprefix }}%
\providecommand \urlprefix  [0]{URL }%
\providecommand \Eprint [0]{\href }%
\providecommand \doibase [0]{https://doi.org/}%
\providecommand \selectlanguage [0]{\@gobble}%
\providecommand \bibinfo  [0]{\@secondoftwo}%
\providecommand \bibfield  [0]{\@secondoftwo}%
\providecommand \translation [1]{[#1]}%
\providecommand \BibitemOpen [0]{}%
\providecommand \bibitemStop [0]{}%
\providecommand \bibitemNoStop [0]{.\EOS\space}%
\providecommand \EOS [0]{\spacefactor3000\relax}%
\providecommand \BibitemShut  [1]{\csname bibitem#1\endcsname}%
\let\auto@bib@innerbib\@empty
\bibitem [{\citenamefont {Bell}(1964)}]{Bell1964}%
  \BibitemOpen
  \bibfield  {author} {\bibinfo {author} {\bibfnamefont {J.~S.}\ \bibnamefont {Bell}},\ }\bibfield  {title} {\bibinfo {title} {{On the Einstein Podolsky Rosen paradox}},\ }\href {https://doi.org/10.1103/PhysicsPhysiqueFizika.1.195} {\bibfield  {journal} {\bibinfo  {journal} {{Physics Physique Fizika}}\ }\textbf {\bibinfo {volume} {1}},\ \bibinfo {pages} {195} (\bibinfo {year} {1964})}\BibitemShut {NoStop}%
\bibitem [{\citenamefont {Brunner}\ \emph {et~al.}(2014)\citenamefont {Brunner}, \citenamefont {Cavalcanti}, \citenamefont {Pironio}, \citenamefont {Scarani},\ and\ \citenamefont {Wehner}}]{Brunner2014}%
  \BibitemOpen
  \bibfield  {author} {\bibinfo {author} {\bibfnamefont {N.}~\bibnamefont {Brunner}}, \bibinfo {author} {\bibfnamefont {D.}~\bibnamefont {Cavalcanti}}, \bibinfo {author} {\bibfnamefont {S.}~\bibnamefont {Pironio}}, \bibinfo {author} {\bibfnamefont {V.}~\bibnamefont {Scarani}},\ and\ \bibinfo {author} {\bibfnamefont {S.}~\bibnamefont {Wehner}},\ }\bibfield  {title} {\bibinfo {title} {Bell nonlocality},\ }\href {https://doi.org/10.1103/RevModPhys.86.419} {\bibfield  {journal} {\bibinfo  {journal} {Rev. Mod. Phys.}\ }\textbf {\bibinfo {volume} {86}},\ \bibinfo {pages} {419} (\bibinfo {year} {2014})}\BibitemShut {NoStop}%
\bibitem [{\citenamefont {Scarani}(2012)}]{Scarani2012}%
  \BibitemOpen
  \bibfield  {author} {\bibinfo {author} {\bibfnamefont {V.}~\bibnamefont {Scarani}},\ }\bibfield  {title} {\bibinfo {title} {{ The device-independent outlook on quantum physics}},\ }\href {http://www.physics.sk/aps/pub.php?y=2012&pub=aps-12-04} {\bibfield  {journal} {\bibinfo  {journal} {{acta physica slovaca}}\ }\textbf {\bibinfo {volume} {62}},\ \bibinfo {pages} {347} (\bibinfo {year} {2012})}\BibitemShut {NoStop}%
\bibitem [{\citenamefont {Ac{\'\i}n}\ and\ \citenamefont {Navascu{\'e}s}(2017)}]{Acin2017BlackBoxUnspeakablesBook}%
  \BibitemOpen
  \bibfield  {author} {\bibinfo {author} {\bibfnamefont {A.}~\bibnamefont {Ac{\'\i}n}}\ and\ \bibinfo {author} {\bibfnamefont {M.}~\bibnamefont {Navascu{\'e}s}},\ }\bibinfo {title} {{Black Box Quantum Mechanics}},\ in\ \href {https://doi.org/10.1007/978-3-319-38987-5_17} {\emph {\bibinfo {booktitle} {Quantum {$[$}Un{$]$}Speakables II: Half a Century of Bell's Theorem}}},\ \bibinfo {editor} {edited by\ \bibinfo {editor} {\bibfnamefont {R.}~\bibnamefont {Bertlmann}}\ and\ \bibinfo {editor} {\bibfnamefont {A.}~\bibnamefont {Zeilinger}}}\ (\bibinfo  {publisher} {Springer International Publishing},\ \bibinfo {address} {Cham},\ \bibinfo {year} {2017})\ pp.\ \bibinfo {pages} {307--319}\BibitemShut {NoStop}%
\bibitem [{\citenamefont {Hensen}\ \emph {et~al.}(2015)\citenamefont {Hensen}, \citenamefont {Bernien}, \citenamefont {Dr{\'e}au}, \citenamefont {Reiserer}, \citenamefont {Kalb}, \citenamefont {Blok}, \citenamefont {Ruitenberg}, \citenamefont {Vermeulen}, \citenamefont {Schouten}, \citenamefont {Abell{\'a}n}, \citenamefont {Amaya}, \citenamefont {Pruneri}, \citenamefont {Mitchell}, \citenamefont {Markham}, \citenamefont {Twitchen}, \citenamefont {Elkouss}, \citenamefont {Wehner}, \citenamefont {Taminiau},\ and\ \citenamefont {Hanson}}]{Hensen2015}%
  \BibitemOpen
  \bibfield  {author} {\bibinfo {author} {\bibfnamefont {B.}~\bibnamefont {Hensen}}, \bibinfo {author} {\bibfnamefont {H.}~\bibnamefont {Bernien}}, \bibinfo {author} {\bibfnamefont {A.~E.}\ \bibnamefont {Dr{\'e}au}}, \bibinfo {author} {\bibfnamefont {A.}~\bibnamefont {Reiserer}}, \bibinfo {author} {\bibfnamefont {N.}~\bibnamefont {Kalb}}, \bibinfo {author} {\bibfnamefont {M.~S.}\ \bibnamefont {Blok}}, \bibinfo {author} {\bibfnamefont {J.}~\bibnamefont {Ruitenberg}}, \bibinfo {author} {\bibfnamefont {R.~F.~L.}\ \bibnamefont {Vermeulen}}, \bibinfo {author} {\bibfnamefont {R.~N.}\ \bibnamefont {Schouten}}, \bibinfo {author} {\bibfnamefont {C.}~\bibnamefont {Abell{\'a}n}}, \bibinfo {author} {\bibfnamefont {W.}~\bibnamefont {Amaya}}, \bibinfo {author} {\bibfnamefont {V.}~\bibnamefont {Pruneri}}, \bibinfo {author} {\bibfnamefont {M.~W.}\ \bibnamefont {Mitchell}}, \bibinfo {author} {\bibfnamefont {M.}~\bibnamefont {Markham}}, \bibinfo {author} {\bibfnamefont {D.~J.}\ \bibnamefont {Twitchen}}, \bibinfo {author}
  {\bibfnamefont {D.}~\bibnamefont {Elkouss}}, \bibinfo {author} {\bibfnamefont {S.}~\bibnamefont {Wehner}}, \bibinfo {author} {\bibfnamefont {T.~H.}\ \bibnamefont {Taminiau}},\ and\ \bibinfo {author} {\bibfnamefont {R.}~\bibnamefont {Hanson}},\ }\bibfield  {title} {\bibinfo {title} {{Loophole-free Bell inequality violation using electron spins separated by 1.3 kilometres}},\ }\href {https://doi.org/10.1038/nature15759} {\bibfield  {journal} {\bibinfo  {journal} {Nature}\ }\textbf {\bibinfo {volume} {526}},\ \bibinfo {pages} {682} (\bibinfo {year} {2015})}\BibitemShut {NoStop}%
\bibitem [{\citenamefont {Rosenfeld}\ \emph {et~al.}(2017)\citenamefont {Rosenfeld}, \citenamefont {Burchardt}, \citenamefont {Garthoff}, \citenamefont {Redeker}, \citenamefont {Ortegel}, \citenamefont {Rau},\ and\ \citenamefont {Weinfurter}}]{Rosenfeld2017}%
  \BibitemOpen
  \bibfield  {author} {\bibinfo {author} {\bibfnamefont {W.}~\bibnamefont {Rosenfeld}}, \bibinfo {author} {\bibfnamefont {D.}~\bibnamefont {Burchardt}}, \bibinfo {author} {\bibfnamefont {R.}~\bibnamefont {Garthoff}}, \bibinfo {author} {\bibfnamefont {K.}~\bibnamefont {Redeker}}, \bibinfo {author} {\bibfnamefont {N.}~\bibnamefont {Ortegel}}, \bibinfo {author} {\bibfnamefont {M.}~\bibnamefont {Rau}},\ and\ \bibinfo {author} {\bibfnamefont {H.}~\bibnamefont {Weinfurter}},\ }\bibfield  {title} {\bibinfo {title} {{Event-Ready Bell Test Using Entangled Atoms Simultaneously Closing Detection and Locality Loopholes}},\ }\href {https://doi.org/10.1103/PhysRevLett.119.010402} {\bibfield  {journal} {\bibinfo  {journal} {Phys. Rev. Lett.}\ }\textbf {\bibinfo {volume} {119}},\ \bibinfo {pages} {010402} (\bibinfo {year} {2017})}\BibitemShut {NoStop}%
\bibitem [{\citenamefont {Storz}\ \emph {et~al.}(2023)\citenamefont {Storz}, \citenamefont {Sch{\"a}r}, \citenamefont {Kulikov}, \citenamefont {Magnard}, \citenamefont {Kurpiers}, \citenamefont {L{\"u}tolf}, \citenamefont {Walter}, \citenamefont {Copetudo}, \citenamefont {Reuer}, \citenamefont {Akin}, \citenamefont {Besse}, \citenamefont {Gabureac}, \citenamefont {Norris}, \citenamefont {Rosario}, \citenamefont {Martin}, \citenamefont {Martinez}, \citenamefont {Amaya}, \citenamefont {Mitchell}, \citenamefont {Abellan}, \citenamefont {Bancal}, \citenamefont {Sangouard}, \citenamefont {Royer}, \citenamefont {Blais},\ and\ \citenamefont {Wallraff}}]{Stortz2023}%
  \BibitemOpen
  \bibfield  {author} {\bibinfo {author} {\bibfnamefont {S.}~\bibnamefont {Storz}}, \bibinfo {author} {\bibfnamefont {J.}~\bibnamefont {Sch{\"a}r}}, \bibinfo {author} {\bibfnamefont {A.}~\bibnamefont {Kulikov}}, \bibinfo {author} {\bibfnamefont {P.}~\bibnamefont {Magnard}}, \bibinfo {author} {\bibfnamefont {P.}~\bibnamefont {Kurpiers}}, \bibinfo {author} {\bibfnamefont {J.}~\bibnamefont {L{\"u}tolf}}, \bibinfo {author} {\bibfnamefont {T.}~\bibnamefont {Walter}}, \bibinfo {author} {\bibfnamefont {A.}~\bibnamefont {Copetudo}}, \bibinfo {author} {\bibfnamefont {K.}~\bibnamefont {Reuer}}, \bibinfo {author} {\bibfnamefont {A.}~\bibnamefont {Akin}}, \bibinfo {author} {\bibfnamefont {J.-C.}\ \bibnamefont {Besse}}, \bibinfo {author} {\bibfnamefont {M.}~\bibnamefont {Gabureac}}, \bibinfo {author} {\bibfnamefont {G.~J.}\ \bibnamefont {Norris}}, \bibinfo {author} {\bibfnamefont {A.}~\bibnamefont {Rosario}}, \bibinfo {author} {\bibfnamefont {F.}~\bibnamefont {Martin}}, \bibinfo {author} {\bibfnamefont {J.}~\bibnamefont
  {Martinez}}, \bibinfo {author} {\bibfnamefont {W.}~\bibnamefont {Amaya}}, \bibinfo {author} {\bibfnamefont {M.~W.}\ \bibnamefont {Mitchell}}, \bibinfo {author} {\bibfnamefont {C.}~\bibnamefont {Abellan}}, \bibinfo {author} {\bibfnamefont {J.-D.}\ \bibnamefont {Bancal}}, \bibinfo {author} {\bibfnamefont {N.}~\bibnamefont {Sangouard}}, \bibinfo {author} {\bibfnamefont {B.}~\bibnamefont {Royer}}, \bibinfo {author} {\bibfnamefont {A.}~\bibnamefont {Blais}},\ and\ \bibinfo {author} {\bibfnamefont {A.}~\bibnamefont {Wallraff}},\ }\bibfield  {title} {\bibinfo {title} {{Loophole-free Bell inequality violation with superconducting circuits}},\ }\href {https://doi.org/10.1038/s41586-023-05885-0} {\bibfield  {journal} {\bibinfo  {journal} {Nature}\ }\textbf {\bibinfo {volume} {617}},\ \bibinfo {pages} {265} (\bibinfo {year} {2023})}\BibitemShut {NoStop}%
\bibitem [{\citenamefont {Bell}(1976)}]{Bell76}%
  \BibitemOpen
  \bibfield  {author} {\bibinfo {author} {\bibfnamefont {J.~S.}\ \bibnamefont {Bell}},\ }\bibfield  {title} {\bibinfo {title} {The theory of local beables},\ }\href@noop {} {\bibfield  {journal} {\bibinfo  {journal} {Epistemological Lett.}\ }\textbf {\bibinfo {volume} {9}},\ \bibinfo {pages} {11} (\bibinfo {year} {1976})},\ \bibinfo {note} {(Reproduced in Ref.~\cite{Bellcollection}.)}\BibitemShut {NoStop}%
\bibitem [{\citenamefont {Wiseman}\ and\ \citenamefont {Cavalcanti}(2017)}]{Wiseman2017}%
  \BibitemOpen
  \bibfield  {author} {\bibinfo {author} {\bibfnamefont {H.~M.}\ \bibnamefont {Wiseman}}\ and\ \bibinfo {author} {\bibfnamefont {E.~G.}\ \bibnamefont {Cavalcanti}},\ }\bibinfo {title} {{\emph{Causarum Investigatio} and the Two Bell's Theorems of John Bell}},\ in\ \href {https://doi.org/10.1007/978-3-319-38987-5_6} {\emph {\bibinfo {booktitle} {Quantum [Un]Speakables II: Half a Century of Bell's Theorem}}},\ \bibinfo {editor} {edited by\ \bibinfo {editor} {\bibfnamefont {R.}~\bibnamefont {Bertlmann}}\ and\ \bibinfo {editor} {\bibfnamefont {A.}~\bibnamefont {Zeilinger}}}\ (\bibinfo  {publisher} {Springer International Publishing},\ \bibinfo {address} {Cham},\ \bibinfo {year} {2017})\ pp.\ \bibinfo {pages} {119--142}\BibitemShut {NoStop}%
\bibitem [{\citenamefont {Cavalcanti}\ and\ \citenamefont {Wiseman}(2021)}]{Cavalcanti2021}%
  \BibitemOpen
  \bibfield  {author} {\bibinfo {author} {\bibfnamefont {E.~G.}\ \bibnamefont {Cavalcanti}}\ and\ \bibinfo {author} {\bibfnamefont {H.~M.}\ \bibnamefont {Wiseman}},\ }\bibfield  {title} {\bibinfo {title} {{Implications of Local Friendliness Violation for Quantum Causality}},\ }\href {https://doi.org/10.3390/e23080925} {\bibfield  {journal} {\bibinfo  {journal} {{Entropy}}\ }\textbf {\bibinfo {volume} {23}},\ \bibinfo {pages} {925} (\bibinfo {year} {2021})}\BibitemShut {NoStop}%
\bibitem [{\citenamefont {Ma}\ \emph {et~al.}(2016)\citenamefont {Ma}, \citenamefont {Yuan}, \citenamefont {Cao}, \citenamefont {Qi},\ and\ \citenamefont {Zhang}}]{Ma2016}%
  \BibitemOpen
  \bibfield  {author} {\bibinfo {author} {\bibfnamefont {X.}~\bibnamefont {Ma}}, \bibinfo {author} {\bibfnamefont {X.}~\bibnamefont {Yuan}}, \bibinfo {author} {\bibfnamefont {Z.}~\bibnamefont {Cao}}, \bibinfo {author} {\bibfnamefont {B.}~\bibnamefont {Qi}},\ and\ \bibinfo {author} {\bibfnamefont {Z.}~\bibnamefont {Zhang}},\ }\bibfield  {title} {\bibinfo {title} {Quantum random number generation},\ }\href {https://doi.org/10.1038/npjqi.2016.21} {\bibfield  {journal} {\bibinfo  {journal} {npj Quantum Information}\ }\textbf {\bibinfo {volume} {2}},\ \bibinfo {pages} {16021} (\bibinfo {year} {2016})}\BibitemShut {NoStop}%
\bibitem [{\citenamefont {Pirandola}\ \emph {et~al.}(2020)\citenamefont {Pirandola}, \citenamefont {Andersen}, \citenamefont {Banchi}, \citenamefont {Berta}, \citenamefont {Bunandar}, \citenamefont {Colbeck}, \citenamefont {Englund}, \citenamefont {Gehring}, \citenamefont {Lupo}, \citenamefont {Ottaviani}, \citenamefont {Pereira}, \citenamefont {Razavi}, \citenamefont {Shamsul~Shaari}, \citenamefont {Tomamichel}, \citenamefont {Usenko}, \citenamefont {Vallone}, \citenamefont {Villoresi},\ and\ \citenamefont {Wallden}}]{Pirandola2020}%
  \BibitemOpen
  \bibfield  {author} {\bibinfo {author} {\bibfnamefont {S.}~\bibnamefont {Pirandola}}, \bibinfo {author} {\bibfnamefont {U.~L.}\ \bibnamefont {Andersen}}, \bibinfo {author} {\bibfnamefont {L.}~\bibnamefont {Banchi}}, \bibinfo {author} {\bibfnamefont {M.}~\bibnamefont {Berta}}, \bibinfo {author} {\bibfnamefont {D.}~\bibnamefont {Bunandar}}, \bibinfo {author} {\bibfnamefont {R.}~\bibnamefont {Colbeck}}, \bibinfo {author} {\bibfnamefont {D.}~\bibnamefont {Englund}}, \bibinfo {author} {\bibfnamefont {T.}~\bibnamefont {Gehring}}, \bibinfo {author} {\bibfnamefont {C.}~\bibnamefont {Lupo}}, \bibinfo {author} {\bibfnamefont {C.}~\bibnamefont {Ottaviani}}, \bibinfo {author} {\bibfnamefont {J.~L.}\ \bibnamefont {Pereira}}, \bibinfo {author} {\bibfnamefont {M.}~\bibnamefont {Razavi}}, \bibinfo {author} {\bibfnamefont {J.}~\bibnamefont {Shamsul~Shaari}}, \bibinfo {author} {\bibfnamefont {M.}~\bibnamefont {Tomamichel}}, \bibinfo {author} {\bibfnamefont {V.~C.}\ \bibnamefont {Usenko}}, \bibinfo {author} {\bibfnamefont
  {G.}~\bibnamefont {Vallone}}, \bibinfo {author} {\bibfnamefont {P.}~\bibnamefont {Villoresi}},\ and\ \bibinfo {author} {\bibfnamefont {P.}~\bibnamefont {Wallden}},\ }\bibfield  {title} {\bibinfo {title} {Advances in quantum cryptography},\ }\href {https://doi.org/10.1364/AOP.361502} {\bibfield  {journal} {\bibinfo  {journal} {Advances in Optics and Photonics}\ }\textbf {\bibinfo {volume} {12}},\ \bibinfo {pages} {1012} (\bibinfo {year} {2020})}\BibitemShut {NoStop}%
\bibitem [{\citenamefont {Zapatero}\ \emph {et~al.}(2023)\citenamefont {Zapatero}, \citenamefont {van Leent}, \citenamefont {Arnon-Friedman}, \citenamefont {Liu}, \citenamefont {Zhang}, \citenamefont {Weinfurter},\ and\ \citenamefont {Curty}}]{Zapatero2023}%
  \BibitemOpen
  \bibfield  {author} {\bibinfo {author} {\bibfnamefont {V.}~\bibnamefont {Zapatero}}, \bibinfo {author} {\bibfnamefont {T.}~\bibnamefont {van Leent}}, \bibinfo {author} {\bibfnamefont {R.}~\bibnamefont {Arnon-Friedman}}, \bibinfo {author} {\bibfnamefont {W.-Z.}\ \bibnamefont {Liu}}, \bibinfo {author} {\bibfnamefont {Q.}~\bibnamefont {Zhang}}, \bibinfo {author} {\bibfnamefont {H.}~\bibnamefont {Weinfurter}},\ and\ \bibinfo {author} {\bibfnamefont {M.}~\bibnamefont {Curty}},\ }\bibfield  {title} {\bibinfo {title} {{Advances in device-independent quantum key distribution}},\ }\href {https://doi.org/10.1038/s41534-023-00684-x} {\bibfield  {journal} {\bibinfo  {journal} {npj Quantum Information}\ }\textbf {\bibinfo {volume} {9}},\ \bibinfo {pages} {10} (\bibinfo {year} {2023})}\BibitemShut {NoStop}%
\bibitem [{\citenamefont {Buhrman}\ \emph {et~al.}(2010)\citenamefont {Buhrman}, \citenamefont {Cleve}, \citenamefont {Massar},\ and\ \citenamefont {de~Wolf}}]{Buhrman2010}%
  \BibitemOpen
  \bibfield  {author} {\bibinfo {author} {\bibfnamefont {H.}~\bibnamefont {Buhrman}}, \bibinfo {author} {\bibfnamefont {R.}~\bibnamefont {Cleve}}, \bibinfo {author} {\bibfnamefont {S.}~\bibnamefont {Massar}},\ and\ \bibinfo {author} {\bibfnamefont {R.}~\bibnamefont {de~Wolf}},\ }\bibfield  {title} {\bibinfo {title} {{Nonlocality and communication complexity}},\ }\href {https://doi.org/10.1103/RevModPhys.82.665} {\bibfield  {journal} {\bibinfo  {journal} {Rev. Mod. Phys.}\ }\textbf {\bibinfo {volume} {82}},\ \bibinfo {pages} {665} (\bibinfo {year} {2010})}\BibitemShut {NoStop}%
\bibitem [{\citenamefont {Pironio}\ \emph {et~al.}(2016)\citenamefont {Pironio}, \citenamefont {Scarani},\ and\ \citenamefont {Vidick}}]{Pironio_2016}%
  \BibitemOpen
  \bibfield  {author} {\bibinfo {author} {\bibfnamefont {S.}~\bibnamefont {Pironio}}, \bibinfo {author} {\bibfnamefont {V.}~\bibnamefont {Scarani}},\ and\ \bibinfo {author} {\bibfnamefont {T.}~\bibnamefont {Vidick}},\ }\bibfield  {title} {\bibinfo {title} {{Focus on device independent quantum information}},\ }\href {https://doi.org/10.1088/1367-2630/18/10/100202} {\bibfield  {journal} {\bibinfo  {journal} {New Journal of Physics}\ }\textbf {\bibinfo {volume} {18}},\ \bibinfo {pages} {100202} (\bibinfo {year} {2016})}\BibitemShut {NoStop}%
\bibitem [{\citenamefont {Cirel'son}(1980)}]{Tsirelson1980}%
  \BibitemOpen
  \bibfield  {author} {\bibinfo {author} {\bibfnamefont {B.~S.}\ \bibnamefont {Cirel'son}},\ }\bibfield  {title} {\bibinfo {title} {{Quantum generalizations of Bell's inequality}},\ }\href {https://doi.org/10.1007/BF00417500} {\bibfield  {journal} {\bibinfo  {journal} {{Letters in Mathematical Physics}}\ }\textbf {\bibinfo {volume} {4}},\ \bibinfo {pages} {93} (\bibinfo {year} {1980})}\BibitemShut {NoStop}%
\bibitem [{\citenamefont {Tsirelson}(1993)}]{Tsirelson1993}%
  \BibitemOpen
  \bibfield  {author} {\bibinfo {author} {\bibfnamefont {B.~S.}\ \bibnamefont {Tsirelson}},\ }\bibfield  {title} {\bibinfo {title} {{Some results and problems on quantum Bell-type inequalities}},\ }\href@noop {} {\bibfield  {journal} {\bibinfo  {journal} {Hadronic Journal Supplement}\ }\textbf {\bibinfo {volume} {8}},\ \bibinfo {pages} {329} (\bibinfo {year} {1993})}\BibitemShut {NoStop}%
\bibitem [{\citenamefont {Popescu}\ and\ \citenamefont {Rohrlich}(1994)}]{Popescu1994}%
  \BibitemOpen
  \bibfield  {author} {\bibinfo {author} {\bibfnamefont {S.}~\bibnamefont {Popescu}}\ and\ \bibinfo {author} {\bibfnamefont {D.}~\bibnamefont {Rohrlich}},\ }\bibfield  {title} {\bibinfo {title} {{Quantum nonlocality as an axiom}},\ }\href {https://doi.org/10.1007/BF02058098} {\bibfield  {journal} {\bibinfo  {journal} {Foundations of Physics}\ }\textbf {\bibinfo {volume} {24}},\ \bibinfo {pages} {379} (\bibinfo {year} {1994})}\BibitemShut {NoStop}%
\bibitem [{\citenamefont {Popescu}(2014)}]{Popescu2014}%
  \BibitemOpen
  \bibfield  {author} {\bibinfo {author} {\bibfnamefont {S.}~\bibnamefont {Popescu}},\ }\bibfield  {title} {\bibinfo {title} {{Nonlocality beyond quantum mechanics}},\ }\href {https://doi.org/10.1038/nphys2916} {\bibfield  {journal} {\bibinfo  {journal} {Nature Physics}\ }\textbf {\bibinfo {volume} {10}},\ \bibinfo {pages} {264} (\bibinfo {year} {2014})}\BibitemShut {NoStop}%
\bibitem [{\citenamefont {Svetlichny}(1987)}]{Svetlichny1987}%
  \BibitemOpen
  \bibfield  {author} {\bibinfo {author} {\bibfnamefont {G.}~\bibnamefont {Svetlichny}},\ }\bibfield  {title} {\bibinfo {title} {{Distinguishing three-body from two-body nonseparability by a Bell-type inequality}},\ }\href {https://doi.org/10.1103/PhysRevD.35.3066} {\bibfield  {journal} {\bibinfo  {journal} {Phys. Rev. D}\ }\textbf {\bibinfo {volume} {35}},\ \bibinfo {pages} {3066} (\bibinfo {year} {1987})}\BibitemShut {NoStop}%
\bibitem [{\citenamefont {Bancal}\ \emph {et~al.}(2012)\citenamefont {Bancal}, \citenamefont {Pironio}, \citenamefont {Acín}, \citenamefont {Liang}, \citenamefont {Scarani},\ and\ \citenamefont {Gisin}}]{Bancal2012}%
  \BibitemOpen
  \bibfield  {author} {\bibinfo {author} {\bibfnamefont {J.-D.}\ \bibnamefont {Bancal}}, \bibinfo {author} {\bibfnamefont {S.}~\bibnamefont {Pironio}}, \bibinfo {author} {\bibfnamefont {A.}~\bibnamefont {Acín}}, \bibinfo {author} {\bibfnamefont {Y.-C.}\ \bibnamefont {Liang}}, \bibinfo {author} {\bibfnamefont {V.}~\bibnamefont {Scarani}},\ and\ \bibinfo {author} {\bibfnamefont {N.}~\bibnamefont {Gisin}},\ }\bibfield  {title} {\bibinfo {title} {{Quantum non-locality based on finite-speed causal influences leads to superluminal signalling}},\ }\href {https://doi.org/10.1038/nphys2460} {\bibfield  {journal} {\bibinfo  {journal} {{Nature Physics}}\ }\textbf {\bibinfo {volume} {8}},\ \bibinfo {pages} {867} (\bibinfo {year} {2012})}\BibitemShut {NoStop}%
\bibitem [{\citenamefont {Cavalcanti}\ and\ \citenamefont {Wiseman}(2012)}]{Cavalcanti2012}%
  \BibitemOpen
  \bibfield  {author} {\bibinfo {author} {\bibfnamefont {E.~G.}\ \bibnamefont {Cavalcanti}}\ and\ \bibinfo {author} {\bibfnamefont {H.~M.}\ \bibnamefont {Wiseman}},\ }\bibfield  {title} {\bibinfo {title} {{ Bell Nonlocality, Signal Locality and Unpredictability (or What Bohr Could Have Told Einstein at Solvay Had He Known About Bell Experiments)}},\ }\href {https://doi.org/10.1007/s10701-012-9669-1} {\bibfield  {journal} {\bibinfo  {journal} {{Foundations of Physics}}\ }\textbf {\bibinfo {volume} {42}},\ \bibinfo {pages} {1329} (\bibinfo {year} {2012})}\BibitemShut {NoStop}%
\bibitem [{\citenamefont {Bancal}\ \emph {et~al.}(2013)\citenamefont {Bancal}, \citenamefont {Barrett}, \citenamefont {Gisin},\ and\ \citenamefont {Pironio}}]{Bancal2013}%
  \BibitemOpen
  \bibfield  {author} {\bibinfo {author} {\bibfnamefont {J.-D.}\ \bibnamefont {Bancal}}, \bibinfo {author} {\bibfnamefont {J.}~\bibnamefont {Barrett}}, \bibinfo {author} {\bibfnamefont {N.}~\bibnamefont {Gisin}},\ and\ \bibinfo {author} {\bibfnamefont {S.}~\bibnamefont {Pironio}},\ }\bibfield  {title} {\bibinfo {title} {Definitions of multipartite nonlocality},\ }\href {https://doi.org/10.1103/PhysRevA.88.014102} {\bibfield  {journal} {\bibinfo  {journal} {Phys. Rev. A}\ }\textbf {\bibinfo {volume} {88}},\ \bibinfo {pages} {014102} (\bibinfo {year} {2013})}\BibitemShut {NoStop}%
\bibitem [{\citenamefont {Barnea}\ \emph {et~al.}(2013)\citenamefont {Barnea}, \citenamefont {Bancal}, \citenamefont {Liang},\ and\ \citenamefont {Gisin}}]{Barnea2013}%
  \BibitemOpen
  \bibfield  {author} {\bibinfo {author} {\bibfnamefont {T.~J.}\ \bibnamefont {Barnea}}, \bibinfo {author} {\bibfnamefont {J.-D.}\ \bibnamefont {Bancal}}, \bibinfo {author} {\bibfnamefont {Y.-C.}\ \bibnamefont {Liang}},\ and\ \bibinfo {author} {\bibfnamefont {N.}~\bibnamefont {Gisin}},\ }\bibfield  {title} {\bibinfo {title} {Tripartite quantum state violating the hidden-influence constraints},\ }\href {https://doi.org/10.1103/PhysRevA.88.022123} {\bibfield  {journal} {\bibinfo  {journal} {Phys. Rev. A}\ }\textbf {\bibinfo {volume} {88}},\ \bibinfo {pages} {022123} (\bibinfo {year} {2013})}\BibitemShut {NoStop}%
\bibitem [{\citenamefont {Liang}\ \emph {et~al.}(2014)\citenamefont {Liang}, \citenamefont {Curchod}, \citenamefont {Bowles},\ and\ \citenamefont {Gisin}}]{Liang2020}%
  \BibitemOpen
  \bibfield  {author} {\bibinfo {author} {\bibfnamefont {Y.-C.}\ \bibnamefont {Liang}}, \bibinfo {author} {\bibfnamefont {F.~J.}\ \bibnamefont {Curchod}}, \bibinfo {author} {\bibfnamefont {J.}~\bibnamefont {Bowles}},\ and\ \bibinfo {author} {\bibfnamefont {N.}~\bibnamefont {Gisin}},\ }\bibfield  {title} {\bibinfo {title} {{Anonymous Quantum Nonlocality}},\ }\href {https://doi.org/10.1103/PhysRevLett.113.130401} {\bibfield  {journal} {\bibinfo  {journal} {Phys. Rev. Lett.}\ }\textbf {\bibinfo {volume} {113}},\ \bibinfo {pages} {130401} (\bibinfo {year} {2014})}\BibitemShut {NoStop}%
\bibitem [{\citenamefont {Branciard}\ \emph {et~al.}(2010)\citenamefont {Branciard}, \citenamefont {Gisin},\ and\ \citenamefont {Pironio}}]{Branciard2010}%
  \BibitemOpen
  \bibfield  {author} {\bibinfo {author} {\bibfnamefont {C.}~\bibnamefont {Branciard}}, \bibinfo {author} {\bibfnamefont {N.}~\bibnamefont {Gisin}},\ and\ \bibinfo {author} {\bibfnamefont {S.}~\bibnamefont {Pironio}},\ }\bibfield  {title} {\bibinfo {title} {{Characterizing the Nonlocal Correlations Created via Entanglement Swapping}},\ }\href {https://doi.org/10.1103/PhysRevLett.104.170401} {\bibfield  {journal} {\bibinfo  {journal} {Phys. Rev. Lett.}\ }\textbf {\bibinfo {volume} {104}},\ \bibinfo {pages} {170401} (\bibinfo {year} {2010})}\BibitemShut {NoStop}%
\bibitem [{\citenamefont {Fritz}(2012)}]{Fritz_2012}%
  \BibitemOpen
  \bibfield  {author} {\bibinfo {author} {\bibfnamefont {T.}~\bibnamefont {Fritz}},\ }\bibfield  {title} {\bibinfo {title} {{Beyond Bell's theorem: correlation scenarios}},\ }\href {https://doi.org/10.1088/1367-2630/14/10/103001} {\bibfield  {journal} {\bibinfo  {journal} {New Journal of Physics}\ }\textbf {\bibinfo {volume} {14}},\ \bibinfo {pages} {103001} (\bibinfo {year} {2012})}\BibitemShut {NoStop}%
\bibitem [{\citenamefont {Bong}\ \emph {et~al.}(2020)\citenamefont {Bong}, \citenamefont {Utreras-Alarcón}, \citenamefont {Ghafari}, \citenamefont {Liang}, \citenamefont {Tischler}, \citenamefont {Cavalcanti}, \citenamefont {Pryde},\ and\ \citenamefont {Wiseman}}]{Bong2020}%
  \BibitemOpen
  \bibfield  {author} {\bibinfo {author} {\bibfnamefont {K.-W.}\ \bibnamefont {Bong}}, \bibinfo {author} {\bibfnamefont {A.}~\bibnamefont {Utreras-Alarcón}}, \bibinfo {author} {\bibfnamefont {F.}~\bibnamefont {Ghafari}}, \bibinfo {author} {\bibfnamefont {Y.-C.}\ \bibnamefont {Liang}}, \bibinfo {author} {\bibfnamefont {N.}~\bibnamefont {Tischler}}, \bibinfo {author} {\bibfnamefont {E.~G.}\ \bibnamefont {Cavalcanti}}, \bibinfo {author} {\bibfnamefont {G.~J.}\ \bibnamefont {Pryde}},\ and\ \bibinfo {author} {\bibfnamefont {H.~M.}\ \bibnamefont {Wiseman}},\ }\bibfield  {title} {\bibinfo {title} {{A strong no-go theorem on the Wigner’s friend paradox}},\ }\href {https://doi.org/10.1038/s41567-020-0990-x} {\bibfield  {journal} {\bibinfo  {journal} {{Nature Physics}}\ }\textbf {\bibinfo {volume} {16}},\ \bibinfo {pages} {1199} (\bibinfo {year} {2020})}\BibitemShut {NoStop}%
\bibitem [{\citenamefont {Utreras-Alarcón}\ \emph {et~al.}(2024)\citenamefont {Utreras-Alarcón}, \citenamefont {Cavalcanti},\ and\ \citenamefont {Wiseman}}]{Utreras-Alarcon2024}%
  \BibitemOpen
  \bibfield  {author} {\bibinfo {author} {\bibfnamefont {A.}~\bibnamefont {Utreras-Alarcón}}, \bibinfo {author} {\bibfnamefont {E.~G.}\ \bibnamefont {Cavalcanti}},\ and\ \bibinfo {author} {\bibfnamefont {H.~M.}\ \bibnamefont {Wiseman}},\ }\bibfield  {title} {\bibinfo {title} {{Allowing Wigner’s friend to sequentially measure incompatible observables}},\ }\href {https://doi.org/10.1098/rspa.2024.0040} {\bibfield  {journal} {\bibinfo  {journal} {{Proceedings of the Royal Society A: Mathematical, Physical and Engineering Sciences}}\ }\textbf {\bibinfo {volume} {480}},\ \bibinfo {pages} {20240040} (\bibinfo {year} {2024})}\BibitemShut {NoStop}%
\bibitem [{\citenamefont {Yīng}\ \emph {et~al.}(2024)\citenamefont {Yīng}, \citenamefont {Ansanelli}, \citenamefont {Biagio}, \citenamefont {Wolfe}, \citenamefont {Schmid},\ and\ \citenamefont {Cavalcanti}}]{Ying2024relatingwigners}%
  \BibitemOpen
  \bibfield  {author} {\bibinfo {author} {\bibfnamefont {Y.}~\bibnamefont {Yīng}}, \bibinfo {author} {\bibfnamefont {M.~M.}\ \bibnamefont {Ansanelli}}, \bibinfo {author} {\bibfnamefont {A.~D.}\ \bibnamefont {Biagio}}, \bibinfo {author} {\bibfnamefont {E.}~\bibnamefont {Wolfe}}, \bibinfo {author} {\bibfnamefont {D.}~\bibnamefont {Schmid}},\ and\ \bibinfo {author} {\bibfnamefont {E.~G.}\ \bibnamefont {Cavalcanti}},\ }\bibfield  {title} {\bibinfo {title} {Relating {W}igner's {F}riend {S}cenarios to {N}onclassical {C}ausal {C}ompatibility, {M}onogamy {R}elations, and {F}ine {T}uning},\ }\href {https://doi.org/10.22331/q-2024-09-26-1485} {\bibfield  {journal} {\bibinfo  {journal} {{Quantum}}\ }\textbf {\bibinfo {volume} {8}},\ \bibinfo {pages} {1485} (\bibinfo {year} {2024})}\BibitemShut {NoStop}%
\bibitem [{\citenamefont {Haddara}\ and\ \citenamefont {Cavalcanti}(2025)}]{Haddara2025}%
  \BibitemOpen
  \bibfield  {author} {\bibinfo {author} {\bibfnamefont {M.}~\bibnamefont {Haddara}}\ and\ \bibinfo {author} {\bibfnamefont {E.~G.}\ \bibnamefont {Cavalcanti}},\ }\bibfield  {title} {\bibinfo {title} {{Local friendliness polytopes in multipartite scenarios}},\ }\href {https://doi.org/10.1103/PhysRevA.111.012206} {\bibfield  {journal} {\bibinfo  {journal} {Phys. Rev. A}\ }\textbf {\bibinfo {volume} {111}},\ \bibinfo {pages} {012206} (\bibinfo {year} {2025})}\BibitemShut {NoStop}%
\bibitem [{\citenamefont {Cavalcanti}(2021)}]{Cavalcanti2021Foundphys}%
  \BibitemOpen
  \bibfield  {author} {\bibinfo {author} {\bibfnamefont {E.~G.}\ \bibnamefont {Cavalcanti}},\ }\bibfield  {title} {\bibinfo {title} {{The View from a Wigner Bubble}},\ }\href {https://doi.org/10.1007/s10701-021-00417-0} {\bibfield  {journal} {\bibinfo  {journal} {{Foundations of Physics}}\ }\textbf {\bibinfo {volume} {51}},\ \bibinfo {pages} {39} (\bibinfo {year} {2021})}\BibitemShut {NoStop}%
\bibitem [{\citenamefont {Wiseman}\ \emph {et~al.}(2023)\citenamefont {Wiseman}, \citenamefont {Cavalcanti},\ and\ \citenamefont {Rieffel}}]{Wiseman2023thoughtfullocal}%
  \BibitemOpen
  \bibfield  {author} {\bibinfo {author} {\bibfnamefont {H.~M.}\ \bibnamefont {Wiseman}}, \bibinfo {author} {\bibfnamefont {E.~G.}\ \bibnamefont {Cavalcanti}},\ and\ \bibinfo {author} {\bibfnamefont {E.~G.}\ \bibnamefont {Rieffel}},\ }\bibfield  {title} {\bibinfo {title} {A ``thoughtful'' {L}ocal {F}riendliness no-go theorem: a prospective experiment with new assumptions to suit},\ }\href {https://doi.org/10.22331/q-2023-09-14-1112} {\bibfield  {journal} {\bibinfo  {journal} {{Quantum}}\ }\textbf {\bibinfo {volume} {7}},\ \bibinfo {pages} {1112} (\bibinfo {year} {2023})}\BibitemShut {NoStop}%
\bibitem [{\citenamefont {Woodhead}(2014)}]{Woodhead2014}%
  \BibitemOpen
  \bibfield  {author} {\bibinfo {author} {\bibfnamefont {E.}~\bibnamefont {Woodhead}},\ }\emph {\bibinfo {title} {{Imperfections and self testing in prepare-and-measure quantum key distribution}}},\ \href@noop {} {Ph.D. thesis},\ \bibinfo  {school} {{Laboratoire d’Information Quantique Universit{\'e} libre de Bruxelles}} (\bibinfo {year} {2014})\BibitemShut {NoStop}%
\bibitem [{\citenamefont {Fine}(1982{\natexlab{a}})}]{Fine1982}%
  \BibitemOpen
  \bibfield  {author} {\bibinfo {author} {\bibfnamefont {A.}~\bibnamefont {Fine}},\ }\bibfield  {title} {\bibinfo {title} {{Hidden Variables, Joint Probability, and the Bell Inequalities}},\ }\href {https://doi.org/10.1103/PhysRevLett.48.291} {\bibfield  {journal} {\bibinfo  {journal} {Phys. Rev. Lett.}\ }\textbf {\bibinfo {volume} {48}},\ \bibinfo {pages} {291} (\bibinfo {year} {1982}{\natexlab{a}})}\BibitemShut {NoStop}%
\bibitem [{\citenamefont {Fine}(1982{\natexlab{b}})}]{Fine1982b}%
  \BibitemOpen
  \bibfield  {author} {\bibinfo {author} {\bibfnamefont {A.}~\bibnamefont {Fine}},\ }\bibfield  {title} {\bibinfo {title} {{Joint distributions, quantum correlations, and commuting observables}},\ }\href {https://doi.org/10.1063/1.525514} {\bibfield  {journal} {\bibinfo  {journal} {Journal of Mathematical Physics}\ }\textbf {\bibinfo {volume} {23}},\ \bibinfo {pages} {1306} (\bibinfo {year} {1982}{\natexlab{b}})}\BibitemShut {NoStop}%
\bibitem [{\citenamefont {Bowles}\ \emph {et~al.}(2021)\citenamefont {Bowles}, \citenamefont {Hirsch},\ and\ \citenamefont {Cavalcanti}}]{Bowles2021}%
  \BibitemOpen
  \bibfield  {author} {\bibinfo {author} {\bibfnamefont {J.}~\bibnamefont {Bowles}}, \bibinfo {author} {\bibfnamefont {F.}~\bibnamefont {Hirsch}},\ and\ \bibinfo {author} {\bibfnamefont {D.}~\bibnamefont {Cavalcanti}},\ }\bibfield  {title} {\bibinfo {title} {{Single-copy activation of {B}ell nonlocality via broadcasting of quantum states}},\ }\href {https://doi.org/10.22331/q-2021-07-13-499} {\bibfield  {journal} {\bibinfo  {journal} {{Quantum}}\ }\textbf {\bibinfo {volume} {5}},\ \bibinfo {pages} {499} (\bibinfo {year} {2021})}\BibitemShut {NoStop}%
\bibitem [{\citenamefont {Boghiu}\ \emph {et~al.}(2023)\citenamefont {Boghiu}, \citenamefont {Hirsch}, \citenamefont {Lin}, \citenamefont {Quintino},\ and\ \citenamefont {Bowles}}]{Boghiu2023}%
  \BibitemOpen
  \bibfield  {author} {\bibinfo {author} {\bibfnamefont {E.-C.}\ \bibnamefont {Boghiu}}, \bibinfo {author} {\bibfnamefont {F.}~\bibnamefont {Hirsch}}, \bibinfo {author} {\bibfnamefont {P.-S.}\ \bibnamefont {Lin}}, \bibinfo {author} {\bibfnamefont {M.~T.}\ \bibnamefont {Quintino}},\ and\ \bibinfo {author} {\bibfnamefont {J.}~\bibnamefont {Bowles}},\ }\bibfield  {title} {\bibinfo {title} {{Device-independent and semi-device-independent entanglement certification in broadcast Bell scenarios}},\ }\href {https://doi.org/10.21468/SciPostPhysCore.6.2.028} {\bibfield  {journal} {\bibinfo  {journal} {SciPost Phys. Core}\ }\textbf {\bibinfo {volume} {6}},\ \bibinfo {pages} {028} (\bibinfo {year} {2023})}\BibitemShut {NoStop}%
\bibitem [{\citenamefont {Brukner}(2017)}]{Brukner2017}%
  \BibitemOpen
  \bibfield  {author} {\bibinfo {author} {\bibfnamefont {{\v{C}}.}~\bibnamefont {Brukner}},\ }\bibinfo {title} {{On the Quantum Measurement Problem}},\ in\ \href {https://doi.org/10.1007/978-3-319-38987-5_5} {\emph {\bibinfo {booktitle} {Quantum [Un]Speakables II: Half a Century of Bell's Theorem}}},\ \bibinfo {editor} {edited by\ \bibinfo {editor} {\bibfnamefont {R.}~\bibnamefont {Bertlmann}}\ and\ \bibinfo {editor} {\bibfnamefont {A.}~\bibnamefont {Zeilinger}}}\ (\bibinfo  {publisher} {Springer International Publishing},\ \bibinfo {address} {Cham},\ \bibinfo {year} {2017})\ pp.\ \bibinfo {pages} {95--117}\BibitemShut {NoStop}%
\bibitem [{\citenamefont {Brukner}(2018)}]{Brukner2018}%
  \BibitemOpen
  \bibfield  {author} {\bibinfo {author} {\bibfnamefont {{\v C}.}~\bibnamefont {Brukner}},\ }\bibfield  {title} {\bibinfo {title} {{A No-Go Theorem for Observer-Independent Facts}},\ }\href {https://doi.org/10.3390/e20050350} {\bibfield  {journal} {\bibinfo  {journal} {{Entropy}}\ }\textbf {\bibinfo {volume} {20}},\ \bibinfo {pages} {350} (\bibinfo {year} {2018})}\BibitemShut {NoStop}%
\bibitem [{\citenamefont {Werner}\ and\ \citenamefont {Wolf}(2001)}]{Werner2001}%
  \BibitemOpen
  \bibfield  {author} {\bibinfo {author} {\bibfnamefont {R.~F.}\ \bibnamefont {Werner}}\ and\ \bibinfo {author} {\bibfnamefont {M.~M.}\ \bibnamefont {Wolf}},\ }\bibfield  {title} {\bibinfo {title} {{Bell Inequalities and Entanglement}},\ }\href@noop {} {\bibfield  {journal} {\bibinfo  {journal} {Quantum Info. Comput.}\ }\textbf {\bibinfo {volume} {1}},\ \bibinfo {pages} {1–25} (\bibinfo {year} {2001})}\BibitemShut {NoStop}%
\bibitem [{\citenamefont {Abramsky}\ and\ \citenamefont {Brandenburger}(2011)}]{Abramsky2011}%
  \BibitemOpen
  \bibfield  {author} {\bibinfo {author} {\bibfnamefont {S.}~\bibnamefont {Abramsky}}\ and\ \bibinfo {author} {\bibfnamefont {A.}~\bibnamefont {Brandenburger}},\ }\bibfield  {title} {\bibinfo {title} {The sheaf-theoretic structure of non-locality and contextuality},\ }\href {https://doi.org/10.1088/1367-2630/13/11/113036} {\bibfield  {journal} {\bibinfo  {journal} {New Journal of Physics}\ }\textbf {\bibinfo {volume} {13}},\ \bibinfo {pages} {113036} (\bibinfo {year} {2011})}\BibitemShut {NoStop}%
\bibitem [{\citenamefont {Clauser}\ and\ \citenamefont {Horne}(1974)}]{Clauser1974}%
  \BibitemOpen
  \bibfield  {author} {\bibinfo {author} {\bibfnamefont {J.~F.}\ \bibnamefont {Clauser}}\ and\ \bibinfo {author} {\bibfnamefont {M.~A.}\ \bibnamefont {Horne}},\ }\bibfield  {title} {\bibinfo {title} {{Experimental consequences of objective local theories}},\ }\href {https://doi.org/10.1103/PhysRevD.10.526} {\bibfield  {journal} {\bibinfo  {journal} {Phys. Rev. D}\ }\textbf {\bibinfo {volume} {10}},\ \bibinfo {pages} {526} (\bibinfo {year} {1974})}\BibitemShut {NoStop}%
\bibitem [{\citenamefont {Garg}\ and\ \citenamefont {Mermin}(1982{\natexlab{a}})}]{Garg1982}%
  \BibitemOpen
  \bibfield  {author} {\bibinfo {author} {\bibfnamefont {A.}~\bibnamefont {Garg}}\ and\ \bibinfo {author} {\bibfnamefont {N.~D.}\ \bibnamefont {Mermin}},\ }\bibfield  {title} {\bibinfo {title} {{Comment on "Hidden Variables, Joint Probability, and the Bell Inequalities"}},\ }\href {https://doi.org/10.1103/PhysRevLett.49.242} {\bibfield  {journal} {\bibinfo  {journal} {Phys. Rev. Lett.}\ }\textbf {\bibinfo {volume} {49}},\ \bibinfo {pages} {242} (\bibinfo {year} {1982}{\natexlab{a}})}\BibitemShut {NoStop}%
\bibitem [{\citenamefont {Fine}(1982{\natexlab{c}})}]{Fine1982Response}%
  \BibitemOpen
  \bibfield  {author} {\bibinfo {author} {\bibfnamefont {A.}~\bibnamefont {Fine}},\ }\bibfield  {title} {\bibinfo {title} {{Fine Responds}},\ }\href {https://doi.org/10.1103/PhysRevLett.49.243} {\bibfield  {journal} {\bibinfo  {journal} {Phys. Rev. Lett.}\ }\textbf {\bibinfo {volume} {49}},\ \bibinfo {pages} {243} (\bibinfo {year} {1982}{\natexlab{c}})}\BibitemShut {NoStop}%
\bibitem [{\citenamefont {Garg}\ and\ \citenamefont {Mermin}(1982{\natexlab{b}})}]{Garg1982b}%
  \BibitemOpen
  \bibfield  {author} {\bibinfo {author} {\bibfnamefont {A.}~\bibnamefont {Garg}}\ and\ \bibinfo {author} {\bibfnamefont {N.~D.}\ \bibnamefont {Mermin}},\ }\bibfield  {title} {\bibinfo {title} {{Correlation Inequalities and Hidden Variables}},\ }\href {https://doi.org/10.1103/PhysRevLett.49.1220} {\bibfield  {journal} {\bibinfo  {journal} {Phys. Rev. Lett.}\ }\textbf {\bibinfo {volume} {49}},\ \bibinfo {pages} {1220} (\bibinfo {year} {1982}{\natexlab{b}})}\BibitemShut {NoStop}%
\bibitem [{\citenamefont {Br{\o}ndsted}(1983)}]{Brondsted1983}%
  \BibitemOpen
  \bibfield  {author} {\bibinfo {author} {\bibfnamefont {A.}~\bibnamefont {Br{\o}ndsted}},\ }\bibinfo {title} {{Convex Polytopes}},\ in\ \href {https://doi.org/10.1007/978-1-4612-1148-8_3} {\emph {\bibinfo {booktitle} {{An Introduction to Convex Polytopes}}}},\ \bibinfo {editor} {edited by\ \bibinfo {editor} {\bibfnamefont {A.}~\bibnamefont {Br{\o}ndsted}}}\ (\bibinfo  {publisher} {Springer},\ \bibinfo {address} {New York},\ \bibinfo {year} {1983})\ pp.\ \bibinfo {pages} {44--97}\BibitemShut {NoStop}%
\bibitem [{\citenamefont {Gr{\"u}nbaum}(2003)}]{Grünbaum2003}%
  \BibitemOpen
  \bibfield  {author} {\bibinfo {author} {\bibfnamefont {B.}~\bibnamefont {Gr{\"u}nbaum}},\ }\bibinfo {title} {Polytopes},\ in\ \href {https://doi.org/10.1007/978-1-4613-0019-9_3} {\emph {\bibinfo {booktitle} {Convex Polytopes}}},\ \bibinfo {editor} {edited by\ \bibinfo {editor} {\bibfnamefont {V.}~\bibnamefont {Kaibel}}, \bibinfo {editor} {\bibfnamefont {V.}~\bibnamefont {Klee}},\ and\ \bibinfo {editor} {\bibfnamefont {G.~M.}\ \bibnamefont {Ziegler}}}\ (\bibinfo  {publisher} {Springer},\ \bibinfo {address} {New York},\ \bibinfo {year} {2003})\ pp.\ \bibinfo {pages} {35--60}\BibitemShut {NoStop}%
\bibitem [{\citenamefont {Avis}\ \emph {et~al.}(1997)\citenamefont {Avis}, \citenamefont {Bremner},\ and\ \citenamefont {Seidel}}]{Avis1997}%
  \BibitemOpen
  \bibfield  {author} {\bibinfo {author} {\bibfnamefont {D.}~\bibnamefont {Avis}}, \bibinfo {author} {\bibfnamefont {D.}~\bibnamefont {Bremner}},\ and\ \bibinfo {author} {\bibfnamefont {R.}~\bibnamefont {Seidel}},\ }\bibfield  {title} {\bibinfo {title} {{How good are convex hull algorithms?}},\ }\href {https://doi.org/https://doi.org/10.1016/S0925-7721(96)00023-5} {\bibfield  {journal} {\bibinfo  {journal} {Computational Geometry}\ }\textbf {\bibinfo {volume} {7}},\ \bibinfo {pages} {265} (\bibinfo {year} {1997})}\BibitemShut {NoStop}%
\bibitem [{\citenamefont {Pitowsky}(1989)}]{Pitowsky1989}%
  \BibitemOpen
  \bibfield  {author} {\bibinfo {author} {\bibfnamefont {I.}~\bibnamefont {Pitowsky}},\ }\bibinfo {title} {Classical correlation polytopes and propositional logic},\ in\ \href {https://doi.org/10.1007/BFb0021188} {\emph {\bibinfo {booktitle} {Quantum Probability --- Quantum Logic}}}\ (\bibinfo  {publisher} {Springer},\ \bibinfo {address} {Berlin, Heidelberg},\ \bibinfo {year} {1989})\ pp.\ \bibinfo {pages} {11--51}\BibitemShut {NoStop}%
\bibitem [{\citenamefont {Barrett}\ \emph {et~al.}(2005)\citenamefont {Barrett}, \citenamefont {Linden}, \citenamefont {Massar}, \citenamefont {Pironio}, \citenamefont {Popescu},\ and\ \citenamefont {Roberts}}]{Barret2005}%
  \BibitemOpen
  \bibfield  {author} {\bibinfo {author} {\bibfnamefont {J.}~\bibnamefont {Barrett}}, \bibinfo {author} {\bibfnamefont {N.}~\bibnamefont {Linden}}, \bibinfo {author} {\bibfnamefont {S.}~\bibnamefont {Massar}}, \bibinfo {author} {\bibfnamefont {S.}~\bibnamefont {Pironio}}, \bibinfo {author} {\bibfnamefont {S.}~\bibnamefont {Popescu}},\ and\ \bibinfo {author} {\bibfnamefont {D.}~\bibnamefont {Roberts}},\ }\bibfield  {title} {\bibinfo {title} {{Nonlocal correlations as an information-theoretic resource}},\ }\href {https://doi.org/10.1103/PhysRevA.71.022101} {\bibfield  {journal} {\bibinfo  {journal} {Phys. Rev. A}\ }\textbf {\bibinfo {volume} {71}},\ \bibinfo {pages} {022101} (\bibinfo {year} {2005})}\BibitemShut {NoStop}%
\bibitem [{\citenamefont {Jones}\ and\ \citenamefont {Masanes}(2005)}]{Jones2005}%
  \BibitemOpen
  \bibfield  {author} {\bibinfo {author} {\bibfnamefont {N.~S.}\ \bibnamefont {Jones}}\ and\ \bibinfo {author} {\bibfnamefont {L.}~\bibnamefont {Masanes}},\ }\bibfield  {title} {\bibinfo {title} {{Interconversion of nonlocal correlations}},\ }\href {https://doi.org/10.1103/PhysRevA.72.052312} {\bibfield  {journal} {\bibinfo  {journal} {Phys. Rev. A}\ }\textbf {\bibinfo {volume} {72}},\ \bibinfo {pages} {052312} (\bibinfo {year} {2005})}\BibitemShut {NoStop}%
\bibitem [{\citenamefont {Pironio}\ \emph {et~al.}(2011)\citenamefont {Pironio}, \citenamefont {Bancal},\ and\ \citenamefont {Scarani}}]{Pironio2011}%
  \BibitemOpen
  \bibfield  {author} {\bibinfo {author} {\bibfnamefont {S.}~\bibnamefont {Pironio}}, \bibinfo {author} {\bibfnamefont {J.-D.}\ \bibnamefont {Bancal}},\ and\ \bibinfo {author} {\bibfnamefont {V.}~\bibnamefont {Scarani}},\ }\bibfield  {title} {\bibinfo {title} {Extremal correlations of the tripartite no-signaling polytope},\ }\href {https://doi.org/10.1088/1751-8113/44/6/065303} {\bibfield  {journal} {\bibinfo  {journal} {Journal of Physics A: Mathematical and Theoretical}\ }\textbf {\bibinfo {volume} {44}},\ \bibinfo {pages} {065303} (\bibinfo {year} {2011})}\BibitemShut {NoStop}%
\bibitem [{\citenamefont {Pironio}(2005)}]{Pironio2005}%
  \BibitemOpen
  \bibfield  {author} {\bibinfo {author} {\bibfnamefont {S.}~\bibnamefont {Pironio}},\ }\bibfield  {title} {\bibinfo {title} {{Lifting Bell inequalities}},\ }\href {https://doi.org/10.1063/1.1928727} {\bibfield  {journal} {\bibinfo  {journal} {Journal of Mathematical Physics}\ }\textbf {\bibinfo {volume} {46}},\ \bibinfo {pages} {062112} (\bibinfo {year} {2005})}\BibitemShut {NoStop}%
\bibitem [{\citenamefont {Goh}\ \emph {et~al.}(2018)\citenamefont {Goh}, \citenamefont {Kaniewski}, \citenamefont {Wolfe}, \citenamefont {V\'ertesi}, \citenamefont {Wu}, \citenamefont {Cai}, \citenamefont {Liang},\ and\ \citenamefont {Scarani}}]{Goh2017}%
  \BibitemOpen
  \bibfield  {author} {\bibinfo {author} {\bibfnamefont {K.~T.}\ \bibnamefont {Goh}}, \bibinfo {author} {\bibfnamefont {J.}~\bibnamefont {Kaniewski}}, \bibinfo {author} {\bibfnamefont {E.}~\bibnamefont {Wolfe}}, \bibinfo {author} {\bibfnamefont {T.}~\bibnamefont {V\'ertesi}}, \bibinfo {author} {\bibfnamefont {X.}~\bibnamefont {Wu}}, \bibinfo {author} {\bibfnamefont {Y.}~\bibnamefont {Cai}}, \bibinfo {author} {\bibfnamefont {Y.-C.}\ \bibnamefont {Liang}},\ and\ \bibinfo {author} {\bibfnamefont {V.}~\bibnamefont {Scarani}},\ }\bibfield  {title} {\bibinfo {title} {{Geometry of the set of quantum correlations}},\ }\href {https://doi.org/10.1103/PhysRevA.97.022104} {\bibfield  {journal} {\bibinfo  {journal} {Phys. Rev. A}\ }\textbf {\bibinfo {volume} {97}},\ \bibinfo {pages} {022104} (\bibinfo {year} {2018})}\BibitemShut {NoStop}%
\bibitem [{\citenamefont {Le}\ \emph {et~al.}(2023)\citenamefont {Le}, \citenamefont {Meroni}, \citenamefont {Sturmfels}, \citenamefont {Werner},\ and\ \citenamefont {Ziegler}}]{Le2023quantumcorrelations}%
  \BibitemOpen
  \bibfield  {author} {\bibinfo {author} {\bibfnamefont {T.~P.}\ \bibnamefont {Le}}, \bibinfo {author} {\bibfnamefont {C.}~\bibnamefont {Meroni}}, \bibinfo {author} {\bibfnamefont {B.}~\bibnamefont {Sturmfels}}, \bibinfo {author} {\bibfnamefont {R.~F.}\ \bibnamefont {Werner}},\ and\ \bibinfo {author} {\bibfnamefont {T.}~\bibnamefont {Ziegler}},\ }\bibfield  {title} {\bibinfo {title} {Quantum {C}orrelations in the {M}inimal {S}cenario},\ }\href {https://doi.org/10.22331/q-2023-03-16-947} {\bibfield  {journal} {\bibinfo  {journal} {{Quantum}}\ }\textbf {\bibinfo {volume} {7}},\ \bibinfo {pages} {947} (\bibinfo {year} {2023})}\BibitemShut {NoStop}%
\bibitem [{\citenamefont {{Barizien, Victor and Bancal, Jean-Daniel}}(2025)}]{Barizien-Bancal2025}%
  \BibitemOpen
  \bibfield  {author} {\bibinfo {author} {\bibnamefont {{Barizien, Victor and Bancal, Jean-Daniel}}},\ }\bibfield  {title} {\bibinfo {title} {{Quantum statistics in the minimal Bell scenario}},\ }\href {https://doi.org/10.1038/s41567-025-02782-3} {\bibfield  {journal} {\bibinfo  {journal} {Nature Physics}\ }\textbf {\bibinfo {volume} {21}},\ \bibinfo {pages} {577} (\bibinfo {year} {2025})}\BibitemShut {NoStop}%
\bibitem [{\citenamefont {Clauser}\ \emph {et~al.}(1969)\citenamefont {Clauser}, \citenamefont {Horne}, \citenamefont {Shimony},\ and\ \citenamefont {Holt}}]{Clauser1969}%
  \BibitemOpen
  \bibfield  {author} {\bibinfo {author} {\bibfnamefont {J.~F.}\ \bibnamefont {Clauser}}, \bibinfo {author} {\bibfnamefont {M.~A.}\ \bibnamefont {Horne}}, \bibinfo {author} {\bibfnamefont {A.}~\bibnamefont {Shimony}},\ and\ \bibinfo {author} {\bibfnamefont {R.~A.}\ \bibnamefont {Holt}},\ }\bibfield  {title} {\bibinfo {title} {{Proposed Experiment to Test Local Hidden-Variable Theories}},\ }\href {https://doi.org/10.1103/PhysRevLett.23.880} {\bibfield  {journal} {\bibinfo  {journal} {Phys. Rev. Lett.}\ }\textbf {\bibinfo {volume} {23}},\ \bibinfo {pages} {880} (\bibinfo {year} {1969})}\BibitemShut {NoStop}%
\bibitem [{\citenamefont {Czechlewski}\ and\ \citenamefont {Paw\l{}owski}(2018)}]{Czechlewski2018}%
  \BibitemOpen
  \bibfield  {author} {\bibinfo {author} {\bibfnamefont {M.}~\bibnamefont {Czechlewski}}\ and\ \bibinfo {author} {\bibfnamefont {M.}~\bibnamefont {Paw\l{}owski}},\ }\bibfield  {title} {\bibinfo {title} {{Influence of the choice of postprocessing method on Bell inequalities}},\ }\href {https://doi.org/10.1103/PhysRevA.97.062123} {\bibfield  {journal} {\bibinfo  {journal} {Phys. Rev. A}\ }\textbf {\bibinfo {volume} {97}},\ \bibinfo {pages} {062123} (\bibinfo {year} {2018})}\BibitemShut {NoStop}%
\bibitem [{\citenamefont {Ac{\'\i}n}\ \emph {et~al.}(2015)\citenamefont {Ac{\'\i}n}, \citenamefont {Fritz}, \citenamefont {Leverrier},\ and\ \citenamefont {Sainz}}]{Acin2015}%
  \BibitemOpen
  \bibfield  {author} {\bibinfo {author} {\bibfnamefont {A.}~\bibnamefont {Ac{\'\i}n}}, \bibinfo {author} {\bibfnamefont {T.}~\bibnamefont {Fritz}}, \bibinfo {author} {\bibfnamefont {A.}~\bibnamefont {Leverrier}},\ and\ \bibinfo {author} {\bibfnamefont {A.~B.}\ \bibnamefont {Sainz}},\ }\bibfield  {title} {\bibinfo {title} {{A Combinatorial Approach to Nonlocality and Contextuality}},\ }\href {https://doi.org/10.1007/s00220-014-2260-1} {\bibfield  {journal} {\bibinfo  {journal} {Communications in Mathematical Physics}\ }\textbf {\bibinfo {volume} {334}},\ \bibinfo {pages} {533} (\bibinfo {year} {2015})}\BibitemShut {NoStop}%
\bibitem [{\citenamefont {Mazzari}\ \emph {et~al.}(2023)\citenamefont {Mazzari}, \citenamefont {Ruffolo}, \citenamefont {Vieira}, \citenamefont {Temistocles}, \citenamefont {Rabelo},\ and\ \citenamefont {Terra~Cunha}}]{Mazzari2023}%
  \BibitemOpen
  \bibfield  {author} {\bibinfo {author} {\bibfnamefont {A.}~\bibnamefont {Mazzari}}, \bibinfo {author} {\bibfnamefont {G.}~\bibnamefont {Ruffolo}}, \bibinfo {author} {\bibfnamefont {C.}~\bibnamefont {Vieira}}, \bibinfo {author} {\bibfnamefont {T.}~\bibnamefont {Temistocles}}, \bibinfo {author} {\bibfnamefont {R.}~\bibnamefont {Rabelo}},\ and\ \bibinfo {author} {\bibfnamefont {M.}~\bibnamefont {Terra~Cunha}},\ }\bibfield  {title} {\bibinfo {title} {{Generalized Bell Scenarios: Disturbing Consequences on Local-Hidden-Variable Models}},\ }\href {https://www.mdpi.com/1099-4300/25/9/1276} {\bibfield  {journal} {\bibinfo  {journal} {Entropy}\ }\textbf {\bibinfo {volume} {25}} (\bibinfo {year} {2023})}\BibitemShut {NoStop}%
\bibitem [{\citenamefont {Rodari}\ \emph {et~al.}(2023)\citenamefont {Rodari}, \citenamefont {Poderini}, \citenamefont {Polino}, \citenamefont {Suprano}, \citenamefont {Sciarrino},\ and\ \citenamefont {Chaves}}]{rodari2023characterizinghybridcausalstructures}%
  \BibitemOpen
  \bibfield  {author} {\bibinfo {author} {\bibfnamefont {G.}~\bibnamefont {Rodari}}, \bibinfo {author} {\bibfnamefont {D.}~\bibnamefont {Poderini}}, \bibinfo {author} {\bibfnamefont {E.}~\bibnamefont {Polino}}, \bibinfo {author} {\bibfnamefont {A.}~\bibnamefont {Suprano}}, \bibinfo {author} {\bibfnamefont {F.}~\bibnamefont {Sciarrino}},\ and\ \bibinfo {author} {\bibfnamefont {R.}~\bibnamefont {Chaves}},\ }\href {https://arxiv.org/abs/2401.00063} {\bibinfo {title} {{Characterizing Hybrid Causal Structures with the Exclusivity Graph Approach}}} (\bibinfo {year} {2023}),\ \Eprint {https://arxiv.org/abs/2401.00063} {arXiv:2401.00063 [quant-ph]} \BibitemShut {NoStop}%
\bibitem [{\citenamefont {Vorob’ev}(1962)}]{Vorobeev1962}%
  \BibitemOpen
  \bibfield  {author} {\bibinfo {author} {\bibfnamefont {N.~N.}\ \bibnamefont {Vorob’ev}},\ }\bibfield  {title} {\bibinfo {title} {{Consistent Families of Measures and Their Extensions}},\ }\href {https://doi.org/10.1137/1107014} {\bibfield  {journal} {\bibinfo  {journal} {Theory of Probability \& Its Applications}\ }\textbf {\bibinfo {volume} {7}},\ \bibinfo {pages} {147} (\bibinfo {year} {1962})}\BibitemShut {NoStop}%
\bibitem [{\citenamefont {Wolf}\ \emph {et~al.}(2009)\citenamefont {Wolf}, \citenamefont {Perez-Garcia},\ and\ \citenamefont {Fernandez}}]{Wolf2009}%
  \BibitemOpen
  \bibfield  {author} {\bibinfo {author} {\bibfnamefont {M.~M.}\ \bibnamefont {Wolf}}, \bibinfo {author} {\bibfnamefont {D.}~\bibnamefont {Perez-Garcia}},\ and\ \bibinfo {author} {\bibfnamefont {C.}~\bibnamefont {Fernandez}},\ }\bibfield  {title} {\bibinfo {title} {{Measurements Incompatible in Quantum Theory Cannot Be Measured Jointly in Any Other No-Signaling Theory}},\ }\href {https://doi.org/10.1103/PhysRevLett.103.230402} {\bibfield  {journal} {\bibinfo  {journal} {Phys. Rev. Lett.}\ }\textbf {\bibinfo {volume} {103}},\ \bibinfo {pages} {230402} (\bibinfo {year} {2009})}\BibitemShut {NoStop}%
\bibitem [{\citenamefont {Quintino}\ \emph {et~al.}(2019)\citenamefont {Quintino}, \citenamefont {Budroni}, \citenamefont {Woodhead}, \citenamefont {Cabello},\ and\ \citenamefont {Cavalcanti}}]{Quintino2019}%
  \BibitemOpen
  \bibfield  {author} {\bibinfo {author} {\bibfnamefont {M.~T.}\ \bibnamefont {Quintino}}, \bibinfo {author} {\bibfnamefont {C.}~\bibnamefont {Budroni}}, \bibinfo {author} {\bibfnamefont {E.}~\bibnamefont {Woodhead}}, \bibinfo {author} {\bibfnamefont {A.}~\bibnamefont {Cabello}},\ and\ \bibinfo {author} {\bibfnamefont {D.}~\bibnamefont {Cavalcanti}},\ }\bibfield  {title} {\bibinfo {title} {{Device-Independent Tests of Structures of Measurement Incompatibility}},\ }\href {https://doi.org/10.1103/PhysRevLett.123.180401} {\bibfield  {journal} {\bibinfo  {journal} {Phys. Rev. Lett.}\ }\textbf {\bibinfo {volume} {123}},\ \bibinfo {pages} {180401} (\bibinfo {year} {2019})}\BibitemShut {NoStop}%
\bibitem [{\citenamefont {Santos}\ and\ \citenamefont {Amaral}(2021)}]{Santos2021}%
  \BibitemOpen
  \bibfield  {author} {\bibinfo {author} {\bibfnamefont {L.}~\bibnamefont {Santos}}\ and\ \bibinfo {author} {\bibfnamefont {B.}~\bibnamefont {Amaral}},\ }\bibfield  {title} {\bibinfo {title} {{Conditions for logical contextuality and nonlocality}},\ }\href {https://doi.org/10.1103/PhysRevA.104.022201} {\bibfield  {journal} {\bibinfo  {journal} {{Phys. Rev. A}}\ }\textbf {\bibinfo {volume} {104}},\ \bibinfo {pages} {022201} (\bibinfo {year} {2021})}\BibitemShut {NoStop}%
\bibitem [{\citenamefont {Budroni}\ \emph {et~al.}(2022)\citenamefont {Budroni}, \citenamefont {Cabello}, \citenamefont {G\"uhne}, \citenamefont {Kleinmann},\ and\ \citenamefont {Larsson}}]{BudroniContextualityreview2022}%
  \BibitemOpen
  \bibfield  {author} {\bibinfo {author} {\bibfnamefont {C.}~\bibnamefont {Budroni}}, \bibinfo {author} {\bibfnamefont {A.}~\bibnamefont {Cabello}}, \bibinfo {author} {\bibfnamefont {O.}~\bibnamefont {G\"uhne}}, \bibinfo {author} {\bibfnamefont {M.}~\bibnamefont {Kleinmann}},\ and\ \bibinfo {author} {\bibfnamefont {J.-A.}\ \bibnamefont {Larsson}},\ }\bibfield  {title} {\bibinfo {title} {{Kochen-Specker contextuality}},\ }\href {https://doi.org/10.1103/RevModPhys.94.045007} {\bibfield  {journal} {\bibinfo  {journal} {Rev. Mod. Phys.}\ }\textbf {\bibinfo {volume} {94}},\ \bibinfo {pages} {045007} (\bibinfo {year} {2022})}\BibitemShut {NoStop}%
\bibitem [{\citenamefont {Zhang}\ \emph {et~al.}(2025)\citenamefont {Zhang}, \citenamefont {Li}, \citenamefont {Hu}, \citenamefont {Xiang}, \citenamefont {Li}, \citenamefont {Guo}, \citenamefont {Tura}, \citenamefont {Gong}, \citenamefont {He},\ and\ \citenamefont {Liu}}]{Zhang2025}%
  \BibitemOpen
  \bibfield  {author} {\bibinfo {author} {\bibfnamefont {C.}~\bibnamefont {Zhang}}, \bibinfo {author} {\bibfnamefont {Y.}~\bibnamefont {Li}}, \bibinfo {author} {\bibfnamefont {X.-M.}\ \bibnamefont {Hu}}, \bibinfo {author} {\bibfnamefont {Y.}~\bibnamefont {Xiang}}, \bibinfo {author} {\bibfnamefont {C.-F.}\ \bibnamefont {Li}}, \bibinfo {author} {\bibfnamefont {G.-C.}\ \bibnamefont {Guo}}, \bibinfo {author} {\bibfnamefont {J.}~\bibnamefont {Tura}}, \bibinfo {author} {\bibfnamefont {Q.}~\bibnamefont {Gong}}, \bibinfo {author} {\bibfnamefont {Q.}~\bibnamefont {He}},\ and\ \bibinfo {author} {\bibfnamefont {B.-H.}\ \bibnamefont {Liu}},\ }\bibfield  {title} {\bibinfo {title} {{Randomness versus Nonlocality in Multiple-Input and Multiple-Output Quantum Scenario}},\ }\href {https://doi.org/10.1103/PhysRevLett.134.090201} {\bibfield  {journal} {\bibinfo  {journal} {Phys. Rev. Lett.}\ }\textbf {\bibinfo {volume} {134}},\ \bibinfo {pages} {090201} (\bibinfo {year} {2025})}\BibitemShut {NoStop}%
\bibitem [{\citenamefont {Ramanathan}\ \emph {et~al.}(2024)\citenamefont {Ramanathan}, \citenamefont {Liu},\ and\ \citenamefont {Pironio}}]{ramanathan2024maximumquantumnonlocalitysufficient}%
  \BibitemOpen
  \bibfield  {author} {\bibinfo {author} {\bibfnamefont {R.}~\bibnamefont {Ramanathan}}, \bibinfo {author} {\bibfnamefont {Y.}~\bibnamefont {Liu}},\ and\ \bibinfo {author} {\bibfnamefont {S.}~\bibnamefont {Pironio}},\ }\href {https://arxiv.org/abs/2408.03665} {\bibinfo {title} {{Maximum Quantum Non-Locality is not always Sufficient for Device-Independent Randomness Generation}}} (\bibinfo {year} {2024}),\ \Eprint {https://arxiv.org/abs/2408.03665} {arXiv:2408.03665 [quant-ph]} \BibitemShut {NoStop}%
\bibitem [{\citenamefont {Barrett}\ \emph {et~al.}(2006)\citenamefont {Barrett}, \citenamefont {Kent},\ and\ \citenamefont {Pironio}}]{Barret2006}%
  \BibitemOpen
  \bibfield  {author} {\bibinfo {author} {\bibfnamefont {J.}~\bibnamefont {Barrett}}, \bibinfo {author} {\bibfnamefont {A.}~\bibnamefont {Kent}},\ and\ \bibinfo {author} {\bibfnamefont {S.}~\bibnamefont {Pironio}},\ }\bibfield  {title} {\bibinfo {title} {{Maximally Nonlocal and Monogamous Quantum Correlations}},\ }\href {https://doi.org/10.1103/PhysRevLett.97.170409} {\bibfield  {journal} {\bibinfo  {journal} {Phys. Rev. Lett.}\ }\textbf {\bibinfo {volume} {97}},\ \bibinfo {pages} {170409} (\bibinfo {year} {2006})}\BibitemShut {NoStop}%
\bibitem [{\citenamefont {Horodecki}\ \emph {et~al.}(2009)\citenamefont {Horodecki}, \citenamefont {Horodecki}, \citenamefont {Horodecki},\ and\ \citenamefont {Horodecki}}]{Horodecki2009}%
  \BibitemOpen
  \bibfield  {author} {\bibinfo {author} {\bibfnamefont {R.}~\bibnamefont {Horodecki}}, \bibinfo {author} {\bibfnamefont {P.}~\bibnamefont {Horodecki}}, \bibinfo {author} {\bibfnamefont {M.}~\bibnamefont {Horodecki}},\ and\ \bibinfo {author} {\bibfnamefont {K.}~\bibnamefont {Horodecki}},\ }\bibfield  {title} {\bibinfo {title} {{Quantum entanglement}},\ }\href {https://doi.org/10.1103/RevModPhys.81.865} {\bibfield  {journal} {\bibinfo  {journal} {Rev. Mod. Phys.}\ }\textbf {\bibinfo {volume} {81}},\ \bibinfo {pages} {865} (\bibinfo {year} {2009})}\BibitemShut {NoStop}%
\bibitem [{\citenamefont {Werner}(1989)}]{Werner1989}%
  \BibitemOpen
  \bibfield  {author} {\bibinfo {author} {\bibfnamefont {R.~F.}\ \bibnamefont {Werner}},\ }\bibfield  {title} {\bibinfo {title} {{Quantum states with Einstein-Podolsky-Rosen correlations admitting a hidden-variable model}},\ }\href {https://doi.org/10.1103/PhysRevA.40.4277} {\bibfield  {journal} {\bibinfo  {journal} {Phys. Rev. A}\ }\textbf {\bibinfo {volume} {40}},\ \bibinfo {pages} {4277} (\bibinfo {year} {1989})}\BibitemShut {NoStop}%
\bibitem [{\citenamefont {Gisin}\ and\ \citenamefont {Peres}(1992)}]{Gisin1992}%
  \BibitemOpen
  \bibfield  {author} {\bibinfo {author} {\bibfnamefont {N.}~\bibnamefont {Gisin}}\ and\ \bibinfo {author} {\bibfnamefont {A.}~\bibnamefont {Peres}},\ }\bibfield  {title} {\bibinfo {title} {{Maximal violation of Bell's inequality for arbitrarily large spin}},\ }\href {https://doi.org/https://doi.org/10.1016/0375-9601(92)90949-M} {\bibfield  {journal} {\bibinfo  {journal} {Physics Letters A}\ }\textbf {\bibinfo {volume} {162}},\ \bibinfo {pages} {15} (\bibinfo {year} {1992})}\BibitemShut {NoStop}%
\bibitem [{\citenamefont {Popescu}\ and\ \citenamefont {Rohrlich}(1992)}]{Popescu1992}%
  \BibitemOpen
  \bibfield  {author} {\bibinfo {author} {\bibfnamefont {S.}~\bibnamefont {Popescu}}\ and\ \bibinfo {author} {\bibfnamefont {D.}~\bibnamefont {Rohrlich}},\ }\bibfield  {title} {\bibinfo {title} {{Generic quantum nonlocality}},\ }\href {https://doi.org/https://doi.org/10.1016/0375-9601(92)90711-T} {\bibfield  {journal} {\bibinfo  {journal} {Physics Letters A}\ }\textbf {\bibinfo {volume} {166}},\ \bibinfo {pages} {293} (\bibinfo {year} {1992})}\BibitemShut {NoStop}%
\bibitem [{\citenamefont {Gachechiladze}\ and\ \citenamefont {Gühne}(2017)}]{Gachechiladze2017}%
  \BibitemOpen
  \bibfield  {author} {\bibinfo {author} {\bibfnamefont {M.}~\bibnamefont {Gachechiladze}}\ and\ \bibinfo {author} {\bibfnamefont {O.}~\bibnamefont {Gühne}},\ }\bibfield  {title} {\bibinfo {title} {{Completing the proof of “Generic quantum nonlocality”}},\ }\href {https://doi.org/https://doi.org/10.1016/j.physleta.2016.10.001} {\bibfield  {journal} {\bibinfo  {journal} {Physics Letters A}\ }\textbf {\bibinfo {volume} {381}},\ \bibinfo {pages} {1281} (\bibinfo {year} {2017})}\BibitemShut {NoStop}%
\bibitem [{\citenamefont {Barrett}(2002)}]{Barret2002}%
  \BibitemOpen
  \bibfield  {author} {\bibinfo {author} {\bibfnamefont {J.}~\bibnamefont {Barrett}},\ }\bibfield  {title} {\bibinfo {title} {{Nonsequential positive-operator-valued measurements on entangled mixed states do not always violate a Bell inequality}},\ }\href {https://doi.org/10.1103/PhysRevA.65.042302} {\bibfield  {journal} {\bibinfo  {journal} {Phys. Rev. A}\ }\textbf {\bibinfo {volume} {65}},\ \bibinfo {pages} {042302} (\bibinfo {year} {2002})}\BibitemShut {NoStop}%
\bibitem [{\citenamefont {V\'ertesi}\ and\ \citenamefont {Brunner}(2012)}]{Vertesi2012}%
  \BibitemOpen
  \bibfield  {author} {\bibinfo {author} {\bibfnamefont {T.}~\bibnamefont {V\'ertesi}}\ and\ \bibinfo {author} {\bibfnamefont {N.}~\bibnamefont {Brunner}},\ }\bibfield  {title} {\bibinfo {title} {{Quantum Nonlocality Does Not Imply Entanglement Distillability}},\ }\href {https://doi.org/10.1103/PhysRevLett.108.030403} {\bibfield  {journal} {\bibinfo  {journal} {Phys. Rev. Lett.}\ }\textbf {\bibinfo {volume} {108}},\ \bibinfo {pages} {030403} (\bibinfo {year} {2012})}\BibitemShut {NoStop}%
\bibitem [{\citenamefont {DiVincenzo}\ \emph {et~al.}(2003)\citenamefont {DiVincenzo}, \citenamefont {Mor}, \citenamefont {Shor}, \citenamefont {Smolin},\ and\ \citenamefont {Terhal}}]{diVincenzo2003}%
  \BibitemOpen
  \bibfield  {author} {\bibinfo {author} {\bibfnamefont {D.~P.}\ \bibnamefont {DiVincenzo}}, \bibinfo {author} {\bibfnamefont {T.}~\bibnamefont {Mor}}, \bibinfo {author} {\bibfnamefont {P.~W.}\ \bibnamefont {Shor}}, \bibinfo {author} {\bibfnamefont {J.~A.}\ \bibnamefont {Smolin}},\ and\ \bibinfo {author} {\bibfnamefont {B.~M.}\ \bibnamefont {Terhal}},\ }\bibfield  {title} {\bibinfo {title} {{Unextendible Product Bases, Uncompletable Product Bases and Bound Entanglement}},\ }\href {https://doi.org/10.1007/s00220-003-0877-6} {\bibfield  {journal} {\bibinfo  {journal} {Communications in Mathematical Physics}\ }\textbf {\bibinfo {volume} {238}},\ \bibinfo {pages} {379} (\bibinfo {year} {2003})}\BibitemShut {NoStop}%
\bibitem [{\citenamefont {Gallego}\ \emph {et~al.}(2012)\citenamefont {Gallego}, \citenamefont {W\"urflinger}, \citenamefont {Ac\'{\i}n},\ and\ \citenamefont {Navascu\'es}}]{Gallego2012}%
  \BibitemOpen
  \bibfield  {author} {\bibinfo {author} {\bibfnamefont {R.}~\bibnamefont {Gallego}}, \bibinfo {author} {\bibfnamefont {L.~E.}\ \bibnamefont {W\"urflinger}}, \bibinfo {author} {\bibfnamefont {A.}~\bibnamefont {Ac\'{\i}n}},\ and\ \bibinfo {author} {\bibfnamefont {M.}~\bibnamefont {Navascu\'es}},\ }\bibfield  {title} {\bibinfo {title} {{Operational Framework for Nonlocality}},\ }\href {https://doi.org/10.1103/PhysRevLett.109.070401} {\bibfield  {journal} {\bibinfo  {journal} {Phys. Rev. Lett.}\ }\textbf {\bibinfo {volume} {109}},\ \bibinfo {pages} {070401} (\bibinfo {year} {2012})}\BibitemShut {NoStop}%
\bibitem [{\citenamefont {Collins}\ \emph {et~al.}(2002)\citenamefont {Collins}, \citenamefont {Gisin}, \citenamefont {Popescu}, \citenamefont {Roberts},\ and\ \citenamefont {Scarani}}]{Collins2002}%
  \BibitemOpen
  \bibfield  {author} {\bibinfo {author} {\bibfnamefont {D.}~\bibnamefont {Collins}}, \bibinfo {author} {\bibfnamefont {N.}~\bibnamefont {Gisin}}, \bibinfo {author} {\bibfnamefont {S.}~\bibnamefont {Popescu}}, \bibinfo {author} {\bibfnamefont {D.}~\bibnamefont {Roberts}},\ and\ \bibinfo {author} {\bibfnamefont {V.}~\bibnamefont {Scarani}},\ }\bibfield  {title} {\bibinfo {title} {{Bell-Type Inequalities to Detect True $\mathit{n}$-Body Nonseparability}},\ }\href {https://doi.org/10.1103/PhysRevLett.88.170405} {\bibfield  {journal} {\bibinfo  {journal} {Phys. Rev. Lett.}\ }\textbf {\bibinfo {volume} {88}},\ \bibinfo {pages} {170405} (\bibinfo {year} {2002})}\BibitemShut {NoStop}%
\bibitem [{\citenamefont {Mitchell}\ \emph {et~al.}(2004)\citenamefont {Mitchell}, \citenamefont {Popescu},\ and\ \citenamefont {Roberts}}]{Mitchell2004}%
  \BibitemOpen
  \bibfield  {author} {\bibinfo {author} {\bibfnamefont {P.}~\bibnamefont {Mitchell}}, \bibinfo {author} {\bibfnamefont {S.}~\bibnamefont {Popescu}},\ and\ \bibinfo {author} {\bibfnamefont {D.}~\bibnamefont {Roberts}},\ }\bibfield  {title} {\bibinfo {title} {{Conditions for the confirmation of three-particle nonlocality}},\ }\href {https://doi.org/10.1103/PhysRevA.70.060101} {\bibfield  {journal} {\bibinfo  {journal} {Phys. Rev. A}\ }\textbf {\bibinfo {volume} {70}},\ \bibinfo {pages} {060101} (\bibinfo {year} {2004})}\BibitemShut {NoStop}%
\bibitem [{\citenamefont {Seevinck}\ and\ \citenamefont {Uffink}(2008)}]{Seevinck2008}%
  \BibitemOpen
  \bibfield  {author} {\bibinfo {author} {\bibfnamefont {M.}~\bibnamefont {Seevinck}}\ and\ \bibinfo {author} {\bibfnamefont {J.}~\bibnamefont {Uffink}},\ }\bibfield  {title} {\bibinfo {title} {Partial separability and entanglement criteria for multiqubit quantum states},\ }\href {https://doi.org/10.1103/PhysRevA.78.032101} {\bibfield  {journal} {\bibinfo  {journal} {Phys. Rev. A}\ }\textbf {\bibinfo {volume} {78}},\ \bibinfo {pages} {032101} (\bibinfo {year} {2008})}\BibitemShut {NoStop}%
\bibitem [{\citenamefont {Peres}(1996)}]{peres1996}%
  \BibitemOpen
  \bibfield  {author} {\bibinfo {author} {\bibfnamefont {A.}~\bibnamefont {Peres}},\ }\bibfield  {title} {\bibinfo {title} {{Collective tests for quantum nonlocality}},\ }\href {https://doi.org/10.1103/PhysRevA.54.2685} {\bibfield  {journal} {\bibinfo  {journal} {Phys. Rev. A}\ }\textbf {\bibinfo {volume} {54}},\ \bibinfo {pages} {2685} (\bibinfo {year} {1996})}\BibitemShut {NoStop}%
\bibitem [{\citenamefont {Popescu}(1995)}]{Popescu1995}%
  \BibitemOpen
  \bibfield  {author} {\bibinfo {author} {\bibfnamefont {S.}~\bibnamefont {Popescu}},\ }\bibfield  {title} {\bibinfo {title} {{Bell's Inequalities and Density Matrices: Revealing ``Hidden'' Nonlocality}},\ }\href {https://doi.org/10.1103/PhysRevLett.74.2619} {\bibfield  {journal} {\bibinfo  {journal} {Phys. Rev. Lett.}\ }\textbf {\bibinfo {volume} {74}},\ \bibinfo {pages} {2619} (\bibinfo {year} {1995})}\BibitemShut {NoStop}%
\bibitem [{\citenamefont {Gisin}(1996)}]{Gisin1996}%
  \BibitemOpen
  \bibfield  {author} {\bibinfo {author} {\bibfnamefont {N.}~\bibnamefont {Gisin}},\ }\bibfield  {title} {\bibinfo {title} {{Hidden quantum nonlocality revealed by local filters}},\ }\href {https://doi.org/https://doi.org/10.1016/S0375-9601(96)80001-6} {\bibfield  {journal} {\bibinfo  {journal} {Physics Letters A}\ }\textbf {\bibinfo {volume} {210}},\ \bibinfo {pages} {151} (\bibinfo {year} {1996})}\BibitemShut {NoStop}%
\bibitem [{\citenamefont {Sen(De)}\ \emph {et~al.}(2005)\citenamefont {Sen(De)}, \citenamefont {Sen}, \citenamefont {Brukner}, \citenamefont {Bu\v{z}ek},\ and\ \citenamefont {{Ż}ukowski}}]{Aditi2005}%
  \BibitemOpen
  \bibfield  {author} {\bibinfo {author} {\bibfnamefont {A.}~\bibnamefont {Sen(De)}}, \bibinfo {author} {\bibfnamefont {U.}~\bibnamefont {Sen}}, \bibinfo {author} {\bibfnamefont {{\v C}.}~\bibnamefont {Brukner}}, \bibinfo {author} {\bibfnamefont {V.}~\bibnamefont {Bu\v{z}ek}},\ and\ \bibinfo {author} {\bibfnamefont {M.}~\bibnamefont {{Ż}ukowski}},\ }\bibfield  {title} {\bibinfo {title} {{Entanglement swapping of noisy states: A kind of superadditivity in nonclassicality}},\ }\href {https://doi.org/10.1103/PhysRevA.72.042310} {\bibfield  {journal} {\bibinfo  {journal} {Phys. Rev. A}\ }\textbf {\bibinfo {volume} {72}},\ \bibinfo {pages} {042310} (\bibinfo {year} {2005})}\BibitemShut {NoStop}%
\bibitem [{\citenamefont {Masanes}(2006)}]{Masanes2006}%
  \BibitemOpen
  \bibfield  {author} {\bibinfo {author} {\bibfnamefont {L.}~\bibnamefont {Masanes}},\ }\bibfield  {title} {\bibinfo {title} {{Asymptotic Violation of Bell Inequalities and Distillability}},\ }\href {https://doi.org/10.1103/PhysRevLett.97.050503} {\bibfield  {journal} {\bibinfo  {journal} {Phys. Rev. Lett.}\ }\textbf {\bibinfo {volume} {97}},\ \bibinfo {pages} {050503} (\bibinfo {year} {2006})}\BibitemShut {NoStop}%
\bibitem [{\citenamefont {Navascu\'es}\ and\ \citenamefont {V\'ertesi}(2011)}]{Navascues2011}%
  \BibitemOpen
  \bibfield  {author} {\bibinfo {author} {\bibfnamefont {M.}~\bibnamefont {Navascu\'es}}\ and\ \bibinfo {author} {\bibfnamefont {T.}~\bibnamefont {V\'ertesi}},\ }\bibfield  {title} {\bibinfo {title} {{Activation of Nonlocal Quantum Resources}},\ }\href {https://doi.org/10.1103/PhysRevLett.106.060403} {\bibfield  {journal} {\bibinfo  {journal} {Phys. Rev. Lett.}\ }\textbf {\bibinfo {volume} {106}},\ \bibinfo {pages} {060403} (\bibinfo {year} {2011})}\BibitemShut {NoStop}%
\bibitem [{\citenamefont {Cavalcanti}\ \emph {et~al.}(2011)\citenamefont {Cavalcanti}, \citenamefont {Almeida}, \citenamefont {Scarani},\ and\ \citenamefont {Ac{\'\i}n}}]{Cavalcanti2011}%
  \BibitemOpen
  \bibfield  {author} {\bibinfo {author} {\bibfnamefont {D.}~\bibnamefont {Cavalcanti}}, \bibinfo {author} {\bibfnamefont {M.~L.}\ \bibnamefont {Almeida}}, \bibinfo {author} {\bibfnamefont {V.}~\bibnamefont {Scarani}},\ and\ \bibinfo {author} {\bibfnamefont {A.}~\bibnamefont {Ac{\'\i}n}},\ }\bibfield  {title} {\bibinfo {title} {{Quantum networks reveal quantum nonlocality}},\ }\href {https://doi.org/10.1038/ncomms1193} {\bibfield  {journal} {\bibinfo  {journal} {Nature Communications}\ }\textbf {\bibinfo {volume} {2}},\ \bibinfo {pages} {184} (\bibinfo {year} {2011})}\BibitemShut {NoStop}%
\bibitem [{\citenamefont {Palazuelos}(2012)}]{Palazuelos2012}%
  \BibitemOpen
  \bibfield  {author} {\bibinfo {author} {\bibfnamefont {C.}~\bibnamefont {Palazuelos}},\ }\bibfield  {title} {\bibinfo {title} {{Superactivation of Quantum Nonlocality}},\ }\href {https://doi.org/10.1103/PhysRevLett.109.190401} {\bibfield  {journal} {\bibinfo  {journal} {Phys. Rev. Lett.}\ }\textbf {\bibinfo {volume} {109}},\ \bibinfo {pages} {190401} (\bibinfo {year} {2012})}\BibitemShut {NoStop}%
\bibitem [{\citenamefont {Cavalcanti}\ \emph {et~al.}(2013)\citenamefont {Cavalcanti}, \citenamefont {Ac\'{\i}n}, \citenamefont {Brunner},\ and\ \citenamefont {V\'ertesi}}]{Cavalcanti2013}%
  \BibitemOpen
  \bibfield  {author} {\bibinfo {author} {\bibfnamefont {D.}~\bibnamefont {Cavalcanti}}, \bibinfo {author} {\bibfnamefont {A.}~\bibnamefont {Ac\'{\i}n}}, \bibinfo {author} {\bibfnamefont {N.}~\bibnamefont {Brunner}},\ and\ \bibinfo {author} {\bibfnamefont {T.}~\bibnamefont {V\'ertesi}},\ }\bibfield  {title} {\bibinfo {title} {{All quantum states useful for teleportation are nonlocal resources}},\ }\href {https://doi.org/10.1103/PhysRevA.87.042104} {\bibfield  {journal} {\bibinfo  {journal} {Phys. Rev. A}\ }\textbf {\bibinfo {volume} {87}},\ \bibinfo {pages} {042104} (\bibinfo {year} {2013})}\BibitemShut {NoStop}%
\bibitem [{\citenamefont {Hirsch}\ \emph {et~al.}(2013)\citenamefont {Hirsch}, \citenamefont {Quintino}, \citenamefont {Bowles},\ and\ \citenamefont {Brunner}}]{Hirsch2013}%
  \BibitemOpen
  \bibfield  {author} {\bibinfo {author} {\bibfnamefont {F.}~\bibnamefont {Hirsch}}, \bibinfo {author} {\bibfnamefont {M.~T.}\ \bibnamefont {Quintino}}, \bibinfo {author} {\bibfnamefont {J.}~\bibnamefont {Bowles}},\ and\ \bibinfo {author} {\bibfnamefont {N.}~\bibnamefont {Brunner}},\ }\bibfield  {title} {\bibinfo {title} {{Genuine Hidden Quantum Nonlocality}},\ }\href {https://doi.org/10.1103/PhysRevLett.111.160402} {\bibfield  {journal} {\bibinfo  {journal} {Phys. Rev. Lett.}\ }\textbf {\bibinfo {volume} {111}},\ \bibinfo {pages} {160402} (\bibinfo {year} {2013})}\BibitemShut {NoStop}%
\bibitem [{\citenamefont {Gallego}\ \emph {et~al.}(2014)\citenamefont {Gallego}, \citenamefont {Würflinger}, \citenamefont {Chaves}, \citenamefont {Acín},\ and\ \citenamefont {Navascués}}]{Gallego2014}%
  \BibitemOpen
  \bibfield  {author} {\bibinfo {author} {\bibfnamefont {R.}~\bibnamefont {Gallego}}, \bibinfo {author} {\bibfnamefont {L.~E.}\ \bibnamefont {Würflinger}}, \bibinfo {author} {\bibfnamefont {R.}~\bibnamefont {Chaves}}, \bibinfo {author} {\bibfnamefont {A.}~\bibnamefont {Acín}},\ and\ \bibinfo {author} {\bibfnamefont {M.}~\bibnamefont {Navascués}},\ }\bibfield  {title} {\bibinfo {title} {{Nonlocality in sequential correlation scenarios}},\ }\href {https://doi.org/10.1088/1367-2630/16/3/033037} {\bibfield  {journal} {\bibinfo  {journal} {New Journal of Physics}\ }\textbf {\bibinfo {volume} {16}},\ \bibinfo {pages} {033037} (\bibinfo {year} {2014})}\BibitemShut {NoStop}%
\bibitem [{\citenamefont {Tendick}\ \emph {et~al.}(2020)\citenamefont {Tendick}, \citenamefont {Kampermann},\ and\ \citenamefont {Bru\ss{}}}]{Tendick2020}%
  \BibitemOpen
  \bibfield  {author} {\bibinfo {author} {\bibfnamefont {L.}~\bibnamefont {Tendick}}, \bibinfo {author} {\bibfnamefont {H.}~\bibnamefont {Kampermann}},\ and\ \bibinfo {author} {\bibfnamefont {D.}~\bibnamefont {Bru\ss{}}},\ }\bibfield  {title} {\bibinfo {title} {{Activation of Nonlocality in Bound Entanglement}},\ }\href {https://doi.org/10.1103/PhysRevLett.124.050401} {\bibfield  {journal} {\bibinfo  {journal} {Phys. Rev. Lett.}\ }\textbf {\bibinfo {volume} {124}},\ \bibinfo {pages} {050401} (\bibinfo {year} {2020})}\BibitemShut {NoStop}%
\bibitem [{\citenamefont {Villegas-Aguilar}\ \emph {et~al.}(2024)\citenamefont {Villegas-Aguilar}, \citenamefont {Polino}, \citenamefont {Ghafari}, \citenamefont {Quintino}, \citenamefont {Laverick}, \citenamefont {Berkman}, \citenamefont {Rogge}, \citenamefont {Shalm}, \citenamefont {Tischler}, \citenamefont {Cavalcanti}, \citenamefont {Slussarenko},\ and\ \citenamefont {Pryde}}]{Villegas-Aguilar2024}%
  \BibitemOpen
  \bibfield  {author} {\bibinfo {author} {\bibfnamefont {L.}~\bibnamefont {Villegas-Aguilar}}, \bibinfo {author} {\bibfnamefont {E.}~\bibnamefont {Polino}}, \bibinfo {author} {\bibfnamefont {F.}~\bibnamefont {Ghafari}}, \bibinfo {author} {\bibfnamefont {M.~T.}\ \bibnamefont {Quintino}}, \bibinfo {author} {\bibfnamefont {K.~T.}\ \bibnamefont {Laverick}}, \bibinfo {author} {\bibfnamefont {I.~R.}\ \bibnamefont {Berkman}}, \bibinfo {author} {\bibfnamefont {S.}~\bibnamefont {Rogge}}, \bibinfo {author} {\bibfnamefont {L.~K.}\ \bibnamefont {Shalm}}, \bibinfo {author} {\bibfnamefont {N.}~\bibnamefont {Tischler}}, \bibinfo {author} {\bibfnamefont {E.~G.}\ \bibnamefont {Cavalcanti}}, \bibinfo {author} {\bibfnamefont {S.}~\bibnamefont {Slussarenko}},\ and\ \bibinfo {author} {\bibfnamefont {G.~J.}\ \bibnamefont {Pryde}},\ }\bibfield  {title} {\bibinfo {title} {{Nonlocality activation in a photonic quantum network}},\ }\href {https://doi.org/10.1038/s41467-024-47354-w} {\bibfield  {journal} {\bibinfo  {journal} {Nature
  Communications}\ }\textbf {\bibinfo {volume} {15}},\ \bibinfo {pages} {3112} (\bibinfo {year} {2024})}\BibitemShut {NoStop}%
\bibitem [{\citenamefont {Augusiak}\ \emph {et~al.}(2014)\citenamefont {Augusiak}, \citenamefont {Demianowicz},\ and\ \citenamefont {Acín}}]{Augusiak_2014}%
  \BibitemOpen
  \bibfield  {author} {\bibinfo {author} {\bibfnamefont {R.}~\bibnamefont {Augusiak}}, \bibinfo {author} {\bibfnamefont {M.}~\bibnamefont {Demianowicz}},\ and\ \bibinfo {author} {\bibfnamefont {A.}~\bibnamefont {Acín}},\ }\bibfield  {title} {\bibinfo {title} {{Local hidden–variable models for entangled quantum states}},\ }\href {https://doi.org/10.1088/1751-8113/47/42/424002} {\bibfield  {journal} {\bibinfo  {journal} {Journal of Physics A: Mathematical and Theoretical}\ }\textbf {\bibinfo {volume} {47}},\ \bibinfo {pages} {424002} (\bibinfo {year} {2014})}\BibitemShut {NoStop}%
\bibitem [{\citenamefont {Hirsch}\ \emph {et~al.}(2017)\citenamefont {Hirsch}, \citenamefont {Quintino}, \citenamefont {V{\'{e}}rtesi}, \citenamefont {Navascu{\'{e}}s},\ and\ \citenamefont {Brunner}}]{Hirsch2017betterlocalhidden}%
  \BibitemOpen
  \bibfield  {author} {\bibinfo {author} {\bibfnamefont {F.}~\bibnamefont {Hirsch}}, \bibinfo {author} {\bibfnamefont {M.~T.}\ \bibnamefont {Quintino}}, \bibinfo {author} {\bibfnamefont {T.}~\bibnamefont {V{\'{e}}rtesi}}, \bibinfo {author} {\bibfnamefont {M.}~\bibnamefont {Navascu{\'{e}}s}},\ and\ \bibinfo {author} {\bibfnamefont {N.}~\bibnamefont {Brunner}},\ }\bibfield  {title} {\bibinfo {title} {{Better local hidden variable models for two-qubit {W}erner states and an upper bound on the {G}rothendieck constant {$K_G(3)$}}},\ }\href {https://doi.org/10.22331/q-2017-04-25-3} {\bibfield  {journal} {\bibinfo  {journal} {{Quantum}}\ }\textbf {\bibinfo {volume} {1}},\ \bibinfo {pages} {3} (\bibinfo {year} {2017})}\BibitemShut {NoStop}%
\bibitem [{\citenamefont {Designolle}\ \emph {et~al.}(2023)\citenamefont {Designolle}, \citenamefont {Iommazzo}, \citenamefont {Besan\ifmmode~\mbox{\c{c}}\else \c{c}\fi{}on}, \citenamefont {Knebel}, \citenamefont {Gel\ss{}},\ and\ \citenamefont {Pokutta}}]{Designolle2023}%
  \BibitemOpen
  \bibfield  {author} {\bibinfo {author} {\bibfnamefont {S.}~\bibnamefont {Designolle}}, \bibinfo {author} {\bibfnamefont {G.}~\bibnamefont {Iommazzo}}, \bibinfo {author} {\bibfnamefont {M.}~\bibnamefont {Besan\ifmmode~\mbox{\c{c}}\else \c{c}\fi{}on}}, \bibinfo {author} {\bibfnamefont {S.}~\bibnamefont {Knebel}}, \bibinfo {author} {\bibfnamefont {P.}~\bibnamefont {Gel\ss{}}},\ and\ \bibinfo {author} {\bibfnamefont {S.}~\bibnamefont {Pokutta}},\ }\bibfield  {title} {\bibinfo {title} {{Improved local models and new Bell inequalities via Frank-Wolfe algorithms}},\ }\href {https://doi.org/10.1103/PhysRevResearch.5.043059} {\bibfield  {journal} {\bibinfo  {journal} {Phys. Rev. Res.}\ }\textbf {\bibinfo {volume} {5}},\ \bibinfo {pages} {043059} (\bibinfo {year} {2023})}\BibitemShut {NoStop}%
\bibitem [{\citenamefont {Barnum}\ \emph {et~al.}(1996)\citenamefont {Barnum}, \citenamefont {Caves}, \citenamefont {Fuchs}, \citenamefont {Jozsa},\ and\ \citenamefont {Schumacher}}]{Barnum1996}%
  \BibitemOpen
  \bibfield  {author} {\bibinfo {author} {\bibfnamefont {H.}~\bibnamefont {Barnum}}, \bibinfo {author} {\bibfnamefont {C.~M.}\ \bibnamefont {Caves}}, \bibinfo {author} {\bibfnamefont {C.~A.}\ \bibnamefont {Fuchs}}, \bibinfo {author} {\bibfnamefont {R.}~\bibnamefont {Jozsa}},\ and\ \bibinfo {author} {\bibfnamefont {B.}~\bibnamefont {Schumacher}},\ }\bibfield  {title} {\bibinfo {title} {{Noncommuting Mixed States Cannot Be Broadcast}},\ }\href {https://doi.org/10.1103/PhysRevLett.76.2818} {\bibfield  {journal} {\bibinfo  {journal} {Phys. Rev. Lett.}\ }\textbf {\bibinfo {volume} {76}},\ \bibinfo {pages} {2818} (\bibinfo {year} {1996})}\BibitemShut {NoStop}%
\bibitem [{\citenamefont {Walleghem}\ \emph {et~al.}(2024{\natexlab{a}})\citenamefont {Walleghem}, \citenamefont {Yīng}, \citenamefont {Wagner},\ and\ \citenamefont {Schmid}}]{walleghem2024connectingextendedwignersfriend}%
  \BibitemOpen
  \bibfield  {author} {\bibinfo {author} {\bibfnamefont {L.}~\bibnamefont {Walleghem}}, \bibinfo {author} {\bibfnamefont {Y.}~\bibnamefont {Yīng}}, \bibinfo {author} {\bibfnamefont {R.}~\bibnamefont {Wagner}},\ and\ \bibinfo {author} {\bibfnamefont {D.}~\bibnamefont {Schmid}},\ }\href {https://arxiv.org/abs/2409.07537} {\bibinfo {title} {{Connecting extended Wigner's friend arguments and noncontextuality}}} (\bibinfo {year} {2024}{\natexlab{a}}),\ \Eprint {https://arxiv.org/abs/2409.07537} {arXiv:2409.07537 [quant-ph]} \BibitemShut {NoStop}%
\bibitem [{\citenamefont {Di~Biagio}\ and\ \citenamefont {Rovelli}(2021)}]{DiBiagio2021}%
  \BibitemOpen
  \bibfield  {author} {\bibinfo {author} {\bibfnamefont {A.}~\bibnamefont {Di~Biagio}}\ and\ \bibinfo {author} {\bibfnamefont {C.}~\bibnamefont {Rovelli}},\ }\bibfield  {title} {\bibinfo {title} {{Stable Facts, Relative Facts}},\ }\href {https://doi.org/10.1007/s10701-021-00429-w} {\bibfield  {journal} {\bibinfo  {journal} {Foundations of Physics}\ }\textbf {\bibinfo {volume} {51}},\ \bibinfo {pages} {30} (\bibinfo {year} {2021})}\BibitemShut {NoStop}%
\bibitem [{\citenamefont {Yang}(2022)}]{Yang2022}%
  \BibitemOpen
  \bibfield  {author} {\bibinfo {author} {\bibfnamefont {J.~M.}\ \bibnamefont {Yang}},\ }\bibfield  {title} {\bibinfo {title} {{Law of Total Probability in Quantum Theory and Its Application in Wigner's Friend Scenario}},\ }\href {https://doi.org/10.3390/e24070903} {\bibfield  {journal} {\bibinfo  {journal} {Entropy}\ }\textbf {\bibinfo {volume} {24}},\ \bibinfo {pages} {903} (\bibinfo {year} {2022})}\BibitemShut {NoStop}%
\bibitem [{\citenamefont {Haddara}\ and\ \citenamefont {Cavalcanti}(2023)}]{Haddara_2023}%
  \BibitemOpen
  \bibfield  {author} {\bibinfo {author} {\bibfnamefont {M.}~\bibnamefont {Haddara}}\ and\ \bibinfo {author} {\bibfnamefont {E.~G.}\ \bibnamefont {Cavalcanti}},\ }\bibfield  {title} {\bibinfo {title} {{A possibilistic no-go theorem on the Wigner’s friend paradox}},\ }\href {https://doi.org/10.1088/1367-2630/aceea3} {\bibfield  {journal} {\bibinfo  {journal} {New Journal of Physics}\ }\textbf {\bibinfo {volume} {25}},\ \bibinfo {pages} {093028} (\bibinfo {year} {2023})}\BibitemShut {NoStop}%
\bibitem [{\citenamefont {Ding}\ \emph {et~al.}(2023)\citenamefont {Ding}, \citenamefont {Wang}, \citenamefont {He}, \citenamefont {Hou}, \citenamefont {Gao},\ and\ \citenamefont {Yan}}]{Ding2023}%
  \BibitemOpen
  \bibfield  {author} {\bibinfo {author} {\bibfnamefont {D.}~\bibnamefont {Ding}}, \bibinfo {author} {\bibfnamefont {C.}~\bibnamefont {Wang}}, \bibinfo {author} {\bibfnamefont {Y.~Q.}\ \bibnamefont {He}}, \bibinfo {author} {\bibfnamefont {T.}~\bibnamefont {Hou}}, \bibinfo {author} {\bibfnamefont {T.}~\bibnamefont {Gao}},\ and\ \bibinfo {author} {\bibfnamefont {F.~L.}\ \bibnamefont {Yan}},\ }\bibfield  {title} {\bibinfo {title} {{Tripartite Wigner’s friend scenario and its test}},\ }\href {https://doi.org/10.1088/1402-4896/acdd33} {\bibfield  {journal} {\bibinfo  {journal} {Physica Scripta}\ }\textbf {\bibinfo {volume} {98}},\ \bibinfo {pages} {075104} (\bibinfo {year} {2023})}\BibitemShut {NoStop}%
\bibitem [{\citenamefont {Zeng}\ \emph {et~al.}(2024)\citenamefont {Zeng}, \citenamefont {Labib},\ and\ \citenamefont {Russo}}]{zeng2024violationslocalfriendlinessquantum}%
  \BibitemOpen
  \bibfield  {author} {\bibinfo {author} {\bibfnamefont {W.~J.}\ \bibnamefont {Zeng}}, \bibinfo {author} {\bibfnamefont {F.}~\bibnamefont {Labib}},\ and\ \bibinfo {author} {\bibfnamefont {V.}~\bibnamefont {Russo}},\ }\href {https://arxiv.org/abs/2409.15302} {\bibinfo {title} {{Towards violations of Local Friendliness with quantum computers}}} (\bibinfo {year} {2024}),\ \Eprint {https://arxiv.org/abs/2409.15302} {arXiv:2409.15302 [quant-ph]} \BibitemShut {NoStop}%
\bibitem [{\citenamefont {Walleghem}\ \emph {et~al.}(2024{\natexlab{b}})\citenamefont {Walleghem}, \citenamefont {Wagner}, \citenamefont {Yīng},\ and\ \citenamefont {Schmid}}]{walleghem2024extendedwignersfriendparadoxes}%
  \BibitemOpen
  \bibfield  {author} {\bibinfo {author} {\bibfnamefont {L.}~\bibnamefont {Walleghem}}, \bibinfo {author} {\bibfnamefont {R.}~\bibnamefont {Wagner}}, \bibinfo {author} {\bibfnamefont {Y.}~\bibnamefont {Yīng}},\ and\ \bibinfo {author} {\bibfnamefont {D.}~\bibnamefont {Schmid}},\ }\href {https://arxiv.org/abs/2310.06976} {\bibinfo {title} {Extended wigner's friend paradoxes do not require nonlocal correlations}} (\bibinfo {year} {2024}{\natexlab{b}}),\ \Eprint {https://arxiv.org/abs/2310.06976} {arXiv:2310.06976 [quant-ph]} \BibitemShut {NoStop}%
\bibitem [{\citenamefont {Walleghem}\ and\ \citenamefont {Catani}(2025)}]{walleghem2025extendedwignersfriendnogo}%
  \BibitemOpen
  \bibfield  {author} {\bibinfo {author} {\bibfnamefont {L.}~\bibnamefont {Walleghem}}\ and\ \bibinfo {author} {\bibfnamefont {L.}~\bibnamefont {Catani}},\ }\href {https://arxiv.org/abs/2502.02461} {\bibinfo {title} {{An extended Wigner's friend no-go theorem inspired by generalized contextuality}}} (\bibinfo {year} {2025}),\ \Eprint {https://arxiv.org/abs/2502.02461} {arXiv:2502.02461 [quant-ph]} \BibitemShut {NoStop}%
\bibitem [{\citenamefont {Temistocles}\ \emph {et~al.}(2019)\citenamefont {Temistocles}, \citenamefont {Rabelo},\ and\ \citenamefont {Cunha}}]{Temistocles2019}%
  \BibitemOpen
  \bibfield  {author} {\bibinfo {author} {\bibfnamefont {T.}~\bibnamefont {Temistocles}}, \bibinfo {author} {\bibfnamefont {R.}~\bibnamefont {Rabelo}},\ and\ \bibinfo {author} {\bibfnamefont {M.~T.}\ \bibnamefont {Cunha}},\ }\bibfield  {title} {\bibinfo {title} {{Measurement compatibility in Bell nonlocality tests}},\ }\href {https://doi.org/10.1103/PhysRevA.99.042120} {\bibfield  {journal} {\bibinfo  {journal} {Phys. Rev. A}\ }\textbf {\bibinfo {volume} {99}},\ \bibinfo {pages} {042120} (\bibinfo {year} {2019})}\BibitemShut {NoStop}%
\bibitem [{\citenamefont {Chiribella}\ \emph {et~al.}(2024)\citenamefont {Chiribella}, \citenamefont {Giannelli},\ and\ \citenamefont {Scandolo}}]{Chiribella2024}%
  \BibitemOpen
  \bibfield  {author} {\bibinfo {author} {\bibfnamefont {G.}~\bibnamefont {Chiribella}}, \bibinfo {author} {\bibfnamefont {L.}~\bibnamefont {Giannelli}},\ and\ \bibinfo {author} {\bibfnamefont {C.~M.}\ \bibnamefont {Scandolo}},\ }\bibfield  {title} {\bibinfo {title} {{Bell Nonlocality in Classical Systems Coexisting with Other System Types}},\ }\href {https://doi.org/10.1103/PhysRevLett.132.190201} {\bibfield  {journal} {\bibinfo  {journal} {Phys. Rev. Lett.}\ }\textbf {\bibinfo {volume} {132}},\ \bibinfo {pages} {190201} (\bibinfo {year} {2024})}\BibitemShut {NoStop}%
\bibitem [{\citenamefont {Bell}(2004)}]{Bellcollection}%
  \BibitemOpen
  \bibfield  {author} {\bibinfo {author} {\bibfnamefont {J.~S.}\ \bibnamefont {Bell}},\ }\href {https://doi.org/DOI: 10.1017/CBO9780511815676} {\emph {\bibinfo {title} {{Speakable and Unspeakable in Quantum Mechanics: Collected Papers on Quantum Philosophy}}}},\ \bibinfo {edition} {2nd}\ ed.\ (\bibinfo  {publisher} {Cambridge University Press},\ \bibinfo {address} {Cambridge},\ \bibinfo {year} {2004})\BibitemShut {NoStop}%
\bibitem [{\citenamefont {Śliwa}(2003)}]{Sliwa2003}%
  \BibitemOpen
  \bibfield  {author} {\bibinfo {author} {\bibfnamefont {C.}~\bibnamefont {Śliwa}},\ }\bibfield  {title} {\bibinfo {title} {{Symmetries of the Bell correlation inequalities}},\ }\href {https://doi.org/https://doi.org/10.1016/S0375-9601(03)01115-0} {\bibfield  {journal} {\bibinfo  {journal} {{Physics Letters A}}\ }\textbf {\bibinfo {volume} {317}},\ \bibinfo {pages} {165} (\bibinfo {year} {2003})}\BibitemShut {NoStop}%
\bibitem [{\citenamefont {{Cezary Śliwa}}(2003)}]{Sliwa2003arXiv}%
  \BibitemOpen
  \bibfield  {author} {\bibinfo {author} {\bibnamefont {{Cezary Śliwa}}},\ }\href {https://arxiv.org/abs/quant-ph/0305190} {\bibinfo {title} {{Symmetries of the Bell correlation inequalities}}} (\bibinfo {year} {2003}),\ \Eprint {https://arxiv.org/abs/0305190} {arXiv:0305190 [quant-ph]} \BibitemShut {NoStop}%
\bibitem [{\citenamefont {Pitowsky}\ and\ \citenamefont {Svozil}(2001)}]{Pitowsky2001}%
  \BibitemOpen
  \bibfield  {author} {\bibinfo {author} {\bibfnamefont {I.}~\bibnamefont {Pitowsky}}\ and\ \bibinfo {author} {\bibfnamefont {K.}~\bibnamefont {Svozil}},\ }\bibfield  {title} {\bibinfo {title} {{Optimal tests of quantum nonlocality}},\ }\href {https://doi.org/10.1103/PhysRevA.64.014102} {\bibfield  {journal} {\bibinfo  {journal} {Phys. Rev. A}\ }\textbf {\bibinfo {volume} {64}},\ \bibinfo {pages} {014102} (\bibinfo {year} {2001})}\BibitemShut {NoStop}%
\end{thebibliography}%

\appendix

\section{ Proof of Theorem \ref*{Theorem:LDisInQuantum}\label{appendixProofthatLDisInQuantum}}

\begin{customthm}{\ref*{Theorem:LDisInQuantum}}
Let $S = (I,M,O)$ be a scenario. Then $\mathbf{LD}(S) \subset \mathbf{Q}(S)$. 
\begin{proof}
    Let $\wp \in \mathbf{LD}(S)$ be an arbitrary behaviour. By definition, $\wp = \sum_{j\in J}\lambda_j D_j$, with $D_j \in \mathbf{PP}_{NS}(S)$ for all $j\in J$, $\lambda_j\geq 0$ for all $j\in J$ and $\sum_{j\in J}\lambda_j =1$. Consider the distributions $D_j(\vec{a}|\vec{x}) = \prod_{i\in I}D_j(a_i|x_i) \in \{0,1\}$. Owing to normalization, there is exactly one outcome $a_i = \alpha^j_{x_i}$, say,  for each $j$ with probability one occurring at each site. We will construct a quantum state and a collection of projective measurements which reproduces $D_j$.  Let $\mathcal{H}_{\vec{A}} =  \bigotimes_{i\in I}\mathcal{H}_{A_i}$ be the Hilbert space of all the agents $i\in I$ and for each agent $i$, their Hilbert space $\mathcal{H}_{A_i} = \bigotimes\limits_{1 \leq m_i \leq |M_i|  } \mathcal{H}_{x^{m_i}_i}$ with $\mathrm{dim}({\mathcal{H}_{x^{m_i}_i}}) = |O_{x_i}|$ for all $i\in I, x_i\in M_i$. Consider the pure state $\rho^j_{\boldsymbol{\alpha}_i} := \ket{\boldsymbol{\alpha^j_i}}\bra{\boldsymbol{\alpha^j_i}} \in \mathbb{S}(\mathcal{H}_{A_i})$ for each agent, where
    \begin{align}
        \ket{\boldsymbol{\alpha^j_i}} = \bigotimes_{1\leq m_i \leq |M_i|}\ket{\alpha^j_{x^{m_i}_i}}
    \end{align}
    and the $\ket{\alpha^j_{x^{m_i}_i}}$ form an orthonormal basis of $\mathcal{H}^A_{x^{m_i}_i}$ for all $i\in I$ with $\alpha^j_{x_i^{m_i}}\in O_{x_i = x_i^{m_i}}$. 
    Define the projections $M_{a_i|x_i}:= \ket{a_{x_i}}\bra{a_{x_i}} \otimes id_{\mathcal{H}_{A_i \setminus \{x_i\}}}$ for all $a_i \in O_{x_i}, x_i\in M_i$ and $ i\in I$. Then 
\begin{align}
    \mathrm{tr}[M_{a_i|x_i}\rho^j_{\boldsymbol{\alpha}_i}] = D^j(a=\alpha_{x_i}| \boldsymbol{\alpha}^j_i),
\end{align}
for each party, as required. By construction, the state $\rho^j =\bigotimes\limits_{i\in I}\rho^j_{\boldsymbol{\alpha}_i} \in \mathbb{S}(\mathcal{H_{\vec{A}}}) $
    can be used to recover the entire collection $D_j$ by use of local measurements $\bigotimes_{i\in I}M_{a_i|x_i}$,  which shows that $\mathbf{PP_{\mathbf{NS}}}(S)\subset \mathbf{Q}(S)$. Since the quantum state space is convex, the mixed state 
    \begin{align}
        \rho = \sum_{j\in J}\lambda_j \rho^j = \sum_{j\in J} \lambda_j(\bigotimes\limits_{i\in I}\rho^j_{\boldsymbol{\alpha}_i} )
    \end{align}
is a well defined quantum state, which owing to linearity of the trace reproduces any $\wp \in \mathbf{LD}(S)$ with $\wp = \sum_{j\in J}\lambda_jP_j$, proving the claim. 
\end{proof}
\end{customthm}
\section{Proof of Theorem \ref*{Theorem:unionof3partypdp'SiSNOTCONVEX}}
\label{appendixSection:ProofofProp}

The complete list of facet inequalities for the tripartite Bell-polytope was solved by Sliwa \cite{Sliwa2003, Sliwa2003arXiv} (and also by Pitowsky and Svozil \cite{Pitowsky2001}), leading to 46 inequivalent classes of inequalities when relabeling symmetries are taken into account. Sliwa listed his results \cite{Sliwa2003arXiv} in terms of the correlators such as
\begin{align}
    \expectation{A_iB_jC_k} = \sum\limits_{a,b,c \in \{\pm 1\}} abc\wp(abc|x_A=i, x_B = j, x_C = k),
    \end{align}
with $i,j,k \in \{1,2\}$ labelling the inputs of parties $A,B,C$, with similar definitions for the single-party expressions such as $\expectation{A_i}$ and two-party expressions such as $\expectation{A_iC_j}$. For convenience, we will adopt a similar notational conventions here. 

We draw attention to class (3) of the list in Ref.~\cite{Sliwa2003arXiv} in particular, and pick three representatives from that class, which we refer to as the Sliwa 3A, 3B, and 3C inequalities, respectively.  \newline

\begin{flushleft}
    Sliwa 3A:
\end{flushleft}\begin{align}\label{Equation:SliwaA1type3}
       \begin{split}
           \expectation{A_1 B_1 C_1} + \expectation{A_2 B_1 C_1} + \expectation{A_1 B_2 C_2} - \expectation{A_2 B_2 C_2} \leq 2.
       \end{split} 
    \end{align}
Sliwa 3B:

\begin{align}\label{Equation:SliwaA2type3} 
       \begin{split}
           \expectation{A_1 B_1 C_1} + \expectation{A_1 B_2 C_1} + \expectation{A_2 B_1 C_2} - \expectation{A_2 B_2 C_2} \leq 2.
       \end{split} 
    \end{align}

\begin{flushleft}
Sliwa 3C 
\end{flushleft}\begin{align}
    \label{Equation:SliwaA3type3}
    \begin{split}
        \expectation{A_1 B_1 C_1} + \expectation{A_1 B_1 C_2} + \expectation{A_2 B_2 C_1} - \expectation{A_2 B_2 C_2} \leq 2.
    \end{split}
\end{align}

The Sliwa 3A inequality of Eq.~\eqref{Equation:SliwaA1type3} is exactly the representative of that class given in Ref.~\cite{Sliwa2003arXiv}. That the other two are in the same class, can be seen by noting that the expressions in 3B and 3C can be obtained from 3A by relabeling the parties  $A-B$ and $A-C$, respectively. These are among the allowed symmetries since each party has the same number of inputs at each site \cite{Sliwa2003,Sliwa2003arXiv}.

\proposition{\label{Proposition:Sliwa3relationsInthreepartitescenario}Let $S$ be a tripartite scenario with $I = \{A,B,C\}$, $M_i = \{1,2\}$ for all $i\in I$ and $O_{x_i} = \{-1,1 \}$ for all $x_i \in M_i$. Let $M^{\{k\}}$ denote the collection $M^{\{k\}}_i = M_i$, if $i=k$ and $ M_i^{\{k\}} = \emptyset$ otherwise. Then the following three statements are true:
\begin{enumerate}[label=\roman*), ref=\ref{Proposition:Sliwa3relatonsInthreepartitescenario}.\roman*]

  \item \label{Proposition:Sliwa3.i}The Sliwa 3A inequality of Eq.~ \eqref{Equation:SliwaA1type3} is valid for the partially deterministic polytope $\mathbf{PD}(S,M^{\{A\}}) $ but not for $\mathbf{PD}(S,M^{\{B\}})$ and $\mathbf{PD}(S,M^{\{C\}})$.
  
    \item \label{Proposition:Sliwa3.ii}The Sliwa 3B inequality of Eq.~\eqref{Equation:SliwaA2type3} is valid for $\mathbf{PD}(S,M^{\{B\}})$ but not for $\mathbf{PD}(S,M^{A})$ and $\mathbf{PD}(S,M^{\{C\}})$.

     \item \label{Proposition:Sliwa3.iii}The Sliwa 3C inequality of Eq.~\eqref{Equation:SliwaA2type3} is valid for $\mathbf{PD}(S,M^{\{C\}})$ but not for $\mathbf{PD}(S,M^{\{A\}})$ and $\mathbf{PD}(S,M^{\{B\}})$.

\end{enumerate}
  }
\begin{proof}
    Since the polytopes $\mathbf{PD}(S,M^{\{k\}})$ are convex, it is sufficient to verify that the claim holds for the extreme points. Furthermore, since the inequalities \eqref{Equation:SliwaA1type3} and \eqref{Equation:SliwaA2type3} are satisfied by all predictable behaviour owing to the fact that they are  facets of the polytope $\mathbf{B}(S)$ as established by Sliwa \cite{Sliwa2003,Sliwa2003arXiv}, it is sufficient to constrain the analysis to only the non-classical, i.e. genuinely partially predictable, extreme points of the polytopes. From Theorem \ref{TheoremextremepointsofPDP} it follows that the non-predictable extreme points of $\mathbf{PD}(S,M^{\{k\}})$ are behaviour products of deterministic distributions for the party $k$ and PR-boxes for the remaining two parties.   We will only  show the claim for the case of Eq.~\eqref{Equation:SliwaA1type3} as the other two cases are analogous. 
    
  For the extreme points of $\mathbf{PD}(S,M^{\{A\}})$ Eq.~\eqref{Equation:SliwaA1type3} reduces to
    \begin{align}
    \begin{split}\label{Equation:Sliwa3EquationReducedforFirstpolytope}
      &  \expectation{A_1}\expectation{B_1 C_1} + \expectation{A_2} \expectation{B_1C_1} + \expectation{A_1}\expectation{B_2 C_2} - \expectation{A_2}\expectation{B_2 C_2} \\
       & = \expectation{B_1 C_1}( \expectation{A_1} + \expectation{A_2}) + \expectation{B_2 C_2}( \expectation{A_1} - \expectation{A_2}) \leq 2,
        \end{split}
    \end{align}
following from the fact that the party $A$ with the deterministic input is uncorrelated with $B$ and $C$. One of the two terms in  Equation \eqref{Equation:Sliwa3EquationReducedforFirstpolytope} always vanishes owing to the fact that $\expectation{A_k} \in \{-1,1 \}$ for all $k \in \{1,2\}$, while the pairwise correlators $\expectation{B_i C_i}$ for a PR box are similarly always bounded to take values in the set $\{-1,1\}$ with the additional constraint that exactly three of the correlators (with the same signature) take the same value. It follows that any behaviour $\wp^{1} \in \mathbf{PD}(S,M^{\{A\}})$ satisfies the inequality \eqref{Equation:SliwaA1type3}, as claimed. Note also, that the equality in Eq.~\eqref{Equation:Sliwa3EquationReducedforFirstpolytope} can be achieved for example by a behaviour $\wp_A^{PR}$ with a PR-box between $B-C$ with the property that $\expectation{B_1C_1}=1$ and $\expectation{A_i}=1$.

It remains to show that the same inequality can be violated by some $\wp^{2}\in \mathbf{PD}(S,M^{\{2\}})$ and $\wp^3 \in \mathbf{PD}(S,M^{\{C\}})$. 

For a non-classical extreme point of $\mathbf{PD}(S,M^{\{B\}})$ Eq.~\eqref{Equation:SliwaA1type3} is, in a similar manner,  seen to be equivalent to 
\begin{align}
\begin{split}
    \expectation{B_1}(\expectation{A_1 C_1} + \expectation{A_2 C_1}) + \expectation{B_2}(\expectation{A_1 C_2} - \expectation{A_2 C_2}) \leq 2.
    \end{split}
\end{align}
This inequality can however be violated up to the value $4>2$, for example, by the partially predictable behaviour $\wp_B^{PR} \in \mathbf{PD}(S,M^{\{B\}})$ for which $\expectation{B_i} = 1$ for $i\in \{1, 2\}$ and the PR-box between $A-C$ with $\expectation{A_2 C_2} = -1$ and $\expectation{A_i B_j} = 1$, otherwise. 

A similar observation is made for a non-classical extreme point $\wp^{C}\in \mathrm{Ext}(\mathbf{PD}(S,M^{\{C\}}))$, that Eq.~\eqref{Equation:SliwaA1type3} implies 

\begin{align}
    \expectation{C_1}(\expectation{A_1B_1} + \expectation{A_2B_1}) +\expectation{C_2}(\expectation{A_1B_2} - \expectation{A_2B_2}) \leq 2,
\end{align}
which can be violated up to the maximum value 4 with a partially predictable point $\wp_C^{PR}$ that has $\expectation{C_i}=1$ and a PR-box for $A-B$ with $\expectation{A_2B_2}=-1$ and $\expectation{A_iA_j}$ otherwise.

Finally, note that the inequalities \eqref{Equation:SliwaA2type3} and \eqref{Equation:SliwaA3type3} can be obtained from \eqref{Equation:SliwaA1type3} by swapping the labels $A \leftrightarrow B$ and $A-C$ respectively,  and hence the same argument can be seen to hold in that case by symmetry, as they correspond to permutations of the deterministic party.   
\end{proof}
\normalfont

With respect to specific values of violations of the Sliwa-3 inequalities, we also state the following immediate observation as a lemma.

\lemma{\label{Lemma:LemmaPartialPRboxesLemma}Let $S$ be a tripartite scenario with $I = \{A,B,C\}$ and $M_i = \{1,2\}$ for all $i\in I$ and $O_{x_i}=\{-1,1\}.$ The partially predictable behaviours
$\wp_A^{PR}, \wp_B^{PR}$ and $\wp_C^{PR}$ with 
\begin{align}
    \expectation{A_j}=1, \expectation{B_1C_1}=1, \expectation{B_2C_2}=-1, \label{Equation:LemmaPartialPRboxesRelations}
\end{align}
and similarly for $\wp_B^{PR}$ and $\wp_{C}^{PR}$, each reach the maximum value of 2, for the inequality valid for their polytope $\mathbf{PD}(S,M^{k}), k\in I$ and the maximum value of 4, for the inequalities of Proposition \ref{Proposition:Sliwa3relationsInthreepartitescenario},  which are not valid for that $k$. 
}
\begin{proof}
    In the proof of Proposition \ref{Proposition:Sliwa3relationsInthreepartitescenario} it was shown that $\wp_A^{PR}$ can reach the maximum value of 2 of the equation valid for it, if $\expectation{A_j}=1$ for both inputs, and $\expectation{B_1C_1}$. For $\wp_B^{PR}$ and $\wp_C^{PR}$, it was shown that they violate the equation valid for $\mathbf{PD}(S,M^{\{A\}})$ up to the maximum value 4, if $\expectation{A_2C_2}=-1$ and $\expectation{A_2B_2}=-1$, respectively. It is clear that the type constraints in \eqref{Equation:LemmaPartialPRboxesRelations} can therefore be satisfied simultaneously while being compatible with the claimed maximal violations. Since the behaviours that satisfy \eqref{Equation:LemmaPartialPRboxesRelations} are related to eachother by a symmetry $A-B$, $A-C$ similarly to the relations of the inequalities \eqref{Equation:SliwaA1type3}-\eqref{Equation:SliwaA3type3} the claim follows.  
\end{proof}
\normalfont
We are now ready to prove Theorem \ref{Theorem:unionof3partypdp'SiSNOTCONVEX}.

\

\begin{customthm}{\ref*{Theorem:unionof3partypdp'SiSNOTCONVEX}}
    Let $S= (I, M, O)$ be a scenario with $|I| =3.$ If the polytopes $\mathbf{PD}(S,M^{I'})$ are not equal for at least two $I' \in \mathcal{I}$, then the union  $\bigcup_{I' \in \mathcal{I}}\mathbf{PD}(S,M^{I'})$ is not convex.
\end{customthm}
\begin{proof}
    From Theorem \ref{Theorem:EquivalenceClassesofPDP's}, see also Example \ref{Example:partialdeterminismInThreepartiteTwo-inputScenario} and Fig.~\ref{fig:TripartiteEquivalenceClasses}, it follows that there are three inequivalent classes of partially deterministic polytopes $\mathbf{PD}(S,M^{k})$, $k\in I$. That at least two of them is contained in the union $\bigcup_{I' \in \mathcal{I}}\mathbf{PD}(S,M^{I'})$ is therefore necessary for the union to not be convex, since otherwise $\mathbf{PD}(S,M^{I'})$ equals the Bell-polytope $\mathbf{B}(S)$, which is convex and contained in every other partially deterministic polytope which are individually convex. 

    We will show that condition is also sufficient.  

    By proposition \ref{Proposition:Sliwa3relationsInthreepartitescenario} a necessary condition for convexity is that any mixture satisfies at least one of the  Sliwa 3 inequalities in Eqs.~\eqref{Equation:SliwaA1type3}-\eqref{Equation:SliwaA3type3}. Consider a uniform mixture   of three behaviours $\wp_A^{PR}, \wp_B^{PR}$ and $\wp^{PR}_{C}$ of the kind described in Lemma \ref{Lemma:LemmaPartialPRboxesLemma}. Since each of the expressions Eqs.~\eqref{Equation:SliwaA1type3}-\eqref{Equation:SliwaA3type3} is linear in probabilities, evaluating each expression with the mixture $\frac{1}{3}\sum_{k}\wp^{PR}_k$ gives the convex weight of the individual values given $\wp_k^{PR}$. Since $\frac{1}{3}(4+4+2)= \frac{8}{3}>2$, and each inequality attains the same value by Lemma \ref{Lemma:LemmaPartialPRboxesLemma}, this mixture is not contained in any single one of the polytopes $\mathbf{PD}(S,M^{\{k\}})$. Thus, the maximal union $\bigcup_{I' \in \mathcal{I}^{\max}} \mathbf{PD}(S,M^{I'})$ in particular is not convex. 

    Similarly, it can be seen that a union of any pair $k,j \in I, k\neq j$, of the partially deterministic polytopes is not convex, by considering the mixture $\frac{1}{2}\wp_k^{PR} + \frac{1}{2}\wp_j^{PR}$ instead, and observing that a value of $3>2$ is obtained for any pair of appropriate inequalities. 
    \end{proof}

\end{document}